\PassOptionsToPackage{names,dvipsnames}{xcolor}
\pdfoutput=1

\newif\ifdraft\draftfalse 
\newif\ifanon\anonfalse   
\newif\iffull\fullfalse   
\newif\iflongrefs\longrefsfalse 
\newif\ifbackref\backreffalse 
\newif\ifsooner\soonerfalse
\newif\iflater\laterfalse
\newif\ifieee\ieeetrue
\newif\ifcamera\cameratrue 
\newif\ifallcites\allcitesfalse
\newif\ifneedspace\needspacefalse
\makeatletter \@input{texdirectives.tex} \makeatother

\ifieee
\documentclass[10pt, conference, compsocconf, letterpaper, times]{IEEEtran}

\def\IEEEbibitemsep{3pt}

\DeclareMathAlphabet{\mathit}{\encodingdefault}{\familydefault}{m}{it}

\else 

\ifcamera
\documentclass[acmsmall]{acmart}
\else
\ifanon
\documentclass[acmsmall,review,anonymous,nonacm]{acmart}
\else
\documentclass[acmsmall,review=false,screen,nonacm]{acmart}
\fi
\fi

\acmJournal{PACMPL}
\acmYear{2018}
\acmMonth{1}
\acmVolume{2}
\acmNumber{POPL}
\acmPrice{}
\acmArticle{65}
\acmMonth{1}
\acmDOI{10.1145/3158153}
\ifcamera\else\copyrightyear{2017}\fi 

\ifanon
\setcopyright{none}             
\else
\setcopyright{rightsretained}   
\fi

\makeatletter
\def\@copyrightpermission{\ifcamera\\\\\\\fi This work is licensed under a \href{https://creativecommons.org/licenses/by/4.0/}{Creative Commons Attribution 4.0 International License}}
\makeatother

\ifcamera\else
\makeatletter
\def\@authorsaddresses{}
\fancypagestyle{firstpagestyle}{%
  \fancyhf{}%
  \renewcommand{\headrulewidth}{\z@}%
  \renewcommand{\footrulewidth}{\z@}%
    \fancyhead[L]{\ifanon\ACM@linecountL\fi}%
}
\fancypagestyle{standardpagestyle}{%
  \fancyhf{}%
  \renewcommand{\headrulewidth}{\z@}%
  \renewcommand{\footrulewidth}{\z@}%
    \fancyhead[LE]{\ifanon\ACM@linecountL\fi\@headfootfont\thepage}%
    \fancyhead[RO]{\@headfootfont\thepage}%
    \fancyhead[RE]{\@headfootfont\@shortauthors}%
    \fancyhead[LO]{\ifanon\ACM@linecountL\fi\@headfootfont\shorttitle}%
    \fancyfoot[RO,LE]{}%
}
\def\@mkbibcitation{}

\makeatother
\fi

\fi 

\usepackage{booktabs}   
\usepackage{subcaption} 

\usepackage[inference]{semantic}
\usepackage{amsmath}\allowdisplaybreaks
\usepackage{mathpartir}
\usepackage{amsthm}
\usepackage{amssymb}
\usepackage{stmaryrd} 
\usepackage{xspace}
\usepackage{latexsym}
\usepackage{ifthen}
\usepackage{mathtools}
\usepackage{color}
\usepackage{listings}
\usepackage{verbatim}
\usepackage{tikz}
\usetikzlibrary{positioning,shadows,arrows,calc,backgrounds,fit,shapes,shapes.multipart,decorations.pathreplacing,shapes.misc,patterns,decorations.markings}
\usepackage{tikzscale}
\usepackage{tikz-qtree}
\usepackage[T1]{fontenc}
\usepackage[scaled=.83]{beramono}
\usepackage{epigraph}
\usepackage{booktabs}
\usepackage{float}
\usepackage{etoolbox}
\usepackage{wrapfig}
\usepackage{scalerel}
\usepackage{nameref}
\usepackage[strict]{changepage}
\usepackage{coqed}
\usepackage{bm}
\usepackage{graphicx}
\usepackage[makeroom]{cancel}
\usepackage{centernot}
\usepackage{thm-restate}
\usepackage{enumitem}
\usepackage{fancyvrb}
\usepackage{varwidth}

\ifieee
\usepackage[numbers,sort]{natbib}
\fi

\definecolor{dkblue}{rgb}{0,0.1,0.5}
\definecolor{dkgreen}{rgb}{0,0.4,0}
\definecolor{dkred}{rgb}{0.6,0,0}
\definecolor{dkpurple}{rgb}{0.7,0,1.0}
\definecolor{purple}{rgb}{0.9,0,1.0}
\definecolor{olive}{rgb}{0.4, 0.4, 0.0}
\definecolor{teal}{rgb}{0.0,0.4,0.4}
\definecolor{azure}{rgb}{0.0, 0.5, 1.0}
\definecolor{gray}{rgb}{0.5, 0.5, 0.5}
\definecolor{dkgray}{rgb}{0.3, 0.3, 0.3}

\ifieee
\usepackage[
    pdftex,%
    pdfpagelabels,%
    colorlinks,%
    linkcolor=dkblue,%
    citecolor=dkblue,%
    filecolor=dkblue,%
    urlcolor=dkblue%
    \ifbackref,backref\fi%
]{hyperref}
\else
\usepackage{hyperref}
\hypersetup{
  breaklinks=true,
  citecolor=black,
  linkcolor=black,
  urlcolor=black,
}
\fi

\usepackage[capitalize]{cleveref}

\def\Snospace~{\S{}}

\def\Nnospace~{}

\usepackage{etoolbox}
\makeatletter
\patchcmd{\hyper@makecurrent}{%
    \ifx\Hy@param\Hy@chapterstring
        \let\Hy@param\Hy@chapapp
    \fi
}{%
    \iftoggle{inappendix}{
        \@checkappendixparam{chapter}%
        \@checkappendixparam{section}%
        \@checkappendixparam{subsection}%
        \@checkappendixparam{subsubsection}%
        \@checkappendixparam{paragraph}%
        \@checkappendixparam{subparagraph}%
    }{}%
}{}{\errmessage{failed to patch}}

\newcommand*{\@checkappendixparam}[1]{%
    \def\@checkappendixparamtmp{#1}%
    \ifx\Hy@param\@checkappendixparamtmp
        \let\Hy@param\Hy@appendixstring
    \fi
}
\makeatletter

\newtoggle{inappendix}
\togglefalse{inappendix}

\apptocmd{\appendix}{\toggletrue{inappendix}}{}{\errmessage{failed to patch}}


\AtBeginEnvironment{tikzpicture}{\catcode`$=3 } 

\newcommand{\MP}[1]{\marco{#1}}

\newcommand*{\EG}{e.g.,\xspace}
\newcommand*{\IE}{i.e.,\xspace}
\newcommand*{\ETAL}{et al.\xspace}
\newcommand*{\ETC}{etc.\xspace}

\newcommand{\mi}[1]{\ensuremath{\mathit{#1}}}
\newcommand{\ii}[1]{\mi{#1}}

\newcommand{\mtt}[1]{\ensuremath{\mathtt{#1}}}
\newcommand{\mf}[1]{\ensuremath{\mathbf{#1}}}

\newcommand{\mc}[1]{\ensuremath{\mathcal{#1}}}
\newcommand{\ms}[1]{\ensuremath{\mathsf{#1}}}
\newcommand{\mb}[1]{\ensuremath{\mathbb{#1}}}

\newcommand{\isdef}[0]{\ensuremath{\mathrel{\overset{\makebox[0pt]{\mbox{\normalfont\tiny\sffamily def}}}{=}}}}

\newcommand{\relmiddle}[1]{\mathrel{}\middle#1\mathrel{}}

\newcommand\bnfdef{\ensuremath{\mathrel{::=}}}

\newcommand{\OB}[1]{\ensuremath{\overline{#1}}}

\newcommand{\myset}[2]{\ensuremath{\left\{#1 ~\relmiddle|~ #2\right\}}}

\newcommand{\divrc}[0]{\ensuremath{\com{\Uparrow}}\xspace}

\newcommand{\termc}[0]{\ensuremath{\com{\Downarrow}}\xspace}

\newcommand*{\QEDA}{\hfill\ensuremath{\blacksquare}}%

\AtEndEnvironment{problem}{\null\hfill\QEDA}
\AtEndEnvironment{example}{}

\Crefname{lstlisting}{Listing}{Listings}
\Crefname{problem}{Problem}{Problems}

\Crefname{equation}{Rule}{Rules}

\newcommand{\compgen}[1]{\compskel{#1}{\S}{\T}}
\newcommand{\cmp}[1]{\compskel{#1}{}{}} 

\newcommand{\funname}[1]{\mtt{#1}}
\newcommand{\fun}[2]{\ensuremath{{\bl{\funname{#1}\left(#2\right)}}}\xspace}
\newcommand{\funn}[2]{\ensuremath{{\funname{#1}\left(#2\right)}}\xspace}
\newcommand{\dom}[1]{\fun{dom}{#1}}

\newcommand{\Nat}[0]{\ensuremath{\mb{N}}\xspace}

\newcommand{\backtrskel}[3]{\ensuremath{\bl{\left\langle\!\left\langle {#1} \right\rangle\!\right\rangle^{#2}_{#3}}}}
\newcommand{\backtr}[1]{\backtrskel{#1}{}{}}

\newcommand{\ctx}[0]{\ensuremath{\mb{C}}}
\newcommand{\ctxs}[0]{\src{\ctx}\xspace}
\newcommand{\ctxt}[0]{\trg{\ctx}\xspace}
\newcommand{\ctxc}[0]{\com{\ctx}\xspace}

\newcommand{\ctxhs}[1]{\ctxs\src{\hole{#1}}\xspace}
\newcommand{\ctxht}[1]{\ctxt\trg{\hole{#1}}\xspace}
\newcommand{\ctxhc}[1]{\ctxc\com{\hole{#1}}\xspace}
\newcommand{\hole}[1]{\ensuremath{\left[#1\right]}}

\newcommand{\evalctx}[0]{\ensuremath{\mb{E}}}
\newcommand{\evalctxs}[1]{\src{\evalctx\hole{#1}}\xspace}
\newcommand{\evalctxt}[1]{\trg{\evalctx\hole{#1}}\xspace}

\newcommand{\Bools}[0]{\src{{Bool}}\xspace}
\newcommand{\Nats}[0]{\src{{Nat}}\xspace}
\newcommand{\Boolt}[0]{\trg{{Bool}}\xspace}
\newcommand{\Natt}[0]{\trg{{Nat}}\xspace}

\newcommand{\trues}[0]{\src{{true}}\xspace}
\newcommand{\falses}[0]{\src{{false}}\xspace}

\newcommand{\truet}[0]{\trg{{true}}\xspace}
\newcommand{\falset}[0]{\trg{false}\xspace}

\newcommand{\truec}[0]{\com{{true}}\xspace}
\newcommand{\falsec}[0]{\com{{false}}\xspace}

\newcommand{\srce}[0]{\src{\emptyset}\xspace}
\newcommand{\trge}[0]{\trg{\emptyset}\xspace}

\newcommand{\SInits}[1]{\ensuremath{\src{\Omega_0}(\src{#1})}\xspace}
\newcommand{\SInitt}[1]{\ensuremath{\trg{\Omega_0}(\trg{#1})}\xspace}

\newcommand{\neutcol}[0]{black}
\newcommand{\stlccol}[0]{RoyalBlue}
\newcommand{\ulccol}[0]{RedOrange}
\newcommand{\commoncol}[0]{black}    

\newcommand{\col}[2]{\ensuremath{{\color{#1}{#2}}}}

\newcommand{\src}[1]{\ms{\col{\stlccol}{#1}}}
\newcommand{\trg}[1]{\mf{\col{\ulccol }{#1}}}
\newcommand{\bl}[1]{\col{\neutcol }{#1}}
\newcommand{\com}[1]{\mi{\col{\commoncol }{#1}}}

\newcommand{\fails}[0]{\src{fail}\xspace}
\newcommand{\failt}[0]{\trg{fail}\xspace}

\newcounter{typerule}
\crefname{typerule}{rule}{rules}

\newcommand{\typeruleInt}[5]{
	\def\thetyperule{#1}%
	\refstepcounter{typerule}%
	\label{tr:#4}%
  \ensuremath{\begin{array}{c}#5 \inference{#2}{#3}\end{array}} 
}
\newcommand{\typerule}[4]{
  \typeruleInt{#1}{#2}{#3}{#4}{\textsf{\scriptsize ({#1})} \\      }
}

\newcounter{criteria}
\crefname{criteria}{}{}
\newcommand{\criteria}[2]{%
	\def\thecriteria{\detokenize{#1}}%
  	\refstepcounter{criteria}%
  	\label{cr:#2}%
  	#1%
}

\pgfdeclarelayer{background}
\pgfdeclarelayer{veryback}
\pgfdeclarelayer{veryback2}
\pgfdeclarelayer{veryback3}
\pgfdeclarelayer{back2}
\pgfdeclarelayer{foreground}
\pgfsetlayers{veryback3,veryback2,veryback,background,back2,main,foreground}

\newcommand{\myfig}[3]{\begin{figure} [!h]
#1
\caption{\label{fig:#2}#3}
\end{figure}}

\newcommand{\mytoprule}[1]{\vspace{1mm}\noindent\hrulefill\ \raisebox{-0.5ex}{\fbox{\ensuremath{#1}}} \hrulefill\hrulefill\hrulefill\vspace{0.5mm}}

\def\botrule{\vspace{0mm}\hrule\vspace{2mm}}

\newcounter{line}

\newcommand{\asm}[1]{\mtt{#1}}

\newcommand{\xto}[1]{\ensuremath{~\mathrel{\xrightarrow{~#1~}}~}}
\newcommand{\Xto}[1]{\ensuremath{~\mathrel{\xRightarrow{~#1~}}~}}

\newcommand{\xtos}[1]{\src{\xto{#1}}}

\newcommand{\nXtos}[1]{\src{\centernot{\Xto{#1}}}}
\newcommand{\xtot}[1]{\trg{\xto{#1}}}

\newcommand{\nXtot}[1]{\trg{\centernot{\Xto{#1}}}}

\newcommand{\xtosb}[1]{\src{\xto{\bl{#1}}}}
\newcommand{\Xtosb}[1]{\src{\Xto{\bl{#1}}}}

\newcommand{\xtotb}[1]{\trg{\xto{\bl{#1}}}}
\newcommand{\Xtotb}[1]{\trg{\Xto{\bl{#1}}}}

\newcommand{\Xtoltb}[1]{\trg{\Xtol{\bl{#1}}}}
\newcommand{\Xtolsb}[1]{\src{\Xtol{\bl{#1}}}}

\newcommand{\Xtol}[1]{\ensuremath{\xRightarrow{~#1~}\Low}}

\newcommand{\Low}[0]{\ensuremath{\!\!\!\!\Rightarrow} }

\definecolor{mygreen}{rgb}{0,0.6,0}
\definecolor{mygray}{rgb}{0.5,0.5,0.5}
\definecolor{mymauve}{rgb}{0.58,0,0.82}

\lstdefinelanguage{Java} 
{morekeywords={abstract, all, and, as, assert, but, check, disj, else, exactly, extends, fact, for, fun, iden, if, iff, implies, in, Int, int, let, lone, module, no, none, not, one, open, or, part, pred, run, seq, set, sig, some, sum, then, univ, package, class, public, private, null, return, new, interface, extern, object, implements, System, static, super, try , catch, throw, throws, Unit, var, val, of, principal, trust},
sensitive=true,
keywordstyle=\bfseries\color{green!40!black},
commentstyle=\itshape\color{purple!40!black},
morecomment=[l][\small\itshape\color{purple!40!black}]{//},
identifierstyle=\color{blue},
stringstyle=\color{orange},
basicstyle=\small,
basicstyle={\small\ttfamily},
numbers=left,
numberstyle=\tiny\color{mygray},
tabsize=2,
numbersep=3pt,
breaklines=true,
lineskip=-2pt,
stepnumber=1,
captionpos=b,
breaklines=true,
breakatwhitespace=false,
showspaces=false,
showtabs=false,
float=!h,
columns=fullflexible,escapeinside={(*@}{@*)},
moredelim=**[is][\color{red!60}]{@}{@},
literate={->}{{$\to$}}1 {^}{{$\mspace{-3mu}\widehat{\quad}\mspace{-3mu}$}}1
{<}{$<$ }2 {>}{$>$ }2 {>=}{$\geq$ }2 {=<}{$\leq$ }2
{<:}{{$<\mspace{-3mu}:$}}2 {:>}{{$:\mspace{-3mu}>$}}2
{=>}{{$\Rightarrow$ }}2 {+}{$+$ }2 {++}{{$+\mspace{-8mu}+$ }}2
{<=>}{{$\Leftrightarrow$ }}2 {+}{$+$ }2 {++}{{$+\mspace{-8mu}+$ }}2
{\~}{{$\mspace{-3mu}\widetilde{\quad}\mspace{-3mu}$}}1
{!=}{$\neq$ }2 {*}{${}^{\ast}$}1 
{\#}{$\#$}1
}

\DeclareMathOperator\ceq{\ensuremath{\mathrel{\simeq_{\mi{ctx}}}}}

\def\teqaux#1{\vcenter{\hbox{\ooalign{\hfil
       \raise6pt \hbox{\scriptsize{T}}\hfil\cr\hfil
       $=$}}}}

\def\beqaux#1{\vcenter{\hbox{\ooalign{\hfil
       \raise6pt \hbox{\scriptsize{\ensuremath{\beta}}}\hfil\cr\hfil
       $=$}}}}

\def\ceqwaux#1{\vcenter{\hbox{\ooalign{\hfil
       \raise6pt \hbox{\scriptsize{w-b}}\hfil\cr\hfil
       $\ceq$}}}}

\def\praux#1{\vcenter{\hbox{\ooalign{\hfil
       \raise6pt \hbox{$\sqsubset$}\hfil\cr\hfil
       $\sim$}}}}

\newcommand{\labelfont}[1]{\ensuremath{\asm{#1}}}
\newcommand{\clpaper}[2]{\ensuremath{\labelfont{call}~ #1~ #2}}
\newcommand{\cl}[2]{\ensuremath{\labelfont{call}~ #1~ #2{?}}}

\newcommand{\rtpaper}[1]{\ensuremath{\labelfont{ret}~ #1}}
\newcommand{\rt}[1]{\ensuremath{\labelfont{ret}~ #1{!}}}

\newcommand{\rdl}[1]{\ensuremath{\labelfont{read}~#1}}
\newcommand{\wrl}[1]{\ensuremath{\labelfont{write}~#1}}
\newcommand{\inpl}[1]{\ensuremath{\labelfont{input}~#1}}
\newcommand{\outl}[1]{\ensuremath{\labelfont{output}~#1}}

\newcommand{\traces}[3]{\ensuremath{\ms{TR}^{#2}_{#3}(#1)}}

\newcommand{\trt}[1]{\trg{\traces{#1}{}{}}}

\newcommand{\behav}[1]{\funn{Behav}{#1}}
\newcommand{\behavs}[1]{\src{\behav{#1}}}
\newcommand{\behavt}[1]{\trg{\behav{#1}}}
\newcommand{\behavc}[1]{\com{\behav{#1}}}
\newcommand{\behavAll}{\funname{Behavs}}

\newcommand{\srcAll}{\funname{\src{Progs}}}

\newcommand{\divc}[0]{\com{\Uparrow}\xspace}
\newcommand{\terc}[0]{\com{\Downarrow}\xspace}

\newcommand{\failactc}[0]{\com{\failact}\xspace}
\newcommand{\failact}[0]{\bot}

\newcommand{\set}[1]{\ensuremath{\widehat{#1}} }

\newcommand{\card}[1]{\ensuremath{|\!|{#1}|\!|}}

\theoremstyle{definition}
\newtheorem{definition}{Definition}[section]

\newtheorem{theorem}[definition]{Theorem}

\newtheorem{example}[definition]{Example}
\newtheorem{corollary}[definition]{Corollary}
\newtheorem{lemma}[definition]{Lemma}
\Crefname{corollary}{Corollary}{Corollaries}
\Crefname{informal}{Definition}{Definition}
\Crefname{assumption}{Assumption}{Assumptions}
\crefname{assumption}{Assumption}{Assumptions}
\Crefname{property}{Property}{Properties}
\crefname{property}{Property}{Properties}

\newcommand{\Bool}[0]{\mtt{Bool}\xspace}

\newcommand{\op}[0]{\ensuremath{\oplus}}

\newcommand{\ret}[1]{{return}~#1}
\newcommand{\retpaper}[1]{{ret}~#1}
\newcommand{\letin}[3]{{let}~#1=#2~{in}~#3}
\newcommand{\letins}[3]{\src{let}~#1\src{=}#2~\src{in}~#3}
\newcommand{\letint}[3]{\trg{let}~#1\trg{=}#2~\trg{in}~#3}
\newcommand{\call}[1]{{call}~#1}

\newcommand{\iftet}[3]{\trg{if}~#1~\trg{then}~#2~\trg{else}~#3}
\newcommand{\ifte}[3]{{if}~#1~{then}~#2~{else}~#3}
\newcommand{\iftes}[3]{\src{if}~#1~\src{then}~#2~\src{else}~#3}

\def\formatCompilers#1{\ms{#1}\xspace}

\newcommand{\scc}{\formatCompilers{SCC}}
\newcommand{\ccc}{\formatCompilers{CCC}}

\def\porc{P}
\def\porcc{C}
\def\criterion#1{\formatCompilers{#1\porc}}
\def\criterionprimed#1{\formatCompilers{#1\porc '}}

\def\pf#1{\ensuremath{\let\porc\porcc #1}}

\newcommand{\rhp}[0]{\criterion{RH}}
\newcommand{\rschp}[0]{\criterion{RSCH}} 
\newcommand{\rkschp}[0]{\criterion{R\mi{K}SCH}}
\newcommand{\rtschp}[0]{\criterion{R2SCH}}
\newcommand{\rtp}[0]{\criterion{RT}}
\newcommand{\tp}{\criterion{T}}
\newcommand{\rfrschp}[0]{\criterion{RFrSCH}}

\newcommand{\rhsp}[0]{\criterion{RHS}}
\newcommand{\rhlp}[0]{\criterion{RHL}}
\newcommand{\rkhsp}[0]{\criterion{R\mathit{K}HS}}
\newcommand{\rkkhsp}[0]{\criterion{R\mathit{(K{+}1)}HS}}
\newcommand{\rthsp}[0]{\criterion{R2HS}}
\newcommand{\rsp}[0]{\criterion{RS}}
\newcommand{\rdp}[0]{\criterion{RD}}

\newcommand{\rrhp}[0]{\criterion{RrH}}
\newcommand{\rrhpprimed}[0]{\criterionprimed{RrH}}
\newcommand{\rkrhp}[0]{\criterion{R\mathit{K}rH}}
\newcommand{\rtrhp}[0]{\criterion{R2rH}}

\newcommand{\rrsp}[0]{\criterion{RrS}}
\newcommand{\rrspprimed}[0]{\criterionprimed{RrS}}
\newcommand{\rrxp}[0]{\criterion{RrX}}
\newcommand{\rrxpprimed}[0]{\criterionprimed{RrX}}
\newcommand{\rrtp}[0]{\criterion{RrT}}
\newcommand{\rrtpprimed}[0]{\criterionprimed{RrT}}

\newcommand{\rkrsp}[0]{\criterion{R\mathit{K}rS}}
\newcommand{\rkrxp}[0]{\criterion{R\mathit{K}rX}}
\newcommand{\rkrtp}[0]{\criterion{R\mathit{K}rT}}

\newcommand{\rtrsp}[0]{\criterion{R2rS}}
\newcommand{\rtrxp}[0]{\criterion{R2rX}}
\newcommand{\rtrtp}[0]{\criterion{R2rT}}

\newcommand{\rfrsp}[0]{\criterion{RFrS}}
\newcommand{\rfrxp}[0]{\criterion{RFrX}}

\newcommand{\rtep}[0]{\criterion{RTE}}
\newcommand{\oep}[0]{\criterion{OE}}
\newcommand{\rtinip}[0]{\criterion{RTINI}}

\newcommand{\rrhpA}[0]{\formatCompilers{RrH\porc'}}

\newcommand{\rhpref}[0]{\Cref{cr:rhp}}
  
\newcommand{\rschpref}[0]{\Cref{cr:rschp}} 

\newcommand{\rtpref}[0]{\Cref{cr:rtp}}

\newcommand{\rhspref}[0]{\Cref{cr:rhsp}}

\newcommand{\rkhspref}[0]{\Cref{cr:rkhsp}}

\newcommand{\rthspref}[0]{\Cref{cr:rthsp}}
\newcommand{\rspref}[0]{\Cref{cr:rsp}}

\newcommand{\rrhpref}[0]{\Cref{cr:rrhp}}
  
\newcommand{\rkrhpref}[0]{\Cref{cr:rkrhp}}
\newcommand{\rtrhpref}[0]{\Cref{cr:rtrhp}}

\newcommand{\rrspref}[0]{\Cref{cr:rrsp}}
  
\newcommand{\rrtpref}[0]{\Cref{cr:rrtp}}

\newcommand{\rkrspref}[0]{\Cref{cr:rkrsp}}
\newcommand{\rkrtpref}[0]{\Cref{cr:rkrtp}}

\newcommand{\rtrspref}[0]{\Cref{cr:rtrsp}}
  
\newcommand{\rtrtpref}[0]{\Cref{cr:rtrtp}}

\newcommand{\rtepref}[0]{\Cref{cr:rtep}}

\newcommand{\rtinipref}[0]{\Cref{cr:rtinip}}

\newcommand{\rtrxpref}[0]{\Cref{cr:rtrxp}}
\newcommand{\rfrxpref}[0]{\ifcamera RFrXP\else\Cref{cr:rfrxp}\fi}

\newcommand{\red}[0]{\ensuremath{\ \hookrightarrow\ }}

\newcommand{\redapp}[1]{\ensuremath{\!\!\!^{#1}\ }}
\newcommand{\redstar}[0]{\redapp{*}}

\newcommand{\redstars}[0]{\src{\redstar}}

\newcommand{\redstart}[0]{\trg{\redstar}}

\DeclareMathOperator\reds{\src{\red}}

\DeclareMathOperator\redt{\trg{\red}}

\AtEndEnvironment{example}{\null\hfill$\boxdot$}

\makeatletter
\xdef\@thefnmark{\@empty}

\newcommand{\Thmref}[1]{\Cref{#1}~(\nameref{#1})}
\makeatother

\newcommand{\subst}[2]{\ensuremath{\bl{[}#1\bl{/}#2\bl{]}}} 
\newcommand{\subs}[2]{\subst{\src{#1}}{\src{#2}}}
\newcommand{\subt}[2]{\subst{\trg{#1}}{\trg{#2}}}

\renewcommand{\emptyset}[0]{\varnothing}

\newcounter{hps}
\crefname{hps}{}{}

\newcommand{\langlett}[0]{L} 
\newcommand{\Lt}[0]{\src{\langlett^{\tau}}\xspace}
\newcommand{\Ld}[0]{\trg{\langlett^{u}}\xspace}

\newcommand{\comptd}[1]{\compskel{#1}{\Lt}{\Ld}}

\newcommand{\backtrdt}[1]{\backtrskel{#1}{\Ld}{\Lt}}

\newcommand{\readexp}[0]{read}
\newcommand{\writeexp}[1]{write~#1}

\newcommand{\checkty}[1]{~has~#1}

\newcommand{\inject}[1]{\src{inject_{#1}}}
\newcommand{\extract}[1]{\src{extract_{#1}}}

\newcommand{\psd}[1]{\src{\hat{#1}}}
\newcommand{\emuldv}[0]{\src{EmulTy}}

\newcommand{\toemul}[1]{\fun{toEmul}{#1}}

\newcommand{\logrelgen}[1]{\ensuremath{\operatorname{\approx}^{#1}}}

\DeclareMathOperator\bothlogrel{\logrelgen{}}
\newcommand{\underapproxlogrelgen}[1]{\ensuremath{\operatorname{\lesssim}^{#1}}}
\DeclareMathOperator\underlogrel{\underapproxlogrelgen{}}
\newcommand{\overapproxlogrelgen}[1]{\ensuremath{\operatorname{\gtrsim}^{#1}}}
\DeclareMathOperator\overlogrel{\overapproxlogrelgen{}}

\newcommand{\arbsim}{\ensuremath{\triangledown}}

\newcommand{\genlogrel}[0]{\arbsim}
\DeclareMathOperator\anylogrel{\genlogrel{}}

\newcommand{\anylogreln}[1]{\ensuremath{\mathrel{\triangledown_{#1}}}}

\newcommand{\genrel}[4]{\ensuremath{\bl{\mc{#1}\left\llbracket\src{#2}\right\rrbracket^{#3}_{#4}}}}
\newcommand{\valrel}[1]{\genrel{V}{#1}{}{\anylogrel}}

\newcommand{\contrel}[1]{\genrel{K}{#1}{}{\anylogrel}}

\newcommand{\termrel}[1]{\genrel{E}{#1}{}{\anylogrel}}

\newcommand{\envrel}[1]{\genrel{G}{#1}{}{\anylogrel}}

\newcommand{\langsp}[1]{\ensuremath{\mi{#1}}}
\newcommand{\langspfun}[2]{\ensuremath{\langsp{#1}(#2)}}

\newcommand{\monotfun}[1]{\fun{\nearrow}{#1}}

\newcommand{\stepsfungen}[2]{\ensuremath{\langspfun{lev^{#1}}{#2}}}

\newcommand{\stepsfun}[1]{\stepsfungen{}{#1}}

\newcommand{\progsfun}[1]{\langspfun{progs}{#1}}
\newcommand{\srcprogfun}[1]{\ensuremath{\langspfun{srcprog}{#1}}}
\newcommand{\trgprogfun}[1]{\ensuremath{\langspfun{trgprog}{#1}}}

\newcommand{\latergen}[1]{\ensuremath{\triangleright^{#1}}}
\DeclareMathOperator\later{\latergen{}}
\newcommand{\laterfun}[1]{\langspfun{\triangleright}{#1}}

\newcommand{\obswfungen}[3]{\langspfun{O^{#1}}{#2}_{#3}}

\newcommand{\obsfun}[2]{\obswfungen{}{#1}{#2}}

\newcommand{\futwgen}[1]{\ensuremath{\sqsupseteq^{#1}}}
\newcommand{\strfutwgen}[1]{\ensuremath{\sqsupset_{\later}^{#1}}}
\DeclareMathOperator\futw{\futwgen{}}

\DeclareMathOperator\strfutw{\strfutwgen{}}

\DeclareMathOperator{\typeop}{type}
\DeclareMathOperator{\inputtypeop}{input\_type}
\newcommand{\bttrace}{\!\uparrow}
\newcommand{\plug}[2]{\ensuremath{#1\!\left[#2\right]}}
\newcommand{\comc}{\com}

\newcommand{\xtoctx}[1]{\ensuremath{~\mathrel{\xrightarrow{~#1~}}_{\text{ctx}}~}}
\newcommand{\Xtoctx}[1]{\ensuremath{~\mathrel{\xRightarrow{~#1~}}_{\text{ctx}}~}}
\newcommand{\Xtolctx}[1]{\ensuremath{~\mathrel{\Xtol{~#1~}}_{\text{ctx}}~}}
\newcommand{\xtosbctx}[1]{\src{\xtoctx{\bl{#1}}}}
\newcommand{\Xtosbctx}[1]{\src{\Xtoctx{\bl{\mathnormal{#1}}}}}
\newcommand{\Xtolsbctx}[1]{\src{\Xtolctx{\bl{\mathnormal{#1}}}}}
\newcommand{\xtotbctx}[1]{\trg{\xtoctx{\bl{#1}}}}

\newcommand{\Xtoltbctx}[1]{\trg{\Xtolctx{\bl{\mathnormal{#1}}}}}

\newcommand{\sem}{\boldsymbol{\rightsquigarrow}}
\newcommand{\sdiv}{\boldsymbol{\circlearrowleft}}

\renewcommand{\comptd}[1]{\ensuremath{\bl{\left.\src{#1}\right\downarrow}}}
\renewcommand{\backtrdt}[1]{\ensuremath{\bl{\left.\trg{#1}\right\uparrow}}}

\newcommand{\comm}[3]{\ifdraft{{\color{#1}[#2: #3]}}\fi}
\newcommand{\marco}[1]{\comm{dkblue}{Marco}{#1}}
\newcommand{\dg}[1]{\comm{dkpurple}{Deepak}{#1}}
\newcommand{\ch}[1]{\comm{teal}{CH}{#1}}
\newcommand{\ca}[1]{\comm{olive}{CA}{#1}}

\newcommand{\jt}[1]{\comm{violet}{JT}{#1}}
\newcommand{\rb}[1]{\comm{orange}{RB}{#1}}

\newcommand{\fstar}{F$^\star$}

\newcommand{\lowstar}{Low$^\star$\xspace}

\newcommand{\citeFull}[2]{\ifallcites\cite{#1,#2}\else\cite{#1}\fi}

\renewcommand{\cmp}[1]{\comptd{#1}} 

\ifieee
\else
\fi
\ifieee\else\fi

\ifieee
\makeatletter
\AtBeginDocument{
  \hypersetup{
    pdftitle = {Journey Beyond Full Abstraction},
    pdfauthor = {\ifanon\else Carmine Abate, Roberto Blanco, Deepak Garg, Catalin Hritcu, Marco Patrignani, Jeremy Thibault\fi}
  }
}
\makeatother

\title{\bf{Journey Beyond Full Abstraction}\\[0.5ex]
\LARGE Exploring Robust Property Preservation for Secure Compilation
\ifneedspace\ifieee\ifanon\vspace{-2em}\else\vspace{-1em}\fi\fi\else\ifanon\vspace{2em}\fi\fi}
\else
\title{Journey Beyond Full Abstraction}
\subtitle{Exploring Robust Property Preservation for Secure Compilation}
\fi

\ifanon
\author{}
\else

\ifieee\large
\author{
  Carmine Abate\textsuperscript{1} \quad
  Roberto Blanco\textsuperscript{1} \quad
  Deepak Garg\textsuperscript{2} \quad
  C\u{a}t\u{a}lin Hri\c{t}cu\textsuperscript{1} \quad
  Marco Patrignani\textsuperscript{3,4} \quad
  J\'er\'emy Thibault\textsuperscript{1}\\[1em]
  \textsuperscript{1}Inria Paris\quad
  \textsuperscript{2}MPI-SWS\quad
  \textsuperscript{3}Stanford University\quad
  \textsuperscript{4}CISPA Helmholz Center for Information Security\\[0em]
}

\else 

\author{Carmine Abate}
\affiliation{
  \ifcamera\institution{Inria}\city{Paris}\country{France}
  \else\institution{Inria Paris}\fi}
\email{carmine.abate@inria.fr}

\author{Roberto Blanco}
\affiliation{
  \ifcamera\institution{Inria}\city{Paris}\country{France}
  \else\institution{Inria Paris}\fi}
\email{roberto.blanco@inria.fr}

\author{Deepak Garg}
\affiliation{
  \institution{MPI-SWS}
  \ifcamera\city{Saarbr\"ucken}\country{Germany}\fi}
\email{dg@mpi-sws.org}

\author{C\u{a}t\u{a}lin Hri\c{t}cu}
\affiliation{
  \ifcamera\institution{Inria}\city{Paris}\country{France}
  \else\institution{Inria Paris}\fi}
\email{catalin.hritcu@gmail.com}

\author{Marco Patrignani}
\orcid{0000-0003-3411-9678}
\affiliation{
  \institution{Stanford University}
  \ifcamera\city{Palo Alto}\country{USA}\fi}
\affiliation{
  \institution{CISPA}
  \ifcamera\city{Saarbr\"ucken}\country{Germany}\fi}
\email{marcopat@mpi-sws.org}

\author{J\'er\'emy Thibault}
\affiliation{
  \ifcamera\institution{Inria}\city{Paris}\country{France}
  \else\institution{Inria Paris}\fi}
\email{jeremy.thibault@ens-rennes.fr}

\makeatletter
\renewcommand{\@shortauthors}{C. Abate, R. Blanco, D. Garg, C. Hri\c{t}cu, M. Patrignani, and J. Thibault}
\makeatother

\fi 

\fi 

\begin{document}


\ifieee\maketitle\fi

\begin{abstract}
Good programming languages provide helpful abstractions for writing
secure code, but the security properties of the source language
are generally not preserved when compiling a program and linking it
with adversarial code in a low-level target language (e.g., a library or
a legacy application).  Linked target code that is compromised or malicious
may, for instance, read and write the compiled program's data and code,
jump to arbitrary memory locations, or smash the stack, blatantly
violating any source-level abstraction.  By contrast, a fully abstract
compilation chain protects source-level abstractions all the way down,
ensuring that linked adversarial target code cannot observe
more about the compiled program than what some linked source code could
about the source program.  However, while research in this area has so
far focused on preserving observational equivalence, as needed for
achieving full abstraction, there is a much larger space of security
properties one can choose to preserve against linked adversarial code.
And the precise class of security properties one chooses crucially
impacts not only the supported security goals and the strength of the
attacker model, but also the kind of protections a secure compilation
chain has to introduce.

We are the first to thoroughly explore a large space of formal secure
compilation criteria based on robust property preservation, i.e., the
preservation of properties satisfied against arbitrary adversarial
contexts.  We study robustly preserving various classes of trace
properties such as safety, of hyperproperties such as noninterference,
and of relational hyperproperties such as trace equivalence. This
leads to many new secure compilation criteria, some of which are
easier to practically achieve and prove than full abstraction, and
some of which provide strictly stronger security guarantees.  For each
of the studied criteria we propose an equivalent ``property-free''
characterization that clarifies which proof techniques apply.  For
relational properties and hyperproperties, which relate the behaviors
of multiple programs, our formal definitions of the property classes
themselves are novel. We order our criteria by their relative strength
and show several collapses and separation results.  Finally, we adapt
existing proof techniques to show that even the strongest of our
secure compilation criteria, the robust preservation of all relational
hyperproperties, is achievable for a simple translation from a
statically typed to a dynamically typed language.

\end{abstract}

\ifieee\else\maketitle\fi

\ifcamera
\fi






\section{Introduction}\label{sec:intro}

Good programming languages provide helpful abstractions for writing
secure code.
Even in unsafe low-level languages like C, safe programs have structured
control flow and obey the procedure call and return discipline.
Languages such as Java, C\#, ML, Haskell, or Rust provide type and
memory safety for all programs and additional abstractions such as
modules and interfaces.
Languages for efficient cryptography such as qhasm~\cite{qhasm},
Jasmin~\cite{AlmeidaBBBGLOPS17}, and \lowstar{}~\cite{lowstar} enforce
a ``constant-time'' coding discipline to rule out certain side-channel
attacks.
Finally, verification languages such as Coq and
\fstar{}~\cite{lowstar,fstar-popl2016}
provide abstractions such as dependent
types, logical pre- and postconditions, and tracking side-effects,
\EG distinguishing pure from stateful computations.
Such abstractions make reasoning about security more tractable and
have, for instance, enabled developing high-assurance
libraries in areas such as cryptography~\cite{AlmeidaBBBGLOPS17,
  haclstarccs, ErbsenPGSC19, Delignat-Lavaud17}.

However, such abstractions are not enforced all the way down by
mainstream compilation chains.
The security properties a program satisfies in the source language are
generally not preserved when compiling the program and linking it with
adversarial target code.
High-assurance cryptographic libraries, for instance, get linked into
real applications such as web
browsers~\cite{HaclInFirefox,ErbsenPGSC19} and web servers, which
include millions of lines of legacy C/C++ code.
Even if the abstractions of the source language ensure that the API of
a TLS library cannot leak the server's private
key~\cite{Delignat-Lavaud17}, such guarantees are completely lost when
compiling the library and linking it into a C/C++ application that can
get compromised via a buffer overflow, simply allowing the adversary
to read the private key from memory~\cite{DurumericKAHBLWABPP14}.
A compromised or malicious application that links in a
high-assurance library can easily read and write its data and
code, jump to arbitrary memory locations, or smash the stack,
blatantly violating any source-level abstraction and breaking any
security guarantee obtained by source-level reasoning.
\ifsooner\ch{Eric: linking followed by various prepositions with, into, in;
  my impression is that it's just fine, but double check}\fi

An idea that has been gaining increasing traction recently is that it
should be possible to build secure compilation chains that protect
source-level abstractions even against linked adversarial target code,
which is generally represented by target language \emph{contexts}.
Research in this area has so far focused on achieving
{\em full abstraction}~\citeFull{AgtenSJP12, MarcosSurvey, AbadiFG02,
  AhmedB11, Ahmed15, NewBA16, AbadiP12,
  JagadeesanPRR11, FournetSCDSL13, PatrignaniASJCP15,
   DevriesePP16, JuglaretHAEP16}{PatrignaniDP16, AhmedB08, LarmuseauPC15},
whose security-relevant direction ensures that even an adversarial
target context cannot observe more about the compiled program than
some source context could about the source program.
In order to achieve full abstraction, the various parts of the secure
compilation chain---including, \EG 
the compiler, linker,
loader, runtime, system, and hardware---have to work together to
provide enough protection to the compiled program, so that whenever
two programs are {\em observationally equivalent} in the source language
(\IE no source context can distinguish them), the two programs obtained
by compiling them are observationally equivalent in the target
language (\IE no target context can distinguish them).

Observational equivalences are, however, not the only class of
security properties one may want to {\em robustly preserve},
\IE preserve against arbitrary adversarial contexts.
One could instead be interested in robustly preserving, for instance,
classes of trace properties such as safety~\cite{LamportS84} or
liveness~\cite{AlpernS85}, or of hyperproperties~\cite{ClarksonS10}
such as hypersafety, including variants of
noninterference~\citeFull{GoguenM82, McLean92, ZdancewicM03,
  AskarovHSS08, SabelfeldS01}{SabelfeldM03}, which cover data confidentiality and integrity.
However, full abstraction is generally not strong enough on its own to
imply the robust preservation of any of these properties (as we show
in \autoref{sec:fa}, and as was also argued by
others~\cite{PatrignaniG17}).
%
At the same time, the kind of protections one has to put in place for
achieving full abstraction seem like
overkill if all one wants is to robustly preserve safety or hypersafety.
Indeed, it is significantly harder to hide the differences between two programs
that are observationally equivalent but otherwise arbitrary, than
to protect the internal invariants and the secret data of a single program.
Thus, a secure compilation chain for robust safety or
hypersafety can likely be more efficient than one for observational
equivalence.
Moreover, hiding the differences between two observationally equivalent
programs is hopeless in the presence of any
side-channels, while robustly preserving safety is not a problem and
even robustly preserving noninterference seems possible in
specific scenarios~\cite{BartheGL18}.
Finally, even when efficiency is not a concern (\EG when security is
enforced by \emph{static} restrictions on target
contexts~\citeFull{Abadi99, AhmedB11, Ahmed15, NewBA16}{AhmedB08}),
proving full abstraction is notoriously challenging even for
simple languages, and conjectures have survived for
decades before being settled~\cite{DevriesePP18}.

Convinced that there is no ``one-size-fits-all'' criterion for secure
interoperability with linked target code,
we explore, for the first time, a large space of secure compilation
criteria based on robust property preservation.
Some of the criteria we introduce are strictly stronger than full
abstraction and, moreover, immediately imply the robust preservation of
well-studied property classes such as safety and hypersafety.
Other criteria we introduce seem easier to practically achieve and
prove than full abstraction.
In general, the richer the class of security properties one tries to
robustly preserve, the harder efficient enforcement becomes, so the
best one can hope for is to strike a pragmatic balance between
security and efficiency that matches each application domain.

For informing such difficult design decisions,
we explore robustly preserving classes
of trace properties (\autoref{sec:prop}),
of hyperproperties (\autoref{sec:hyperprop}),
and of relational hyperproperties (\autoref{sec:rel}).
All these property notions are phrased in terms of execution traces,
which for us are (finite or infinite) sequences of events such as inputs
from and outputs to an external environment~\cite{Leroy09, KumarMNO14}.
Trace properties such as safety~\cite{LamportS84} restrict what
happens along individual program traces, while
hyperproperties~\cite{ClarksonS10} such as noninterference generalize
this to predicates over multiple traces of a program.
In this work we generalize this further to a new class we call
{\em relational hyperproperties},
which relate the traces of {\em different} programs.
An example of relational hyperproperty is trace equivalence, which
requires that two programs produce the same set of traces.
We work out many interesting subclasses that are also novel, such as
{\em relational trace properties}, which relate \emph{individual}
traces of multiple programs.
For instance, ``On every input, program A's output is less than program B's''
is a relational trace property.

\begin{figure*}[!ht]
\input{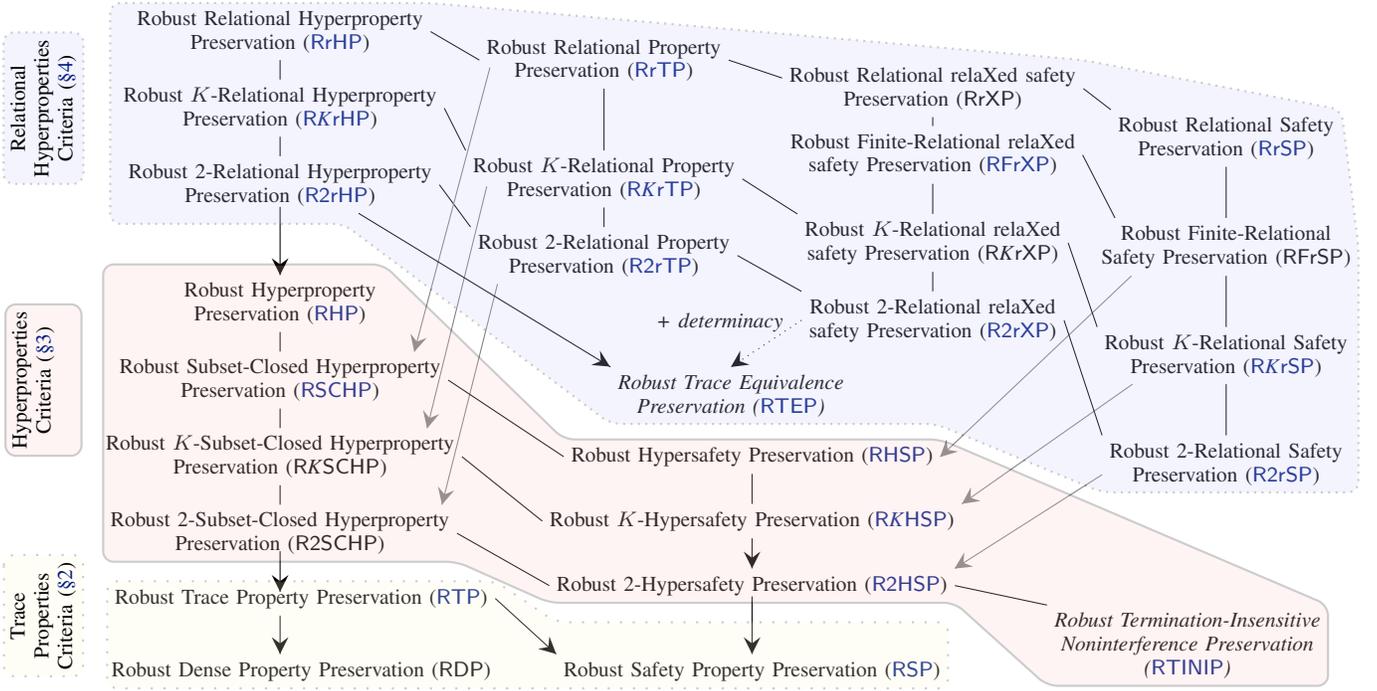}
\caption{\ifieee\else\footnotesize\fi
         Partial order with the secure compilation criteria studied in this paper.
         Criteria higher in the diagram imply
         the lower ones to which they are connected by edges.
         Criteria based on trace properties are grouped in a yellow area,
         those based on hyperproperties are in a red
         area, and those based on relational hyperproperties
         are in a blue area.
         Criteria with an \emph{italics name}
         preserve a \emph{single} property that belongs to the class they are connected to;
         the dotted edge requires an additional determinacy assumption.
         Finally, each edge with a thick arrow denotes a {\em strict} implication that we have proved as a separation result.
}
\label{fig:order}
\ifneedspace\vspace{-1em}\fi
\end{figure*}
We order the secure compilation criteria we introduce by their relative
strength as illustrated by the partial order in \autoref{fig:order}.
In this Hasse diagram edges represent logical implication from higher
criteria to lower ones, so the higher a criterion is, the harder it is
to achieve and prove.
Intuitively, the criteria based on the robust preservation of trace
properties (in the yellow area) only require sandboxing the context
(\IE linked adversarial code) and protecting the internal invariants of
the program from it, i.e.,\ {\em only data integrity}.
The criteria based on hyperproperties (in the red area) require
additionally hiding the data of the program from the context, i.e.,\
{\em data confidentiality}.
Finally, the criteria based on relational hyperproperties (in the blue
area) require additionally hiding the code of the program from the context, i.e.,\ {\em code confidentiality}.

While most \iffull of the \fi implications in the diagram
\iffull are unsurprising and \fi
follow directly from the inclusion between the property
classes~\cite{ClarksonS10}, {\em strict} inclusion between property
classes does not imply strict implication between criteria.
Robustly preserving two distinct
property classes can in fact lead to equivalent criteria, as 
happens in general for hyperliveness and hyperproperties
(\autoref{sec:hyperliveness}) and, in the presence of
source-level reflection or internal nondeterminism,
for many criteria involving hyperproperties and relational hyperproperties
(\autoref{sec:context-composition}).
To show the absence of more collapses,
we also prove various separation results, for instance
that {\em Robust Safety Property Preservation} (\rsp)
%
%
is {\em strictly} weaker than {\em Robust Trace Property Preservation} (\rtp).
For this, we design
(counterexample) compilation chains that satisfy the weaker criterion
but not the stronger one.

For each introduced secure compilation criterion we also discovered an
{\em equivalent ``property-free'' characterization} that is generally better
tailored for proofs and that provides important insights into what
kind of techniques one can use to prove the criterion.
For instance, for proving \rsp and \rtp we can produce a different source
context to explain {\em each} attack trace, while for proving stronger
criteria such as {\em Robust Hyperproperty Preservation} (\rhp) we
have to produce a single source context that works for {\em any} attack trace.

We also formally study the relation between our new security
criteria and full abstraction (\autoref{sec:fa}) proxied by 
 the robust preservation of trace equivalence (\rtep), which in
determinate languages---\IE languages without internal nondeterminism---was
shown to coincide with observational equivalence~\cite{Engelfriet85, ChevalCD13}.
In one direction, \rtep follows unconditionally from {\em
  Robust 2-relational Hyperproperty Preservation}, which is one of our
stronger criteria.
However, if the source and target languages are determinate
and we make some mild extra assumptions (such as input
totality~\cite{ZakinthinosL97, FocardiG95}) \rtep follows even from the
weaker {\em Robust 2-relational relaXed safety Preservation} (\rtrxp).
Here, the challenge was identifying these extra assumptions and
showing that they are sufficient to establish \rtep.
In the other direction, we adapt a counterexample proposed
by~\citet{PatrignaniG17} to show that \rtep (and thus full
abstraction), even in conjunction with compositional compiler
correctness, does \emph{not} imply even the weakest of our criteria,
\rsp, \rdp, and \rtinip.

Finally, we show that two proof techniques originally developed for
full abstraction can be readily adapted to prove our new secure
compilation criteria (\autoref{sec:example}).
First, we use a ``universal embedding''~\cite{NewBA16} to prove that
the strongest of our secure compilation criteria, Robust Relational
Hyperproperty Preservation (\rrhp), is achievable for a simple
translation from a statically typed to a dynamically typed first-order
language with first-order functions and I/O.
Second, we use the same simple translation to illustrate that for
proving {\em Robust Finite-relational relaXed safety Preservation}
(\criteria{\rfrxp}{rfrxp}) we can employ a ``trace-based
back-translation''~\cite{PatrignaniASJCP15, JeffreyR05}, a slightly
less powerful but more generic technique that we extend to back-translate
a finite set of finite execution prefixes into a source context.
This second technique is applicable to all criteria implied by \rfrxp,
which includes robust preservation of safety, of hypersafety, and
in a determinate setting also of trace (and thus observational) equivalence.


In summary, our paper makes five {\bf contributions}:

\smallskip\noindent
{\bf C1.} We phrase the formal security guarantees
obtained by protecting compiled programs from adversarial contexts in
terms of robustly preserving classes of properties.
We are the first to explore a large space of security criteria
based on this idea,
including criteria that provide strictly stronger security guarantees
than full abstraction, and also criteria that are easier to
practically achieve and prove, which is important for building more
realistic secure compilation chains.

\smallskip\noindent
{\bf C2.} We carefully study each new secure compilation criterion and the
non-trivial relations between them.
For each criterion we propose a property-free characterization
that clarifies which proof techniques apply.
For relating the criteria, we order them by their relative strength,
show several interesting collapses, and prove several challenging
separation results.

\smallskip\noindent
{\bf C3.} We introduce \emph{relational} properties and hyperproperties,
which are new property classes
of independent interest, even outside of secure compilation.
%
%
%

\smallskip\noindent
{\bf C4.} We formally study the relation between our \iffull new \fi security
criteria and full abstraction.
In one direction, we show that determinacy is enough for robustly
preserving classes of relational properties and hyperproperties to
imply preservation of observational equivalence\iffull (\IE the
direction of full abstraction interesting for security)\fi.
In the other direction, we show that, even when assuming compiler
correctness, full abstraction does not imply even our weakest criteria.

\smallskip\noindent
{\bf C5.} We show that two existing proof techniques
originally developed for full abstraction can be readily adapted
to our new criteria, which is important since good proof
techniques are difficult to find in this
space~\citeFull{NewBA16,MarcosSurvey}{domipoplta}.

\smallskip
The paper closes with discussions of
related~(\autoref{sec:related}) and future work~(\autoref{sec:future}).
The \ifcamera online \fi appendix \ifcamera at
{\tt \url{https://arxiv.org/abs/1807.04603}} \fi
contains omitted technical details.
%
Many of the theorems formally or informally mentioned in the paper
were also mechanized in the Coq proof assistant and are marked with \CoqSymbol;
this development has around 4400 lines of code and is
available \ifanon for review both as a supplementary material
attached to the submission and in anonymized form \fi at
{\tt\ifanon
\input{hash}
\url{https://github.com/JourneyBeyondFullAbstraction/Anonymous/tree/\hash}
\else
\url{https://github.com/secure-compilation/exploring-robust-property-preservation}
\fi}

\section{Robustly Preserving Trace Properties}
\label{sec:prop}

In this section we look at robustly preserving classes of
{\em trace properties}, and first study the robust preservation of
{\em all} trace properties and its relation
to correct compilation (\autoref{sec:rtp}).
We then look at robustly preserving {\em safety properties}
(\autoref{sec:rsp}), which are the trace properties that can be
falsified by a finite trace prefix (\EG a program never performs a
certain dangerous system call).
These criteria are grouped in the Trace Properties yellow area in
\autoref{fig:order}.
We also carefully studied the robust preservation of {\em liveness
  properties}, but it turns out that the very definition of liveness
is highly dependent on the specifics of the program execution traces,
which makes that part more technical.
For saving space and avoiding a technical detour, we relegate to the
appendix \ifcamera\else (\autoref{app:sec:traces-details}) \fi the details of our
CompCert-inspired trace model, as well as the part about liveness.

\subsection{Robust Trace Property Preservation (\rtp)}
\label{sec:rtp}

Like all secure compilation criteria we study in this paper, the
\rtp criterion below is a {\em generic}
property of an arbitrary {\em compilation chain}, which includes a source and a
target language, each with a notion of partial programs ($P$) and
contexts ($C$) that can be linked together to produce whole programs
($C[P]$), and each with a trace-producing semantics for whole
programs ($C[P] \sem t$).
The sets of partial programs and of contexts of the source
and target languages are unconstrained parameters of our secure
compilation criteria; our criteria make no assumptions about
their structure,
or whether the program or the context
gets control initially once linked and executed (\EG the
context could be an application that embeds a library program or the
context could be a library that is embedded into an application program).\footnote{
  One limitation of our formal setup, is that for simplicity we assume
  that any partial program can be linked with any context, irrespective
  of their interfaces (\EG types or specs). One can extend our criteria to
  take interfaces into account, as we illustrate
  in \ifcamera the appendix \else \autoref{app:sec:proofs-appendix} \fi
  for the example in \autoref{sec:example}.}
The traces produced by the source and target
semantics\footnote{\iffull Throughout the \else In this \fi
  paper we assume for simplicity that traces are
  exactly the same in both the source and target language, as is also
  the case in the CompCert verified C compiler~\cite{Leroy09}.
\ifanon
  In the appendix (\autoref{app:sec:different-traces}), we illustrate
  how to lift this restriction for the specific case of \rtp.
\else
  We hope to lift this restriction in the future (\autoref{sec:conclusion}).
\fi}
are arbitrary for \rtp, but for \rsp 
we have to consider traces with a specific structure (finite or
infinite sequences of events drawn from an arbitrary set).
Intuitively, traces capture the interaction between a
whole program and its external environment, including for instance
user input, output to a terminal, network communication, system calls,
\ETC~\cite{Leroy09, KumarMNO14}.
As opposed to a context, which is just a piece of a program, the
environment's behavior is not (and often {\em cannot} be) modeled by
the programming language, beyond the (often nondeterministic) interaction events
that we store in the trace.
%
Finally, a compilation chain includes a compiler: the compilation of a
partial source program $\src{P}$ is a partial target program
we write $\cmp{P}$.\footnote{For easier reading, we use
  a \src{blue, sans\text{-}serif} font for \src{source} elements,
 an \trg{orange, bold} font for \trg{target} elements and a
  \com{\commoncol, italic} font generically for elements of either language.}


 
\iffull As also mentioned in the introduction, t\else T\fi{}he
responsibility of enforcing
secure compilation does not have to rest just with the compiler, but
may be freely shared by various parts of the compilation chain.
In particular, to help enforce security, the target-level {\em linker} could
disallow linking with a suspicious context (\EG one that is not
well-typed~\citeFull{Abadi99, AhmedB11, Ahmed15, NewBA16}{AhmedB08})
or could always allow linking but
introduce protection barriers between the program and the context (\EG
by instrumenting the program~\cite{DevriesePP16, NewBA16} or the
context~\cite{sfi_sosp1993, Tan17, AbateABEFHLPST18} to introduce
dynamic checks).
Similarly, the semantics of the target language can include various
protection mechanisms (\EG processes with different virtual address
spaces~\citeFull{ProvosFH03}{ReisG09, Kilpatrick03, wedge_nsdi2008, GudkaWACDLMNR15},
protected enclaves~\cite{PatrignaniASJCP15},
capabilities~\citeFull{WatsonWNMACDDGL15,
  ChisnallDGBJWMM17, SkorstengaardDB19}{ElKorashyTPDGP18},
\ifallcites
tags~\cite{micropolicies2015, AbateABEFHLPST18}\else etc.\fi).
Finally, the compiler might have to refrain from aggressive
optimizations that would break security~\cite{DSilvaPS15, SimonCA18, BartheGL18}.
Our secure compilation criteria are agnostic to the concrete
enforcement mechanism used by the compilation chain to protect the
compiled program from the adversarial target context.

{\em Trace properties} are defined simply as sets of allowed
traces~\cite{LamportS84}.
A whole program $C[P]$ {\em satisfies} a trace property $\pi$ when the
set of traces produced by $C[P]$ is included in the set $\pi$ or, formally,
$\{t ~|~ C[P] \sem t\} \subseteq \pi$.
More interestingly, we say that a partial program $P$ {\em robustly
  satisfies}~\cite{KupfermanV99, GordonJ04, SwaseyGD17} a trace property
$\pi$ when $P$ linked with \emph{any} (adversarial) context $C$ satisfies $\pi$.
%
%
Armed with this, 
{\em Robust Trace Property Preservation} (\rtp) is defined as the
preservation of robust satisfaction of all trace properties.
So if a partial source program $\src{P}$ robustly satisfies a trace
property $\pi \in 2^\ii{Trace}$
(wrt.\ all source contexts) then its compilation $\cmp{P}$ must also
robustly satisfy $\pi$ (wrt.\ all target contexts).
If we unfold all intermediate definitions, a compilation chain
satisfies \criteria{\rtp}{rtp} iff:
\begin{align*}
  \rtp :\quad
  \forall \pi \in 2^\ii{Trace}.~\forall\src{P}.~
    &
    (\forall\src{C_S}~t \ldotp
    \mathrel{\src{C_S\hole{P} \sem}} t \Rightarrow t\in\pi)
    \Rightarrow
    \\
    &
    (\forall \trg{C_T}~t \ldotp
    \mathrel{\trg{C_T\hole{\cmp{P}} \sem}} t \Rightarrow t\in\pi)
\end{align*}


%

This definition directly captures which 
properties (specifically, all trace properties) of the source are robustly preserved by the compilation chain.
However, in order to prove that a compilation chain satisfies \rtp we
propose an equivalent~(\CoqSymbol) ``property-free'' characterization,
which we call \criteria{\pf{\rtp}}{pfrtp} (for ``\rtp Characterization''):
\begin{align*}
  \pf{\rtp}:\quad\ \forall\src{P}.~ \forall\trg{C_T}.~ \forall t.~
        \mathrel{\trg{C_T\hole{\cmp{P}} \sem}} t \Rightarrow
        \mathrel{\exists\src{C_S}\ldotp \src{C_S\hole{P}\sem}} t
\end{align*}
\pf{\rtp} requires that, given a compiled program $\cmp{P}$ and a
target context $\trg{C_T}$ which together produce an attack trace $t$,
we can generate a source context $\src{C_S}$ that causes trace $t$
to be produced by $\src{P}$.
When proving that a compilation chain satisfies
\pf{\rtp} we can pick a different context $\src{C_S}$ for each $t$ and,
in fact, try to construct $\src{C_S}$ from \iffull the \fi
trace $t$ \iffull itself \fi or from
the execution $\trg{C_T\hole{\cmp{P}}} \mathrel{\trg{\sem}} t$.

We present similar property-free characterizations for each of our criteria (\autoref{fig:order}). However, for criteria stronger than \rtp \iffull (such as {\em Robust Hyperproperty
  Preservation}  or \rhp)\fi,
a single context $\src{C_S}$ will have to work for more than one trace.
In general, the shape of the property-free characterization explains
what information can be used to produce the source
context $\src{C_S}$ when proving a compilation chain secure.

\paragraph{Relation to compiler correctness}
\pf{\rtp} is similar to ``backward simulation'' (\pf{\tp}), a standard
compiler \emph{correctness} criterion~\cite{Leroy09}. Let $W$ denote a
whole program.
\begin{align*}
\pf{\tp}:\quad
  \forall\src{W}.~ \forall t.~ \trg{\cmp{W}} \mathrel{\trg{\sem}} t \Rightarrow
  \src{W} \mathrel{\src{\sem}} t
\end{align*}
Maybe slightly less known is that this property-free characterization of
correct compilation also has an equivalent property-full characterization
as the preservation of all trace properties:
$$\begin{multlined}
\criteria{\tp}{tp}:\quad
  \forall \pi \in 2^\ii{Trace}.~\forall\src{W}.~\\
  \qquad\qquad(\forall t \ldotp
    \src{W} \mathrel{\src{\sem}} t \Rightarrow t\in\pi)
    \Rightarrow
    (\forall t \ldotp
    \trg{\cmp{W}} \mathrel{\trg{\sem}} t \Rightarrow t\in\pi)
\end{multlined}$$
The major difference compared to \rtp
is that \tp only preserves the trace properties
of whole programs and does not consider adversaries.
In contrast, \rtp allows linking a compiled partial program with
arbitrary target contexts and protects the program so that all
{\em robustly satisfied} trace properties are preserved.
In general, \rtp and \tp are incomparable.
However, \rtp strictly implies \tp when whole
programs ($W$) are a subset of partial programs ($P$) and,
additionally, the semantics of {\em whole} programs is independent of any
linked context (\IE
$\forall W~t~C.~
W \sem t \iff
C[W] \sem t$, which happens, intuitively, when the whole program
starts execution and, being whole, never calls into the context).
\ifsooner\ch{Have we actually proved this in Coq? Carmine says no,
  but maybe we should eventually do it}\fi
%
%

\ifsooner
\ch{The names \scc and \ccc are made up, but maybe that's okay
  given that I anyway wrote \scc and \ccc in a way that's specific to
  our setting. Could at least make sure that we don't badly overlap
  with any widely-used names.}
\ch{Renamed \formatCompilers{BCC} to \tp, but that no longer matches
  \scc and \ccc? Does coming up with better names for \scc and \ccc
  involve coming up with a property-full characterization for them?}
\ch{Arthur asks: What do the acronyms SCC and CCC stand for?
  Separate Compiler Correctness (SCC) and Compositional Compiler Correctness (CCC)}
\fi


More compositional criteria for compiler correctness have also been
proposed\iffull in the literature\fi~\citeFull{KangKHDV15,
  NeisHKMDV15}{PercontiA14, StewartBCA15}.
At a minimum such criteria allow linking with contexts that are the
compilation of source contexts~\cite{KangKHDV15}, which
\iffull in our setting \fi can be formalized as follows:
\[\begin{multlined}
\criteria{\scc}{scc}:\quad
  \forall\src{P}.~ \forall\src{C_S}.~ \forall t.~
    \trg{\cmp{C_S}\hole{\cmp{P}}} \mathrel{\trg{\sem}} t \Rightarrow
    \src{C_S\hole{P}} \mathrel{\src{\sem}} t
\end{multlined}\]
More permissive criteria allow linking with any target context that
behaves like some source context~\cite{NeisHKMDV15}, which in our
setting can be written as:
\[\begin{multlined}
\criteria{\ccc}{ccc}:~
  \forall\src{P}~\trg{C_T}~\src{C_S}~t.~
    \trg{C_T} {\approx} \src{C_S} \wedge
    \trg{C_T\hole{\cmp{P}}} \mathrel{\trg{\sem}} t \Rightarrow
    \src{C_S\hole{P}} \mathrel{\src{\sem}} t
\end{multlined}\]
Here $\approx$ relates equivalent partial programs in the target and
the source, and could, for instance, be instantiated with a
cross-language logical relation~\cite{NeisHKMDV15, Ahmed15}.\iffull\ch{More
  references in \cite[Section 6.1.2]{PatrignaniG17}}\fi{}
%
%
\rtp is incomparable to \iffull these \fi \scc and \ccc
\iffull criteria \fi.
On the one hand,
\rtp allows linking with {\em arbitrary} target-level
contexts, which is not allowed by \scc and \ccc,
and requires inserting strong protection barriers.
On the other hand, in \rtp all source-level reasoning has to be done
with respect to an {\em arbitrary} source context, while with
\scc and \ccc one can reason about a known source context.
\iffull
\ch{might cut sentence for space:}%
Technically, \pf{\rtp} does not imply \scc, since
even if we instantiate \pf{\rtp} with \cmp{C_S} for \trg{C_T}, what
we obtain in the source is
$\exists \src{C_S'}\ldotp \src{C_S'\hole{P}}\sem t$,
for some $\src{C_S'}$ that is unrelated to the original $\src{C_S}$.
Similarly, \pf{\rtp} does not imply \ccc, which is strictly stronger than
\scc under the natural assumption that $\cmp{C_S} \approx \src{C_S}$.
\fi


\iflater
\ch{It would be nice if eventually the informal arguments here would
  actually be backed some proofs, even if just on paper!
  Added to my TODO list}
\fi




\subsection{Robust Safety Property Preservation (\rsp)}
\label{sec:rsp}

Robust safety preservation is an interesting criterion for secure
compilation because it is easier to achieve and prove than most
criteria of \autoref{fig:order}, while still being quite
expressive~\cite{SwaseyGD17, GordonJ04}.

Recall that a trace property is a safety property if, within any (possibly
infinite) trace that violates the property, there exists a finite
``bad prefix'' that violates it.
We write $m \leq t$ for the prefix relation between a finite trace
prefix $m$ and a trace $t$\ifcamera\else (and give a precise
definition in \autoref{app:sec:traces-details})\fi.
Using this we define safety properties in the usual
way~\cite{AlpernS85, LamportS84, schneider1997concurrent}:
\[
\ii{Safety} \triangleq \{\pi \in 2^\ii{Trace} ~|~ \forall t \not\in \pi.~
  \exists m \leq t.~ \forall t' \geq m.~ t' \not\in \pi \}\\
\]

The definition of \criteria{\rsp}{rsp} simply restricts the preservation of robust
satisfaction from all trace properties in \rtp to only safety properties;
otherwise the definition is exactly the same:
\begin{align*}
\rsp:\quad
  \forall \pi \in \ii{Safety}.~\forall\src{P}.~
    &
    (\forall\src{C_S}~t \ldotp
    \src{C_S\hole{P}} \mathrel{\src{\sem}} t \Rightarrow t\in\pi)
    \Rightarrow
    \\
    &
    (\forall \trg{C_T}~t \ldotp
    \trg{C_T\hole{\cmp{P}}} \mathrel{\trg{\sem}} t \Rightarrow t\in\pi)
\end{align*}

One might wonder how safety properties can be \emph{robustly} satisfied in the
source, given that execution traces can contain events emitted not
just by the partial program but also by the adversarial context, which
could trivially emit ``bad events'' and, hence, violate any safety
property.
A first alternative is for the semantics of the source language to
simply prevent the context from producing any events, maybe other than 
termination\iffull, as we do in the compilation chain presented in the of
\autoref{sec:example}\fi, or, at least, prevent the
context from producing any events the safety properties of interest consider bad.
The compilation chain has then to ``sandbox'' the context to restrict
the events it can produce~\cite{sfi_sosp1993, Tan17}.
%
A second alternative is for the source semantics to record enough
information in the trace so that one can determine the origin of each
event---the partial program or the context\iffull---as done, for instance, in the appendix by what we call informative traces of \autoref{sec:instance-tracebased}\fi.
Then, safety properties in which the context's events are never bad can
be robustly satisfied.
With this second alternative, the obtained global guarantees are
weaker, \EG one cannot enforce that the whole program never makes a
dangerous system call, but only that the partial program cannot be
tricked by the context into making it.

The equivalent (\CoqSymbol) property-free characterization for \rsp
requires one to back-translate a program ($\src{P}$), a target context
($\trg{C_T}$), and a {\em finite} bad trace prefix
($\trg{C_T\hole{\cmp{P}}} \mathrel{\trg{\sem}} m$)
into a source context ($\src{C_S}$) producing the same finite trace prefix ($m$)
in the source ($\src{C_S\hole{P}} \mathrel{\src{\sem}} m$):
$$\begin{multlined}
  \criteria{\pf{\rsp}}{pfrsp}:\quad \forall\src{P}.~ \forall\trg{C_T}.~ \forall m.~
        \trg{C_T\hole{\cmp{P}}} \mathrel{\trg{\sem}} m \Rightarrow
\exists \src{C_S}\ldotp \src{C_S\hole{P}} \mathrel{\src{\sem}} m
\end{multlined}$$
Syntactically, the only change with respect to \pf{\rtp}
is the switch from whole traces $t$ to finite trace prefixes $m$.
As for \pf{\rtp}, we can pick a different context $\src{C_S}$
for each execution $\trg{C_T\hole{\cmp{P}}} \mathrel{\trg{\sem}} m$.
(In our formalization we define $W \mathrel{\sem} m$ generically as
$\exists t {\geq} m.~W \mathrel{\sem} t$.)
The fact that for \pf{\rsp} these are {\em finite} execution prefixes
can significantly simplify the back-translation into 
source contexts (as we show in \autoref{sec:instance-tracebased}).
\MP{ integrated above }
%

It is trivially true that \rtp implies \rsp, since the former robustly preserves all trace properties while the latter only robustly preserves safety properties.
We have also proved that \rtp \emph{strictly implies} \rsp.
\begin{theorem}\label{thm:rsp-doesnt-imply-rtp}
$\rtp \Rightarrow \rsp$, but $\rsp \not \Rightarrow \rtp$
\Coqed
\end{theorem}
\begin{proof}[Proof sketch]
As explained above, $\rtp \Rightarrow \rsp$ is trivial.
Showing strictness requires constructing a counterexample compilation chain to
the reverse implication.
%
We take any target language that can produce infinite traces.
%
We take the source language to be a variant of the target with the
same partial programs, but where we extend whole programs and contexts with
a bound on the number of events they can produce before being
terminated.
Compilation simply erases this bound.
\ifallcites
(While this construction might seem artificial, languages with a fuel
bound are gaining popularity~\cite{Ethereum}.)
\fi
This compilation chain satisfies \rsp (equivalently, \pf{\rsp})
but not \rtp.
To show that it satisfies \pf{\rsp}, we simply back-translate a target
context $\trg{C_T}$ and a finite trace prefix $m$ to a source context
%
$\src{(\trg{C_T},\ii{length}(\com{m}))}$ that uses the length of $m$
as the allowed bound, so this context can still produce $m$ in the
source without being prematurely terminated.
However, this compilation chain does not satisfy \rtp, since in the source all
executions are finite and, hence, no infinite target trace can be simulated by \emph{any} source context.
\end{proof}

\section{Robustly Preserving Hyperproperties}
\label{sec:hyperprop}

So far, we have studied the robust preservation of trace properties,
which are properties of {\em individual} traces of a program.
In this section we generalize this to {\em hyperproperties}, which are
properties of {\em multiple} traces of a program~\cite{ClarksonS10}.
A well-known hyperproperty is noninterference\iffull, which has many
variants\fi~\citeFull{GoguenM82, McLean92, ZdancewicM03,
  AskarovHSS08}{SabelfeldM03}, \iffull but \else which \fi usually
requires considering two traces of a program that differ on secret inputs.
Another hyperproperty is bounded mean response time over all executions.
We study robust preservation of many subclasses of hyperproperties:
all hyperproperties (\autoref{sec:rhp}),
subset-closed 
hyperproperties (\autoref{sec:rschp}),
hypersafety and $K$-hypersafety (\autoref{sec:rhsp}),
and hyperliveness (\autoref{sec:hyperliveness}).
These criteria are in the red area
in \autoref{fig:order}.



\subsection{Robust Hyperproperty Preservation (\rhp)}
\label{sec:rhp}


While trace properties are sets of traces,
hyperproperties are sets of sets of traces~\cite{ClarksonS10}.
We call the set of traces of a whole program $W$ the {\em behavior} of $W$:
$\behav{W} = \{ t ~|~ W \sem t \}$. A hyperproperty
is a set of allowed behaviors.
Program $W$ satisfies hyperproperty $H$ if
the behavior of $W$ is a member of $H$,
\IE $\behav{W} \in H$, or, equivalently,
$\{ t ~|~ W \sem t \} \in H$.
Contrast this to $W$ satisfying trace property $\pi$,
which holds if the behavior of $W$ is a subset of the set $\pi$,
\IE $\behav{W} \subseteq \pi$, or, equivalently,
$\forall t.~ W \sem t \Rightarrow t \in \pi$.
So while a trace property determines whether each individual trace of a
program should be allowed or not, a hyperproperty determines whether
the set of traces of a program, its behavior, should be allowed or not.
\iflater\MP{ begin example: hyperproperties VS properties}\fi
For instance, the trace property $\pi_{123} = \{t_1, t_2, t_3\}$ is
satisfied by programs with behaviors such as $\{t_1\}$, $\{t_2\}$, $\{t_2, t_3\}$,
and $\{t_1, t_2, t_3\}$, but a program with
behavior $\{t_1, t_4\}$ does not satisfy $\pi_{123}$.
A hyperproperty like $H_{1+23} = \{ \{t_1\}, \{t_2, t_3\}\}$
is satisfied only by programs with behavior $\{t_1\}$ or
with behavior $\{t_2, t_3\}$.
A program with behavior $\{t_2\}$ does not satisfy $H_{1+23}$,
so hyperproperties can express that if some traces (\EG $t_2$) are possible
then some other traces (\EG $t_3$) should also be possible.
A program with behavior $\{t_1,t_2,t_3\}$ also does not satisfy $H_{1+23}$,
so hyperproperties can express that if some traces (\EG $t_2$ and
$t_3$) are possible then some other traces (\EG $t_1$) should not be possible.
Finally, trace properties can be easily lifted to hyperproperties:
A trace property $\pi$ becomes the hyperproperty $[\pi] = 2^{\pi}$,
the powerset of $\pi$.

We say that a partial program $P$ {\em robustly satisfies} a
hyperproperty $H$ if it satisfies $H$ for any context $C$.
Given this we define \criteria{\rhp}{rhp} as the preservation of robust satisfaction
of arbitrary hyperproperties:
\begin{align*}
\rhp:\quad
  \forall H \in 2^{2^\ii{Trace}}.~\forall\src{P}.~
    &
    (\forall\src{C_S} \ldotp
    \src{\behav{C_S\hole{P}}} \in H)
    \Rightarrow
    \\&
    (\forall \trg{C_T} \ldotp
    \trg{\behav{C_T\hole{\cmp{P}}}} \in H)
\end{align*}

The equivalent (\CoqSymbol) \iffull property-free \fi characterization of \rhp
is \criteria{\pf{\rhp}}{pfrhp}
\iffull not very surprising\fi
:
\begin{align*}
  \pf{\rhp}:\quad&\ \forall\src{P}.~ \forall\trg{C_T}.~ \exists \src{C_S}.~
            \trg{\behav{C_T\hole{\cmp{P}}}} = \src{\behav{C_S\hole{P}}}
  \\
 \pf{\rhp}:\quad&\ \forall\src{P}.~ \forall\trg{C_T}.~ \exists \src{C_S}.~ \forall t.~
        \trg{C_T\hole{\cmp{P}} \sem} t \iff
        \src{C_S\hole{P}\sem} t
\end{align*}
This requires that, for every partial program $\src{P}$ and target
context $\trg{C_T}$, there is a (back-translated) source context
$\src{C_S}$ that perfectly preserves the set of traces of $\trg{C_T\hole{\cmp{P}}}$ when linked to $\src{P}$.
There are two differences from \rtp:
(1)~the $\exists \src{C_S}$ and $\forall t$ quantifiers are swapped,
    so the back-translated $\src{C_S}$
    must work for all traces $t$, and
(2)~the implication in \pf{\rtp} ($\Rightarrow$)
    became a two-way implication in \pf{\rhp} ($\iff$), so
    compilation has to perfectly preserve the set of traces.
In particular the compiler cannot refine behavior (remove traces), \EG it
cannot implement nondeterministic scheduling via a deterministic
scheduler.

\smallskip
In the following subsections we study restrictions of \rhp to various
subclasses of hyperproperties. To prevent duplication we define
$\rhp(X)$ to be the robust satisfaction of a class $X$ of hyperproperties
(so \rhp above is simply $\rhp(2^{2^\ii{Trace}})$):
\begin{align*}
\rhp(X):\quad
  \forall H \in X.~\forall\src{P}.~
    &(\forall\src{C_S} \ldotp
    \src{\behav{C_S\hole{P}}} \in H)
    \Rightarrow\\
    & (\forall \trg{C_T} \ldotp
    \trg{\behav{C_T\hole{\cmp{P}}}} \in H)
\end{align*}

\subsection{Robust Subset-Closed Hyperproperty Preservation
  (\rschp)}
\label{sec:rschp}

If one restricts robust preservation to only subset-closed
hyperproperties then refinement of behavior is again allowed.
A hyperproperty $H$ is subset-closed, written $H {\in} \ii{SC}$,
if for any two behaviors \iffull $b_1$ and $b_2$ so that\fi $b_1 {\subseteq} b_2$,
if $b_2 {\in} H$ then $b_1 {\in} H$.
For instance, the lifting $[\pi]$ of any trace property $\pi$ is
subset-closed, but the hyperproperty $H_{1+23}$ above is not.
It can be made subset-closed by allowing all smaller behaviors:
$H_{1+23}^\ii{SC} = \{\emptyset, \{t_1\}, \{t_2\}, \{t_3\}, \{t_2, t_3\}\}$
is subset-closed.

{\em Robust Subset-Closed Hyperproperty
  Preservation} (\criteria{\rschp}{rschp}) is simply defined as $\rhp(\ii{SC})$.
The equivalent (\CoqSymbol) property-free characterization of \pf{\rschp}
simply gives up the $\Leftarrow$ direction of \pf{\rhp}:
$$\begin{multlined}
  \pf{\rschp}:\quad \forall\src{P}.~ \forall\trg{C_T}.~ \exists \src{C_S}.~ \forall t.~
        \trg{C_T\hole{\cmp{P}}} \mathrel{\trg{\sem}} t \Rightarrow
        \src{C_S\hole{P}} \mathrel{\src{\sem}} t
\end{multlined}$$

The most interesting subclass of subset-closed hyperproperties is
hypersafety, which we discuss \iffull in the next sub-section\else next\fi.
The appendix \ifcamera\else(\autoref{app:sec:k-subset-closed})\fi
also studies \iffull the subclass of \fi  $K$-subset-closed
hyperproperties~\cite{mastroeni2018verifying},
which can be seen as generalizing $K$-hypersafety below.


\subsection{Robust Hypersafety Preservation (\rhsp)}
\label{sec:rhsp}
\label{sec:rnip}

Hypersafety is a generalization of safety that is very
important in practice, since several important notions of
noninterference are hypersafety, such as
termination-insensitive noninterference~\cite{SabelfeldS01, Fenton74,
  AskarovHSS08},
observational determinism~\cite{McLean92, ZdancewicM03, Roscoe95}, and
nonmalleable information flow~\cite{CecchettiMA17}.

According to Alpern and Schneider~\cite{AlpernS85}, the ``bad thing''
that a safety property disallows must be {\em finitely observable} and
{\em irremediable}.
For safety the ``bad thing'' is a finite trace prefix that cannot
be extended to any trace satisfying the safety property.
For hypersafety, \citet{ClarksonS10} generalize
the ``bad thing'' to a finite set of finite trace prefixes
that they call an {\em observation},
drawn from the set $\ii{Obs} = 2^\ii{FinPref}_\ii{Fin}$,
which denotes the set of all finite subsets of finite prefixes.
They then lift the prefix relation to sets:
an observation $o \in \ii{Obs}$
is a prefix of a behavior $b \in 2^\ii{Trace}$, written $o {\leq} b$,
if $\forall m \in o.~ \exists t \in b.~ m {\leq} t$.
Finally, they define hypersafety analogously to safety,
but the domains involved include an extra level of sets:
\[
\ii{Hypersafety} {\triangleq} \{ H ~|~
  \forall b {\not\in} H.~
    (\exists o {\in} \ii{Obs}.~ o {\leq} b \wedge
      (\forall b' {\geq} o.~ b' {\not\in} H)) \}
\]
Here the ``bad thing'' is an observation $o$ that cannot be
extended to a behavior $b'$ satisfying the hypersafety property $H$.
We use this to define
{\em  Robust Hypersafety Preservation}~(\criteria{\rhsp}{rhsp})
as $\rhp(\ii{Hypersafety})$
and propose the following equivalent~(\CoqSymbol) characterization for it:
$$\begin{multlined}
\pf{\rhsp}:~
  \forall\src{P}.~ \forall\trg{C_T}.~ \forall o \in \ii{Obs}.~\\
    \quad\qquad o \leq \trg{\behav{C_T\hole{\cmp{P}}}} \Rightarrow
    \exists \src{C_S}.~
    o \leq \src{\behav{C_S\hole{P}}}
\end{multlined}$$
This says that to prove \rhsp one needs to be able to back-translate a
partial program $\src{P}$, a context $\trg{C_T}$, and a prefix $o$ of
the behavior of $\trg{C_T\hole{\cmp{P}}}$, to a source context
$\src{C_S}$ so that the behavior of $\src{C_S\hole{P}}$ extends $o$.
It is possible to use the finite set of finite executions
corresponding to observation $o$ to drive this back-translation
(as we do in \autoref{sec:instance-tracebased}).

For hypersafety the involved observations are finite sets but their
cardinality is otherwise unrestricted.
In practice though, most hypersafety properties can be falsified by
very small sets: counterexamples to termination-insensitive
noninterference~\cite{SabelfeldS01, Fenton74, AskarovHSS08} and observational
determinism~\cite{McLean92, ZdancewicM03, Roscoe95} are observations containing
2 finite prefixes, while counterexamples to nonmalleable information
flow~\cite{CecchettiMA17} are observations containing 4 finite prefixes.
To account for this, \citet{ClarksonS10}
introduce $K$-hypersafety as a restriction of hypersafety to
observations of a fixed cardinality $K$.
Given $\ii{Obs}_K = 2^\ii{FinPref}_\ii{Fin(K)}$, the set of
observations with cardinality $K$, all definitions and results above
can be ported to $K$-hypersafety by simply replacing $\ii{Obs}$ with
$\ii{Obs}_K$. Specifically, we denote by $\criteria{\rkhsp}{rkhsp}$ the
criterion $\rhp(\ii{\mbox{$K$-}Hypersafety})$.

The set of lifted safety properties,
$\{[\pi] ~|~ \pi \in \ii{Safety} \}$, is precisely the same as
$1$-hypersafety, since the counterexample for them is a single finite prefix.
For a more interesting example, termination-insensitive
noninterference (\ii{TINI})~\cite{AskarovHSS08, SabelfeldS01,
  Fenton74} can be defined as follows in our setting:
\begin{align*}
\ii{TINI} \triangleq \{ b ~|~ \forall t_1~t_2 {\in} b.~
  &(t_1~\ii{terminating} \land t_2~\ii{terminating} ~\\
  &\land \ii{pub\text{-}inputs}(t_1){=}\ii{pub\text{-}inputs}(t_2))
\\
&\Rightarrow
   \ii{pub\text{-}events}(t_1){=}\ii{pub\text{-}events}(t_2) \}
\end{align*}
%
%
This \iffull definition \fi
requires that trace events are either inputs or outputs, each of
them associated with a security level: public or secret.
\ii{TINI} ensures that for any two terminating traces of the program
behavior for which the two sequences of public inputs are the same,
the two sequences of public events---inputs and outputs---are also the same.
\ii{TINI} is $2$-hypersafety, since $b \not\in \ii{TINI}$ implies
that there exist finite traces $t_1$ and $t_2$ that agree on the
public inputs but not on all public events, so we can simply take
$o = \{ t_1, t_2 \}$.
Since the traces in $o$ are already terminated,
any extension $b'$ of $o$ can only add extra traces,
\IE $\{ t_1, t_2 \} \subseteq b'$, so $b' \not\in \ii{TINI}$
as needed to conclude that \ii{TINI} is in $2$-hypersafety.
In \autoref{fig:order}, we write {\em Robust Termination-Insensitive
  Noninterference Preservation} (\criteria{\rtinip}{rtinip})
for $\rhp(\{\ii{TINI}\})$.

\subsection{Separation Between Properties and Hyperproperties}
\label{sec:hypersafety-separations}

\ifsooner
\ch{\ii{TINI} seems quite weak.
How about the following noninterference definition
  that works with infinite traces:
\[
\ii{NI?} = \{ b ~|~ \forall t_1, t_2 \in b.~
  \ii{pub\text{-}inputs}(t_1)=\ii{pub\text{-}inputs}(t_2) \Rightarrow
   \ii{pub\text{-}events}(t_1)=\ii{pub\text{-}events}(t_2) \}
\]
Q1: Does it make sense?
Q1.5: How is it related to \ii{TSNI}? Or progress-sensitivity?
It doesn't require
co-termination (which is exactly what makes \ii{TSNI} not be
hypersafety)? It does require co-termination if pub-events
includes termination, so it does seem related to TSNI!\\
Q2: Is it hypersafety? A counterexample to this is 2 potentially
infinite traces $t_1$ and $t_2$ that have the same public inputs but
that have different public outputs. The first different public outputs
must happen at a finite position though, so one could hope to take
the 2 prefixes up to that position in $o$.
These prefixes can still be extended, but is that a problem?
They will keep being different!?\\
Q3: How about something stronger, that works like determinacy,
as long as the public inputs are the same, the public outputs are the
same? Is this related to observational determinism?
\[
\begin{array}{l}
\ii{NI??} = \{ b ~|~ \forall t_1, t_2 \in b.~
  \ii{pub\text{-}events}(t_1)=\ii{pub\text{-}events}(t_2)
\\\qquad\qquad\qquad\qquad\qquad
  \vee~\exists m_1 m_2.~ \exists i_1 \neq i_2 \text{ public inputs}.~
\\\qquad\qquad\qquad\qquad\qquad\quad
  \ii{pub\text{-}events}(m_1)=\ii{pub\text{-}events}(m_2) \wedge
   m_1 \cdot i_1 \leq t_1 \land m_2 \cdot i_2 \leq t_2 \}
\end{array}
\]
This is basically the determinacy of the \ii{pub\text{-}events}
with respect to the public inputs.
}
\fi

\ifsooner
\ch{Is there any way to phrase this in a way that doesn't involve
  designating secret/public inputs/outputs. For instance, in terms of
  knowledge: a target context doesn't gain out any more information
  about \cmp{P}'s data than a source context gets. Is this the same as
  universally quantifying over all input/output labelings? There is
  definitely a difference between a compiler that preserves security
  for a single (known to it!) labeling vs a compiler that preserves
  it for all labelings (which is basically the same as treating
  everything as secret). We might need to clarify the fact that
  there are two different choices here, and just by looking at
  \ii{TINI} one cannot even tell which one we took}
\fi

\iflater
\ch{Would reactive noninterference be well suited for our setting?
  The paper looks complicated though~\cite{BohannonP10}.}
\fi


Enforcing \rhsp is strictly more demanding than enforcing \rsp.
Because even \rthsp (robust 2-hypersafety preservation) implies 
\rtinip, a compilation chain satisfying \rthsp has to make
sure that a target-level context cannot infer more information about
the internal data of $\cmp{P}$ than a source context could infer about
the data of $\src{P}$.
By contrast, a \rsp compilation chain can allow arbitrary {\em reads} of $\cmp{P}$'s
internal data, even if $\src{P}$'s data is private at the source level.
Intuitively, for proving \pf{\rsp}, the source context produced by
back-translation can guess any secret $\cmp{P}$ receives in the {\em
  single} considered execution, but for \criteria{\rthsp}{rthsp}
the single source context needs to work for {\em two} different executions,
potentially with two different secrets, so guessing is no longer an option.
We use this idea to prove a more general separation result
$\rtp \not\Rightarrow \rtinip$, by exhibiting a toy compilation chain
in which private variables are readable in the target language, but
not in source.
%
%
\begin{theorem}\label{thm:rtp-doesnt-imply-rtinip}
$\rtp \not\Rightarrow \rtinip$
\end{theorem}
This implies a strict separation between all criteria based on
hyperproperties (the red area in \autoref{fig:order}, having \rtinip
as the bottom) and all the ones based on trace properties (the yellow
area in \autoref{fig:order} having \rtp as the top).



Using a more complex counterexample involving a system of $K$ linear
equations, we have also shown that, for any $K$, robust preservation
of $K$-hypersafety, does not imply robust preservation of $(K{+}1)$-hypersafety.

\begin{theorem}\label{thm:rkhsp-doesnt-imple-rkkhsp}
$\forall K.~ \rkhsp \not\Rightarrow \rkkhsp$
\end{theorem}

\subsection{Where Is Robust Hyperliveness Preservation?}
\label{sec:hyperliveness}

{\em Robust Hyperliveness Preservation} (\rhlp) does not appear in
\autoref{fig:order}, because it is provably equivalent to \rhp (or, equivalently, \pf{\rhp}).
We define \rhlp as $\rhp(\ii{Hyperliveness})$ for the following standard
definition of \ii{Hyperliveness}~\cite{ClarksonS10}:
\[
\ii{Hyperliveness} \triangleq
  \{ H ~|~ \forall o \in \ii{Obs}.~ \exists b {\geq} o.~ b \in H  \}
\]
The proof that \rhlp implies \pf{\rhp}  (\CoqSymbol) involves showing that
$\{ b ~|~ b {\not=} \trg{\behav{C_T\hole{\cmp{P}}}} \}$,
the hyperproperty allowing all behaviors other than
$\trg{\behav{C_T\hole{\cmp{P}}}}$, is hyperliveness.
\ifsooner\ch{Using Alix-style proof this hyperproperty could change
  to turn $\not=$ to $=$?}\fi
Another way to obtain this result \iffull though \fi is from the fact that, as in
previous models~\cite{AlpernS85}, each
hyperproperty can be decomposed as the intersection of two
hyperliveness properties.
This collapse of {\em preserving} hyperliveness and {\em preserving} all
hyperproperties happens irrespective of the adversarial contexts.

\section{Robustly Preserving Relational Hyperproperties}
\label{sec:rel}

Trace properties and hyperproperties are predicates on the behavior of
a single program.
However, we may be interested in showing that compilation robustly preserves
\emph{relations} between the behaviors of two or more programs.
For example, suppose we optimize a partial source program $\src{P_1}$
to $\src{P_2}$ such that $\src{P_2}$ runs
faster than $\src{P_1}$ in any source context.
We may want compilation to preserve this ``runs faster than''
\emph{relation} between the two program behaviors against arbitrary
target contexts.
Similarly, in any source context, the behaviors of
$\src{P_1}$ and $\src{P_2}$ may be equal and we may want the compiler
to preserve such trace equivalence~\citeFull{BaeldeDH17,
  DelauneH17}{ChevalKR18} in arbitrary target contexts.
This last criterion, which we call \emph{Robust Trace Equivalence
  Preservation} (\rtep) in \autoref{fig:order}, is interesting because
in various determinate settings~\cite{ChevalCD13, Engelfriet85} it
coincides with preserving observational equivalence, the
security-relevant part of full abstraction (see \autoref{sec:fa}).

In this section, we study the robust preservation of such
\emph{relational hyperproperties} and several interesting subclasses,
still relating the behaviors of multiple programs.
%
Unlike hyperproperties and trace properties, relational
hyperproperties have not been defined as a general concept in the
literature, so even their definitions are new. We describe relational
hyperproperties and their robust preservation in
\autoref{sec:rel:hyper}, then look at subclasses induced by what we call
\emph{relational properties} (\autoref{sec:rel:trace})
and \emph{relational safety properties} (\autoref{sec:rel:safety}).
The appendix \ifcamera\else(\autoref{app:sec:rel-trace-prop-criteria})\fi
presents a few other subclasses.
The corresponding secure compilation criteria are grouped in the blue area
in \autoref{fig:order}. In \autoref{sec:rel:sep-rel} we show that, in
general, none of these \iffull criteria for robust preservation of
relations between programs \else relational criteria \fi are implied
by any non-relational \iffull robust preservation \fi criterion (from
\autoref{sec:prop} and \autoref{sec:hyperprop}), while in
\autoref{sec:context-composition} we show two specific situations in
which most relational criteria collapse to non-relational ones.






\subsection{
Relational Hyperproperty Preservation (\rrhp)}
\label{sec:rel:hyper}

We define a \emph{relational hyperproperty} as a predicate (relation)
on a sequence of behaviors of some length. A sequence of programs of the same
length is then said to have the relational hyperproperty if their
behaviors collectively satisfy the predicate. Depending on the arity
of the predicate, we get different subclasses of relational
hyperproperties. For arity $1$, the resulting subclass describes
relations on the behavior of individual programs, which coincides with
hyperproperties (\autoref{sec:hyperprop}). For arity $2$, the
resulting subclass consists of relations on the behaviors of two
programs. Both examples described at the beginning of this section lie
in this subclass. This generalizes to any finite arity $K$ (predicates
on behaviors of $K$ programs), and to the infinite arity.

Next, we define the robust preservation of these subclasses.
\iffull
For arity $1$, robust relational hyperproperty
preservation coincides with robust hyperproperty preservation, \rhp,
of \autoref{sec:rhp}.\ch{dupe}
\fi
For arity $2$, \emph{robust 2-relational hyperproperty preservation},
\rtrhp, is defined as follows:
$$\begin{multlined}
\criteria{\rtrhp}{rtrhp}:\,
  \forall R \in 2^{(\behavAll^2)}.~\forall\src{P_1}~\src{P_2}.~\\
  ~(\forall\src{C_S} \ldotp
  (\src{\behav{{C_S}\hole{P_1}}},\,\src{\behav{C_S\hole{P_2}}}) \in R)
  \Rightarrow\\
  ~(\forall\trg{C_T} \ldotp
  (\trg{\behav{{C_T}\hole{\cmp{P_1}}}},\,\trg{\behav{C_T\hole{\cmp{P_2}}}}) \in R)
\end{multlined}$$
\rtrhp says that for any binary relation $R$ on behaviors
of programs, if the behaviors of $\src{P_1}$ and $\src{P_2}$ satisfy
$R$ in every source context, then so do the behaviors of $\cmp{P_1}$
and $\cmp{P_2}$ in every target context. In other words, a compiler
satisfies \rtrhp iff it preserves any relation between pairs of
program behaviors that hold in all contexts. In particular, such a
compilation chain preserves trace equivalence in all contexts (\IE
\criteria{\rtep}{rtep}),
which we obtain by instantiating $R$ with equality in the
above definition (\CoqSymbol).
If execution time is recorded on
program traces, then such a compilation chain also preserves relations like
``the average execution time of $\src{P_1}$ across all inputs is no
more than the average execution time of $\src{P_2}$ across all
inputs'' and ``$\src{P_1}$ runs faster than $\src{P_2}$ on all
inputs'' (i.e., $\src{P_1}$ is an improvement of $\src{P_2}$).
This last property can also be described as a relational predicate on
pairs of traces (rather than behaviors); we return to this point in
\autoref{sec:rel:trace}.

\criteria{\rtrhp}{pfrtrhp} has an equivalent (\CoqSymbol)
property-free variant that does not mention relations $R$:
\begin{align*}
  \pf{\rtrhp}: \forall\src{P_1}\,\src{P_2}\,\trg{C_T}.\exists \src{C_S}.\,
  &
  \trg{\behav{C_T\hole{\cmp{P_1}}}} {=} \src{\behav{C_S\hole{P_1}}}
  \\
  \mathrel{\wedge} \
  &
  \trg{\behav{C_T\hole{\cmp{P_2}}}} {=} \src{\behav{C_S\hole{P_2}}}
\end{align*}
\pf{\rtrhp} is a \iffull direct \fi generalization of \pf{\rhp} from
\autoref{sec:rhp}, but now the same source context $\src{C_S}$ has to
simulate the behaviors of \emph{two} target programs,
$\trg{C_T\hole{\cmp{P_1}}}$ and $\trg{C_T\hole{\cmp{P_2}}}$.

\rtrhp generalizes to any finite arity $K$ in the obvious way,
yielding \criteria{\rkrhp}{rkrhp}.
Finally, this also generalizes to the infinite arity. We call this
\emph{Robust Relational Hyperproperty Preservation} (\criteria{\rrhp}{rrhp}):
$$\begin{multlined}
\criteria{\rrhp}{rrhp}:\,
  \forall R \in 2^{(\behavAll^\omega)}.~\forall\src{P_1},..,\src{P_K},...~\\
  ~(\forall\src{C_S} \ldotp
  (\src{\behav{{C_S}\hole{P_1}}},..,\src{\behav{C_S\hole{P_K}}},..) \in R)
  \Rightarrow\\
  ~(\forall\trg{C_T} \ldotp
  (\trg{\behav{{C_T}\hole{\cmp{P_1}}}},..,\trg{\behav{C_S\hole{\cmp{P_K}}}},..) \in R)
\end{multlined}$$

%
\rrhp is the strongest criterion we study and, hence, it is the
highest point in \autoref{fig:order}.
This includes robustly preserving predicates on all programs of the
language, although we have not yet found practical uses for this.
More interestingly, \rrhp has a very natural equivalent
property-free characterization, \criteria{\pf{\rrhp}}{pfrrhp},
requiring for every target context $\trg{C_T}$, a source context
$\src{C_S}$ that can simulate the behavior of $\trg{C_T}$ for
\emph{any} program:
\begin{align*}
\pf{\rrhp}:\quad
\forall\trg{C_T}.~ \exists \src{C_S}.~\forall \src{P}.~
  \trg{\behav{C_T\hole{\cmp{P}}}} {=} \src{\behav{C_S\hole{P}}}
\end{align*}

It is instructive to compare the property-free characterizations of
the preservation of robust trace properties (\pf{\rtp}),
hyperproperties (\pf{\rhp}), and relational hyperproperties
(\pf{\rrhp}). In \pf{\rtp}, the source context \src{C_S} may depend on
the target context \trg{C_T}, the source program \src{P} and a given
trace $t$. In \pf{\rhp}, \src{C_S} may depend only on \trg{C_T} and
\src{P}. In \pf{\rrhp}, \src{C_S} may depend only on \trg{C_T}. This
directly reflects the increasing expressive power of trace properties,
hyperproperties, and relational hyperproperties, as predicates on
traces, behaviors (set of traces), and sequences of behaviors, respectively.



\subsection{
Relational Trace Property Preservation (\rrtp)}
\label{sec:rel:trace}

\emph{Relational (trace) properties} are the subclass of relational
hyperproperties that are fully characterized by relations on
\emph{individual} traces of multiple programs. For example, the
relation ``$\src{P_1}$ runs faster than $\src{P_2}$ on every input''
is a 2-ary relational property characterized by pairs of traces, one
from $\src{P_1}$ and the other from $\src{P_2}$, which either differ
in the input or where the execution time in $\src{P_1}$'s trace is
less than that in $\src{P_2}$'s trace.
\iflater
\ca{To be honest,
  \textbf{formally} they are two different mathematical objects, sets
  of tuples of traces and sets of tuples of behaviors.}\ch{In general,
  property class inclusion is up to a coercion, which is described
  below, so I don't understand the complaint. Would you also claim
  that trace properties are formally not hyperproperties, just because
  there is a representation change?}%
\fi
Formally, relational properties of arity $K$ are a subclass of relational hyperproperties
of the same arity. A $K$-ary relational hyperproperty is a relational
(trace) property if there is a $K$-ary relation $R$ on \emph{traces}
such that $P_1,..,P_K$ are related by the relational hyperproperty iff
$(t_1,\ldots,t_k) {\in} R$ for any $t_1 {\in} \behav{P_1}, \ldots, t_k
{\in} \behav{P_K}$.  \iffull

\fi
Next, we define the robust preservation of relational properties of
different arities. For arity 1, this coincides with \iffull robust trace
property preservation, \fi \rtp  from \autoref{sec:rtp}. For arity 2, we
define \emph{Robust 2-relational \iffull Trace\fi Property
Preservation}\iffull (\rtrtp) as follows\fi:
\begin{align*}
&
\rtrtp:
  \forall R \in 2^{(\ii{Trace}^2)}.~\forall\src{P_1}~\src{P_2}.
  \\&
  \left(\begin{aligned}
    \forall\src{C_S}\,t_1\,t_2 \ldotp
    (\src{C_S\hole{P_1}} \mathrel{\src{\sem}} t_1
    \mathrel{\wedge}
    \src{C_S\hole{P_2}} \mathrel{\src{\sem}} t_2)
    \Rightarrow (t_1{,}t_2) {\in} R
    \end{aligned}\right) 
    \Rightarrow
    \\&
  \left(\begin{aligned}
    \forall \trg{C_T}\,t_1\,t_2 \ldotp
    (\trg{C_T\hole{\cmp{P_1}}} \trg{\sem} t_1
    \wedge
    \trg{C_T\hole{\cmp{P_2}}} \trg{\sem} t_2)
    \Rightarrow (t_1{,} t_2) {\in} R
    \end{aligned}\right)
\end{align*}
%
\criteria{\rtrtp}{rtrtp} is weaker than its relational hyperproperty
counterpart, \rtrhp (\autoref{sec:rel:hyper}): Unlike \rtrhp, \rtrtp
does not imply the robust preservation of relations like ``the average
execution time of $\src{P_1}$ across all inputs is no more than the
average execution time of $\src{P_2}$ across all inputs'' (a relation
between average execution times of $\src{P_1}$ and $\src{P_2}$ cannot
be characterized by any relation between individual traces of
$\src{P_1}$ and $\src{P_2}$).
\iffull
\ch{What's missing here is also relating
  this to 2-subset closed} \dg{We can't do this since we didn't define
  2-subset closed in the paper's body. Let's not get into this.}
\fi

\rtrtp also has an equivalent (\CoqSymbol)
\iffull property-free \fi characterization\iffull,
\pf{\rtrtp}\fi:
\begin{align*}
\pf{\rtrtp}:~ &\forall\src{P_1}~\src{P_2}~\trg{C_T}~t_1~t_2.~
\\&\ 
(\trg{C_T\hole{\cmp{P_1}}} \mathrel{\trg{\sem}} t_1
\mathrel{\wedge}
\trg{C_T\hole{\cmp{P_2}}} \mathrel{\trg{\sem}} t_2)
\Rightarrow
\\&\  \ 
\exists \src{C_S}\ldotp (\src{C_S\hole{P_1}} \mathrel{\src{\sem}} t_1
\mathrel{\wedge}
\src{C_S\hole{P_2}} \mathrel{\src{\sem}} t_2 )
\end{align*}
Establishing \pf{\rtrtp} requires constructing a source context
$\src{C_S}$ that can simultaneously simulate a given trace of
$\trg{C_T \hole{ \cmp{P_1}}}$ and a given trace of $\trg{C_T
  \hole{\cmp{P_2}}}$.
\rtrtp generalizes from arity 2 to any finite arity $K$ (yielding
\criteria{\rkrtp}{rkrtp}) and the infinite one
(yielding \criteria{\rrtp}{rrtp}) in the obvious way.

\subsection{Robust Relational Safety Preservation (\rrsp)}
\label{sec:rel:safety}

\emph{Relational safety properties} are a natural generalization of
safety and hypersafety properties to multiple programs, and an
important subclass of relational trace properties. Several interesting
relational trace properties are actually relational safety
properties. For instance, if we restrict the earlier relational trace
property ``\src{P_1} runs faster than \src{P_2} on all inputs'' to
terminating programs it becomes a relational safety property,
characterized by pairs of bad terminating prefixes, where both
prefixes have the same input, and the left prefix shows termination no
earlier than the right prefix.

Formally, a relation $R \in 2^{(\ii{Trace}^K)}$ is 
{\em $K$-relational safety} if for every $K$ ``bad'' traces
$(t_1, \ldots, t_K) \not\in R$,
there exist $K$ ``bad'' finite prefixes $m_1, \ldots, m_k$ such that
$\forall i.~ m_i \leq t_i$, and any $K$ traces
$(t_1', \ldots, t_K')$ pointwise extending $m_1, \ldots, m_k$
are also not in the relation, \IE
$\forall i.~ m_i \leq t_i'$ implies
$(t_1', \ldots, t_K') \not\in R$.
Then, \emph{Robust 2-relational Safety Preservation}
(\criteria{\rtrsp}{rtrsp}) is simply defined
by restricting \rtrtp to only 2-relational safety properties.
The equivalent (\CoqSymbol) property-free characterization for \rtrsp
is the following:
\begin{align*}
  \pf{\rtrsp} :~& \forall\src{P_1}~\src{P_2}~\trg{C_T}~m_1~m_2.~\\
  &\quad (\trg{C_T\hole{\cmp{P_1}}} \rightsquigarrow m_1
  \mathrel{\wedge}
  \trg{C_T\hole{\cmp{P_2}}} \rightsquigarrow m_2)
  \Rightarrow\\
  &\quad \exists \src{C_S}\ldotp (\src{C_S\hole{P_1}}\rightsquigarrow m_1
  \mathrel{\wedge}
  \src{C_S\hole{P_2}}\rightsquigarrow m_2)
\end{align*}

The only difference from the stronger \pf{\rtrtp}
(\autoref{sec:rel:trace}) is between considering full traces and only
finite prefixes. Again,
\rtrsp generalizes to any finite arity $K$ (yielding
\criteria{\rkrsp}{rkrsp}) and the infinite one
(yielding \criteria{\rrsp}{rrsp}) in the obvious way.
\iflater
\ca{It is not exactly the obvious way, we use Theorem R2rSP\_R2rSC' in Coq (but we mention it in the appendix)}
\rb{Could you add a reference to the point in the appendix where this happens?}
\ch{The obviousness claim here is just about the {\em standard
    property-full criteria}.  The fact that going to property-free
  happens in more steps or that for strange reasons we started from a
  weirder property-full criterion is irrelevant here.  Don't want to
  make this section super repetitive just for obvious stuff.}
\fi


\iffull
\subsection{Robust Non-Relational Preservation \iffull Does Not \else
  Doesn't \fi Imply Robust Relational Preservation}
\else
\subsection{Separation Between Relational and Non-Relational}
\label{sec:rel:sep-rel}

Relational (hyper)properties (\autoref{sec:rel:hyper},
\autoref{sec:rel:trace}) and hyperproperties (\autoref{sec:hyperprop})
are different but both have a ``relational'' nature:
relational (hyper)properties are relations on the behaviors or traces
of multiple programs, while hyperproperties are relations on multiple
traces of the same program.
So one may wonder whether there is any case in which the
\emph{robust preservation} of a class of relational (hyper)properties
is equivalent to that of a class of hyperproperties.
Could a compiler that robustly
preserves all hyperproperties (\rhp, \autoref{sec:rhp}) also robustly
preserves at least some class of 2-relational (hyper)properties?
In \autoref{sec:context-composition} we show special
cases in which this is indeed the case, while here 
we now show that {\em in general} \rhp does not imply the robust
preservation of any subclass of relational properties that we have
described so far (except, of course, relational properties of
arity 1, that are just hyperproperties). 
Since \rhp is the strongest non-relational robust preservation 
criterion that we study, this also means that no non-relational
robust preservation criterion implies any relational robust 
preservation criterion in \autoref{fig:order}. 
So, all edges from relational to non-relational criteria in
\autoref{fig:order} are strict implications.

To prove this, we build a compilation chain satisfying \rhp,
but not \rtrsp, the weakest relational criterion in
\autoref{fig:order}. 

\begin{theorem}\label{thm:rhp-doesnt-imply-r2rsp}
$\rhp \not\Rightarrow \rtrsp$
\end{theorem}
\begin{proof}[Proof sketch]
Consider a source language that lacks code introspection, and a target
language that is exactly the same, but additionally has a primitive
with which the context can read the code of the compiled \iffull partial \fi
program as data~\cite{Smith84}.
Consider the trivial compiler that is syntactically
the identity. It is clear that this compiler satisfies \rhp
since the added operation of code introspection offers no advantage to
the context when we consider properties of a single program, as is the
case in \rhp. More precisely, in establishing \pf{\rhp}, the
property-free characterization of \rhp, given a target context
$\trg{C_T}$ and a program $\src{P}$, we can construct a simulating
source context $\src{C_S}$ by modifying $\trg{C_T}$ to hard-code
$\src{P}$ wherever $\trg{C_T}$ performs code introspection. This works
as $\src{C_S}$ can depend on $\src{P}$ in \pf{\rhp}.

Now consider two programs that differ only in some dead code, that
both read a value from the context and write it back verbatim to the
output. These two programs satisfy the relational safety property
``the outputs of the two programs are equal'' in any \emph{source}
context. However, there is a trivial \emph{target} context that causes
the compiled programs to break this relational property. This context
reads the code of the program it is linked to, and provides $1$ as
input if it happens to be the first of our two programs and $2$
otherwise. Consequently, in this target context, the two programs
produce outputs $1$ and $2$ and do not have this relational safety
property in all contexts. Hence, this compiler does not satisfy
\rtrsp. Technically, the trick of hard-coding the program in
$\src{C_S}$ no longer works since there are two different programs
here.
\end{proof}

This proof provides a fundamental insight: To robustly preserve any
subclass of relational (hyper)properties, compilation must ensure that
target contexts cannot learn anything about the \emph{syntactic
  program} they interact with beyond what source contexts can also
learn.
%
When the target language is low-level, hiding code attributes can be
difficult:
it may require padding the code segment of the compiled program to a
fixed size, and cleaning or hiding any code-layout-dependent data like
code pointers from memory and registers when passing control to the
context.
These complex protections are not necessary for any non-relational
preservation criteria (even \rhp), but are already known to be
necessary for fully abstract compilation to low-level
code~\citeFull{PatrignaniASJCP15, JuglaretHAEP16,
  Kennedy06}{PatrignaniDP16}.
They can also be trivially circumvented if the context has access to
any side-channels, \EG it can measure time via a different thread.
In fact, in such settings trying to hide the source code can be seen
as a hopeless attempt at ``security through obscurity'', which is
widely rejected by cryptographers since the early days~\cite{Kerckhoffs1883}.

\subsection{Composing Contexts Using Full Reflection or Internal
  Nondeterminism in the Source Language}
\label{sec:context-composition}

The proof of the previous separation theorem strongly relies on the absence of
code introspection in the source language.
However, if source contexts can obtain {\em complete} intrinsic
information about the programs they are linked with, then \rhp implies
\rtrhp.
Such ``full reflection'' facilities are available in languages such as
Lisp~\cite{Smith84} and Smalltalk.
\ifallcites\footnote{Full reflection was shown to cause observational
  equivalence to degenerate to syntactical identity~\cite{WandF88, Felleisen91}.}\fi
For proving this collapse we inspect
the alternative characterizations, \pf{\rhp} and \pf{\rtrhp}.
The main difference between these two criteria, as explained
in \autoref{sec:rel:hyper}, is that the source context
\src{C_S} obtained by \pf{\rtrhp} depends on two, possibly distinct
programs \src{P_1} and \src{P_2} and a target context \trg{C_T}, while
every possible source context obtained by \pf{\rhp} depends on one
single program.
Hence, by applying \pf{\rhp} once for \src{P_1} and once for \src{P_2},
with the same context \trg{C_T}, we obtain two source contexts
\src{C_{S_1}} and \src{C_{S_2}} that are a priori unrelated.
Without further hypotheses, one cannot show \pf{\rtrhp}.
However, with full reflection we can define a source context
\src{C_S'} that behaves exactly like \src{C_{S_1}} when linked with
\src{P_1}, and like \src{C_{S_2}} otherwise.
%
%
%
%
We can use this construction to show not only that \rhp
implies \rtrhp, but also that robust preservation of each class of
finite-relational properties collapses to the corresponding 
hyperproperty-based criterion:
\begin{theorem}\label{thm:main-reflection}
If the source language has full reflection then
$\rhp {\Rightarrow} \rkrhp$,
$\rschp {\Rightarrow} \rkrtp$, and
$\rthsp {\Rightarrow} \rfrsp$.\Coqed
\end{theorem}
\iffull
\IE \rhp implies \rkrhp, \rhsp implies
\rfrsp, and \rschp\ch{no longer introduced} implies robust preservation of
relational subset-closed hyperproperties (\CoqSymbol).
\fi

One may wonder whether some other condition exists that makes robust
preservation of relational hyperproperty classes collapse even to the
corresponding \emph{trace-property-based} criteria (\autoref{sec:prop}).
This is indeed the case when the source language has an \emph{internal
  nondeterministic choice operator $\oplus$}, such that the
behavior of $\src{P_1} \oplus \src{P_2}$ is at least the union of the
behaviors of $\src{P_1}$ and $\src{P_2}$.
Such an operator is standard in process calculi~\cite{booksangio}.
To illustrate this we show that \pf{\rtp} implies \pf{\rtrtp}.
Note that \pf{\rtrtp} produces a source context
\src{C_S} that depends on a target context, two source programs
$\src{P_1}$ and $\src{P_2}$ and two, possibly incomparable, traces
$t_1$ and $t_2$.
\pf{\rtp} produces a context depending only on a single trace of a
single source program.
We can apply \pf{\rtp} twice: once for $t_1$ and $\src{P_1}$
obtaining \src{C_{S_1}} and once for $t_2$ and $\src{P_2}$
obtaining \src{C_{S_2}}.
To prove \pf{\rtrtp} we need to build a source context that
over-approximates the behaviors of both \src{C_{S_1}} and
\src{C_{S_2}}.
This context can be $\src{C_{S_1}} \oplus \src{C_{S_2}}$.
Hence, in this setting \pf{\rtp} (\rtp) implies \pf{\rtrtp} (\rtrtp).
%
This result generalizes to any finite arity.
\begin{theorem}\label{thm:main-nondet-choice}
If the source language has an internal nondeterministic choice
operator on contexts then
$\rtp \Rightarrow \rkrtp$,
$\rschp \Rightarrow \rfrschp$, and
$\rsp \Rightarrow \rfrsp$.\Coqed
\end{theorem}

Notice that since contexts are finite objects, the techniques above
only produce collapses in cases where finitely many source contexts
need to be composed.
Criteria relying on infinite-arity relations such as \rrhp and \rrtp
are thus not impacted by these collapses.
%
%
The appendix \ifcamera\else(\autoref{app:sec:composing-contexts})\fi has
more details and collapsed variants of \autoref{fig:order}.

\section{Where Is Full Abstraction?}
\label{sec:fa}

Full abstraction---the preservation and reflection of observational
equivalence---is a well-studied criterion for secure compilation
(\autoref{sec:related}).
The security-relevant direction of full
abstraction is {\em Observational Equivalence Preservation} (\oep)~\cite{DevriesePP16,MarcosSurvey}:
$$\begin{multlined}
\criteria{\oep}{oep}:
  \forall\src{P_1}~\src{P_2}.~ \src{P_1 \approx P_2} \Rightarrow
  \trg{\cmp{P_1} \approx \cmp{P_2}}
\end{multlined}$$
One natural question is how \oep relates to our
criteria of robust preservation.

Here we answer this question for languages without internal
nondeterminism.
In such {\em determinate}~\cite{Engelfriet85,Leroy09} settings
observational equivalence coincides with trace equivalence in all
contexts~\cite{Engelfriet85, ChevalCD13} and, hence, \oep coincides
with robust trace-equivalence preservation (\rtep).
As explained in \autoref{sec:rel:hyper}, it is obvious that \rtep is
an instance of \rtrhp, obtained by choosing equality as the relation $R$.
However, for determinate languages with {\em input
  totality}~\cite{ZakinthinosL97, FocardiG95} (if the program accepts
one input value, it has to also accept any other input value) we have
proved that even the weaker \rtrtp implies \rtep (\CoqSymbol).
This proof also requires that if a whole program can produce every
finite prefix of an infinite trace then it can also produce the
complete trace, but we have showed that this holds for the infinite
traces produced in a standard way by {\em any} determinate small-step
semantics.
%
%
%
%
Under these assumptions, we have in fact proved that \rtep follows
from the even weaker {\em Robust 2-relational relaXed safety Preservation}
(\criteria{\rtrxp}{rtrxp}).
The class {\em 2-relational relaXed safety} is a variant of {\em
  2-relational Safety} from \autoref{sec:rel:safety}; with this relaxed variant
``bad'' prefixes $x_1$ and $x_2$ are allowed to end with silent divergence
(denoted as $\ii{XPref}$):
\[
\begin{array}{l}
R \in \text{2-}\ii{relational\ relaXed\ safety} \iff\\
\qquad
\forall (t_1,t_2) {\not\in} R.~ \exists x_1\,x_2 {\in} \ii{XPref}.~
\forall t_1' {\geq} x_1\,t_2' {\geq} x_2.~ (t_1',t_2') {\not\in} R 
\end{array}
\]

\begin{theorem}\label{thm:main-rtrxp-rtep}
Assuming a determinate source language and a determinate and input
total small-step semantics for the target language,
$\rtrxp{} \Rightarrow \rtep$.\Coqed
\end{theorem}

\ifsooner
\dg{If this is the only use of relaXed safety in this paper, I would
  be really inclined to leave it out of the paper and be content with
  writing \rtrtp implies \rtep. I really don't see enough discussion
  to justify introducing a nuanced notion so late in the paper.}
\ch{I'm a bit reluctant to water down one of our most interesting
  theorems. For once, it justifies a big addition to Figure 1}
\fi






In the other direction, we adapt an existing
counterexample~\cite{PatrignaniG17} to show that \rtep (and, hence, for
determinate languages also \oep) does \emph{not} imply \rsp or any of
the criteria above it in \autoref{fig:order}.
%
Fundamentally, \rtep only requires preserving \emph{equivalence} of behavior.
%
%
Consequently, an \rtep compiler can insert code that violates any
security property, as long as it doesn't alter these
equivalences~\cite{PatrignaniG17}.
Worse, even when the \rtep compiler is also required to be correct
(\IE \tp, \scc, and \ccc from \autoref{sec:rtp}),
the compiled program only needs to properly deal with
interactions with target contexts that
behave like source ones, and can behave insecurely when
interacting with target contexts that have no source equivalent.
%


\begin{theorem}\label{thm:rtep-useless}
There exists a compiler between two deterministic
languages that satisfies \rtep, \tp, \scc, and \ccc,
but that does not satisfy \rsp.
\end{theorem}
\begin{proof}
Consider a source language where a partial program receives a natural
number or boolean from the context, and produces a number output,
which is the only event. We compile to a restricted language that only
has numbers by mapping booleans \src{true} and \src{false} to \trg{0}
and \trg{1} respectively.
The compiler's only interesting aspect is that it translates a source
function $\src{P} = \src{f(x{:}\Bools) \mapsto e}$ that inputs
booleans to $\trg{\comptd{\src{P}}} =$ \scalebox{.95}{%
  $\trg{f(x{:}\Natt) {\mapsto} \ifte{x {<}
      2}{\comptd{\src{e}}}{\ifte{x {<} 3}{f(x)}{42}}}$}.
%
The compiled function checks if its input is a valid boolean (\trg{0}
or \trg{1}). If so, it executes $\comptd{\src{e}}$.  Otherwise, it
behaves insecurely, silently diverging on input $2$ and outputting
$42$ on inputs $3$ or more. This compiler does not satisfy \rsp since
the source program $\src{f(x{:}\Bools) {\mapsto} 0}$ robustly
satisfies the safety property ``never output 42'', but the program's
compilation does not.

On the other hand, it is easy to see that this compiler is correct
since a compiled program behaves exactly like its source counterpart
on correct inputs. It is also easily seen to satisfy \rtep, since the
additional behaviors added by the compiler (silently diverging on
input $2$ and outputting $42$ on inputs $3$ or more) are independent
of the source code (they only depend on the type), so these cannot be
used by any target context to distinguish two compiled programs.

In the appendix\ifcamera\else (\autoref{app:sec:rtep-useless})\fi, we use the same
counterexample compilation chain to also show that \rtep does not
imply the robust preservation of (our variant of) liveness properties.
We also use a simple extension of this compilation chain to show that \rtep
does not imply \rtinip either. The idea is similar: we add a secret
external input to the languages and when receiving an out of bounds
argument the compiled code simply leaks the secret input, which breaks
\rtinip, but not \rtep. 
\qedhere

%
\end{proof}

\section{Proof Techniques for \rrhp and \rfrxp}
\label{sec:example}

This section demonstrates that the criteria we introduce can be proved
by adapting existing back-translation techniques.
%
We introduce a statically typed source language and a similar
dynamically typed target one (\autoref{sec:src-trg-instances}), as
well as a simple translation between the two
(\autoref{sec:comp-instances}).
We then describe the essence of two very different secure compilation
proofs for this compilation chain, both based on techniques originally
developed for showing fully abstract compilation.
The first proof shows (a typed variant of)
\rrhp (\autoref{sec:instance-top-lattice}), the
strongest criterion from \autoref{fig:order}, using a {\em
  context-based} back-translation, which provides a ``universal
embedding'' of a target context into a source context~\cite{NewBA16}.
The second proof shows a slightly weaker criterion, {\em Robust
  Finite-relational relaXed safety Preservation} (\rfrxp;
\autoref{sec:instance-tracebased}), but which is still very useful, as
it implies robust preservation of arbitrary safety and hypersafety
properties as well as \rtep.
%
%
This second proof relies on a {\em trace-based}
back-translation~\cite{PatrignaniASJCP15, JeffreyR05}, extended to
produce a context from {\em a finite set} of finite execution prefixes.
These finiteness restrictions are offset by a more generic proof
technique that only depends on the context-program interaction (\EG
calls and returns), while ignoring all other language details.
For space reasons, we leave the details of the proofs for\ifcamera{}
the appendix\else{} \autoref{app:sec:proofs-appendix}\fi.

\ifsooner
\jt{ And mention that this criterion still imply RTEP/ RTINIP, etc. }
\jt{ Should I write about that? I can do it but I'm afraid space is going
  to be an issue. }
\ch{said that in the intro, so enough for now.}
\fi

\subsection{Source and Target Languages}\label{sec:src-trg-instances}
The two languages we consider are simple first-order languages with
named procedures and boolean and natural values.
The source language \Lt is typed while the target language \Ld is untyped.
A program in either language is a collection of function definitions,
each function body is a pure expression that can perform comparison
and natural operations ($\oplus$), conditional branching, recursive calls, and use let-in
bindings.
Expressions can also read naturals from the environment and write naturals
to the environment, both of which generate trace events.
\Ld has all the features of \Lt and adds a primitive \trg{e \checkty{\tau}} 
to dynamically check whether an expression $\trg{e}$ has type $\trg{\tau}$.
A context \ctxc{} can call functions and perform general computation
on the returned values, but it cannot directly generate {\com{\readexp}}
and {\com{\writeexp{e}}} events, as those are security-sensitive. 
Since contexts are single expressions, we disallow callbacks from the
program to the context: thus calls go from context to program, and
returns from program to context.
\begin{align*}
  \mi{Programs}~\src{P} \bnfdef&\ \src{\OB{I};\OB{F}}
  \qquad\qquad\quad
  \mi{Contexts}~\ctxs{} \bnfdef\ \src{e}
  \\
  \mi{Types}~\src{\tau} \bnfdef&\ \Bools \mid \Nats
  \qquad
  \mi{Interfaces}~\src{I} \bnfdef\ \src{f:\tau\to\tau}
  \\
  \mi{Functions}~\src{F} \bnfdef&\ \src{f(x:\tau):\tau \mapsto \retpaper e}
  \\
  \mi{Expressions}~\src{e} \bnfdef&\ \src{x} \mid \trues \mid \falses \mid \src{n}\in\mb{N} \mid \src{e \op e} \mid \src{e \geq e} 
  \\
  \mid&\ \src{\letin{x:\tau}{e}{e}} \mid \src{\ifte{e}{e}{e}} 
  \\
  \mid&\ \src{\call{f}~e} \mid \src{\readexp} \mid \src{\writeexp{e}} \mid \fails 
  \\
  \mi{Programs}~\trg{P} \bnfdef&\ \trg{\OB{I};\OB{F}}
  \qquad\qquad\quad\ \ 
  \mi{Contexts}~\ctxt \bnfdef\ \trg{e}
  \\
  \mi{Types}~\trg{\tau} \bnfdef&\ \Boolt \mid \Natt 
  \qquad
  \mi{Interfaces}~\trg{I} \bnfdef\ \trg{f}
  \\
  \mi{Functions}~\trg{F} \bnfdef&\ \trg{f(x) \mapsto \retpaper e}
  \\
  \mi{Expressions}~\trg{e} \bnfdef&\ \trg{x} \mid \truet \mid \falset \mid \trg{n}\in\mb{N} \mid \trg{e \op e} \mid \trg{e \geq e} 
  \\
    \mid&\ \trg{\letin{x}{e}{e}} \mid \trg{\ifte{e}{e}{e}} 
  \\
    \mid&\ \trg{\call{f}~e} \mid \trg{\readexp} \mid \trg{\writeexp{e}} \mid \failt \mid \trg{e \checkty{\tau}}
  \\
  \mi{Labels}~\bl{\lambda} \bnfdef&\ \bl{\epsilon} \mid \bl{\alpha}
  \\
  \mi{Actions}~\bl{\alpha} \bnfdef&\ \bl{\rdl{n}} \mid \bl{\wrl{n}} \mid \terc \mid \divc \mid \failactc
\end{align*}

\iffull
\Lt is typed using a simple type system that, given its simplicity, is
elided.
\fi

Each language has a standard small-step operational
semantics 
(omitted for brevity), as well as
a big-step trace semantics ($\Omega\sem{\OB{\alpha}}$, as in previous
sections). 
%
The initial state of a program $P$ plugged into a context $\ctxc{}$ is
denoted as 
$P\triangleright\ctxc$ and the behavior of such a program is
the set of traces that can be produced by the semantics:
\begin{align*}
  \behavc{\ctxhc{P}}
  &= \myset{ \bl{\OB{\alpha}} }{ P\triangleright\ctxc \sem\OB{\alpha} }
\end{align*}  

\subsection{Compiler}\label{sec:comp-instances}
The compiler \comptd{\cdot} takes programs of \Lt and generates programs of
\Ld,
by replacing static type annotations with dynamic type checks of
function arguments upon function invocation:

\begin{align*}
  \comptd{\src{I_1,\cdots,I_m;F_1,\cdots,F_n}} 
    =&\ 
    \trg{\comptd{\src{I_1}},\cdots,\comptd{I_m};\comptd{F_1},\cdots,\comptd{F_n}}
  \\
  \comptd{\src{f:\tau\to\tau'}} 
    =&\ 
    \trg{f}
  \\
  \comptd{
    \begin{aligned}
      &
      \src{f(x:\tau):\tau'\mapsto}
      \\
      &\ 
      \src{\retpaper e}    
    \end{aligned}
  }
  =&\
        \trg{
          \left(\begin{aligned}
            \trg{f(x)\mapsto} 
            &\trg{\retpaper\ }
            \iftet{
              \trg{x \checkty{\comptd{\tau}}}
              \\&\ 
            }{
              \comptd{e}
            }{\failt}
        \end{aligned}\right)} 
\end{align*}
\begin{align*}
  \comptd{\Nats} 
    =&\ 
    \Natt
  &
  \comptd{\Bools} 
    =&\ 
    \Boolt
        \\
  \comptd{\trues} 
    =&\ 
    \truet
  &
  \comptd{\falses} 
    =&\ 
    \falset
  \\
  \comptd{\src{n}} 
    =&\ 
    \trg{n}
  &
  \comptd{\src{x}} 
    =&\ 
    \trg{x}
  \\
  \comptd{\src{e\op e'}} 
    =&\ 
    \trg{\comptd{e}\op \comptd{e'}}
  &
  \comptd{\src{e\geq e'}} 
    =&\ 
    \trg{\comptd{e}\geq \comptd{e'}}
  \\
  \comptd{\src{\readexp}} 
    =&\ 
    \trg{\readexp}
  &
  \comptd{\src{\writeexp{e}}} 
    =&\ 
    \trg{\writeexp{\comptd{e}}}
        \\
  \comptd{\src{\call{f}~e}} 
    =&\ 
    \trg{\call{f}~\comptd{e}}
  \\
  \comptd{
    \begin{aligned}
      &
      \letins{\src{x:\tau}}{\src{e}
      \\
      &\ }{\src{e'}}
    \end{aligned}
  } 
    =&\ 
    \trg{
      \begin{aligned}
        &
          \letint{\trg{x}}{\comptd{e}
        \\
        &\
          }{\comptd{e'}}
      \end{aligned}
    }
  &
  \comptd{
    \begin{aligned}
      &
        \iftes{\src{e}}{\src{e'}
      \\
      &\
        }{\src{e''}}
    \end{aligned}
  }
    =&\ 
    \trg{
      \begin{aligned}
        &
          \iftet{\comptd{e}}{\comptd{e'}
        \\
        &\
        }{\comptd{e''}}
      \end{aligned}
    }
\end{align*}

\subsection{Proof of \iffull Robust Relational Hyperproperty
  Preservation \else \rrhp\fi by Context-Based Back-Translation}
\label{sec:instance-top-lattice}
To prove that \comptd{\cdot} attains
\iffull  Robust Relational Hyperproperty Preservation \else \rrhp \fi,
we need a way to back-translate target contexts into source contexts.
To this end we use a universal embedding, a technique previously proposed
for proving fully abstract compilation~\cite{NewBA16}.
The back-translation 
needs to generate a source context that respects source-level constraints; in this case, the resulting source context must be well-typed.
To ensure this, we use \Nats as an {\em universal back-translation
  type} in the produced source contexts.
The intuition of the back-translation is that it will encode \truet as \src{0}, \falset as \src{1} and an arbitrary natural number \trg{n} as \src{n+2}.
Based on this encoding, we translate values between
regular source types and the back-translation type.
Specifically, we define the following shorthand for the back-translation:
\src{\inject{\tau} (e)} takes an expression \src{e} of type \src{\tau} and returns an expression of back-translation type; \src{\extract{\tau} (e)} takes an expression \src{e} of the back-translation type and returns an expression of type \src{\tau}.
\begin{align*}
  \src{\inject{\Nats} (e)}
    =&\
    \src{e+2}
  \\
  \src{\inject{\Bools} (e)}
    =&\
    \src{\ifte{e}{1}{0}}
  \\
  \src{\extract{\Nats} (e)}
  =&\
        \src{
          \left(\begin{aligned}
              \letins{\src{x}}{\src{e}}{
                \iftes{
                  \src{x\geq 2}
                }{
                  \src{x-2}
                }{
                  \fails
                }
            }
          \end{aligned}\right)
        }
  \\
  \src{\extract{\Bools} (e)}
  =&\
        \src{
          \left(\begin{aligned}
           &\letins{\src{x}}{\src{e}}
            {\iftes{\src{x \geq 2}}{\fails\\&\ }
              {\iftes{\src{x+1\geq 2}}{\trues}{\falses}}}
        \end{aligned}\right)}
\end{align*}
\src{\inject{\tau} (e)} never incurs runtime errors, but \src{\extract{\tau} (e)} may.
This mimics the ability of target contexts to write ill-typed code (\EG \trg{3+\truet}) which we must be able to back-translate and whose semantics we must preserve (see \Cref{ex:backtr}).

Concretely, the back-translation is defined inductively on the
structure of target contexts:
\iffull
For conciseness we omit the list of function interfaces
\src{I} needed for the \trg{\call{\cdot}} case (i.e., what the
back-translated context links against).
\jt{Try to clarify this case a  bit.}
\fi
\begin{align*}
  \backtrdt{\truet} 
    =&\ \src{1}
  &
  \backtrdt{\falset} 
    =&\ \src{0}
  &
  \backtrdt{\trg{n}} 
    =&\ \src{n+2}
  &
  \backtrdt{\trg{x}} 
    =&\ \src{x}
\end{align*}
\vspace*{-1.6em}
\begin{align*}
  \backtrdt{\trg{e \geq e'}} 
    =&\ \src{
      \begin{aligned}[t]
        &
        \letins{\src{x1}:\Nats}{\extract{\Nats}(\backtrdt{\trg{e}})
        \\
        &\
        }{\letins{\src{x2}:\Nats}{\extract{\Nats}(\backtrdt{\trg{e'}})
        \\
        &\ \ 
        }{\inject{\Bools}(\src{x1\geq x2})}}
      \end{aligned}
    }
  \\
  \backtrdt{\trg{e \op e'}} 
    =&\ \src{
      \begin{aligned}[t]
        &
        \letins{\src{x1}:\Nats}{\extract{\Nats}(\backtrdt{\trg{e}})
        \\
        &\
        }{\letins{\src{x2}:\Nats}{\extract{\Nats}(\backtrdt{\trg{e'}})
        \\
        &\ \ 
        }{\inject{\Nats}(\src{x1\op x2})}}
      \end{aligned}
    }
  \\
  \backtrdt{\trg{
      \letint{\trg{x}}{\trg{e}
      }{\trg{e'}}
  }} 
    =&\ \src{
        \letins{\src{x:\Nats}}{\backtrdt{\trg{e}}
        }{\backtrdt{\trg{e'}}}
      }
  \\
  \backtrdt{\trg{
    \left(\begin{aligned}
      &
      \iftet{\trg{e}}{
      \\
      &\
      \trg{e'}}{\trg{e''}}
    \end{aligned}\right)
    }} 
    =&\ \src{
        \iftes{\extract{\Bools}(\backtrdt{\trg{e}})
        }{
        \backtrdt{\trg{e'}}}{\backtrdt{\trg{e''}}}
    }
\\
  \backtrdt{\trg{e \checkty{\Boolt}}} =&\ \src{
          \letins{\src{x:\Nats}}{\backtrdt{\trg{e}}
           }{\src{\ifte{x \geq 2}{0}{1}}}
  }
  \\
  \backtrdt{\trg{e \checkty{\Natt}}} =&\ \src{
          \begin{aligned}
          &
          \letins{\src{x:\Nats}}{\backtrdt{\trg{e}}
         }{\src{\ifte{x \geq 2}{1}{0}}}
          \end{aligned}}
          \\
    \backtrdt{\trg{\call{f}~e}} 
    =&\ 
    \src{\inject{\tau'}(\call{f}~\extract{\tau}(\backtrdt{\trg{e}}))}
  \\
    &\
    \text{ if }\src{f:\tau\to\tau'}\in\src{\OB{I}}
  \\
  \backtrdt{\failt}
    =&\
    \fails 
\end{align*}

\begin{example}[Back-Translation]\label{ex:backtr}
  Through the back-translation of two simple target contexts we explain why \backtrdt{\cdot} is correct and why it needs \inject{\cdot} and \extract{\cdot}.
\iflater
  \rb{It would be nice to have a name for those two that we could use here.
      See question for Marco above.}
\fi

  Consider the context \trg{\ctxt_1 = 3 * 5}, which reduces to
        \trg{15} irrespective of the program it links against.
  The back-translation must intuitively ensure that \backtrdt{\ctxt_1} reduces to \src{17}, which is the back-translation of \trg{15}.
  If we unfold the definition of \backtrdt{\ctxt_1} we have the following (given that \backtrdt{3}=\src{5} and \backtrdt{5}=\src{7}):
  \begin{align*}
    \src{
      \begin{aligned}[t]
        &
        \letins{
          \src{x1}:\Nats
        }{
          \extract{\Nats}(5)
        \\
        &\
        }{
          \letins{
            \src{x2}:\Nats
          }{
            \extract{\Nats}(7)
          }{
            \inject{\Nats}(\src{x1* x2})
          }
        }
      \end{aligned}
    }
  \end{align*}
  By examining the code of \extract{\Nats} we see that in both cases it will just perform a subtraction by 2, turning \src{5} and \src{7} respectively into \src{3} and \src{5}.
  So after some reduction steps we arrive at the following term: \src{\inject{\Nats}(\src{3* 5})}.
  The inner multiplication then returns \src{15} and its injection returns \src{17}, which is also the result of \backtrdt{\trg{15}}.

  Let us now consider a different context, \trg{\ctxt_2 = \falset + 3}.
  We know that no matter what program links against it, it will reduce to \failt.
  Its statically well-typed back-translation is:
  \begin{align*}
    \src{
      \begin{aligned}[t]
        &
        \letins{
          \src{x1}:\Nats
        }{
          \extract{\Nats}(0)
        \\&\ 
        }{
          \letins{\src{x2}:\Nats}{\extract{\Nats}(7)}{\inject{\Nats}(\src{x1* x2})}
        }
      \end{aligned}
    }
  \end{align*}
  By looking at its code we can see that the execution of \src{\extract{\Nats}(0)} will indeed result in \fails, which is what we want and expect, as that is precisely the back-translation of \failt.
\end{example}


The \rrhp proof for this compilation chain uses a simple logical
relation that includes cases for both terms of source type
(intuitively used for compiler correctness) and for terms of
back-translation type~\cite{NewBA16,DevriesePP16}.

\subsection{Proof of \iffull Robust Finite-Relational Relaxed Safety
  Preservation\else \rfrxp \fi by Trace-Based Back-Translation}
\label{sec:instance-tracebased}


Proving that this simple compilation chain attains \pf{\rfrxp} does
not require back-translating a target context, as we only need to
build a source context that can reproduce a finite set of finite trace
prefixes, but that is not necessarily equivalent to the original
target context.  We describe this back-translation on an example
leaving again details to \ifcamera the online appendix\else
\autoref{app:sec:proofs-appendix}\fi.
\begin{example}[Back-Translation of Traces]\label{ex:backtr_traces}

  Consider the following two programs\iffull
    (whose interfaces are omitted for brevity)\fi:
  \begin{align*}
    \src{P_1} =&\ (\src{f(x{:}\Nats):\Nats \mapsto \retpaper x},
    \src{g(x{:}\Nats):\Bools \mapsto \retpaper \trues}) 
    \\
    \src{P_2} =&\ (\src{f(x{:}\Nats):\Nats \mapsto \retpaper \readexp},
    \src{g(x{:}\Nats):\Bools \mapsto \retpaper \trues})
  \end{align*}

  Their compiled counterparts are almost identical, with the only addition of dynamic type checks on function arguments:
  \begin{align*}
    \comptd{P_1} =&\ \trg{f(x) \mapsto \retpaper (\ifte{x \checkty{\Natt}}{x}{\failt})}, \\
    &\trg{g(x) \mapsto \retpaper (\ifte{x \checkty{\Natt}}{\truet}{\failt})} \\
    \comptd{P_2} =&\ \trg{f(x) \mapsto \retpaper (\ifte{x \checkty{\Natt}}{\readexp}{\failt})}, \\
    &\trg{g(x) \mapsto \retpaper (\ifte{x\checkty{\Natt}}{\truet}{\failt})}
  \end{align*}

\ifsooner
\ch{The ret syntax is silly}
\fi

  Now, consider the following target context:
  \begin{align*}
  \ctxt 
  = \trg{
    \begin{aligned}[t]
      &
      \letint{
        \trg{x1}
      }{
        \trg{\call{f}~5}
      \\&\ 
      }{
        \iftet{
          \trg{x1 \geq 5}
        }{
          \trg{\call{g}~(x1)}
        }{
          \trg{\call{g}~(\falset)}
        }
      }
    \end{aligned}
    }	
  \end{align*}

  The two programs plugged into this context can generate (at least) the following traces (where \(\termc\) indicates termination and \(\failactc\) indicates failure):
  \begin{align*}
    \trg{\ctxht{\comptd{P_1}}} &\sem \termc
    &
    \trg{\ctxht{\comptd{P_2}}} &\sem \rdl{5}; \termc
    &
    \trg{\ctxht{\comptd{P_2}}} &\sem \rdl{0}; \failactc
  \end{align*}
  In the execution of \trg{\ctxht{\comptd{P_1}}},
  the program executes completely and terminates,
  producing no side effects.
  In the first execution of \trg{\ctxht{\comptd{P_2}}}, the program reads \(5\),
  and the {\em then} branch of the context's conditional is executed.
  In the second execution of \trg{\ctxht{\comptd{P_2}}}, the program reads \(0\),
  the {\em else} branch of the context's conditional is executed
  and the program fails in \(\trg{g}\) after detecting a type error.

  \MP{ some shortening done below }%
  These traces alone are not enough to construct a source context
  since 
  they do not record information about the control flow of
  program executions, 
  specifically on which function produces which input or
  output.
  To recover this information 
  %
  we enrich execution prefixes with
  information about calls (from context to program) and returns (from program to context).
  The enriched rules on
  calls and returns now generate events to model these control flows.
  If a call or
  return occurs internally within the program, no trace event
  is generated 
  since they are 
  not relevant for back-translating the context.
%
  The revised semantics is almost identical to the original, and allows exactly
  the same program executions, only producing more informative traces.
  Hence, the original execution can be enriched in a valid way for the new semantics.
\ifsooner
  \rb{I know what you want to say here, but this reads hazily.
      What I understand here is this: the new semantics allows exactly the same
      set of program executions, but producing more informative traces, so any
      execution can be enriched while remaining valid.}
  \jt{Rephrased. Is it better now?}
\fi
  \begin{align*}
	\mi{Labels}~\bl{\lambda} \bnfdef&\ \dots \mid \bl{\beta}
  &
        \mi{Interactions}~\bl{\beta} \bnfdef&\ \bl{\clpaper{f}{v}} \mid \bl{\rtpaper{v}}
  \end{align*}

  The traces produced by the compiled programs plugged into the context become:
  \begin{align*}
    \trg{\ctxht{\comptd{P_1}}} &\sem \clpaper{f}{5}; \phantom{\rdl{5};} \rtpaper{5}; \clpaper{g}{5}; \rtpaper{\truec}; \termc\\
    \trg{\ctxht{\comptd{P_2}}} &\sem \clpaper{f}{5}; \rdl{5}; \rtpaper{5}; \clpaper{g}{5}; \rtpaper{\truec}; \termc\\
    \trg{\ctxht{\comptd{P_2}}} &\sem \clpaper{f}{5}; \rdl{0}; \rtpaper{0}; \clpaper{g}{\falsec}; \failactc
  \end{align*}

  In our languages, reads and writes can only be performed by
  programs, while the context only performs a sequence of calls to the
  program, possibly performing some computation and branching on return
  values.
  Thus, the role of the back-translated source is to perform
  the appropriate calls to the program, depending of the values returned.
  The inner workings of the programs, that is inputs, outputs, and internal
  calls and returns, are not a concern of the back-translation and are
  obtained through compiler correctness.
%
  %
\ifsooner
  \rb{The 'context' description reads a little vaguely. The last sentence
      could also use some additional exposition.}
  \jt{Tried rewriting a bit.}
\fi
  Furthermore, the context is shared by all executions, but each
  execution has its own program. Hence, since I/O occurs only in the
  program, the only source of variation among all executions come from
  the program.

  From this, one can conclude that the context is a deterministic
  expression, calling the program, and branching on the returned
  values. This can be seen in the way traces are organized: ignoring
  the I/O, the traces form a tree (\autoref{fig:backtr-trace}, on the
  left).
  This tree can be translated to a source context using nested
  conditionals as depicted below (\autoref{fig:backtr-trace}, on the
  right, dotted lines indicated what the back-translation generates
  for each action in the tree).
  When additional branches are missing (e.g., there is no third
  trace that analyzes the first return or no second trace that
  analyses the second return on the left execution), the
  back-translation inserts \fails in the code -- they are dead code
  branches (marked with a **).

\ifsooner
  \rb{On non-determinism: partly true, but also because I/O can only occur
      in the program. Assuming those are enriched traces. After this, refine
      the intuition a bit more.}
  \jt{Rewritten.}
\fi
 
  \myfig{\small
  \begin{center}
  \begin{tikzpicture}[remember picture]
  \node (root) {$\clpaper{f}{5}$}
  	child {
  		node[xshift = -1em,yshift = 1.5em] (1l) {$\rtpaper{5}$}
  		child{
  			node[yshift = 1.5em] (2l) {$\clpaper{g}{5}$}
  			child{
  				node[yshift = 1.5em] (3l) {$\rtpaper{\truec}$}
  				child{
  					node[yshift = 2.5em] (4l) {$\termc$}
  				}
  			}
  		}
  	}
  	child {
  		node[align=center,xshift = 1em,yshift = 1.5em] (1r) {\\ $\rtpaper{0}$}
  		child{
  			node[yshift = 1.5em] (2r) {$\clpaper{g}{\falsec}$}
  			child{
  				node[yshift = 2.5em] (3r) {$\failact$}
  			}
  		}
  	};

  	\node[align = left, right of= root, xshift = 11em, yshift= -5em] (code){ 
		\src{
  			\begin{aligned}[t]
  				&
  				\letins{\src{x}}{\src{\call{f}~{5}\tikz\node(croot){};}
  				\\ 
  				&\
  				}{
  				\begin{aligned}[t]
  					&
  					\iftes{\src{x == 5} \tikz\node(c1l){};
  					\\ 
  					&\ 
  					}{
  					\letins{\src{y}}{\src{\call{g}~{5}\tikz\node(c2l){};}}{
  					\\ 
  					&\ \ \qquad
            \iftes{\src{y==\trues}\tikz\node(c3l){};}{\src{0}\tikz\node(c4l){};
            \\
            &\ \ \qquad \ }{\fails\tikz\node(cextra2){};}
  					}
  					\\ 
  					&\
  					}{
	  				\begin{aligned}[t]
	  					&
	  					\iftes{\src{x==0}\tikz\node(c1r){}; 
	  					\\ 
	  					&\ 
	  					}{
	  					\fails\tikz\node(c2r){};
	  					}{\fails\tikz\node(cextra){};}	
	  				\end{aligned}
  					}
  				\end{aligned}
  				}
  			\end{aligned}
  		}
  	};
  	\draw[dotted] ([xshift=-1em]root.0) 
    -| ([xshift=-1em]croot.90);
  	\draw[dotted] ([xshift=1em]1l.90) |- ([yshift = .8em,xshift = 1em]c1l.0) |- ([xshift=-1em]c1l.0);
  	\draw[dotted] ([xshift=1em]2l.90) |- ([yshift = 1.5em,xshift = 2em]2l.90) -| ([xshift=-1em]c2l.90);
  	\draw[dotted] ([xshift=1em]3l.-90) |- ([yshift = -6em,xshift=-1em]c3l.-90) -| ([xshift=-1em]c3l.-90);
     \draw[dotted] (4l.0)  -| ([xshift=-1em]c4l.0); 

  	\draw[dotted] ([xshift=-1em]1r.0) -| ([xshift=-9em, yshift = .5em]c1r.90) -| ([xshift=-1em]c1r.90);
  	\draw[dotted] ([xshift=2em]2r.-90) |- ([yshift = .5em,xshift=-1em]c2r.90) -| ([xshift=-1em]c2r.90);
  	\draw[dotted] ([xshift=-.5em]3r.0) |- ([yshift = .5em,xshift=-1em]c2r.90) -| ([xshift=-1em]c2r.90);

  	\node[right of = cextra2] (cdest) {**};
  	\draw[dotted] ([xshift=-1.5em]cextra.90) -- (cdest);
  	\draw[dotted] ([xshift=-.5em]cextra2.0) -- (cdest);
.
  \end{tikzpicture}
  \end{center}
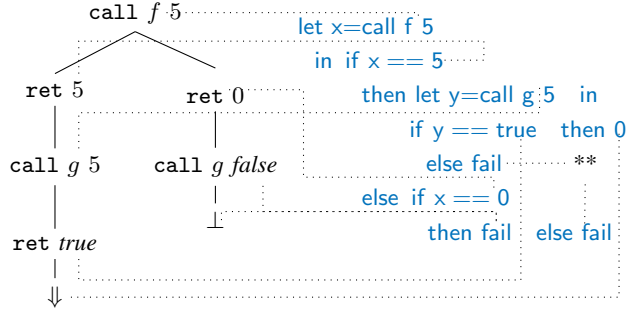
  }{backtr-trace}{Example of a back-translation of traces.}  

  To prove \rfrxp we show correctness of the back-translation, which
  ensures that the back-translated source context produces exactly the
  original non-informative traces.
  This is, however, not completely true of informative traces (that
  track calls and returns).
  Since calling \(\src{g}\) with a boolean is ill-typed, our
  back-translation shifts the failure from the program to the context,
  so the picture links $\clpaper{g}{\falsec}$ action to a
  \fails. The call is never executed at the source level.
\end{example}





\section{Related Work}
\label{sec:related}





\paragraph{Full Abstraction}
was originally used as a criterion for secure compilation in the
seminal work of \citet{Abadi99} and has since received a lot of
attention~\cite{MarcosSurvey}.
\citet{Abadi99} and, later, \citet{Kennedy06} identified
failures of full abstraction in the Java to JVM and C\# to CIL compilers, some of
which were fixed, but also others for which fixing was deemed
too costly compared to the perceived practical security gain.
%
\citet{AbadiFG02} proved full abstraction of secure channel
implementations using cryptography, but to prevent network traffic
attacks they had to introduce noise in their translation, which in
practice would consume network bandwidth.
%
Ahmed~\ETAL~\citeFull{AhmedB11, Ahmed15, NewBA16}{AhmedB08} proved the
full abstraction of type-preserving compiler passes for simple
functional languages.
%
\citet{AbadiP12} and
\citet{JagadeesanPRR11} expressed the protection
provided by address space layout randomization as a probabilistic
variant of full abstraction.
\citet{FournetSCDSL13} devised a fully abstract compiler from a subset
of ML to JavaScript.
Patrignani~\ETAL~\citeFull{PatrignaniASJCP15}{LarmuseauPC15,PatrignaniDP16}
studied fully abstract compilation to machine code, starting from
single modules written in simple, idealized object-oriented and
functional languages and targeting a hardware isolation mechanism
similar to Intel's SGX~\cite{sgx}.

Until recently, most formal work on secure interoperability with
linked target code was focused only on fully abstract compilation.
The goal of our work is to explore a
diverse set of secure compilation criteria, some of them formally
stronger than (the interesting direction of) full abstraction at least
in various determinate settings, and thus potentially harder to
achieve and prove, some of them apparently easier to achieve and prove than full
abstraction, but most of them not directly comparable to full abstraction.
This exploration clarifies the trade-off between security guarantees
and efficient enforcement for secure compilation:
On one extreme, \rtp robustly preserves only trace properties, but
does not require enforcing confidentiality; on the other extreme,
robustly preserving relational properties gives very strong
guarantees, but requires enforcing that both the private data and the
code of a program remain hidden from the context, which is often much
harder to achieve.
The best criterion to apply depends on the application domain, but our
framework can be used to address interesting design questions such as
the following:
{\em (1)~What secure compilation criterion, when violated, would the
developers of practical compilers be willing to fix at least in
principle?}
The work of \citet{Kennedy06} indicates that fully abstract
compilation is not such a good answer to this question, and we wonder
whether \rtp or \rhp could be better answers.
%
{\em (2)~What secure compilation criterion would the
translations of \citet{AbadiFG02} still satisfy if they did not
introduce (inefficient) noise to prevent network traffic analysis?}
\citet{AbadiFG02} explicitly leave this problem open in their
paper, and we believe one answer could be \rtp, since it does not
require preserving any confidentiality.

We also hope that our work can help eliminate common misconceptions
about the security guarantees provided (or not) by full abstraction.
For instance, \citet{FournetSCDSL13} illustrate the difficulty of
achieving security for JavaScript code using a simple example policy
that (1) restricts message sending to only correct URLs and (2)
prevents leaking certain secret data.
Then they go on to prove full abstraction apparently in the hope of
preventing contexts from violating such policies.
However, part (1) of this policy is a safety property and part (2) is
hypersafety, and as we showed in \autoref{sec:context-composition}
fully abstract compilation does not imply the robust preservation of
such properties.
In contrast, proving \rhsp would directly imply this, without putting
any artificial restrictions on code introspection, which are
unnecessarily required by full abstraction.
Unfortunately, this is not the only work in the literature that uses
full abstraction even when it is not the right hammer.


\iffull
Our exploration also forced us to challenge the assumptions and design
decisions of prior work. This is most visible in our attempt to use as
generic and realistic a trace model as possible.
To start, this meant moving away from the standard assumption in
the hyperproperties literature~\cite{ClarksonS10}
that all traces are infinite, and switching instead to a
trace model inspired by CompCert's~\cite{Leroy09} with both
terminating and non-terminating traces, and where non-terminating
traces can be finite but not finitely observable (to account for
silent divergence).
This more realistic model required us to find a class of trace
properties to replace liveness.
%
\ch{Don't think the assumption is properly explained there though}

\ch{There was also discussion about FA traces often being very
  impoverished (only one bit of information), while for us they are
  usually richer -- i.e. the richer they are the more properties we
  can state. Will probably need to use Git history to dig that up.}
\fi


\paragraph{Development of \rsp}
Two pieces of concurrent work have examined more carefully how to
attain and prove one of the weakest of our criteria, \rsp
(\autoref{sec:rsp}).
{}\citet{PatrignaniG18} show \rsp for compilers from simple sequential
and concurrent languages to capabilities~\cite{WatsonWNMACDDGL15}.
They observe that if the source language has a verification system for robust safety and
compilation is limited to verified programs, then \rsp can be
established without directly resorting to back-translation. (This observation has also been made independently by Dave Swasey in private
communication to us.)
%
%
\citet{AbateABEFHLPST18} aim at devising secure compilation
chains for protecting mutually distrustful components written in an
unsafe language like C. They show that by moving away from the full
abstraction variant used in earlier work~\cite{JuglaretHAEP16} to a
variant of our \rsp criterion from \autoref{sec:rsp}, they can support
a more realistic model of dynamic component compromise, while at the
same time obtaining a criterion that is easier to achieve and prove
than full abstraction.

\paragraph{Hypersafety Preservation}
The high-level idea of specifying secure compilation as the
preservation of properties and hyperproperties
goes back to the work of \citet{PatrignaniG17}. However, that work's
technical development is limited to one criterion---the preservation
of finite prefixes of program traces by compilation. Superficially,
this is similar to one of our criteria, \rhsp, but there are several
differences even from \rhsp. First, \citet{PatrignaniG17} do not
consider adversarial contexts explicitly. This might suffice for their
setting of closed reactive programs, where traces are inherently fully
abstract (so considering the adversarial context is irrelevant), but
not in general. Second, they are interested in designing a criterion
that accommodates specific fail-safe like mechanisms for low-level
enforcement, so the preservation of hypersafety properties is not
perfect, and one has to show, for every relevant property, that the
criterion is meaningful. However, \citet{PatrignaniG17} consider
translations of trace symbols induced by compilation, an extension
that would also be interesting for our criteria (\ifanon
as shown for \rtp in \autoref{app:sec:different-traces}
\else\autoref{sec:conclusion}\fi).

\paragraph{Proof techniques}
{}\citet{NewBA16} present a back-translation technique based on a
universal type embedding in the source for the purpose of proving full
abstraction of translations from typed to untyped languages. In 
\autoref{sec:instance-top-lattice} we adapted the same
technique to show \rrhp for a simple translation from a statically
typed to a dynamically typed language with first-order functions and
\iffull input-output\else I/O\fi.
\citet{DevriesePP16} show that even when a precise
universal type does not exist in the source, one can use an
approximate embedding that only works for a certain number of
execution steps.  They illustrate such an approximate back-translation
by proving full abstraction for a compiler from the simply-typed to
the untyped $\lambda$-calculus.

\citet{JeffreyR05\ifallcites, JeffreyR05b\fi} introduced a
``trace-based'' back-translation technique.
They were interested in proving full abstraction for so-called trace
semantics. This technique was then adapted to show
full abstraction of compilation chains to low-level target languages
\citeFull{PatrignaniASJCP15}{PatrignaniC15,PatrignaniDP16,AgtenSJP12}.
In 
\autoref{sec:instance-tracebased}, we showed how
these trace-based techniques can be extended to prove all the criteria
below \rfrxp in \autoref{fig:order}, which includes robust
preservation of safety, of noninterference, and in a determinate
setting also of observational equivalence.

%
While many other proof techniques have been previously
proposed~\citeFull{AbadiFG02, AhmedB11, AbadiP12, FournetSCDSL13,
JagadeesanPRR11}{AhmedB08}, proofs of full abstraction remain
notoriously difficult, even for simple translations, with
apparently simple conjectures surviving for decades before being
finally settled~\cite{DevriesePP18}. It will be interesting to
investigate which existing full abstraction techniques can be
repurposed to show the stronger criteria from \autoref{fig:order}.
%
For instance, it will be interesting to determine the strongest
criterion from \autoref{fig:order} for which an approximate
back-translation~\cite{DevriesePP16} can be used.
%



\paragraph{Source-level verification of robust satisfaction}

While this paper studies the \emph{preservation} of robust properties
in compilation chains, formally verifying that a partial source
program robustly satisfies a specification is a challenging problem
\iffull in itself\else too\fi.
So far, most of the research has focused on techniques for proving
observational equivalence~\citeFull{JeffreyR05,
  DelauneH17}{JeffreyR05b, AbadiBF18, ChevalKR18}\ifsooner\ch{bisimulations too?}\fi{}
or trace equivalence~\cite{BaeldeDH17, ChevalCD13}.
Robust satisfaction of trace properties has been model checked for
systems modeled by nondeterministic Moore machines and properties
specified by branching temporal logic~\cite{KupfermanV99}.
Robust safety, the robust satisfaction of safety properties, was
studied for the analysis of security
protocols~\citeFull{GordonJ04}{Backes:Hritcu:Maffei:11, Backes:Hritcu:Maffei:08bb},
and more recently for compositional verification~\cite{SwaseyGD17}.
Verifying the {\em robust} satisfaction of relational hyperproperties
beyond observational equivalence and trace equivalence seems to be an
open research problem.
For addressing it, one can hopefully take inspiration in extensions of
relational Hoare logic~\cite{benton04relational} for dealing with
cryptographic adversaries represented as procedures parameterized by
oracles~\cite{BartheDGKSS13}.


\paragraph{Other Kinds of Secure Compilation}
In this paper we investigated the various kinds of security guarantees
one can obtain from a compilation chain that protects the compiled
program against linked adversarial low-level code.
While this is an instance of {\em secure compilation}~\cite{dagstuhl-sc-2018},
this emerging area is much broader.
Since there are many ways in which a compilation chain can be ``more
secure'', there are also many different notions of secure compilation,
with different security goals and attacker models.
A class secure compilation chains is aimed at providing a ``safer''
semantics for unsafe low-level languages like C and C++, for instance
ensuring memory safety~\cite{NagarakatteMZ15, cheri_asplos2015, CheckedC}.
Other secure compilation work is targeted at closing down
side-channels: for instance by preserving the secret independence
guarantees of the source code~\cite{BartheGL18}, or making sure that
the code erasing secrets is not simply optimized away by the unaware
compilers~\cite{SimonCA18, BessonDJ18, DSilvaPS15, DengN18}.
Closer to our work is the work on building compartmentalizing
compilation chains~\cite{AbateABEFHLPST18, BessonBDJW19,
  WatsonWNMACDDGL15, GudkaWACDLMNR15} for unsafe languages like C and C++.
In particular, as mentioned above, \citet{AbateABEFHLPST18} have
recently showed how \rsp can be extended to express the security
guarantees obtained by protecting mutually distrustful components
against each other.

\iffull
\paragraph{Relation to provably correct compilers}
%
%
The closest related work is on ``compositional compiler correctness'',
but that work doesn't do security as it assumes that the linked
target-level context is well-behaved, for instance it behaves
like a source-level context.
\ch{If that stays there then just reiterate from S2.1?}
\fi

\section{Conclusion and Future Work}
\label{sec:conclusion}
\label{sec:future}

This paper proposes a foundation for secure interoperability
with linked target code by exploring many different criteria based on
robust property preservation (\autoref{fig:order}).
Yet the road to building {\em practical} secure compilation chains
achieving any of these criteria remains long and challenging.
Even for \rsp, scaling up to realistic programming languages and
efficiently enforcing protection of the compiled program without
restrictions on the linked context is
challenging~\cite{AbateABEFHLPST18, PatrignaniG18}.
For \rthsp the problem is even harder, because one also needs to
protect the secrecy of the program's data, which is especially
challenging in a realistic model in which the context can observe
side-channels like timing.
Here, an {\rtinip}-like property might be the best one can hope for in practice.
%

In this paper we assumed for simplicity that traces are exactly the
same in both the source and target language, and while this assumption
is currently true for other work like CompCert~\cite{Leroy09} as well,
it is a restriction nonetheless. We plan to lift this restriction in
the future.
\ifanon
In the appendix (\autoref{app:sec:different-traces}), we illustrate
how to lift this restriction for the \rtp, and in the future we hope
to generalize all other criteria in a similar way.
\fi

\iflater
\ch{from sets of traces/behaviors to probability distribution
-- seems too open ended and we didn't think that much about it,
also \citet{ClarksonS10} claim that this is already handled by
hyperproperties?}
\fi

\iflater
\ch{Just a thought: would it be conceivable to move from sets of
  traces to trees of traces? I guess that's only relevant for
  internally nondeterministic languages, and we anyway don't care too
  much about those? Still, could it help us relate to observational
  equivalence in the fully no-deterministic setting? How would one
  obtain the internal branching though ... isn't that something
  very intensional? Also have a look at Li-yao et al's CPP paper.}
\fi

\ifanon\else
\begin{acks}\small
We are grateful to
Akram El-Korashy,
Arthur Azevedo de Amorim,
\c{S}tefan Ciob\^{a}c\u{a},
Dominique Devriese,
Guido Mart\'inez,
Marco Stronati,
Dave Swasey,
\'Eric Tanter,
and the anonymous reviewers
for their valuable feedback
and in many cases also for participating in various discussions.
This work was in part supported
by the
\ifieee
\href{https://erc.europa.eu}{\iffull European Research Council\else ERC\fi}
\else
\grantsponsor{1}{European Research Council}{https://erc.europa.eu/}
\fi
under ERC Starting Grant SECOMP (\ifieee 715753\else\grantnum{1}{715753}\fi),
by the German Federal Ministry of
Education and Research (BMBF) through funding for the CISPA-Stanford
Center for Cybersecurity (FKZ: 13N1S0762), and
by DARPA grant SSITH/HOPE (FA8650-15-C-7558).
\end{acks}
\fi

\ifanon\clearpage\fi

\ifcamera\else
\onecolumn
\appendices

\setcounter{theorem}{0}
\renewcommand{\thetheorem}{\thesection.\arabic{theorem}}

\let\oldlabel\label
\renewcommand{\label}[1]{\oldlabel{app:#1}}
\let\oldautoref\autoref
\renewcommand{\autoref}[1]{\oldautoref{app:#1}}
\let\oldCref\Cref
\renewcommand{\Cref}[1]{\oldCref{app:#1}}
\let\oldref\ref
\renewcommand{\ref}[1]{\oldref{app:#1}}
\let\oldnameref\nameref
\renewcommand{\nameref}[1]{\oldnameref{app:#1}}

\newcommand{\LL}{\mathcal{L}}
\theoremstyle{definition}
\newcommand{\nequiv}{\not\equiv}
\renewcommand{\compgen}[1]{\cmp{#1}}

\jt{In some parts of the appendices, there is way too much space between paragraphs (appendix B for instance). I suspected wrong. Seems it has to do with LaTeX trying to fill the page completely. I added some ``clearpage'' to temporarily solve the issue.}
\section{Notations}\label{sec:not}

\newcommand*\circled[1]{\tikz[baseline=(char.base)]{
            \node[shape=circle,draw,inner sep=2pt] (char) {\ensuremath{#1}};}}
\newcommand*\termevent{\circled{\varepsilon}}

We use \src{blue, sans\text{-}serif} font for \src{\mi{source}}
elements, \trg{red, bold} font for \trg{\mi{target}} elements and
\com{\commoncol, italic} for elements common to both languages (to
avoid repeating similar definitions twice).  Thus, \src{P} is a
source-level program, \trg{P} is a target-level program and \com{P} is
generic notation for either a source-level or a target-level program.

\begin{align*}
  &\mi{Whole\ Programs} &W
  \\
  &\mi{Partial\ Programs} &P
  \\
  &\mi{Contexts} &C
  \\
  &\mi{Termination\ Events} &\varepsilon
  \\
  &\mi{Events} &e
  \\
  &\mi{Finite\ Trace\ Prefixes} &m &\triangleq
    \\
    &  \quad \mi{(terminated)}
    && e_1 \cdot \cdots \cdot e_n \termevent
    \\
    &  \quad \mi{(not\ yet\ terminated)}
    && e_1 \cdot \cdots \cdot e_n \circ
  \\
 &\mi{relaXed\ Trace\ Prefixes} &x &\triangleq
    \\
    &  \quad \mi{(terminated)}
    && e_1 \cdot \cdots \cdot e_n \termevent
    \\
    &  \quad \mi{(not\ yet\ terminated)}
    && e_1 \cdot \cdots \cdot e_n \circ
    \\ 
    &  \quad \mi{(silent\ divergence)}
    && e_1 \cdot \cdots \cdot e_n \sdiv     
  \\
  &\mi{Traces} &t &\triangleq
    \\
    &  \quad \mi{(program\ termination)}
    && e_1 \cdot \cdots \cdot e_n \termevent
    \\
    &  \quad \mi{(silent\ divergence)}
    && e_1 \cdot \cdots \cdot e_n \sdiv
    \\
    &  \quad \mi{(infinitely\ reactive)}
    && e_1 \cdot \cdots \cdot e_n \cdot \cdots
  \\
  &\mi{Prefix\ relation} &&m \leq t
  \\
  &\mi{The\ set\ of\ all\ traces} &&\com{Trace}
  \\
  &\mi{The\ set\ of\ all\ finite\ trace\ prefixes} &&\com{FinPref}
  \\
  &\mi{The\ set\ of\ all\ relaxed\ trace\ prefixes} &&\com{XPref}
  \\
  &\mi{Semantics\ of\ W} &&W \sem t
  \\
        &\mi{Behavior\ of\ W}
    &\behavc{W} 
    &= \myset{ t }{ W \sem t }
  \\
  &\mi{Set\ with\ elements\ from\ }X &&\quad~2^X
  \\
  &\mi{Set\ of\ size\ }K\mi{\ with\ elements\ from\ }X &&\quad~2_K^X
        \\
  &\mi{Set\ literal} &\set{ x }
    &\triangleq \{ x_1, x_2, \cdots \}
  \\
  &\mi{Property} &\pi&\in 2^\com{Trace}
  \\
  &\mi{Behavior (the\ set\ of\ traces\ of\ a\ program)} &b&\in 2^\com{Trace}
  \\
  &\mi{Hyperproperty} &H&\in 2^{2^{\com{Trace}}}
  \\
  &\mi{Cardinality} &\card{\cdot}
\end{align*}

In addition to trace-based, whole-program semantics,
we also define a finite prefix-based semantics,
\( \com{W} \sem m\) as \( \exists t\geq m\ldotp \com{W}
\sem t\). The notations introduced for finite prefixes (\(\leq\), \(\geq\), \(\sem\), etc.)
are used not only for finite trace prefixes, but also for relaxed trace prefixes.


\clearpage
\section{Safety and Dense Properties with Event-Based Traces}
\label{sec:traces-details}


We start by presenting the CompCert-inspired model for program
execution traces we use in this work (\autoref{sec:traces}).
In this model {\em safety properties} are defined in the standard way
as the trace properties that can be falsified by a finite trace prefix
(\oldautoref{sec:rsp}).
Perhaps more surprisingly, in this trace model the role generally
played by liveness is taken by what we call {\em dense properties},
which we define simply as the trace properties that can only be
falsified by non-terminating traces (\EG a reactive program that runs
forever eventually answers every network request it receives).
Next, to validate the claim that dense properties indeed play \iffull
in our model \fi the same role that liveness plays in previously
proposed trace models~\cite{AlpernS85, LamportS84,
  schneider1997concurrent, PasquaM17}, we prove several related properties
(\autoref{sec:dense-in-depth}), including the fact that every trace
property is the intersection of a safety property and a dense property
(this is our variant of a standard decomposition
result~\cite{AlpernS85}), and the fact that our definition
of dense properties is unique (\autoref{sec:dense-in-depth}).
Finally, we study the robust preservation of dense properties
(\rdp; \autoref{sec:rdp}).

\subsection{Event-Based Trace Model for Safety and Liveness}
\label{sec:traces}



For defining safety and liveness, traces need a bit of structure, and
for this we use a variant of CompCert's {\em realistic} trace
model~\cite{Leroy09}.\footnote{Our trace model is close to that of
  CompCert, but as opposed to CompCert, in this paper we use the word
  ``trace'' for the result of a single program execution and later
  ``behavior'' for the set of all traces of a program
  (\oldautoref{sec:hyperprop}).}
This model is different from the trace
models generally used for studying safety and liveness of reactive
systems~\citeFull{AlpernS85, LamportS84, schneider1997concurrent,
  ClarksonS10}{lamport2002specifying, manna2012temporal}
(\EG in a transition system or a process calculus).
A first important difference is that in CompCert's model, traces are
built from events, not from states.
This is important for efficient compilation, since taking these events
to be relatively coarse-grained gives the compiler more freedom to perform
program optimizations.
For instance, CompCert is inspired by the C programming language
standard and defines the outcome of the program to be a trace of all
I/O and volatile operations it performs, plus an indication of whether
and how it terminates.
%

\ifsooner
\ch{If we added ``final events'' we could make a stronger claim here
  that the CompCert trace model is an instance / variant of what we do.}
\rb{We should be ready for the stronger claim. Main body, here?}
\fi

The {\em events} in our traces are drawn from an arbitrary nonempty
set\iffull (and a few of our results require at least 2 events\ch{this is a detail,
  and unclear how this will change when we have ``final events''})\fi.
Intuitively, traces $t$ are finite or infinite lists of events, where
a finite trace means that the program terminates (possibly with some
related information recording the cause of termination, such as an
exit code) or enters an
unproductive infinite loop after producing all the events in the list.
%
This kind of trace model is natural for usual programming languages
where most programs do indeed terminate
%
%
and is standard for formally correct compilers~\cite{Leroy09, KumarMNO14}.
It is different, however, from the trace model usually considered for
abstract modeling of reactive systems, which considers only infinite
traces~\citeFull{schneider1997concurrent,
  ClarksonS10}{manna2012temporal, lamport2002specifying}
and where a common trick to force all traces to be infinite
is to use stuttering on the final state of an execution
to represent termination~\cite{ClarksonS10}.
In our model, however, events are observable and infinitely repeating the
last event would result in a trace of a non-terminating execution, so
we have to be honest about the fact that terminating executions
produce finite traces.
%
%
Moreover, working with traces of events also means that execution
steps can be silent (i.e., add no events to the trace) and one has to
distinguish termination from silent divergence (a non-terminating execution),
although both of them produce a finite number of events.
So in our model terminating traces are those that end in an explicit
termination event and can thus no longer be extended; all other
traces, whether silently divergent or infinite, are non-terminating.
The proper treatment of program termination and silent divergence 
distinguishes the realistic trace model we use here from previous
theoretical work that extends safety and liveness to finite and
infinite traces~\cite{Rosu12, PasquaM17}.

Using this realistic trace model directly impacts the meaning of
safety, which we try to keep as standard and natural as possible, and
also created the need for a new definition of \emph{dense properties}
to take the place of liveness. 
\ca{I do not completely agree with ``.. impacts the meaning of Safety''.
    I think the meaning of Safety is finitely refutable and  is not changed. }


%


\paragraph{Safety Properties}

The main component of the characterization of safety properties is a definition of
\emph{finite} trace prefixes, which capture the finite observations that can
be made about an execution, for instance by a reference monitor.
We take the stance that a reference monitor can observe that the
program has terminated.
To reflect this, in our trace model finite trace prefixes are lists of
events in which it is observable whether a prefix is terminated and
can no longer be extended, or whether it is not yet terminated and can still
be extended with further events.
Moreover, while termination and silent
divergence are two different terminal trace events, no monitor can
distinguish between the two in finite time, since one cannot tell whether a
program that seems to be looping will {\em eventually} terminate.
Technically, in our model finite trace prefixes $m$ are lists with two
different final constructors: \circled{\varepsilon} for a prefix
terminated with final event $\varepsilon$ (which for instance
distinguishes successful from erroneous termination) and $\circ$ for
not yet terminated prefixes.
In contrast, traces can end either with \circled{\varepsilon}
if the program terminates or with $\sdiv$ if the program silently diverges,
or they can go on infinitely.
The prefix relation $m \leq t$ is defined between a finite prefix $m$
and a trace $t$ according to the intuition above:
$\circled{\varepsilon} \leq \circled{\varepsilon}$,
$\circ \leq t$, and
$e \cdot m' \leq e \cdot t'$ whenever $m' \leq t'$
(where $\cdot$ is concatenation).

The definition of safety properties is then unsurprising
(as already seen in \oldautoref{sec:rsp}):
\[
\ii{Safety} \triangleq \{\pi \in 2^\ii{Trace} ~|~ \forall t \not\in \pi.~
  \exists m \leq t.~ \forall t' \geq m.~ t' \not\in \pi \}\\
  \]

A trace property $\pi$ is \ii{Safety} if, within any trace $t$ that
violates $\pi$, there exists a finite ``bad prefix'' $m$ that can
only be extended to traces $t'$ that also violate $\pi$.
%
\begin{example}
For instance, the trace property
$\pi_{\square\lnot e} = \{ t ~|~ e \not\in t \}$, stating that the bad event
$e$ never occurs in the trace, is \ii{Safety}, since for every trace
$t$ violating $\pi_{\square\lnot e}$ there exists a finite prefix
$m = m' {\cdot} e {\cdot} \circ$ (some prefix $m'$ followed by $e$ and
then by the unfinished prefix symbol $\circ$) that is a prefix of
$t$, and every trace extending $m$ still contains $e$, so it
continues to violate $\pi_{\square\lnot e}$.
\end{example}


\begin{example}
Consider the property $\pi_{\square\lnot \varepsilon} = \{ t ~|~ \forall \varepsilon. \termevent \not\in t \}$
that rejects all terminating traces and accepts all non-terminating traces.
This is a safety property, the justification of which crucially relies on allowing $\termevent$ in the finite
trace prefixes.
For any finite trace $t = e_1 \cdot \ldots \cdot e_n \termevent$ rejected by $\pi_{\square\lnot \varepsilon}$, there
exists a bad prefix $m = e_1 \cdot \ldots \cdot e_n  \termevent$ such that all
extensions of $m$ are also rejected by $\pi_{\square\lnot \varepsilon}$.
This last condition is trivial since the prefix $m$ is terminating (i.e., ends
with $\termevent$) and can thus only be extended to $t$ itself.
\end{example}

\begin{example}
The trace property
$\pi_{\lozenge e}^\ii{term} = \{ t ~|~ t~\ii{terminating} \Rightarrow e \in t \}$
states that in every {\em terminating} trace the event $e$ must eventually
happen. This is also a safety property in our model, since for each terminating
trace $t = e_1 \cdot \ldots \cdot e_n \termevent$ violating $\pi_{\lozenge e}^\ii{term}$
there exists a bad prefix $m = e_1 \cdot \ldots \cdot e_n \termevent$ that can only
be extended to traces that also violate $\pi_{\lozenge e}^\ii{term}$, \IE only to $t$ itself.
\end{example}

Generally speaking, all trace properties (like $\pi_{\square\lnot \varepsilon}$ and
$\pi_{\lozenge e}^\ii{term}$) that only reject terminating traces and therefore
allow all non-terminating traces are safety properties in our model. That is, if
$\forall t~\ii{non\text{-}terminating}.~ t \in \pi$, then $\pi$ is a safety property.
%
Consequently, for any property $\pi$, the derived trace property
$\pi_S = \pi \cup \{t ~|~ t~\ii{non\text{-}terminating} \}$ is a safety property.

\paragraph{Dense Properties}
In our trace model the liveness definition of
\citet{AlpernS85} does not have its intended intuitive
meaning, so instead we focus on the main properties that the Alpern
and Schneider liveness definition satisfies in the infinite state-based trace
model and, in particular, that each trace property can be decomposed as
the intersection of a safety property and a liveness property.
We discovered that in our model the following \iffull surprisingly \fi simple
notion of {\em dense properties} satisfies \iffull all \fi the characterizing
properties of liveness and is, in fact, uniquely determined by these
properties\iffull and the definition of safety above \fi:
%
\[
\ii{Dense} \triangleq \{\pi \in 2^\ii{Trace} ~|~
  \forall t~\ii{terminating}.~ t \in \pi \}
\]
We say that a trace property $\pi$ is \ii{Dense} if it allows all
terminating traces; or, conversely, it can only be violated by
non-terminating traces.
For instance, the property
$\pi_{\lozenge e}^{\lnot\ii{term}} = \{ t ~|~ t~\ii{non\text{-}terminating} \Rightarrow e \in t \}$,
stating that the \iffull good \fi event $e$ will eventually happen along every
non-terminating trace is a dense property, since it accepts all terminating traces.
The property
$\pi_{\square\lnot\sdiv} = \{ t ~|~ \sdiv \not\in t \}
= \{ t ~|~ t~\ii{non\text{-}terminating} \Rightarrow t~\ii{infinite} \}$
stating that the program does not silently diverge is also dense.
Again, more examples are given below:


\begin{example}
The property $\pi_{\square\lozenge e}^{\lnot\ii{term}} =
\{ t ~|~ t~\ii{non\text{-}terminating} \Rightarrow
  t~\ii{infinite} \wedge \forall m \leq t.~ \exists m'.~ m {\cdot} m' {\cdot} e {\leq} t \}$
states that event $e$ happens infinitely often in
any non-terminating trace.
%
Because it allows all terminating traces, it is a dense property.
\end{example}

\begin{example}
The property
$\pi_{\lozenge \varepsilon} = \{ t ~|~ t~\ii{terminating} \} = \{ t ~|~ \exists \varepsilon. \termevent \in t \}$
contains exactly all terminating traces and
rejects all non-terminating
traces. It is therefore the minimal dense property of our trace model.
\end{example}

Trivially, any property becomes dense in our model if we modify it to
accept all terminating traces. That is, given any property $\pi$,
the derived $\pi_L = \pi \cup \{ t ~|~ t~\ii{terminating} \}$ is
dense.

\begin{example}
Take the safety property $\pi_{\square \lnot e} = \{ t ~|~ e \not\in t \}$,
which forbids an event $e$ from appearing in traces.
The modified dense property $\pi_{\square \lnot e}^{\lnot\ii{term}}$
states that event $e$ never occurs along non-terminating traces:
$\pi_{\square \lnot e}^{\lnot\ii{term}} = \{ t ~|~ t~\ii{non\text{-}terminating} \Rightarrow
  e \not\in t \}$.
\end{example}

%

\subsection{Theory of Dense Properties}
\label{sec:dense-in-depth}

We have proved that our definition of dense properties satisfies the main
properties of Alpern and Schneider's related concept of liveness~\cite{AlpernS85},
including its topological characterization, and in particular the following fact.

\begin{theorem}\label{thm:tp-safety-cap-dense}
Any trace property can be decomposed into the intersection of a safety
property and of a dense property (\CoqSymbol):
$\forall \pi. \exists \pi_S \in \ii{Safety}. \exists \pi_D \in \ii{Dense}. \pi = \pi_S \cap \pi_D$.
\end{theorem}
%

\begin{proof}
The proof of this decomposition theorem is in fact very simple in our model.
Given any trace property $\pi$, define
$\pi_S = \pi \cup \{t ~|~ t~\ii{non\text{-}terminating} \}$ and
$\pi_D = \pi \cup \{t ~|~ t~\ii{terminating} \}$.
As discussed above, $\pi_S \in \ii{Safety}$ and $\pi_D \in \ii{Dense}$.
Finally, $\pi_S \cap \pi_L =
(\pi \cup \{t ~|~ t~\ii{non\text{-}terminating} \}) \cap
(\pi \cup \{t ~|~ t~\ii{terminating} \}) = \pi$.
\end{proof}

\begin{example}
In our trace model, the property
$\pi_{\lozenge e} = \{ t ~|~ e \in t \}$ is neither
safety nor dense. However, it can be decomposed as the intersection of
$\pi_{\lozenge e}^\ii{term}$ (a safety property)
and $\pi_{\lozenge e}^{\lnot\ii{term}}$ (a dense property).
\end{example}
Concerning the relation between dense properties
and the liveness definition of \cite{AlpernS85},
the two are in fact equivalent {\em in our model}, but this seems to
be a coincidence and only happens because Alpern and Schneider's
definition completely loses its original intent in our model, as the
following theorem and simple proof suggests.

\begin{theorem}
$\forall \pi {\in} 2^\ii{Trace}.~
\pi {\in} \ii{Dense} {\iff} \forall m. \exists t. m {\leq} t \wedge t {\in} \pi$
\Coqed
\end{theorem}

%
%
\begin{proof}
We will prove each of the directions in turn.

To show the $\Rightarrow$ direction, take some
$\pi \in \ii{Dense}$ and some finite prefix $m$. We can construct
$t_{m\varepsilon}$ from $m$ by simply replacing any final $\circ$ with $\termevent$,
for some designated $\varepsilon$.
By definition $m {\leq} t_{m\varepsilon}$ and moreover,
since $t_{m\varepsilon}$ is terminating and
$\pi \in \ii{Dense}$, we can conclude that $t \in \pi$.

To show the $\Leftarrow$ direction, take some
$\pi \in 2^\ii{Trace}$ and some terminating trace $t$; since $t$ is terminating
we can choose $m = t$ and since this finite prefix extends only to $t$
we immediately obtain $t \in \pi$.
\end{proof}
\iflater
\ch{TODO Extend this trace model with ``final events''
  that can only appear at the end of the trace (see {\tt
    coq/TODO.org}) for details. This is important for the story of
  this section (relation to CompCert) and our instance from 6.
  For now only did silent divergence right!}
\fi
We now show that our definition of dense properties is uniquely
determined given the trace model, the definition of safety, and three
conditions (see \Cref{thm:densethm}) usually satisfied by the class of
liveness properties \cite{AlpernS85}.  The key idea consists in
looking at safety properties from a topological point of view
\cite{AlpernS85, ClarksonS10}.  Conditions in \Cref{thm:densethm}
provide a characterization of another topological class of interest,
that is shown to be exactly the class we called \textit{Dense}
(\Cref{thm:densedense}).

\begin{definition}[Trace Topology \cite{AlpernS85, ClarksonS10}]
 $\mathcal{A}$ is the topology on $\ii{Trace}$ defined by
 its closed set being all and only the Safety properties.  
\end{definition}

\begin{theorem} \label{thm:densethm} Let $X \subseteq 2^{\ii{Trace}}$ such that  
\begin{align*}
   \ii{i)} & ~\ii{Safety} \cap X = \{ \ii{True} \}  & (\emph{trivial intersection})
  \\
   \ii{ii)} &  ~\forall \pi \in 2^{\ii{Trace}}. ~\exists S \in \ii{Safety} ~\exists x \in X. ~ \pi = S \cap x & (\emph{decomposition})
  \\
   \ii{iii)} & ~ \forall x_1 x_2 \in X. ~\forall S \in \ii{Safety}. ~ x_1 = x_2 \cap S \Rightarrow x_2 = x_1 \wedge S = \ii{True} & (\emph{unique decomposition for $X$})
\end{align*}

Then $X$ is the class of the dense sets in $\mathcal{A}$.  
\end{theorem}

\begin{proof} See file TopologyTrace.v, Theorem X\_dense\_class.
\end{proof}

\begin{theorem} \label{thm:densedense} \textit{Dense} is the class of the dense sets
                in the topology $\mathcal{A}$.  
\end{theorem}

\begin{proof}
  See file TopologyTrace.v, Lemma Dense\_dense.
\end{proof}

\begin{corollary} Assume $X \subseteq 2^{\ii{Trace}}$ satisfies the assumptions 
                  of \Cref{thm:densethm}, then $X = \textit{Dense}$
\end{corollary}

\begin{proof}
  See file TopologyTrace.v, X\_Dense\_class.
\end{proof}

\medskip

A property of legacy trace models that does not hold in our model is that any trace property can
be decomposed as the intersection of two liveness properties~\cite{AlpernS85}. 
To show it, first recall that if a set is dense, then every set including it is still dense. 
This means that if the topology allows for two disjoint dense sets $D_1 \cap D_2 = \emptyset$,
we can always write an arbitrary property $\pi$ as intersection of two dense sets. 

\begin{equation*}
  \pi = (D_1 \cup \pi) \cap (D_2 \cup \pi) 
\end{equation*}

This happens for instance in the trace model of Clarkson \emph{et al.},
where it is possible to write an arbitrary property as intersection of two liveness properties
(that play the role of the dense sets) \cite{AlpernS85, ClarksonS10} and is strictly related to 
the fact that only infinite traces are considered. 
In our trace model it is not possible to have disjoint dense sets as they must all include the set 
of all finite traces. It follows that a property discarding some terminating trace cannot have a similar decomposition.

\subsection{Robust Dense Property Preservation (\rdp)}
\label{sec:rdp}


\criteria{\rdp}{rdp} restricts \rtp to only dense properties:
\begin{align*}
\rdp:
  \forall \pi \in \ii{Dense}.~\forall\src{P}.~
    &(\forall\src{C_S}~t \ldotp
    \src{C_S\hole{P} \sem} t \Rightarrow t\in\pi)
    \Rightarrow
    \\&
    (\forall \trg{C_T}~t \ldotp
    \trg{C_T\hole{\cmp{P}} \sem} t \Rightarrow t\in\pi)
\end{align*}

Again, one might wonder how one can get dense properties to be
\emph{robustly} satisfied in the source and then preserved by
compilation.
As for robust safety, one concern is that the context may perform bad
events to violate the dense property. This can be handled in the same
way as for robust safety (\oldautoref{sec:rsp}). An additional concern is
that the context may refuse to give back control (but not terminate) or
silently diverge, thus violating a dense property such as ``along
every infinite trace, an infinite number of good outputs are produced''.
For this, the enforcement mechanism may use time-outs on
the context, forcing it to relinquish control to the partial program
periodically. Alternatively, we may add information to traces about
whether the context or the partial program produces an event, and
weaken dense properties of interest to include traces in which the
context keeps control forever.



The property-free variant of \rdp, called
\criteria{\pf{\rdp}}{pfrdp}, restricts 
\pf{\rtp} to only back-translating {\em non-terminating} traces:
\begin{align*}
  \pf{\rdp}:\ \forall\src{P}.~ \forall\trg{C_T}.~ \forall t~\ii{non\text{-}terminating}.~
        &
        \trg{C_T\hole{\cmp{P}}} \mathrel{\trg{\sem}} t \Rightarrow
        \\&
        \exists \src{C_S}\ldotp \src{C_S\hole{P}} \mathrel{\src{\sem}} t
\end{align*}
Non-terminating traces are either infinite or silently divergent.
We are not aware of good ways to make use of \emph{infinite} executions
$\trg{C_T\hole{\cmp{P}}} \sem t$ to produce a {\em finite} context
$\src{C_S}$, so, unlike for \pf{\rsp}, back-translation proofs of
\pf{\rdp} will likely have to rely only on $\trg{C_T}$ and $\src{P}$,
not $t$, to construct $\src{C_S}$.

Finally, we have proved that \rtp \emph{strictly implies} \rdp
(\CoqSymbol; \autoref{sec:sepCoq}).
The counterexample compilation chain we use for showing the separation
is roughly the inverse of the one we used for \rsp
(\oldautoref{thm:rsp-doesnt-imply-rtp}).
We take the source to be arbitrary, with the sole assumption that
there exists a program $\src{P_{\Omega}}$ that can produce a single
infinite trace $w$ irrespective of the context.
We compile programs by simply pairing them with a constant bound on
the number of steps, \IE $\cmp{P} = \trg{(\src{P},}\,k\trg{)}$.
%
%
On the one hand, \pf{\rdp} holds vacuously, as target programs
cannot produce infinite traces.
On the other hand, this compilation chain does not have \rtp, since
the property $\pi = \{ w \}$ is robustly satisfied by $\src{P_{\Omega}}$ in
the source but not by its compilation
$\trg{(\src{P_\Omega},}\,k\trg{)}$ in the target.

This separation result does not hold in models
with only infinite traces, wherein any trace property can be
decomposed as the intersection of two liveness properties~\cite{AlpernS85}.
In fact, in that model, the analogue of \rdp---Robust Liveness Property Preservation---and \rtp trivially coincide.

Further, neither \rdp nor \rsp implies the other. This follows because
every property can be written as the intersection of a safety and a
dense property  (\iffull see \fi\autoref{sec:dense-in-depth}).
So, if \rdp implies \rsp, then \rdp must imply \rtp, which we just
proved to not hold. By a dual argument, \rsp does not imply \rdp.
More details are given in \autoref{sec:sepCoq}.

\ifsooner

\subsection{Towards More Informative Traces}

\ch{This will be an imperative change we do everywhere (starting with
  Appendix A), not just a potential change we just describe at the
  end.}  \jt{Yet we probably don't have enough space to do the change
  in the main paper.\ch{There's nothing about this in the paper after
    the recent cuts, this is all about this appendices!} Should we add
  this section at the very beginning of the appendices,\ stating that
  in the main paper we did it with only one termination event, and now
  we're going to use this model, and that it doesn't change any result
  presented in the paper?}\ch{No, we should not do any of this stuff,
  we should just get rid of bullet {\bf everywhere} and only use final events.}
\rb{After the change this sub-section is redundant.}

So far in this section, we have assumed that terminating traces and terminating
finite prefixes consist of a sequence of events and a marker to indicate
termination, the solid bullet $\bullet$. However, it may be desirable to
distinguish between several possible reasons for termination, \IE normal and
erroneous termination (possible with some associated data), each reason (and
variation thereof) being considered different from all the others. Such an
extension removes the final gap of expressivity between our trace model and the
realistic model adopted by CompCert.

In \autoref{sec:traces}, this richer trace model decorates the bullet with an
abstract type of termination events, represented by $\varepsilon$, and writes
\circled{\varepsilon}. The various definitions on traces, prefixes, properties,
etc., are adjusted accordingly.

\fi

\clearpage
\section{Secure Compilation Criteria}\label{sec:new-crit}
This appendix describes all the new secure compilation criteria considered
in this work, depending on what class of properties they robustly
preserve: arbitrary trace properties (\Cref{sec:prop}), safety
properties (\Cref{sec:safety}), dense properties (\Cref{sec:dense});
arbitrary hyperproperties (\Cref{sec:hyper}), subset-closed
hyperproperties (\Cref{sec:subset-closed}), including
\(K\)-subset-closed hyperproperties (\Cref{sec:k-subset-closed}),
hypersafety (\Cref{sec:hypersafety}), including \(K\)-hypersafety
(\Cref{sec:k-hypersafety}), hyperliveness (\Cref{sec:hyperliveness});
and arbitrary relational hyperproperties (\Cref{sec:rel-hyper}) and
properties (\Cref{sec:rel-prop}), their \(K\)- and \(2\)-relational
variants (\Cref{sec:k-rel-hyper}, \Cref{sec:k-rel-prop}), and safety
relational properties (\Cref{sec:rel-safety}), including the finite,
\(K\)-, and \(2\)-relational variants (\Cref{sec:fin-rel-safety}). We
also describe the relaxed (X) variants of relational safety
(\Cref{sec:rel-xsafety}).

Each of these sections gives two definitions: a criterion that is
explicit about the class of properties it robustly preserves, and an
equivalent characterization that is \emph{property free}, and is thus
better suited for proofs.


As in the introduction, we organize these criteria in the diagram from
\Cref{fig:order}, where criteria above imply criteria below, and arrows
indicate the strict separation between the two criteria, that is the
existence of a compilation chain satisfying the lower criterion but
not the higher. These separation results are described in \Cref{sec:sep}.
\vspace{1em}

\hspace{-2em} 
\begin{minipage}{\textwidth}
\let\autoref\oldautoref
\begin{center}
{\centering
\begin{tikzpicture}[node distance=4mm,remember picture]\footnotesize
\tikzset{myptr/.style={decoration={markings,mark=at position 1 with {\arrow[scale=2,>=stealth]{>}}},postaction={decorate}}}
	\node[align=center] (rrhp) { Robust Relational Hyperproperty\\  Preservation (\rrhpref) };
	\node[align=center, below = of rrhp.south, yshift=.5em] (rkrhp) { Robust $K$-Relational Hyperproperty\\ Preservation  (\rkrhpref)};
	\node[align=center, below = of rkrhp.south, yshift=.5em] (r2rhp) { Robust 2-Relational Hyperproperty\\ Preservation (\rtrhpref)};

	\node[align=center, right = of rrhp.south, xshift = 8em ] (rrpp) { Robust Relational Property\\ Preservation  (\rrtpref)};
	\node[align=center, below = of rrpp.south, yshift=-1.5em] (rkrpp) { Robust $K$-Relational Property\\ Preservation  (\rkrtpref)};
	\node[align=center, below = of rkrpp.south, yshift=.5em] (r2rpp) { Robust 2-Relational Property\\ Preservation  (\rtrtpref)};

	\node[align=center, right = of rrpp.south, xshift = 7em] (rrxp) { Robust Relational relaXed safety\\ Preservation (\criterion{RrX})};
	\node[align=center, below = of rrxp.south, yshift=1em] (rfrxp) { Robust Finite-Relational relaXed\\ safety Preservation  (\rfrxpref)};
	\node[align=center, below = of rfrxp.south, ] (rkrxp) { Robust $K$-Relational relaXed\\safety Preservation  (\criterion{R\mi{K}rX})};
	\node[align=center, below = of rkrxp.south, yshift=.5em] (r2rxp) { Robust 2-Relational relaXed\\safety Preservation  (\rtrxpref)};

	\node[align=center, right = of rrxp.south, xshift=7em, yshift = -1em] (rrsp) { Robust Relational Safety\\  Preservation (\rrspref)};
	\node[align=center, below = of rrsp.south, yshift=-1em] (rfrsp) { Robust Finite-Relational \\Safety Preservation  (\rfrsp)};
	\node[align=center, below = of rfrsp.south, yshift=-1em] (rkrsp) { Robust $K$-Relational Safety\\ Preservation  (\rkrspref)};
	\node[align=center, below = of rkrsp.south, yshift=-1em] (r2rsp) { Robust 2-Relational Safety\\ Preservation  (\rtrspref)};

	\node[align=center, below = of r2rhp.south, yshift = -1.5em] (h) { Robust Hyperproperty \\ Preservation (\rhpref)};
	\node[align=center, below = of h.south, yshift=.5em] (sc) { Robust Subset-Closed Hyperproperty\\ Preservation  (\rschpref)};
	\node[align=center, below = of sc.south, yshift=.5em] (ksc) { Robust $K$-Subset-Closed Hyperproperty \\ Preservation  (\rkschp)};
	\node[align=center, below = of ksc.south, yshift=.5em] (2sc) { Robust 2-Subset-Closed Hyperproperty \\ Preservation (\rtschp)};
	\node[align=center, below = of 2sc.south, yshift=.5em,xshift=1em] (p) { Robust Trace  Property Preservation (\rtpref)};

	\node[align=center] at ([xshift = 7em]ksc.east -| r2rpp) (hs) { Robust Hypersafety Preservation (\rhspref)};
	\node[align=center, below = of hs.south] (khs) { Robust $K$-Hypersafety Preservation  (\rkhspref)};
	\node[align=center, below = of khs.south] (2hs) { Robust 2-Hypersafety Preservation (\rthspref)};
	\node[align=center, below = of 2hs.south, yshift=-1em] (sp) { Robust Safety Property Preservation (\rspref)};

	\node[align=center] at (p.south |- sp.west) (lp) { Robust Dense Property Preservation (\rdp)};

	\node[align=center, below = of r2rpp.south, yshift=-2.5em, font=\itshape, xshift = 6em, 
	] (fac2) { Robust Trace Equivalence \\Preservation (\rtepref) };
	\node[align=center, below = of 2hs.east, xshift = 11em, font=\itshape, yshift=.5em, 
		] (rnip) { Robust Termination-Insensitive\\ Noninterference Preservation\\ (\rtinipref) };

	\draw[-] (rrhp.south) to (rkrhp.north);
	\draw[-] (rkrhp.south) to (r2rhp.north);
	\draw[myptr] ([yshift=.3em]r2rhp.south) -- (h.north);

	\draw[-] (rrhp.east) to (rrpp.west);
	\draw[-] (rkrhp.east) to (rkrpp.north west);
	\draw[-] (r2rhp.east) to (r2rpp.north west);

	\draw[-] (rrpp.south) to (rkrpp.north);
	\draw[-] (rkrpp.south) to (r2rpp.north);

	\draw[myptr,draw = gray] ([xshift =.5em, yshift =1em]rrpp.south west) -- (sc.12);
	\draw[myptr,draw = gray] ([xshift =1em, yshift =1em]rkrpp.south west) -- (ksc.11);
	\draw[myptr,draw = gray] (r2rpp.195) -- (2sc.10);
	\draw[-] (rrpp.east) -- ([yshift=.5em]rrxp.west) ;
	\draw[-] (rkrpp.east) -- (rkrxp.north west) ;
	\draw[-] (r2rpp.east) -- (r2rxp.north west) ;

	\draw[-] (h.south) to (sc.north);
	\draw[-] (sc.south) to (ksc.north);
	\draw[-] (ksc.south) to (2sc.north);
	\draw[myptr] ([yshift=.5em]2sc.south) to ([xshift=-1em,yshift=-.5em]p.north);

	\draw[-] (sc.east) to (hs.west);
	\draw[-] (hs.south) to (khs.north);
	\draw[myptr,draw = black] (khs.south) to (2hs.north);
	\draw[myptr] ([yshift=.3em]2hs.south) to (sp.north);

	\draw[-] (ksc.east) to (khs.west);
	\draw[-] (2sc.east) to (2hs.west);
	\draw[myptr] (p.east) to (sp.north west);
	\draw[myptr] ([xshift=-1em]p.south) to ([xshift=-1em]lp.north);

	\draw[-] (rrxp.south) to (rfrxp.north);
	\draw[-] (rfrxp.south) to (rkrxp.north);
	\draw[-] (rkrxp.south) to (r2rxp.north);
	
	\draw[-] (rrsp.south) -- (rfrsp.north) ;
	\draw[-] (rfrsp.south) to (rkrsp.north);
	\draw[-] (rkrsp.south) to (r2rsp.north);
	
	\draw[myptr,draw=gray] (rfrsp.197) -- (hs.0);
	\draw[myptr,draw=gray] (rkrsp.197) -- (khs.north east);
	\draw[myptr,draw=gray] (r2rsp.185) -- (2hs.north east);

	\draw[-] (rrxp.east) -- (rrsp.north west);
	\draw[-] (rfrxp.east) -- ([xshift=1em]rfrsp.north west);
	\draw[-] (rkrxp.east) -- (rkrsp.north west);
	\draw[-] (r2rxp.east) -- (r2rsp.north west);

    \draw[myptr] (r2rhp.-20) -- (fac2.north west);
	\draw[myptr, dotted] (r2rxp.west) -- (fac2.90) node [left,align=left,pos=.05, font=\itshape] { + determinacy } ;  
	\draw[-] (2hs.east) to (rnip.north west);


	\draw[rounded corners, thick, fill=blue!20!white, opacity=0.2, dotted]
	(rrhp.north) -- (rrhp.north east) -- (rrxp.north east) -- (rrsp.north east) -- (rfrsp.east) |- (r2rsp.south west) -- ([xshift=4em]fac2.south east) -- (fac2.south west) -- ([xshift=3em,yshift=-.5em]r2rhp.south)
	-| ([xshift=-0.5em]r2rhp.north west) |- (rrhp.north);

	\node[draw=black, dotted, left = of rkrhp.180, align = center, rotate=90, anchor = center, yshift = 2em, rounded corners, thick, fill=blue!20!white, opacity=0.2, dotted, text opacity = 1](rhptitle){Relational \\ Hyperproperties \\ Criteria (\autoref{sec:rel})};

	\draw[rounded corners, thick, fill=red!20!white, opacity=0.2]
	([yshift=.5em]h.north) -- ([yshift=.5em]h.north east) -- (hs.north west) -- (hs.north east) 
	-- (rnip.north east) |- (rnip.south west) -- (2hs.south east)
	-- (2hs.south west) -- (2sc.south east) -- (2sc.south west)
	|- ([yshift=.5em]h.north west) -- ([yshift=.5em]h.north);

	\node[draw=black, align = center, rotate=90, anchor = center, rounded corners, thick, fill=red!20!white, opacity=0.2, text opacity = 1] at (rhptitle |- sc)  (hptitle){Hyperproperties \\ Criteria (\autoref{sec:hyperprop})\\};

	\draw[rounded corners, thick, fill=yellow!20!white, opacity=0.2, loosely dotted]
	(p.north) -- (p.north east) -- ([yshift=1.5em]sp.north west) -- ([yshift=1.5em]sp.north east) -- (sp.south east) -| (p.north west) -- (p.north);

	\node[draw=black, align = center, rotate=90, anchor = center, rounded corners, thick, fill=yellow!20!white, opacity=0.2, loosely dotted, text opacity = 1] at (rhptitle |- p.-90)  (ptitle){Trace \\ Properties \\ Criteria (\autoref{sec:prop})};


\end{tikzpicture}
}
\end{center}
\begin{tikzpicture}[remember picture, overlay]\footnotesize
\draw[rounded corners, thick, dashed, fill=green!80!black, opacity=0.05]
	(rfrsp.north) -| (rfrsp.east) |- (sp.south west) |- (hs.north west) -- (fac2.north west) -- ([xshift=-1.5em]rkrxp.west) -- (rfrxp.north west) -- (rfrxp.north east) -- (rfrsp.north) -- (rfrsp.north);

	\node[draw=black, right = of r2rsp.0, align = center, rotate=-90, anchor = center, yshift = -1em, rounded corners, thick, fill=green!80!black, opacity=0.1, dashed, text opacity = 1](tbbttitle){Trace-based \\ Back-translation Proofs};
\end{tikzpicture}
\end{minipage}
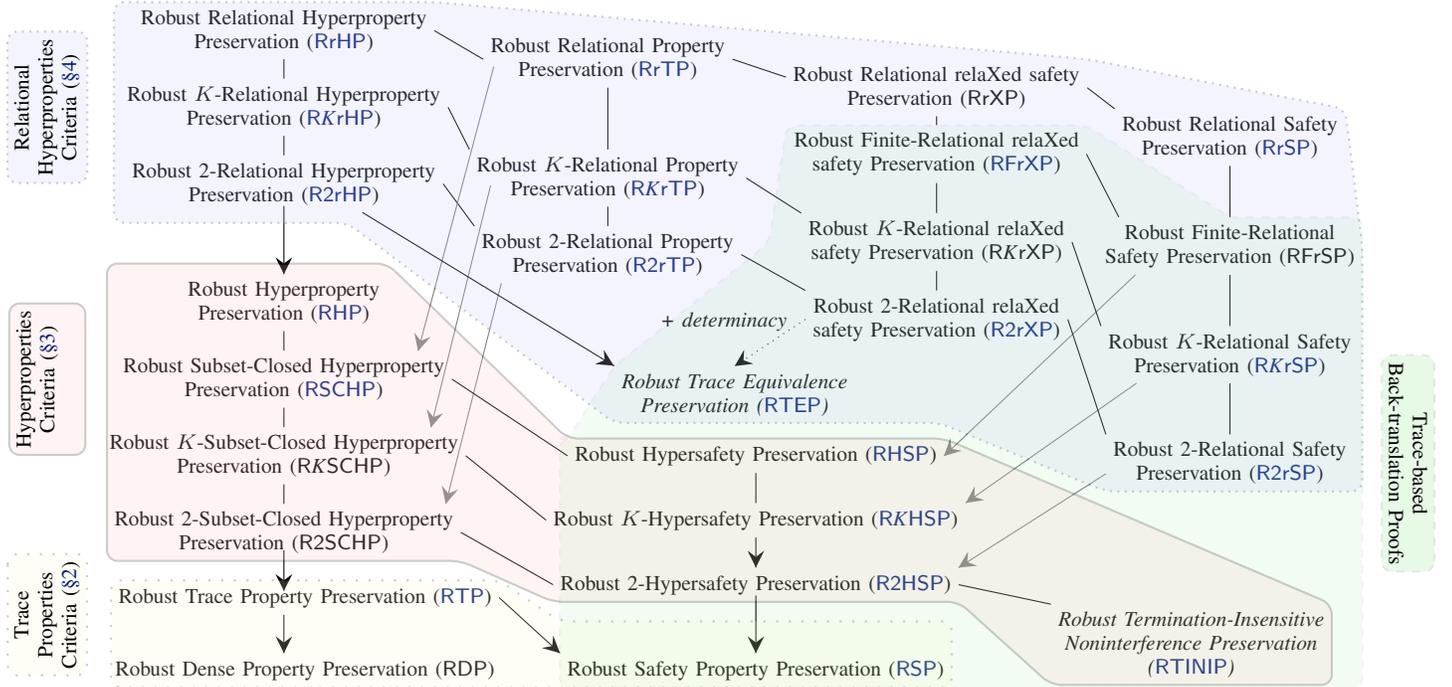
\myfig{}{order}{Partial order of the secure compilation criteria studied in this paper.}


\subsection{Trace Property-Based Criteria}\label{sec:prop-criteria}

We start by describing the three criteria at the bottom of the
lattice, \rtp, \rsp, and \rdp, corresponding to the robust
preservation of \emph{trace properties}, defined as sets of
allowed traces.

These criteria state that for any property \(\pi\) in the class
they preserve, if a source program can only produce traces
belonging to \(\pi\), when linked with any source context,
then the same is true of the compiled program linked with any target
context.

\subsubsection{Robust Trace Property Preservation}\label{sec:prop}
The first of these criteria is called \emph{Robust Trace Property
  Preservation}, or \rtp, and corresponds to the robust preservation
of all trace properties.

\begin{definition}[Robust Trace Property Preservation (\rtp)]\label{def:rtp}
  $$
  \begin{multlined}
    \criteria{\rtp}{rtp}:\quad \forall \pi \in 2^\ii{Trace}\ldotp\forall\src{P}\ldotp
    (\forall\src{C_S}~t \ldotp
    \src{C_S\hole{P} \rightsquigarrow}\ t \Rightarrow t\in\pi) \Rightarrow\\
    (\forall \trg{C_T}~t \ldotp
    \trg{C_T\hole{\cmp{P}} \rightsquigarrow}\ t \Rightarrow t\in\pi)\quad
  \end{multlined}
  $$
\end{definition}

The property-free characterization of \rtp is \pf{\rtp}. This
characterization captures the fact that if a target program can
produce a given trace, then the source can also produce this
trace. Intuitively, this corresponds to the fact that any violation of
a trace property in the target can be explained by the same violation
in the source.

\begin{definition}[Equivalent Characterization of \rtp (\pf{\rtp})]\label{def:rtc}
  \begin{align*}
    \pf{\rtp}:\quad \forall\src{P}.~ \forall\trg{C_T}.~ \forall t.~
    &\trg{C_T\hole{\cmp{P}}} \rightsquigarrow t \Rightarrow\\
    &\exists \src{C_S}\ldotp \src{C_S\hole{P}}\rightsquigarrow t
  \end{align*}
\end{definition}


\begin{theorem}[\rtp and \pf{\rtp} are equivalent]\label{thm:rtp-eq}
	\begin{align*}
		\rtp \iff \pf{\rtp}
	\end{align*}
\end{theorem}
\begin{proof}
  See file Criteria.v, theorem RTC\_RTP for a Coq proof.
  The proof is simple, but still illustrative for how such proofs work in general:

  \begin{description}
  \item[$\Rightarrow$] Let $\src{P}$ be arbitrary.
    We need to show that
    \(
    \forall\trg{C_T}\ldotp\forall t\ldotp
    \trg{C_T\hole{\cmp{P}} \sem}\ t \Rightarrow
    \exists \src{C_S}\ldotp \src{C_S\hole{P}\sem}\ t
    \).
    We can directly conclude this by applying $\rtp$ to $\src{P}$ and the property
    $\pi = \{ t ~|~ \exists \src{C_S}.~\src{C_S\hole{P}\sem} t \}$;
    for this application to be possible we need to show that
    \(\forall\src{C_S}~t\ldotp
    \src{C_S\hole{P} \sem}\ t \Rightarrow
    \exists \src{C_S'}\ldotp\src{C_S'\hole{P} \sem}\ t \),
    which is trivial if taking $\src{C_S'} = \src{C_S}$.

  \item[$\Leftarrow$] Given a compilation chain that satisfies \pf{\rtp} and
    some $\src{P}$ and $\pi$ so that 
    $\forall\src{C_S}~t \ldotp
    \src{C_S\hole{P} \sem}\ t \Rightarrow t\in\pi\;(H)$
    we have to show that
    $\forall \trg{C_T}~t \ldotp
    \trg{C_T\hole{\cmp{P}} \sem}\ t \Rightarrow t\in\pi$.
    Let $\trg{C_T}$ and $t$ so that $\trg{C_T\hole{\cmp{P}} \sem}\ t$,
    we still have to show that $t\in\pi$.
    We can apply \pf{\rtp} to obtain 
    $\exists \src{C_S}\ldotp \src{C_S\hole{P}\sem} t$,
    which we can use to instantiate $H$ to conclude that $t \in \pi$.\qedhere
  \end{description}
\end{proof}

\subsubsection{Robust Safety Property Preservation}\label{sec:safety}


Robust Safety Property Preservation is the criterion corresponding to the
robust preservation of \emph{safety} properties, \IE properties that can be
finitely refuted: for any safety property, and any trace not in the property,
there exists a finite \emph{bad prefix} of the trace that can not be extended
to belong to the property.

\begin{definition}[Safety Property]\label{def:safety}
  We define the set of safety properties, \ii{Safety}:
  \[
  \ii{Safety} \triangleq \myset{\pi \in 2^\ii{Trace} }{ \forall t {\not\in} \pi.~
  \exists m {\leq} t.~ \forall t' {\geq} m.~ t' {\not\in} \pi }
  \]

  A property \(\pi\) is a safety property if and only if \(\pi \in \ii{Safety}\).
\end{definition}

\begin{definition}[Robust Safety Property Preservation (\rsp)]\label{def:rsp}
  $$
  \begin{multlined}
    \criteria{\rsp}{rsp}:\quad \forall \pi \in \ii{Safety}.~\forall\src{P}.~
    (\forall\src{C_S}~t \ldotp \src{C_S\hole{P} \rightsquigarrow}\ t
    \Rightarrow t\in\pi) \Rightarrow\\
    (\forall \trg{C_T}~t \ldotp
    \trg{C_T\hole{\cmp{P}} \rightsquigarrow}\ t \Rightarrow t\in\pi)\quad
  \end{multlined}
  $$
\end{definition}

The equivalent property-free characterization, \pf{\rsp}, captures the
fact that a safety property can be refuted by one finite bad prefix
\(m\): all finite violation of a safety property at the target level
can be explained by the same finite violation at the source level.

\begin{definition}[Equivalent Characterization of \rsp (\pf{\rsp})]\label{def:rsc}
	\begin{align*}
	  \pf{\rsp}:\quad \forall\src{P}.~ \forall\trg{C_T}.~ \forall m.~
          &\trg{C_T\hole{\cmp{P}} \rightsquigarrow}\ m \Rightarrow\\
          &\exists \src{C_S}\ldotp \src{C_S\hole{P}\rightsquigarrow}\ m
	\end{align*}
\end{definition}

\begin{theorem}[\rsp and \pf{\rsp} are equivalent]\label{thm:rsp-eq}
  \begin{align*}
    \rsp \iff \pf{\rsp}
  \end{align*}
\end{theorem}
\begin{proof}
	See file Criteria.v, theorem RSC\_RSP.
\end{proof}

\subsubsection{Robust Dense Property Preservation}\label{sec:dense}
Robust Dense Property Preservation is the criterion corresponding to the robust preservation
of \emph{dense} properties. Dense properties are the properties that include
all finite traces, and roughly correspond to \emph{liveness} in our model. See \Cref{sec:traces-details} and \Cref{sec:dense-in-depth} for more details.

A more detailed view of Robust Dense Property Preservation is given in \Cref{sec:rdp}.

\begin{definition}[Dense Property]\label{def:dense}
  We define the set of dense property, \ii{Dense}:
  \[
  \ii{Dense} \triangleq \myset{\pi \in 2^\ii{Trace} }{
    \forall t~\ii{terminating}.~ t \in \pi }
  \]
  A property \(\pi\) is a dense property if and only if \(\pi \in \ii{Dense}\).
\end{definition}

\begin{definition}[Robust Dense Property Preservation (\rdp)]\label{def:rdp}
  $$
  \begin{multlined}
    \criteria{\rdp}{rdp}:\quad
    \forall \pi \in \ii{Dense}.~\forall\src{P}.~
    (\forall\src{C_S}~t \ldotp
    \src{C_S\hole{P} \rightsquigarrow}\ t \Rightarrow t\in\pi)
    \Rightarrow\\
    (\forall \trg{C_T}~t \ldotp
    \trg{C_T\hole{\cmp{P}} \rightsquigarrow}\ t \Rightarrow t\in\pi)\quad
  \end{multlined}
  $$
\end{definition}

The property-free characterization \pf{\rdp} captures the fact that
dense properties can be violate only by infinite traces.  \jt{Not sure what to explain here}
\begin{definition}[Equivalent Characterization of \rdp (\pf{\rdp})]\label{def:rdc}
  \begin{align*}
    \pf{\rdp}:\quad \forall\src{P}.~ \forall\trg{C_T}.~ \forall t~\text{infinite}.~
    &\trg{C_T\hole{\cmp{P}}\rightsquigarrow}\ t \Rightarrow\\
    &\exists \src{C_S}\ldotp \src{C_S\hole{P}\rightsquigarrow}\ t
  \end{align*}
\end{definition}

\begin{theorem}[\rdp and \pf{\rdp} are equivalent]\label{thm:rdp-eq}
	\begin{align*}
		\rdp \iff \pf{\rdp}
	\end{align*}
\end{theorem}
\begin{proof}
	See file Criteria.v, theorem RDC\_RDP.
\end{proof}

\subsection{Hyperproperty-Based Criteria}\label{sec:hyper-criteria}

The criteria in this section describe the robust preservation of
\emph{hyperproperty}, that is sets of allowed program
behaviors. Formally, a hyperproperty is an element \(H\) of the set
\(2^{2^\ii{Trace}}\), and a program \(P\) satisfies this hyperproperty
\(H\) if and only if \(\behav{P} \in H\). Hyperproperties allow to
express more security properties than trace properties, such as
noninterference for instance.

Again, these criteria state that for any hyperproperty \(H\) in the
class they preserve, if a source program's behavior when linked with
any any source context belongs to \(H\), then the same is true for the
compiled program.

Note that the behavior being considered is not the set of traces generated
by the program when linked with all source contexts, but only the set
of traces generated when linked with a particular context.
\jt{Not sure if this is clear or useful, someone else decides}
\rb{Personally, I understood it without further clarification. It seems the only
    interpretation that makes sense.}

\subsubsection{Robust Hyperproperty Preservation}\label{sec:hyper}
Robust Hyperproperty Preservation is the criterion corresponding to the
robust preservation of all hyperproperties.

\begin{definition}[Robust Hyperproperty Preservation (\rhp)]\label{def:rhp}
  $$
  \begin{multlined}
    \criteria{\rhp}{rhp}:\quad
    \forall H \in 2^{2^\ii{Trace}}.~\forall\src{P}.~
    (\forall\src{C_S} \ldotp
    \src{\behav{C_S\hole{P}}} \in H)
    \Rightarrow
    \\
    (\forall \trg{C_T} \ldotp
    \trg{\behav{C_T\hole{\cmp{P}}}} \in H)\quad
  \end{multlined}
  $$
\end{definition}

The equivalent characterization, \pf{\rhp}, states that all behaviors of the
compiled program are behaviors of the source program: if the compiled program
violate a hyperproperty with a particular behavior, then the source program
does too.

\begin{definition}[Equivalent Characterization of \rhp (\pf{\rhp})]\label{def:rhc}
  \begin{align*}
    \pf{\rhp}:\quad &\forall\src{P}.~ \forall\trg{C_T}.~ \exists \src{C_S}.~
    \trg{\behav{C_T\hole{\cmp{P}}}} = \src{\behav{C_S\hole{P}}}\\
  \end{align*}
\end{definition}
Unfolding this definition gives:
\[
\pf{\rhp}:\quad \forall\src{P}.~ \forall\trg{C_T}.~ \exists \src{C_S}.~ \forall t.~
\left(\trg{C_T\hole{\cmp{P}} \rightsquigarrow}\ t \iff
\src{C_S\hole{P}\rightsquigarrow}\ t\right)
\]

\begin{theorem}[\rhp and \pf{\rhp} are equivalent]\label{thm:hrp-eq}
	\begin{align*}
		\rhp \iff \pf{\rhp}
	\end{align*}
\end{theorem}
\begin{proof}
	See file Criteria.v, theorem RHC\_RHP.
\end{proof}

\subsubsection{Robust Subset-Closed Hyperproperty Preservation}\label{sec:subset-closed}

\jt{I don't really understand subset-closed hyperproperty well enough to describe this section.
  I think \(k\)-SCHP are refutable by \(k\) traces, but I don't see how that extend to
  \rschp}
\jt{No definition of SC?} 

In general a program satisfies a certain hyperproperty
if its set of traces, its behavior, is in the hyperproperty.
With subset-closed hyperproperties (\oldautoref{sec:rschp}), if a set
of traces is accepted then so are all smaller sets of traces.
Subset closed hyperproperties can therefore be used to formalize the
notion of \emph{refinement}~\cite{ClarksonS10}.

\begin{definition}[Subset-Closed Hyperproperties]\label{def:sch}
  We define the set of Subset-Closed Hyperproperties, \ii{SC}:
  \[
  \ii{SC} \triangleq \myset{H}{\forall b_1 \subseteq b_2. b_2 \in H \Rightarrow b_1 \in H } 
  \]

  A hyperproperty \(H\) is subset-closed if and only if \(H \in \ii{SC}\).
\end{definition}

\begin{definition}[Robust Subset-Closed Hyperproperty Preservation (\rschp)]\label{def:rschp}
  $$
  \begin{multlined}
    \criteria{\rschp}{rschp}:\quad
    \forall H \in \ii{SC}.~\forall\src{P}.~
    (\forall\src{C_S} \ldotp
    \src{\behav{C_S\hole{P}}} \in H)
    \Rightarrow\\
    (\forall \trg{C_T} \ldotp
    \trg{\behav{C_T\hole{\cmp{P}}}} \in H)\quad
  \end{multlined}
  $$
\end{definition}

The equivalent characterization of $\rschp$ states that the behaviors of a compiled program
in an arbitrary target context are the refinement of the behaviors of the original 
program in some source context.

\begin{definition}[Equivalent Characterization of \rschp (\pf{\rschp})]\label{def:rschc}
  \begin{align*}
    \pf{\rschp}:\quad \forall\src{P}.~ \forall\trg{C_T}.~ \exists \src{C_S}.~ \forall t.~
    \trg{C_T\hole{\cmp{P}} \rightsquigarrow}\ t \Rightarrow
    \src{C_S\hole{P}\rightsquigarrow}\ t
  \end{align*}
\end{definition}

\begin{theorem}[\rschp and \pf{\rschp} are equivalent]\label{thm:rschp-eq}
	\begin{align*}
		\rschp \iff \pf{\rschp}
	\end{align*}
\end{theorem}
\begin{proof}
	See file Criteria.v, RSCHC\_RSCHP.
\end{proof}

\subsubsection{Robust \texorpdfstring{$K$}{K}-Subset-Closed Hyperproperty Preservation}\label{sec:k-subset-closed}

While for \(K\)-Hypersafety a set of \(K\) bad finite prefixes is
enough to refuse a behavior, for \(K\)-subset-closed hyperproperties,
\(K\) complete traces could be necessary.


\begin{definition}[\(K\)-Subset-Closed Hyperproperties \cite{mastroeni2018verifying}]\label{def:ksch} 
  We define the set of \(K\)-Subset-Closed Hyperproperties, \ii{KSC}:
  \[
  \ii{KSC} \triangleq \myset{H}{ \forall b. ~ b \notin H \iff (\exists T_K \subseteq b. ~(|T_K| \leq K \wedge T_K \notin H)) } 
  \]

  A hyperproperty \(H\) is \(K\)-subset-closed if and only if \(H \in \ii{KSC}\).
\end{definition}

\begin{definition}[Robust $K$-Subset-Closed Hyperproperty Preservation (\rkschp)]\label{def:rkschp}
  $$
  \begin{multlined}
    \criteria{\rkschp}{rkschp}:\quad
    \forall H \in \ii{KSC}.~\forall\src{P}.~
    (\forall\src{C_S} \ldotp
    \src{\behav{C_S\hole{P}}} \in H)
    \Rightarrow\\
    (\forall \trg{C_T} \ldotp
    \trg{\behav{C_T\hole{\cmp{P}}}} \in H)\quad
  \end{multlined}
  $$
\end{definition}


\begin{definition}[Equivalent Characterization of \rkschp (\pf{\rkschp})]\label{def:rkschc}
	\begin{align*}
		\pf{\rkschp}:\quad
		&\
		\forall \src{P}, \trg{C_T}. \forall \com{\set{t}}.\card{\com{\set{t}}}=K.
                \\
		&\
		(\com{\set{t}}\subseteq \behavt{ \trg{C_T}\hole{ \compgen{ \src{P} } } } )
		\Rightarrow
		\exists \src{C_S}.(\com{\set{t}}\subseteq \behavs{ \src{C_S}\hole{ \src{P} }} )
	\end{align*}
\end{definition}

\begin{theorem}[\rkschp and \pf{\rkschp} are equivalent]\label{thm:rkschp-eq}
	\begin{align*}
		\rkschp \iff \pf{\rkschp}
	\end{align*}
\end{theorem}
\begin{proof}
	Analogous to that of \Cref{thm:rtschp-eq} below.
\end{proof}

$\pf{\rtschp}$ is an instance of \Cref{def:rkschc}
with $\card{\com{\set{t}}}=2$.
Similarly, $\rtschp$ is an instance of \Cref{def:rkschp} for $\ii{2SC}$.
\begin{theorem}[\rtschp and \pf{\rtschp} are equivalent]\label{thm:rtschp-eq}
	\begin{align*}
		\rtschp \iff \pf{\rtschp}
	\end{align*}
\end{theorem}
\begin{proof}
	See file Criteria.v, theorem R2SCHC\_R2SCHPC.
\end{proof}

\subsubsection{Robust Hypersafety Preservation}\label{sec:hypersafety}

Robust Hypersafety Preservation (\oldautoref{sec:rhsp})
is the criterion corresponding to the robust
preservation of hypersafety properties (aka. safety hyperproperties), \IE
hyperproperties that can be refuted by a finite number of finite trace prefixes.

Hypersafety is a generalization of safety that captures many important
security properties, such as noninterference.
Informally, a hypersafety property disallows a certain finite observation,
\IE a finite set of finite prefixes.  This observation
\(o\in\ii{Obs}\) is a ``bad observation'', and all its extensions
cannot satisfy the hyperproperty.

\begin{definition}[Observations]\label{def:obs}
  We define the set of observations, \ii{Obs}:
  \[
  \ii{Obs} \triangleq 2^{\ii{FinPref}}_{\ii{Fin}}
  \]
\end{definition}

\begin{definition}[Hypersafety Property]\label{def:hypersafety}
  We define the set of hypersafety properties, \ii{Hypersafety}:
  \begin{align*}  
    \ii{Hypersafety} \triangleq \myset{ H }{
      \forall b \not\in H.~
      (\exists o \in \ii{Obs}.~ o {\leq} b \wedge
      (\forall b' {\geq} o.~ b' \not\in H)) }
  \end{align*}

  A hyperproperty \(H\) is a safety hyperproperty, or a hypersafety property, if and only if
  \(H \in \ii{Hypersafety}\).
\end{definition}

\begin{definition}[Robust Hypersafety Preservation (\rhsp)]\label{def:rhsp}
  $$
  \begin{multlined}
  \criteria{\rhsp}{rhsp}:\quad
  \forall H \in \ii{Hypersafety}.~\forall\src{P}.~
  (\forall\src{C_S} \ldotp
  \src{\behav{C_S\hole{P}}} \in H)
  \Rightarrow\\
  (\forall \trg{C_T} \ldotp
  \trg{\behav{C_T\hole{\cmp{P}}}} \in H)\quad
  \end{multlined}
  $$
\end{definition}

The property-free characterization captures the fact that if a
compiled program can produce an observation \(o\) (in the sense that
it contains a prefix of each trace of the program) that refutes a
hypersafety property, then the same observation can also be produced
by the source program.
\begin{definition}[Equivalent Characterization of \rhsp (\pf{\rhsp})]\label{def:rhsc}
  \begin{align*}
    \pf{\rhsp}:\quad
    \forall\src{P}.~ \forall\trg{C_T}.~ \forall o \in \ii{Obs}.~
    &o \leq \trg{\behav{C_T\hole{\cmp{P}}}} \Rightarrow\\
    &\exists \src{C_S}.~
    o \leq \src{\behav{C_S\hole{P}}}
  \end{align*}
\end{definition}

\begin{theorem}[\rhsp and \pf{\rhsp} are equivalent]\label{thm:rhsp-eq}
	\begin{align*}
		\rhsp \iff \pf{\rhsp}
	\end{align*}
\end{theorem}
\begin{proof}
	See file Criteria.v, theorem RHSC\_RHSP.
\end{proof}

\subsubsection{Robust \texorpdfstring{$K$}{K}-Hypersafety Preservation}\label{sec:k-hypersafety}

\(K\)-hypersafety properties are hypersafety properties that can be
refuted by observations of size at most \(K\), that is one need only
\(K\) appropriately chosen finite prefixes to prove a program doesn't
satisfy the hyperproperty.

\begin{definition}[\(K\)-Observations]\label{def:kobs}
  We define the set of \(K\)-observations, \IE observations of cardinal at most \(K\):
  \[ \ii{Obs}_K \triangleq 2^{\ii{FinPref}}_{\ii{Fin(K)}} \]
\end{definition}

\begin{definition}[\(K\)-Hypersafety Property]\label{def:khypersafety}
  We define the set of \(K\)-hypersafety properties:
  \begin{align*}
    \ii{KHypersafety} \triangleq \myset{ H }{
      \forall b \not\in H.~
      (\exists o \in \ii{Obs}_K.~ o {\leq} b \wedge
      (\forall b' {\geq} o.~ b' \not\in H)) }
  \end{align*}
  A hyperproperty \(H\) is \(K\)-hypersafety if and only if \(H\in\ii{KHypersafety}\).
\end{definition}

\begin{definition}[Robust $K$-Hypersafety Preservation (\rkhsp)]\label{def:rkhsp}
  $$
  \begin{multlined}
    \criteria{\rkhsp}{rkhsp}:\quad
    \forall H \in \ii{KHypersafety}.~\forall\src{P}.~
    (\forall\src{C_S} \ldotp
    \src{\behav{C_S\hole{P}}} \in H)
    \Rightarrow\\
    (\forall \trg{C_T} \ldotp
    \trg{\behav{C_T\hole{\cmp{P}}}} \in H)\quad
  \end{multlined}
  $$
\end{definition}

The property-free characterization has the same intuition, except it is restricted to
behaviors of size 
\(K\).
\begin{definition}[Equivalent Characterization of \rkhsp (\pf{\rkhsp})]\label{def:rkhsc}
  \begin{align*}
    \pf{\rkhsp}:\quad
    \forall \src{P}, \trg{C_T}.~ \forall \com{\set{m}}.~ \card{\com{\set{m}}}=K
    \Rightarrow~
    &
    \com{\set{m}}\leq \behavt{ \trg{C_T}\hole{ \compgen{ \src{P} } } }
    \Rightarrow\\
    &\exists \src{C_S}. \com{\set{m}}\leq \behavs{ \src{C_S}\hole{ \src{P} }}
  \end{align*}
\end{definition}
\jt{\(o\) or \(\set{m}\)?}

\begin{theorem}[\rkhsp and \pf{\rkhsp} are equivalent]\label{thm:rkhsc-eq}
  \begin{align*}
    \rkhsp \iff \pf{\rkhsp}
  \end{align*}
\end{theorem}
\begin{proof}
	Analogous to \Cref{thm:rthsc-eq} below.
\end{proof}

\criteria{\rthsp}{rthsp} is an instance of \Cref{def:rkhsp}
for $K=2$.
Similarly, $\pf{\rthsp}$ is an instance of \Cref{def:rkhsc} for $K=2$.
\begin{theorem}[\rthsp and \pf{\rthsp} are equivalent]\label{thm:rthsc-eq}
	\begin{align*}
		\rthsp \iff \pf{\rthsp}
	\end{align*}
\end{theorem}
\begin{proof}
	See file Criteria.v, theorem R2HSC\_R2HSP.
\end{proof}

A particular instance of \rthsp is \criteria{\rtinip}{rtinip}
(\oldautoref{sec:rnip}).

\subsubsection{Robust Hyperliveness Preservation}\label{sec:hyperliveness}

\begin{definition}[Hyperliveness Property]\label{def:hyperliveness}
  We define the set of hyperliveness properties (or liveness hyperproperties) \ii{Hyperliveness}:
  \begin{align*}
    \ii{Hyperliveness} \triangleq
    \myset{ H }{ \forall o \in \ii{Obs}.~ \exists b {\geq} o.~ b \in H  }
  \end{align*}
  A hyperproperty \(H\) is a hyperliveness property if and only if \(H\in\ii{Hyperliveness}\).
\end{definition}


\begin{definition}[Robust Hyperliveness Preservation (\rhlp)]\label{def:rhlp}
  $$
  \begin{multlined}
    \criteria{\rhlp}{rhlp}:\quad
    \forall H \in \ii{Hyperliveness}.~\forall\src{P}.~
    (\forall\src{C_S} \ldotp
    \src{\behav{C_S\hole{P}}} \in H)
    \Rightarrow\\
    (\forall \trg{C_T} \ldotp
    \trg{\behav{C_T\hole{\cmp{P}}}} \in H)\quad
  \end{multlined}
  $$
\end{definition}

We give no property-free characterization for \rhlp,
since \rhlp collapses with \rhp, as was pointed out in
\oldautoref{sec:hyperliveness}:

\begin{theorem}[\rhp and \rhlp are equivalent]\label{thm:rhp-rhlp-eq}
  \begin{align*}
    \rhp \iff \rhlp
  \end{align*}
\end{theorem}
\begin{proof}
	See file Criteria.v, theorem RHLP\_RHP.
\end{proof}

\subsection{Relational Trace Property-Based Criteria}
\label{sec:rel-trace-prop-criteria}

Relational trace properties are a generalization of trace properties
to allow comparing individual runs of different programs. For instance,
relational trace properties allow expressing properties such as ``program
\(\src{P_1}\) runs faster than \(\src{P_2}\) on every input''.

\subsubsection{Robust \texorpdfstring{\(K\)}{K}-Relational Trace Property Preservation}\label{sec:k-rel-prop}

A \(K\)-relational trace property is a relational trace property of
arity \(K\), that is a relation \(R\) between \(K\) traces. \\
Given \(K\) programs, this programs are said to satisfy the
\(K\)-relation \(R\) if and only if for any traces \(t_1, \dots, t_K\)
they can produce when linked with the same context, \((t_1, \dots,
t_K) \in R\). 
Here, we only give an explicit definition in the case of \(2\)-relations. These definitions
can be lifted trivially to the case of \(K\)-relations.

\begin{definition}[Robust 2-Relational Trace Property Preservation (\rtrtp)]\label{def:rtrtp}
  $$
  \begin{multlined}
    \criteria{\rtrtp}{rtrtp}:\quad
    \forall R \in 2^\ii{(Trace^2)}.~\forall\src{P_1}~\src{P_2}.~
    (\forall\src{C_S}~t_1~t_2 \ldotp
    (\src{C_S\hole{P_1} \rightsquigarrow}\ t_1
    \mathrel{\wedge}
    \src{C_S\hole{P_2} \rightsquigarrow}\ t_2)
    \Rightarrow (t_1,t_2) \in R)
    \Rightarrow\\
    (\forall \trg{C_T}~t_1~t_2 \ldotp
    (\trg{C_T\hole{\cmp{P_1}} \rightsquigarrow}\ t_1
    \mathrel{\wedge}
    \trg{C_T\hole{\cmp{P_2}} \rightsquigarrow}\ t_2)
    \Rightarrow (t_1, t_2) \in R)
  \end{multlined}
  $$
\end{definition}

The equivalent characterization captures the following intuition: if compiled programs
are unrelated by a relation \(R\) because of certain traces, they the source programs
are also unrelated because of the same traces.
\begin{definition}[Equivalent Characterization of \rtrtp (\pf{\rtrtp})]\label{def:rtrtc}
  \begin{align*}
    \pf{\rtrtp}:\quad \forall\src{P_1}~\src{P_2}.~ \forall\trg{C_T}.~ \forall t_1~t_2.~
    &(\trg{C_T\hole{\cmp{P_1}}} \rightsquigarrow t_1
    \mathrel{\wedge}
    \trg{C_T\hole{\cmp{P_2}}} \rightsquigarrow t_2)
    \Rightarrow\\
    &\exists \src{C_S}\ldotp (\src{C_S\hole{P_1}}\rightsquigarrow t_1
    \mathrel{\wedge}
    \src{C_S\hole{P_2}}\rightsquigarrow t_2)
  \end{align*}
\end{definition}

\begin{theorem}[\rtrtp and \pf{\rtrtp} are equivalent]\label{thm:rtrtp-eq}
  \begin{align*}
    \rtrtp \iff \pf{\rtrtp}
  \end{align*}
\end{theorem}
\begin{proof}
	See file Criteria.v, theorem R2rTC\_R2rTP.
\end{proof}


The definitions of \criteria{\rkrtp}{rkrtp} and \pf{\rkrtp} are an
easy generalization.

\begin{theorem}[\rkrtp and \pf{\rkrtp} are equivalent]\label{thm:krprc-rkrpp-eq}
  \begin{align*}
    \rkrtp \iff \pf{\rkrtp}
  \end{align*}
\end{theorem}
\begin{proof}
  Analogous to \Cref{thm:rtrtp-eq}.
\end{proof}

\subsubsection{Robust Relational Trace Property Preservation}\label{sec:rel-prop}

Relational trace properties (\oldautoref{sec:rel:trace}) are a
generalization of the previous relational trace properties, allowing
comparing individual runs of countably many programs. They are defined
as predicate over (infinite) sequence of programs.

\begin{definition}[Robust Relational Trace Property Preservation (\rrtp)]
 $$\begin{multlined}
\criteria{\rrtp}{rrtp}:\,
  \forall R \in 2^{(\ii{Trace}^\omega)}.~\forall\src{P_1},..,\src{P_K},...~\\
  ~(\forall\src{C_S} \ldotp \forall t_1, .., t_k, .. 
    (\forall i. \src{C_S \hole{P_i}} \rightsquigarrow t_i) \Rightarrow (t_1, .., t_k, ..) \in R) 
    \Rightarrow \\
    (\forall\trg{C_T} \ldotp \forall t_1, .., t_k, .. 
    (\forall i. \trg{C_T \hole{P_i}} \rightsquigarrow t_i) \Rightarrow (t_1, .., t_k, ..) \in R)
\end{multlined}$$
\end{definition}

\begin{definition}[Equivalent Property-Full Characterization of (\rrtpprimed)]\label{def:rrtp}
  $$
  \begin{multlined}
    \criteria{\rrtpprimed}{rrtpprimed}:\quad
    \forall R \in 2^{(\srcAll \to \ii{Trace})}.~
    (\forall\src{C_S} \ldotp \forall f \ldotp (\forall \src{P} \ldotp \src{C_S\hole{P}} \rightsquigarrow f(\src{P}))
    \Rightarrow f \in R) \Rightarrow \\
    (\forall \trg{C_T} \ldotp \forall f \ldotp (\forall \src{P} \ldotp \trg{C_T\hole{\cmp{P}}} \rightsquigarrow f(\src{P}))
    \Rightarrow f \in R)\quad
  \end{multlined}
  $$
\end{definition}

\begin{definition}[Equivalent Property-Free Characterization of \rrtp (\pf{\rrtp})]\label{def:rrtc}
  \begin{align*}
    \pf{\rrtp}:\quad \forall f: \srcAll \to \ii{Trace}.~ \forall\trg{C_T}.~
    &(\forall \src{P}.~\trg{C_T\hole{\cmp{P}}} \rightsquigarrow f(\src{P})) \Rightarrow\\
    &\exists \src{C_S}\ldotp (\forall \src{P}.~\src{C_S\hole{P}}\rightsquigarrow f(\src{P}))
  \end{align*}
\end{definition}

\begin{theorem}[\rrtpprimed and \pf{\rrtp} are equivalent]\label{thm:rrtpprimed-eq}
  \begin{align*}
    \rrtpprimed \iff \pf{\rrtp}
  \end{align*}
\end{theorem}
\begin{proof}
  See file Criteria.v, theorem RrTC\_RrTP'.
\end{proof}

\begin{theorem}[\rrtpprimed and \rrtp]\label{thm:rrtp-omega}
 Assuming the set $\srcAll$ is countable,
 \begin{align*}
    \rrtpprimed \iff \rrtp
  \end{align*}  
\end{theorem}

\begin{proof}
  Same argument used in \Cref{thm:rrhp-alter}.
\end{proof}

\begin{theorem}[\rrtp and \pf{\rrtp}]\label{thm:rrtp-eq}
 Assuming the set $\srcAll$ is countable,
   \begin{align*}
    \rrtp \iff \pf{\rrtp}
  \end{align*}
\end{theorem}
\begin{proof}
  Follows from \Cref{thm:rrtp-omega} and \Cref{thm:rrtpprimed-eq}.
\end{proof}

\subsubsection{Robust Relational Safety Preservation}\label{sec:rel-safety}
See \oldCref{sec:rel:safety} for a more detailed account of robust
relational safety preservation.

\subsubsection{Robust Finite-relational Safety Preservation}\label{sec:fin-rel-safety}

A relation $R \in 2^{Trace^K}$ is a $K$-ary relational safety property if for every ``bad'' $K$-trace $(t_1, \ldots, t_K) \not\in R$,
there exists a set of prefixes $m_1, \ldots, m_k \in \ii{FinPref}$ such that $m_i \leq t_i, ~i = 1, \ldots, K$,
and every $K$-trace $(t_1', \ldots, t_K')$ that extends the set of ``bad'' prefixes pointwise is also not in the relation, \IE 
$m_i \leq t_i', ~i = 1, \ldots, K$ implies $(t_1', \ldots, t_K') \not\in R$. \\ 
We provide the definition for preservation of the robust satisfaction
of relational safety of arity 2 (\Cref{def:rtrsp}), the reader can
easily deduce the definition for arity $K$, that we denote by $\rkrsp$.

\iffull At arity $K = 1$, this coincides with Robust Safety Property
Preservation, \rsp (cf. \autoref{sec:rsp}, \autoref{def:rsp}).\fi
At arity 2, we define \emph{Robust 2-relational Safety Preservation}
(\criteria{\rtrsp}{rtrsp}) as follows (cf. \autoref{def:rtrsp}).
\begin{definition}[Robust 2-Relational Safety Preservation (\rtrsp)]\label{def:rtrsp}
  \begin{align*}
    \rtrsp:~&\forall R \in \text{2-}\ii{relational ~Safety}.~\forall\src{P_1}~\src{P_2}. \\
    &\left(
    \forall \src{C_S}~t_1~t_2.
    (\src{C_s\hole{P_1} \sem}\ t_1 \wedge \src{C_s\hole{P_2} \sem}\ t_2) \implies (t_1, t_2) \in R
    \right)
    \implies\\
    &\left(
    \forall \trg{C_T}~t_1~t_2.
    (\trg{C_T\hole{\cmp{P_1}} \sem}\ t_1 \wedge \trg{C_T\hole{\cmp{P_2}} \sem}\ t_2) \implies
    (t_1, t_2) \in R
    \right)
  \end{align*}
\end{definition}
%
%
We show that $\rtrsp$ can be written in the following form, more
convenient to work with.
\begin{definition}[Equivalent Characterization of \rtrsp (\pf{\rtrsp})]\label{def:rtrsc}
  \begin{align*}
    \pf{\rtrsp} :\quad \forall\src{P_1}~\src{P_2}.~ \forall\trg{C_T}.~ \forall m_1~m_2.~&
    (\trg{C_T\hole{\cmp{P_1}}} \rightsquigarrow m_1
    \mathrel{\wedge}
    \trg{C_T\hole{\cmp{P_2}}} \rightsquigarrow m_2)
    \Rightarrow\\
    &\exists \src{C_S}\ldotp (\src{C_S\hole{P_1}}\rightsquigarrow m_1
    \mathrel{\wedge}
    \src{C_S\hole{P_2}}\rightsquigarrow m_2)
  \end{align*}
\end{definition}
\begin{theorem}[\rtrsp and \pf{\rtrsp} are equivalent]\label{thm:rtrsp-eq}
  \begin{align*}
    \rtrsp \iff \pf{\rtrsp}
  \end{align*}
\end{theorem}
\begin{proof}
  See Criteria.v, theorem R2rSC\_R2rSP.
\end{proof}

The following theorem gives us an alternative formulation of $\rtrsp$,
in terms of preservation of robust satisfaction of relations over
finite prefixes. \Cref{thm:rtrspA} allows us to define in a more
elegant way the criteria for arbitrary (but finite) relational safety
(\Cref{def:pf-rfrsp}) as well infinite ones (\Cref{def:rrsp}).
A similar theorem holds for $\rkrsp$.

\begin{theorem}[Characterization of \rtrsp using finite prefixes]\label{thm:rtrspA}
  \begin{align*}
    \rtrsp \iff &\forall R \in 2^{(\ii{FinPref}^2)}.~\forall\src{P_1}~\src{P_2}. \\
    &\left(
    \forall \src{C_S}~m_1~m_2.
    (\src{C_s\hole{P_1} \sem}\ m_1 \wedge \src{C_s\hole{P_2} \sem}\ m_2) \implies (m_1, m_2) \in R
    \right)
    \implies\\
    &\left(
    \forall \trg{C_T}~m_1~m_2.
    (\trg{C_T\hole{\cmp{P_1}} \sem}\ m_1 \wedge \trg{C_T\hole{\cmp{P_2}} \sem}\ m_2) \implies
    (m_1, m_2) \in R
    \right)
  \end{align*}
\end{theorem}

\begin{proof}
  See Criteria.v, Theorem R2rSP\_R2rSC'.
\end{proof}

%
Notice that in the second script we quantify over arbitrary relations over finite prefixes. This captures the main difference between this criterion and the stronger \iffull 2-Relational
Trace Property Preservation,\fi \rtrtp: it considers finite prefixes rather than full traces. This is
also the case in the equivalent property free characterization, \pf{\rtrsp}.
\iffull
\rtrsp generalizes from arity 2 to any finite arity in the obvious
way. It also generalizes to the infinite arity (preservation of
predicates on all programs of the language, that can be characterized
by relations on finite prefixes). We call this \emph{Robust Relational
  Safety Preservation} or \rrsp. The definition and its property-free
characterization \pf{\rrsp} (\CoqSymbol) are identical to \rrtp and
\pf{\rrtp} (\autoref{sec:rel:trace}), respectively, except that we
replace traces with finite prefixes everywhere.
\fi

The definitions of \criteria{\rkrsp}{rkrsp} and \pf{\rkrsp} are an
easy generalization.

\begin{theorem}[\rkrsp and \pf{\rkrsp} are equivalent]\label{thm:krsrc-rkrsp-eq}
  \begin{align*}
    \rkrsp \iff \pf{\rkrsp}
  \end{align*}
\end{theorem}
\begin{proof}
  Analogous to \Cref{thm:rtrsp-eq}.
\end{proof}

Finally, we define \rfrsp as the union over all \(K\) of the \rkrsp%
\!\!s. \jt{ugly hack}

\begin{definition}[Robust Finite-relational Safety Preservation (\rfrsp)]\label{def:pf-rfrsp}
  \begin{align*}
    \rfrsp:\quad
    &
    \forall K,\src{P_1},\cdots,\src{P_k},\com{R} \in 2^{(\ii{FinPref^k})}.
    \\
    &
    (\forall \src{C_S}, \com{m_1}, \cdots, \com{m_k}, 
    (
    \src{C_S\hole{P_1}\sem\com{m_1}} 
    \wedge \cdots \wedge
	   {\src{C_S\hole{P_k}\sem\com{m_k}}}
	   )
	   \\
	   &\quad \Rightarrow
	   (\com{m_1},\cdots,\com{m_k})\in\com{R}) \Rightarrow\\	   
	   &
	   (\forall \trg{C_T}.
	   (
	       {\trg{C_T\hole{\compgen{\src{P_1}}}\sem\com{m_1}}}
	       \wedge \cdots \wedge {\trg{C_T\hole{\compgen{\src{P_k}}}\sem\com{m_k}}}
	       )
	       \\
	       &\quad
	       \Rightarrow (\com{m_1},\cdots,\com{m_k})\in\com{R})
  \end{align*}
\end{definition}
The intuition for this property-free criterion is the same as
for finite-relational properties, except it only requires considering
finite prefixes.
\begin{definition}[Equivalent Characterization of \rfrsp (\pf{\rfrsp})]\label{def:rfrsp}
  \begin{align*}
    \pf{\rfrsp}:\quad \forall K.~\forall\src{P_1}\ldots\src{P_K}.~ \forall\trg{C_T}.~ \forall m_1 \ldots m_K.~
    &(\trg{C_T\hole{\cmp{P_1}} \rightsquigarrow}\ m_1
    \mathrel{\wedge} \ldots \mathrel{\wedge}
    \trg{C_T\hole{\cmp{P_K}} \rightsquigarrow}\ m_K)
    \Rightarrow\\
    &\exists \src{C_S}\ldotp (\src{C_S\hole{P_1}\rightsquigarrow}\ m_1
    \mathrel{\wedge} \ldots \mathrel{\wedge}
    \src{C_S\hole{P_K}\rightsquigarrow}\ m_K)
  \end{align*}
\end{definition}

\begin{theorem}[\rfrsp and \pf{\rfrsp} are equivalent]\label{thm:rfrsp-eq}
  \begin{align*}
    \rfrsp \iff \pf{\rfrsp}
  \end{align*}
\end{theorem}
\begin{proof}
	Analogous to \Cref{thm:rtrsp-eq}.
\end{proof}




\begin{definition}[Robust relational Safety Preservation (\rrsp)]\label{def:rrsp}
  $$
  \begin{multlined}
    \criteria{\rrsp}{rrsp}:\quad
    \forall R \in 2^{(\ii{FinPref})^{\omega}}. \forall P_1,.., P_k, ..
    (\forall\src{C_S} \ldotp \forall m_1, .., m_k, .. \ldotp (\forall i \ldotp \src{C_S\hole{P_i}} \rightsquigarrow m_i) 
     \Rightarrow (m_1, .., m_k, ..) \in R)
    \Rightarrow\\
    (\forall\trg{C_T} \ldotp \forall m_1, .., m_k, .. \ldotp (\forall i \ldotp \trg{C_T\hole{\cmp{P_i}}} \rightsquigarrow m_i) 
     \Rightarrow (m_1, .., m_k, ..) \in R)\quad
  \end{multlined}
  $$
\end{definition}

\begin{definition}[Equivalent Property-Full Characterization of \rrspprimed (\rrspprimed)]\label{def:rrspprimed}
  $$
  \begin{multlined}
    \criteria{\rrspprimed}{rrspprimed}:\quad
    \forall R \in 2^{(\srcAll \to \ii{FinPref})}.~
    (\forall\src{C_S} \ldotp \forall f \ldotp (\forall \src{P} \ldotp \src{C_S\hole{P}} \rightsquigarrow f(\src{P}))
    \Rightarrow R(f))
    \Rightarrow\\
    (\forall \trg{C_T} \ldotp \forall f \ldotp (\forall \src{P} \ldotp \trg{C_T\hole{\cmp{P}}} \rightsquigarrow f(\src{P}))
    \Rightarrow R(f))\quad
  \end{multlined}
  $$
\end{definition}

\begin{definition}[Equivalent Property-Free Characterization of \rrsp (\pf{\rrsp})]\label{def:rrsc}
  \begin{align*}
    \pf{\rrsp} :\quad \forall f: \srcAll \to \ii{FinPref}.~ \forall\trg{C_T}.~
    &(\forall \src{P}.~\trg{C_T\hole{\cmp{P}}} \rightsquigarrow f(\src{P})) \Rightarrow\\
    &\exists \src{C_S}\ldotp (\forall \src{P}.~\src{C_S\hole{P}}\rightsquigarrow f(\src{P}))
  \end{align*}
\end{definition}

\begin{theorem}[\pf{\rrsp} and \rrspprimed are equivalent]\label{thm:rrspprimed-eq}
  \begin{align*}
    \rrspprimed \iff \pf{\rrsp}
  \end{align*}
\end{theorem}
\begin{proof}
	See file Criteria.v, theorem RrSC\_RrSP'.
\end{proof}

\begin{theorem}[\rrspprimed and \rrsp]\label{thm:rrsp-omega}
 Assuming the set $\srcAll$ is countable,
 \begin{align*}
    \rrspprimed \iff \rrsp
  \end{align*}  
\end{theorem}

\begin{proof}
  Same argument used in \Cref{thm:rrhp-alter}.
\end{proof}

\begin{theorem}[\rrsp and \pf{\rrsp}]\label{thm:rrsp-eq}
Assuming the set $\srcAll$ is countable,
   \begin{align*}
    \rrsp \iff \pf{\rrsp}
  \end{align*}
\end{theorem}
\begin{proof}
  Follows from \Cref{thm:rrsp-omega} and \Cref{thm:rrspprimed-eq}.
\end{proof}

\subsubsection{Robust Relational Relaxed Safety Preservation}\label{sec:rel-xsafety}

Relational Relaxed Safety properties generalize relational safety properties
as they consider $\ii{XPref}$ instead of $\ii{FinPref}$. The semantics of a programming
language can capture more than just finite prefixes of complete traces justifying the definition of $\ii{Xpref}$ (see \Cref{sec:not}).
For instance, silent divergence is not finitely observable, but can be represented and produced by small-step semantics, for
instance with a rule such as \(\typerule{Silent-Div}{\forall n, e
  \xto{\epsilon}^n e'}{e \sem\Uparrow}{ex-div}\) where
\(\Uparrow\) represents silent divergence. \\ \\ 
%
%
The criteria in this section are defined exactly as the criteria in the previous section,
except they deal with extended prefixes instead of finite prefixes.

\begin{definition}[Robust relational relaXed safety Preservation (\rrxp)]\label{def:rrxp}
  $$
  \begin{multlined}
    \criteria{\rrsp}{rrsp}:\quad
    \forall R \in 2^{(\ii{XPref})^{\omega}}. \forall P_1,.., P_k, ..
    (\forall\src{C_S} \ldotp \forall x_1, .., x_k, .. \ldotp (\forall i \ldotp \src{C_S\hole{P_i}} \rightsquigarrow x_i) 
     \Rightarrow (x_1, .., x_k, ..) \in R)
    \Rightarrow\\
    (\forall\trg{C_T} \ldotp \forall x_1, .., x_k, .. \ldotp (\forall i \ldotp \trg{C_T\hole{\cmp{P_i}}} \rightsquigarrow x_i) 
     \Rightarrow (x_1, .., x_k, ..) \in R)\quad
  \end{multlined}
  $$
\end{definition}

\begin{definition}[ (\rrxpprimed)]\label{def:rrxpprimed}
  $$
  \begin{multlined}
    \criteria{\rrxp}{rrxp}:\quad
    \forall R \in 2^{(\srcAll \to \ii{XPref})}.~
    (\forall\src{C_S} \ldotp \forall f \ldotp (\forall \src{P} \ldotp \src{C_S\hole{P}} \rightsquigarrow f(\src{P}))
    \Rightarrow R(f))
    \Rightarrow\\
    (\forall \trg{C_T} \ldotp \forall f \ldotp (\forall \src{P} \ldotp \trg{C_T\hole{\cmp{P}}} \rightsquigarrow f(\src{P}))
    \Rightarrow R(f))
  \end{multlined}
  $$
\end{definition}

\begin{definition}[Equivalent Characterization of \rrxpprimed (\pf{\rrxp})]\label{def:rrxc}
  \begin{align*}
    \pf{\rrxp} :\quad \forall f: \srcAll \to \ii{XPref}.~ \forall\trg{C_T}.~
    &(\forall \src{P}.~\trg{C_T\hole{\cmp{P}}} \rightsquigarrow f(\src{P})) \Rightarrow\\
    &\exists \src{C_S}\ldotp (\forall \src{P}.~\src{C_S\hole{P}}\rightsquigarrow f(\src{P}))
  \end{align*}
\end{definition}

\begin{theorem}[\pf{\rrxp} and \rrxpprimed are equivalent]\label{thm:rrxpprimed-eq}
  \begin{align*}
    \rrxp \iff \pf{\rrxp}
  \end{align*}
\end{theorem}
\begin{proof}
	See file Criteria.v, theorem RrXC\_RrXP'.
\end{proof}

\begin{theorem}[\rrxpprimed and \rrxp]\label{thm:rrxp-omega}
 Assuming the set $\srcAll$ is countable,
 \begin{align*}
    \rrxpprimed \iff \rrxp
  \end{align*}  
\end{theorem}

\begin{proof}
  Same argument used in \Cref{thm:rrhp-alter}.
\end{proof}

\begin{theorem}[\rrxp and \pf{\rrxp}]\label{thm:rrxp-eq}
 Assuming the set $\srcAll$ is countable,
   \begin{align*}
    \rrxp \iff \pf{\rrxp}
  \end{align*}
\end{theorem}
\begin{proof}
  Follows from \Cref{thm:rrxp-omega} and \Cref{thm:rrxpprimed-eq}.
\end{proof}

\subsubsection{Robust Finite-Relational Relaxed Safety Preservation}\label{sec:fin-rel-Xsafety}

\ch{Broke the way the \rfrxp looks below, since I had to take it
  out of align to properly add a label}

\begin{definition}[Robust Finite-relational relaXed safety Preservation (\rfrxp)]\label{def:pf-rfrxp}
~\\
  $
    \criteria{\rfrxp}{rfrxp}:\quad
  $
  \begin{align*}
    &
    \forall K,\src{P_1},\cdots,\src{P_K},\com{R} \in 2^{(\ii{XPref^K})}.
    \\
    &
    (\forall \src{C_S}, \com{x_1}, \cdots, \com{x_K}, 
    (
    \src{C_S\hole{P_1}\sem\com{x_1}} 
    \wedge \cdots \wedge
	   {\src{C_S\hole{P_K}\sem\com{x_K}}}
	   )
	   \\
	   &\quad \Rightarrow
	   (\com{x_1},\cdots,\com{x_K})\in\com{R}) \Rightarrow\\	   
	   &
	   (\forall \trg{C_T}.
	   (
	       {\trg{C_T\hole{\compgen{\src{P_1}}}\sem\com{x_1}}}
	       \wedge \cdots \wedge {\trg{C_T\hole{\compgen{\src{P_K}}}\sem\com{x_K}}}
	       )
	       \\
	       &\quad
	       \Rightarrow (\com{x_1},\cdots,\com{x_K})\in\com{R})
  \end{align*}
\end{definition}

\begin{definition}[Equivalent Characterization of \rfrxp (\pf{\rfrxp})]\label{def:rfrxp}
  \begin{align*}
    \pf{\rfrxp}:\quad \forall K.~\forall\src{P_1}\ldots\src{P_K}.~ \forall\trg{C_T}.~ \forall x_1 \ldots x_K.~
    &(\trg{C_T\hole{\cmp{P_1}} \rightsquigarrow}\ x_1
    \mathrel{\wedge} \ldots \mathrel{\wedge}
    \trg{C_T\hole{\cmp{P_K}} \rightsquigarrow}\ x_K)
    \Rightarrow\\
    &\exists \src{C_S}\ldotp (\src{C_S\hole{P_1}\rightsquigarrow}\ x_1
    \mathrel{\wedge} \ldots \mathrel{\wedge}
    \src{C_S\hole{P_K}\rightsquigarrow}\ x_K)
  \end{align*}
\end{definition}

\begin{theorem}[\rfrxp and \pf{\rfrxp} are equivalent]\label{thm:rfrxp-eq}
  \begin{align*}
    \rfrxp \iff \pf{\rfrxp}
  \end{align*}
\end{theorem}
\begin{proof}
	Analogous to \Cref{thm:rtrxp-eq}.
\end{proof}

\begin{definition}[Robust 2-relational relaXed safety Preservation (\rtrxp)]\label{def:rtrxp}
  $$
  \begin{multlined}
    \criteria{\rtrxp}{rtrxp}:\quad
    \forall R \in 2^\ii{(XPref^2)}.~\forall\src{P_1}~\src{P_2}.~
    (\forall\src{C_S}~x_1~x_2 \ldotp
    (\src{C_S\hole{P_1} \rightsquigarrow}\ x_1
    \mathrel{\wedge}
    \src{C_S\hole{P_2} \rightsquigarrow}\ x_2)
    \Rightarrow (x_1,x_2) \in R)
    \Rightarrow\\
    (\forall \trg{C_T}~x_1~x_2 \ldotp
    (\trg{C_T\hole{\cmp{P_1}} \rightsquigarrow}\ x_1
    \mathrel{\wedge}
    \trg{C_T\hole{\cmp{P_2}} \rightsquigarrow}\ x_2)
    \Rightarrow (x_1, x_2) \in R)
  \end{multlined}
  $$
\end{definition}

\begin{definition}[Equivalent Characterization of \rtrxp (\pf{\rtrxp})]\label{def:rtrxc}
  \begin{align*}
    \pf{\rtrxp} :\quad \forall\src{P_1}~\src{P_2}.~ \forall\trg{C_T}.~ \forall x_1~x_2.~&
    (\trg{C_T\hole{\cmp{P_1}}} \rightsquigarrow x_1
    \mathrel{\wedge}
    \trg{C_T\hole{\cmp{P_2}}} \rightsquigarrow x_2)
    \Rightarrow\\
    &\exists \src{C_S}\ldotp (\src{C_S\hole{P_1}}\rightsquigarrow x_1
    \mathrel{\wedge}
    \src{C_S\hole{P_2}}\rightsquigarrow x_2)
  \end{align*}
\end{definition}
\begin{theorem}[\rtrxp and \pf{\rtrxp} are equivalent]\label{thm:rtrxp-eq}
  \begin{align*}
    \rtrxp \iff \pf{\rtrxp}
  \end{align*}
\end{theorem}
\begin{proof}
  See file Criteria.v, R2rXC\_R2rXP'.
\end{proof}


\begin{theorem}[\rkrxp and \pf{\rkrxp} are equivalent]\label{thm:krxrc-rkrxp-eq}
  \begin{align*}
    \rkrxp \iff \pf{\rkrxp}
  \end{align*}
\end{theorem}
\begin{proof}
	Analogous to \Cref{thm:rtrxp-eq}.
\end{proof}

\subsection{Relational Hyperproperty-Based Criteria}\label{sec:rel-criteria}
\jt{Still very confused about these functions}

\subsubsection{Robust Relational Hyperproperty Preservation}\label{sec:rel-hyper}

\begin{definition}[Robust Relational Hyperproperty Preservation (\rrhp)]\label{def:rrhp}
  $$\begin{multlined}
    \criteria{\rrhp}{rrhp}:\quad
    \forall R \in 2^{(\behavAll^\omega)}.~\forall\src{P_1},..,\src{P_K},...~
    (\forall\src{C_S} \ldotp
    (\src{\behav{{C_S}\hole{P_1}}},..,\src{\behav{C_S\hole{P_K}}},..) \in R)
    \Rightarrow\\
    (\forall\trg{C_T} \ldotp
    (\trg{\behav{{C_T}\hole{\cmp{P_1}}}},..,\trg{\behav{C_S\hole{\cmp{P_K}}}},..) \in R)
  \end{multlined}$$
\end{definition}

\rrhp has an equivalent definition (\rrhpA below) that is also \emph{not}
property-free. We introduce this definition for technical reasons as
we use it in proofs later.

\begin{definition}[(\rrhp')]\label{def:rrhpprimed}
  $$\begin{multlined}
    \rrhpA:\quad
    \forall R \in 2^{(\srcAll \to \behavAll)}.~
    (\forall\src{C_S} \ldotp
    (\lambda \src{P}\ldotp \src{\behav{C_S\hole{P}}}) \in R)
    \Rightarrow\\
    (\forall \trg{C_T} \ldotp (\lambda \src{P}\ldotp \trg{\behav{C_T\hole{\cmp{P}}}})
    \in R)\quad
  \end{multlined}$$
\end{definition}

\begin{theorem}[\rrhp and \rrhpprimed are equivalent]\label{thm:rrhp-alter}
Assuming the set $\srcAll$ is countable,
	\begin{align*}
		\rrhp \iff \rrhpprimed
	\end{align*}
\end{theorem}
\begin{proof}
($\Rightarrow$) Following the definition of \rrhpA, assume $R \in
  2^{(\srcAll \to \behavAll)}$. We need to prove:
  $$\begin{multlined}
    (\forall\src{C_S} \ldotp
    (\lambda \src{P}\ldotp \src{\behav{C_S\hole{P}}}) \in R)
    \Rightarrow\\
    (\forall \trg{C_T} \ldotp (\lambda \src{P}\ldotp \trg{\behav{C_T\hole{\cmp{P}}}})
    \in R)
  \end{multlined}$$
  Let $\mathtt{G}$ be a bijective function from
  source programs to $\Nat$.\footnote{When the source language has
    fewer programs than $\omega$, the proof isn't very different.}
  Define $R' \in 2^{(\behavAll^\omega)}$ as follows:
  $$
  R' = \{ (b_1,..,b_k,..) ~|~ (\lambda \src{P}. b_\mathtt{G(\src{P})}) \in R\}
  $$
  For $i \in \Nat$, let $\src{Q_i} = \mathtt{G}^{-1}(i)$. Instantiate
  \rrhp to $R'$ and $\src{Q_1},..,\src{Q_K},..$. We get:
  $$\begin{multlined}
    (\forall\src{C_S} \ldotp
    (\src{\behav{{C_S}\hole{Q_1}}},..,\src{\behav{C_S\hole{Q_K}}},..) \in R')
    \Rightarrow\\
    (\forall\trg{C_T} \ldotp
    (\trg{\behav{{C_T}\hole{\cmp{Q_1}}}},..,\trg{\behav{C_S\hole{\cmp{Q_K}}}},..) \in R')
  \end{multlined}$$
  Plugging in the definition of $R'$ above, this becomes:
  $$\begin{multlined}
    (\forall\src{C_S} \ldotp
    (\lambda \src{P}\ldotp \src{\behav{{C_S}\hole{Q_{\mathtt{G}(P)}}}}) \in R)
    \Rightarrow\\
    (\forall\trg{C_T} \ldotp
    (\lambda \src{P}\ldotp \trg{\behav{{C_T}\hole{\cmp{Q_{\mathtt{G}(P)}}}}}) \in R)
  \end{multlined}$$
  However, by definition, $\src{Q_{\mathtt{G}(P)}} = \src{P}$. So, the above is equal to
  $$\begin{multlined}
    (\forall\src{C_S} \ldotp
    (\lambda \src{P}\ldotp \src{\behav{{C_S}\hole{P}}}) \in R)
    \Rightarrow\\
    (\forall\trg{C_T} \ldotp
    (\lambda \src{P}\ldotp \trg{\behav{{C_T}\hole{\cmp{P}}}}) \in R)
  \end{multlined}$$
  which is what we had to prove.

  \medskip
  
  \noindent ($\Leftarrow$)
  Following the definition of \rrhp, assume $R \in 2^{(\behavAll^\omega)}$ and some infinite sequence $\src{P_1},..,\src{P_K},..$. We have to show:
  $$\begin{multlined}
    (\forall\src{C_S} \ldotp
    (\src{\behav{{C_S}\hole{P_1}}},..,\src{\behav{C_S\hole{P_K}}},..) \in R)
    \Rightarrow\\
    (\forall\trg{C_T} \ldotp
    (\trg{\behav{{C_T}\hole{\cmp{P_1}}}},..,\trg{\behav{C_S\hole{\cmp{P_K}}}},..) \in R)
  \end{multlined}$$
  Define $R' \in 2^{(\srcAll \to \behavAll)}$ as follows:
  \[ R' = \{ f ~|~ (f(\src{P_1}),..,f(\src{P_K}),..) \in R \}
  \]
  Instantiating \rrhpA to $R'$, we get:
  $$\begin{multlined}
    (\forall\src{C_S} \ldotp
    (\lambda \src{P}\ldotp \src{\behav{C_S\hole{P}}}) \in R')
    \Rightarrow\\
    (\forall \trg{C_T} \ldotp (\lambda \src{P}\ldotp \trg{\behav{C_T\hole{\cmp{P}}}})
    \in R')
  \end{multlined}$$
  Expanding the definition of $R'$, this immediately reduces to what we wanted to show.
\end{proof}

\begin{definition}[Equivalent Characterization of \rrhp (\pf{\rrhp})]\label{def:rrhc}
	\begin{align*}
		\pf{\rrhp}:\quad
\forall\trg{C_T}.~ \exists \src{C_S}.~\forall \src{P}.~
  \trg{\behav{C_T\hole{\cmp{P}}}} = \src{\behav{C_S\hole{P}}}
	\end{align*}
\end{definition}

\begin{theorem}[\rrhpA and \pf{\rrhp} are equivalent]\label{thm:rrhpprimed-eq}
	\begin{align*}
		\rrhpA \iff \pf{\rrhp}
	\end{align*}
\end{theorem}
\begin{proof}
	See file Criteria.v, theorem RrHC\_RrHP'.
\end{proof}

\begin{theorem}[\rrhp and \pf{\rrhp}]\label{thm:rrhp-eq}
  Assuming the set $\srcAll$ is countable,
   \begin{align*}
    \rrhp \iff \pf{\rrhp}
  \end{align*}
\end{theorem}
\begin{proof}
  Follows from \Cref{thm:rrhp-alter} and \Cref{thm:rrhpprimed-eq}.
\end{proof}

\subsubsection{Robust \texorpdfstring{\(K\)}{K}-Relational Hyperproperty Preservation}\label{sec:k-rel-hyper}

\(K\)-relational hyperproperties are relations between the behaviors
of several programs. \(K\) programs satisfy a \(K\)-relational hyperproperty if and only if, when
plugged into any same context, their behaviors are related.

The criteria are as expected, generalizing the intuition of hyperproperties for multiple programs.
\begin{definition}[Robust 2-Relational Hyperproperty Preservation (\rtrhp)]\label{def:rtrhp}
	$$\begin{multlined}
		\criteria{\rtrhp}{rtrhp}:\quad
  \forall R \in 2^{(\behavAll^2)}.~\forall\src{P_1}~\src{P_2}.~
    (\forall\src{C_S} \ldotp
  (\src{\behav{{C_S}\hole{P_1}}},\,\src{\behav{C_S\hole{P_2}}}) \in R)
    \Rightarrow\\
    (\forall\trg{C_T} \ldotp
  (\trg{\behav{{C_T}\hole{\cmp{P_1}}}},\,\trg{\behav{C_S\hole{\cmp{P_2}}}}) \in R)\quad
	\end{multlined}$$
\end{definition}

\begin{definition}[Equivalent Characterization of \rtrhp (\pf{\rtrhp})]\label{def:rtrhc}
	\begin{align*}
		\pf{\rtrhp}:\quad \forall\src{P_1}~\src{P_2}.~ \forall\trg{C_T}.~ \exists \src{C_S}.~&
  \trg{\behav{C_T\hole{\cmp{P_1}}}} = \src{\behav{C_S\hole{P_1}}} \mathrel{\wedge} \\
  & \trg{\behav{C_T\hole{\cmp{P_2}}}} = \src{\behav{C_S\hole{P_2}}}
	\end{align*}
\end{definition}

\begin{theorem}[\rtrhp and \pf{\rtrhp} are equivalent]\label{thm:rtrhp-eq}
	\begin{align*}
		\rtrhp \iff \pf{\rtrhp}
	\end{align*}
\end{theorem}
\begin{proof}
	See file Criteria.v, theorem R2rHC\_R2rHP.
\end{proof}

To obtain \criteria{\rkrhp}{rkrhp} and \pf{\rkrhp},
take the definitions of \rtrhp and \pf{\rtrhp} above
and replace $\forall\src{P_1},\src{P_2}$ with $\forall\src{P_1},\cdots,\src{P}_K$.

\begin{theorem}[\rkrhp and \pf{\rkrhp} are equivalent]\label{thm:rkrhp-eq}
	\begin{align*}
		\rkrhp \iff \pf{\rkrhp}
	\end{align*}
\end{theorem}
\begin{proof}
	Analogous to \Cref{thm:rtrhp-eq}.
\end{proof}


%

\subsection{Comparison of Proof Obligations}\label{sec:obligations}
We briefly compare the robust preservation of (variants of) relational
hyperproperties (\rrhp), 
relational trace
properties (\rrtp), 
and relational safety
properties (\rrsp, this subsection) in terms of the difficulty of
back-translation proofs.
For this, it is
instructive to look at the property-free characterizations. In a proof
of \rrsp or any of its variants, we must construct a source context
$\src{C_S}$ that can induce a given set of \emph{finite prefixes of
  traces}, one from each of the programs being related. In \rrtp and
its variants, this obligation becomes harder---now the constructed
$\src{C_S}$ must be able to induce a given set of \emph{full
  traces}. In \rrhp and its variants, the obligation is even
harder---$\src{C_S}$ must be able to induce entire behaviors (sets of
traces) from each of the programs being related. Thus, the increasing
strength of \rrsp, \rrtp and \rrhp is directly reflected in
their corresponding proof obligations.

Furthermore, looking just at the different variants of relational safety,
we note that the number of trace
prefixes the constructed context $\src{C_S}$ must simultaneously
induce in the source programs is exactly the arity of the corresponding relational
property. Constructing $\src{C_S}$ from a finite number of prefixes is
much easier than constructing $\src{C_S}$ from an infinite number of
prefixes. Consequently, it is meaningful to define a special point in
the partial order of \autoref{fig:order} that is the join of $\rkrsp$
for all finite $K$s.
This point is the criterion we call \emph{Robust Finite-Relational
  Safety Preservation} (see \Cref{sec:fin-rel-safety}), or \rfrsp.
%


\clearpage
\section{Which of our Criteria Imply
  Robust Trace Equivalence Preservation?}
This section extends \oldautoref{thm:rhp-doesnt-imply-r2rsp} from
\oldautoref{sec:fa}. While \rtep is always implied by \rtrhp,
we show that in many cases, \rtep is a consequence of weaker relational criteria.

\subsection{Relational Criteria and Robust Trace Equivalence Preservation}

We start by recalling the definition of \criteria{\rtep}{rtep}, which is an instance of \rtrhp:
\begin{definition}[Robust Trace Equivalence Preservation (\rtep)]\label{def:rtep}
	$$\begin{multlined}
		\criteria{\rtep}{rtep}:\quad
  \forall\src{P_1}~\src{P_2}.~
    (\forall\src{C_S} \ldotp
  \src{\behav{{C_S}\hole{P_1}}}=\src{\behav{C_S\hole{P_2}}})
    \Rightarrow\\
    (\forall\trg{C_T} \ldotp
  \trg{\behav{{C_T}\hole{\cmp{P_1}}}}=\trg{\behav{C_S\hole{\cmp{P_2}}}})
	\end{multlined}$$
\end{definition}

In general, as explained in \oldautoref{sec:fa}, \rtep is implied by \rtrhp.
\begin{theorem} $\rtrhp \Rightarrow \rtep$. 
\end{theorem}
\begin{proof} The thesis immediately follows by instantiating $\rtrhp$ with the equality relation, 
              see file Robustdef.v, theorem R2rHP\_RTEP for a formal proof. 
\end{proof}

Similarly, in a deterministic setting, \rtep is implied by \rtrtp.
\begin{theorem} For deterministic source languages $\rtrtp \Rightarrow \rtep$.
\end{theorem}
\begin{proof} See file Criteria.v, theorem  R2rTP\_RTEP.
\end{proof} 

The determinism of the source language is a strong assumption though.
We show that \rtrtp (and even the weaker \rtrxp) imply \rtep even if
we weaken the determinism assumption to just determinacy, if we add
two more assumptions on the target language: input totality, and
``safety-like'' behavior.

\begin{definition}[Determinate Languages]\label{def:determinate}
We say a language is determinate \textit{iff} 
\begin{equation*}
 \forall W. ~\forall t_1 ~t_2.~  W \rightsquigarrow t_1 \wedge W \rightsquigarrow t_1 \Rightarrow
 t_1  ~\mathcal{R}  ~t_2
\end{equation*}
where  
\begin{equation*}
  \begin{split}
   t_1  ~\mathcal{R}  ~t_2 \iff & t_1 = t_2 ~\vee \\
                               & ~\exists m. ~\exists e_1 e_2 \in \mathit{Input}. 
                                 ~ e_1 \neq e_2 ~ \wedge ~ m ~:: ~e_1 \leq t_1
                                 ~\wedge ~m ~:: ~e_2 \leq t_2
  \end{split}                                                                       
\end{equation*}
\end{definition}

Intuitively, determinacy states that a language has no internal
non-determinism, or equivalently that the only source of
non-determinism is the inputs from the environment.

\begin{definition}[Input Totality] \label{def:intot}
  We say a language satisfies input totality \textit{iff} 
  \begin{equation*}
    \forall W. ~\forall m. ~\forall e_1 e_2 \in \mathit{Input}. 
     ~W \rightsquigarrow^* m ~:: ~e_1 \Rightarrow W \rightsquigarrow^* m ~:: ~e_2
  \end{equation*}
\end{definition}
Intuitively, input totality states that whenever a program receives an input from the environment,
then it could have received any other input as well.
Both determinacy~\cite{Engelfriet85, ChevalCD13} and input
totality~~\cite{ZakinthinosL97, FocardiG95} are standard assumptions
and are for instance satisfied by the CompCert compiler~\cite{Leroy09}.

\begin{definition}[``Safety-like'' semantics]
Given a language $\LL$, it semantics is ``safety-like'' \emph{iff}
\begin{equation*}
\forall W. \forall t \text{~infinite}. W \cancel{\rightsquigarrow} t \Rightarrow
\exists m. \exists e. W \rightsquigarrow m \wedge m :: e \leq t \wedge W \cancel{\rightsquigarrow} m :: e
\end{equation*}
\end{definition}

Intuitively, any infinite trace that cannot not produced by a program can be
explained as a finite prefix of that trace that \emph{can} produced by the
program, but after which the next event can no longer be produced by it.
While this property is non-trivial, in \autoref{sec:semSafe} we show
that any small-step semantics satisfying a particular kind of
determinacy always satisfies this property.

We can now state the following theorems:
\begin{theorem} \label{thm:rtrtpRTEP} If the following assumptions hold
  \begin{enumerate}
  \item The source language is determinate.
  \item The target language satisfies input totality.
  \item The target language is ``safety-like''.
  \end{enumerate}
then $\rtrtp \Rightarrow \rtep$.
\end{theorem}

\begin{proof}
  We give a sketch of the proof here, see file R2rTP\_RTEP.v, theorem
  R2rTP\_RTEP for a complete proof.

  Two contextually equivalent programs $\src{P_1}, \src{P_2}$ have the same behavior in any source context.
  This means that for any two traces $t_1\in \behav{\src{C_s\hole{P_1}}}$ and $t_2 \in \behav{\src{C_s\hole{P_2}}}$,
  since $\behav{\src{C_s\hole{P_1}}} = \behav{\src{C_s\hole{P_2}}}$
  we can use the determinacy of the source language (1) to obtain that $t_1 \mathcal{R} ~t_2$, for the relation $\mathcal{R}$ used to define determinacy.
  This allows us to instantiate $\rtrtp$ with the relational property $\mathcal{R}$ and deduce that for an arbitrary target context $\trg{C_T}$, 
  programs $\trg{C_T\hole{\cmp{P_1}}}$ and $\trg{C_T\hole{\cmp{P_2}}}$ can only produce traces also related by $\mathcal{R}$.
  Together with hypotheses 2 and 3 this is enough to show mutual inclusion of the target behaviors. 
  We show that for an arbitrary $\trg{C_T}$, $\behav{\trg{C_T}\hole{\cmp{P_1}}} \subseteq \behav{\trg{C_T}\hole{\cmp{P_2}}}$, the other inclusion is
  symmetric. \\
  Assume by contradiction that there exists $t_1 \in \behav{\trg{C}\hole{\cmp{P_1}}} \setminus \behav{\trg{C}\hole{\cmp{P_2}}}$. 
  Let $m_{max} \leq t_1$ be given by hypothesis 3. Therefore there exists $t_2 \in \behav{\trg{C_T}\hole{\cmp{P_2}}}$
  such that $m_{max} \leq t_2$, with $t_1 \neq t_2$ but still $t_1 \mathcal{R} ~t_2$. By determinacy, $t_1$ and $t_2$ have a common prefix $m$, 
  and there exist two input events $e_1 \neq e_2$ such that $m :: e_i \leq t_i, ~ i = 1, 2$. By maximality of $m_{max}$ it must be 
  $m \leq m_{max}$. The inequality cannot be strict, otherwise both $m :: e_1, ~m :: e_2 \leq m_{max}$. 
  In case $m = m_{max}$ apply input totality and deduce that $\trg{C_T\hole{\cmp{P_2}} \leadsto^*} m_{max} :: e_2$
  contradicting the maximality of $m_{max}$.
\end{proof}

The next result was discussed previously as \oldautoref{thm:main-rtrxp-rtep}.

\begin{theorem} \label{thm:xthm}  Under the same assumptions of \Cref{thm:rtrtpRTEP}, $ \criterion{R2rX} \Rightarrow \rtep $. 

\end{theorem} 

\begin{proof} 
  See file R2rXP\_RTEP.v, theorem R2rXP\_RTEP for a complete proof.
  The argument is very similar to the one in \Cref{thm:rtrtpRTEP}, the relation $\mathcal{R}$ is adapted to $\ii{XPref}$
  as following.
  
  \begin{equation*}
  \begin{split}
   x_1  ~\mathcal{R}_X  ~x_2 \iff & x_1 \leq x_2 ~\vee x_2 ~\leq x_1 ~\vee \\
                                 & ~\exists m. ~\exists e_1 e_2 \in \mathit{Input}. 
                                   ~ e_1 \neq e_2 ~ \wedge ~ m ~:: ~e_1 \leq x_1
                                                   ~\wedge ~m ~:: ~e_2 \leq x_2
  \end{split}                                                                       
\end{equation*} 
 $\mathcal{R}$ holds for traces produces by contextually equivalent source programs $\src{P_1}, \src{P_2}$, 
 and by unfolding $\src{C_S\hole{\cmp{P_i}} \leadsto^*} x_i, ~ i = 1, 2$, $ x_1 \mathcal{R}_X ~x_2$ holds as well. \\
 Proceed as in \Cref{thm:rtrtpRTEP} to show mutual inclusion of behaviors. Determinacy ensures that
 if $t_1 \neq t_2$ then there exist two non comparable $x_1, x_2$ such that $x_1 \leq t_1, ~ x_2 \leq t_2$ both with 
 $m_{max}$ as common prefix and still $x_1 \mathcal{R}_X ~x_2$ and we can conclude with the same argument as in \Cref{thm:rtrtpRTEP}.
 It is crucial to observe that considering $\ii{XPref}$ instead of $\ii{FinPref}$, $x_1, x_2$ can be considered non comparable.
 This can be proved by case analysis on the two traces, in particular if $t_1 = m :: \circlearrowleft$ and $t_2 = m :: \termevent$, 
 $x_1 = m :: \circlearrowleft$ and $x_2 = m :: \termevent$ are two non comparable $x$-prefixes still related by $\mathcal{R}_X$
 but all finite prefixes will be comparable, so that input totality is not useful to reach a contradiction.
\end{proof}

While under the rather liberal condition above \rtrsp does {\bf not}
imply \rtep, it does imply \rtep in the very special case that target
programs cannot produce any silently diverging traces, for instance
because in the target language is terminating. This is a technical
result that we use in a later proof (\autoref{thm:rhp-not-rtep}).

\begin{theorem}\label{thm:rtrsp-rtep} Under the following assumptions:
  \begin{enumerate}
  \item The source language is determinate.
  \item The target language satisfies input totality.
  \item The target language is ``safety-like''.
  \item Target programs cannot produce silently diverging traces.
  \end{enumerate}
then $\rtrsp \Rightarrow \rtep$.
\end{theorem}

\begin{proof}
  See file R2rSP\_RTEP.v, theorem R2rSP\_RTEP for a complete direct
  proof.\ifsooner\ch{It seems stupid to give a complete direct proof when
   the result is trivial:}\fi{}
  Here we just highlight that \rtrxp and \rtrsp are equivalent
  under the very strong hypothesis 4, so we can simply apply
  \autoref{thm:xthm} above.
\end{proof}
While assumption (4) above is very strong, it does hold for strictly
terminating languages, and in all other cases one can use
\autoref{thm:xthm}.

\subsection{Safety-Like Small-Step Semantics} \label{sec:semSafe}

In this section we state and prove the property that many small-step
semantics have the previous ``safety-like'' behavior, in the sense
that we can determine whether an infinite trace cannot be produced by
a program after a finite number of steps.

 First, we state our semantic model and its
basic constituents. 

\begin{definition}[Small-step semantics]\label{def:smallstep}
A small-step semantics is defined in terms of the following abstract components:

\begin{itemize}
\item Program states are represented by \emph{configurations}, $c$.
\item An \emph{initial relation} characterizes initial program states.
\item A \emph{step relation}, $c \xrightarrow{e} c'$ between pairs of states,
      producing an event. Its reflexive and transitive closure is denoted
      ${\xrightarrow{e_1 \cdots e_n}}^*$.
\item A well-founded \emph{order relation} on elements of a type of ``measures.''
\end{itemize}

Events can be either \emph{visible} or \emph{silent}. A configuration
is \emph{stuck} when there is no configuration it can step to; it can \emph{loop
silently} if there is an infinite sequence of silent steps starting from it.

A small-step semantics relates program configurations and the traces they
produce; the relation is moreover parameterized by an element of the type of
measures. In our trace model, there are four possible cases, starting from a
configuration $c$:

\begin{itemize}
\item If $c$ is stuck, the semantics produces the terminating trace $\termevent$
      with some associated information $\varepsilon$.
\item If $c$ can loop silently, the semantics produces the silently diverging
      trace $\sdiv$.
\item If $c$ can silently step $c \xrightarrow{\emptyset}^* c'$ to a $c'$ while
      decreasing its ordering measure with respect to $c$, the semantics recurses
      on $c'$.
\item If $c$ can step with some visible events $c \xrightarrow{m}^* c'$, the
      semantics emits $m$ and recurses on $c'$.
\end{itemize}

\end{definition}

The addition of the well-founded order relation between measures is used to
avoid the usual problem of infinite stuttering on silent events, which is
properly captured by silent divergence. Between two visible events there must
mediate a finite number of silent events. This requirement is enforced by having
the ordered measure decrease when silent steps are taken (there are no
restrictions on ordering between states connected by visible events). A similar
device is used, for example, in the CompCert verified compiler.

The final result holds for a wide class of reasonable languages.
The following determinacy condition is sufficient to prove the result.

\begin{definition}[Weak determinacy]
Two program configurations are related if they produce the same traces under the
semantics; we write $c_1 R c_2$ for this.

Under weak determinacy, if a pair of states is related and each element of this
initial pair steps to another state producing the same sequence of events, the
pair of final states is also related:
\begin{equation*}
\forall c_1. \forall c_1'. \forall m. \forall c_2. \forall c_2'.
        c_1 R c_1' \Rightarrow
        c_1 \xrightarrow{m}^* c_2 \Rightarrow
        c_1' \xrightarrow{m}^* c_2' \Rightarrow
        c_2 R c_2'
\end{equation*}
\end{definition}

Thus stated, the ``safety-like'' quality of small-step semantics follows easily.

\begin{theorem}
Assuming weak determinacy holds, all small-step semantics (that can be encoded
by the scheme of \Cref{def:smallstep}) are ``safety-like.''
\end{theorem}
\begin{proof}
See file SemanticsSafetyLike.v, theorem tgt\_sem.
\end{proof}


\section{Separation Results}\label{sec:sep}

The implications represented by arrows in \Cref{fig:order} are
strict, that is, the two criteria linked by an arrow are \emph{not} equivalent.
This section justifies these separation results by giving, for each of
them, counterexample compilation chains that satisfy the criterion
occurring lower in the diagram (pointed to by the arrow), but not the
upper one.
Finally, in \autoref{sec:rtep-useless}
we prove that \rtep does not imply even the weakest criteria
in our diagram (\rsp and \rdp), even when also assuming compiler
correctness (\tp, \scc, and \ccc).

\subsection{\rsp and \rdp Do Not Imply \rtp} \label{sec:sepCoq}

In this section, we show that the robust preservation of \emph{either}
all safety properties (\Cref{sepS}) \emph{or} of all dense properties
(\Cref{sepL}) is not enough to guarantee the robust preservation of
all trace properties.
(Note that, as a corollary to the decomposition result in
\autoref{thm:tp-safety-cap-dense}, a compiler that preserves all safety
properties \emph{and} all dense properties preserves all properties.)
The two compilation chains in this section have been formalized in the
Coq; see file Separation.v for more details.
This section expands upon the description from \oldautoref{sec:rsp}
(for safety properties) and \autoref{sec:rdp} (for dense properties).


Take an arbitrary language $\LL$ described by a small-step semantics. Assume it is possible to
write a non-terminating program in $\LL$, e.g., a program that produces some infinite trace.
Assume moreover that such a program is independent from the context
with which it is linked (for instance, it is already whole).
To keep things concrete, we consider a standard \emph{while} language
as our $\LL$ and the following non-terminating program $P_{\Omega}$,
where \verb|n|${\in}\mathbb{N}$:

\begin{center}
  \begin{BVerbatim}
  while (true) {
    output(n);
  }
  \end{BVerbatim}
\end{center}


Next, define a language transformer $\phi(\LL)$, which produces a new
language that is identical to $\LL$, except that it bounds program
executions by a certain number of steps (its \textit{``fuel''}).
In particular:
\begin{itemize}
\item If $C$ is a context in $\LL$, then for every $n \in \mathbb{N}$, $(n, C)$ is a context in $\phi(\LL)$ 
      with fuel $n$.
\item Plugging in $\phi(\LL)$ is defined by $(n,C) [ P ]_{\phi(\LL)} \equiv (n, C [ P ]_{\LL})$.
      Subscripts will be omitted when doing so introduces no ambiguities.
\item The semantics of $\phi(\LL)$ extends the semantics of $\LL$ as follows.
      If the amount fuel is $0$, no
      steps are allowed.
      Otherwise, every time
      a step would be taken in $\LL$, the same step is taken in $\phi(\LL)$ and the amount of fuel is decremented by one.
\end{itemize}

\begin{lemma} \label{sepS} \rsp  $\not \Rightarrow$ \rdp
\end{lemma}

\begin{proof}
Take $\phi(\LL)$ as source language, $\LL$ as target, and the compiler to be
the projection of contexts of $\phi(\LL)$ on their second component. We are going to show
that all safety properties that are robustly satisfied in the source
are also robustly satisfied in the target, but not all dense properties are preserved.

Let $S \in \mathit{Safety}$. Assume that all safety properties are robustly preserved,
i.e., that for every program $P$, every source context $(n, C)$ and every trace $t$,
\begin{equation*}
  (n, C [ P ]) \rightsquigarrow t \Rightarrow t \in S
\end{equation*}
In addition, assume for contradiction that there exists some target context $C'$ and trace $t'$ such that
\begin{equation*}
  C' [ P\downarrow ] \rightsquigarrow t' \wedge t' \not \in S
\end{equation*}
where $P \downarrow = P$.
By definition of safety, there exists $m \leq t'$ such that every continuation $t''$
of $m$ violates the property,
\begin{equation*}
  \forall t''. ~ m \leq t'' \Rightarrow t'' \not \in S
\end{equation*}
Consider the source context $(|m|, C')$ where $|m|$ is the length of $m$. Denote by $t_m$ the trace
that contains the events of $m$ followed by a termination marker.
Since $m \leq t_m$ we have that $t_m \not \in S$.
However, $(|m|, C') \rightsquigarrow t_m$,
 which implies that $t_m \in S$, a contradiction.

Next, we produce a dense property that is not robustly preserved by this compiler.
Consider
\begin{equation*}
  L = \{ t | ~t \mathit{~is ~finite} \vee t = output(42)^\omega \}
\end{equation*}
Observe that $L$ is a dense property as it includes all finite traces.
Since programs in the source can produce only finite traces, these will be in $L$.
In the target, however, the program $P = P \downarrow$
\begin{center}
  \begin{BVerbatim}
  while (true) {
    output(41);
  }
  \end{BVerbatim}
\end{center}
is no longer forced to stop after a finite number of steps, and produces an infinite
trace different from $output(42)^\omega$.
\end{proof}

\begin{lemma}\label{sepL}  \rdp $\not \Rightarrow$ \rsp
\end{lemma}

\begin{proof}
  Take $\LL$ as source language, $\phi(\LL)$ as target, and the compiler to be 
  the identity. We are going to show that all
  dense properties are robustly preserved but not all safety properties are robustly preserved. \\
  Let $L$ be a dense property. Every trace $t$ produced by a program in the target is finite,
  so that by definition of \textit{Dense}, $t \in L$. 
  %
  %
  Consider now the following property:
  \begin{equation*}
    S =\{ \mathit{output}(42)^{\omega} \}
  \end{equation*}
  $S$ is a safety property because for every trace $t \not \in S$,
  $t$ starts with a number (possibly zero) of $\mathit{output}(42)$ events, followed either
  by some other event $e \neq \mathit{output}(42)$ or terminated by $\termevent$ for some $\varepsilon$, i.e.,
  \begin{equation*}
    \mathit{output}(42)^{n}; e \leq t \vee \mathit{output}(42)^{n}; \termevent \leq t
  \end{equation*}
  Here, every continuation of $ \mathit{output}(42)^{n}; e$ is different from
  $\mathit{output}(42)^{\omega}$, and different from every finite trace.
  Finally, consider the program $P = P \downarrow$
\begin{center}
  \begin{BVerbatim}
  while (true) {
    output(42);
  }
  \end{BVerbatim}
\end{center}
  which, in the source, produces the infinite trace $\mathit{output}(42)^{\omega} \in S$
  regardless of the context. In the target, only traces of length $k$ can be produced,
  which are not in $S$.
\end{proof}

\begin{theorem}\label{thm:rsp-rdp-rtp}
  Neither \rsp nor \rdp separately imply \rtp.
\end{theorem}

\begin{proof}
Follows directly from \Cref{sepS} and \Cref{sepL} and
\autoref{thm:tp-safety-cap-dense}.
\end{proof}

In our previous discussion, \oldautoref{thm:rsp-doesnt-imply-rtp} corresponds to
the non-trivial direction of \autoref{thm:rsp-rdp-rtp}.
\newcommand{\separationfile}[0]{{\tt separation-results.txt}}



\subsection{\rtp Does Not Imply \rtinip}
\label{sec:rtp-doesnt-imply-rtinip}

In this section we prove \oldautoref{thm:rtp-doesnt-imply-rtinip} from
\oldautoref{sec:hypersafety-separations}:

\begin{theorem}
There is a compiler that satisfies \rtp but not \rtinip.
\end{theorem}

\begin{proof}
  We consider the following source language that works over integers;
  has traces with exactly two events, one input followed by one output;
  and has exactly one program \(\src{P}\):
  \begin{center}
    \begin{BVerbatim}[formatcom=\color{\stlccol}]
x = input;
y = f();
output y;
    \end{BVerbatim}
  \end{center}
  where \src{f()} is a pure function provided by the context.
  The target language is the same as the source language, except the
  context has the ability to directly read the program \(\trg{P}\)'s
  private variables, like \(\trg{x}\).
  The compiler is the identity.

  This compiler satisfies Robust Trace Property Preservation (\rtp).
  The reason is that in the source,
  program \src{P} can generate every possible trace
  given an appropriate source context: to generate the trace \(t =
  [\inpl{i}, \outl{o}]\), take the context whose \src{f()} returns
  the integer $o$. Basically, this source context simply guesses the
  output values from the single trace $t$.

  However, this compiler does not satisfy \rtinip.
  If we take the input to be private and the output to be public,
  then for our language \ii{TINI} is equivalent to the following 
  \(2\)-hypersafety property $H$:
\[
  H = \{b ~|~ \forall {i_1},{o_1},{i_2},{o_2}.~ 
            [\inpl{i_1}, \outl{o_1}] \in b \wedge
            [\inpl{i_2}, \outl{o_2}] \in b \Rightarrow
              o_1 = o_2 \}
\]

  In the source, any \src{f()} defined by the context must be a
  constant function. This is because the context is purely functional
  and has no access to the
  input stream or \src{P}'s local variables, hence \src{f()}'s result
  cannot depend on any changeable quantity. If \src{f()} returns
  a constant \(c\) and we look at two source traces
            $[\inpl{i_1}, \outl{o_1}]$ and
            $[\inpl{i_2}, \outl{o_2}]$,
  then \(o_1 = o_2 = c\) and thus the source program
  satisfies the hyperproperty $H$.

  However, in the target, it is possible to write a context function \trg{f()}
  that breaks $H$: \trg{f() \{ return (x);\}}.
  This function reads \trg{P}'s local variable \trg{x} (which the target
  context is capable of accessing) and returns its value. Hence, with
  this context, the program's outputs depend on its inputs. In particular,
  $[\inpl{1}, \outl{1}]$ and
  $[\inpl{2}, \outl{2}]$ are two traces where the output vary,
  so this context (and consequently the compilation chain)
  breaks the 2-hypersafety property $H$.
\end{proof}



\subsection{\rkhsp Does Not Imply \rkkhsp}

In this section, we prove \oldautoref{thm:rkhsp-doesnt-imple-rkkhsp}
from \oldautoref{sec:hypersafety-separations} by exhibiting a
counterexample compiler, parametric in \(K\), that has Robust
\(K\)-Hypersafety Preservation, \rkhsp, but not Robust \((K+1)\)-Hypersafety
Preservation, \rkkhsp, for an arbitrary \(K\).

Our source language is a standard \emph{while} language with read and write events
to standard I/O. It has traces of length exactly two: one input (read) event
followed by one output (write) event. This language's inputs are always in
the natural range \([1, \dots, K+1]\), while its internal values and outputs are
real numbers. The language has exactly one program, \(\src{P}\), shown below. The
context provides the functions \(\src{f_1()},\dots \src{f_K()}\).

\begin{center}
  \begin{BVerbatim}[formatcom=\color{\stlccol}]
x = read();
switch (x) {
  case x = i where 1 <= i <= K:
     y = x + (sum {f_j() | 1 <= j <= K && i <> j});
     break;
  case x = K + 1:
     y = K + f_1();
}
write (y)
  \end{BVerbatim}
\end{center}

Our target language is identical to the source, with the exception that the context now has
access to the private state of the program, so it can read the
local variable \(\trg{x}\). In the source language, the context lacks this
capability.

The compiler under consideration, \(\cmp{\cdot}\), is the identity,
i.e., it maps \(\src{P}\) to its identical counterpart \(\trg{P}\).

\begin{lemma}\label{thm:ex-rkhsp}
  The compiler \(\cmp{\cdot}\) satisfies \rkhsp.
\end{lemma}
\begin{proof}

  We prove this by showing that for any finite \(K\)-set of prefixes
  \(\{ [\rdl{a_1}, \wrl{b_1}], ..., [\rdl{a_K}, \wrl{b_K}] \}\) that
  the program \(\trg{C_T\hole{P}}\) can produce (for some target
  context \trg{C_T}), there is some source context \src{C_S} that
  produces these \(K\) prefixes as well. This property immediately
  implies that the compiler has \rkhsp.

  To prove this property, note that if all \(K\) prefixes
  \([\rdl{a_1}, \wrl{b_1}], \dots, [\rdl{a_K}, \wrl{b_K}]\) can be
  produced by the target, then, since the target is still
  deterministic, we must have: \(\forall i, \forall j, a_i = a_j
  \implies b_i = b_j\). Thus, we can assume without loss of generality
  that all \(a_i\)s are distinct. It follows that \(I =
  \{a_1,\dots,a_K\}\) is a \(K\)-subset of \(\{1,\dots,K+1\}\), so \(I\)
  must be missing exactly one element in the set of allowed inputs \(\{1, \dots, K+1\}\).

  We proceed by case analysis on the missing element:
  \begin{itemize}
    \item The missing element is \(K+1\).
      We can assume without loss
      of generality (by reordering \(I\) if needed) that \(a_i =
      i\), and
      therefore \(I =
      \{a_1,\dots,a_K\} = \{1,\dots,K\}\).
      We now set up a system of linear equations whose solution
      characterizes the source context \(\src{C_S}\).  Let \(x_i\) be a variable that
      represents the value of \(\src{f_i()}\) (note that
      \(\src{f_i()}\) must be a constant function in the source
      context). Then, we formulate the system the equations:
      \begin{align*}
        \phantom{x_1 +}      x_2 + \dots + x_{K-1} + x_K &= b_1 - a_1\\
        x_1 + \phantom{x_2}      + \dots + x_{K-1} + x_K &= b_2 - a_2\\
        &\dots \\
        x_1 + x_2 + \dots + x_{K-1} + \phantom{x_K}     &= b_K - a_K
      \end{align*}

      This system has a unique solution. To see this, first add all
      the equations. This yields:
      \[
      (K-1) (x_1 + \dots + x_K) = (b_1 + \dots + b_K) - (a_1 + \dots + a_K).
      \]
      This yields an equation \(x_1 + \dots + x_K = c\) for some \(c\).
      Subtracting the first equation in the system (corresponding to $b_1 - a_1$)
      from this sum gives us \(x_1\).
      Subtracting the second equation gives \(x_2\), and so on.
      Hence we obtain a value \(x_i\) that each \(\src{f_i()}\) must return
      in the source in order to produce the required outcome.
      This is the required \(\src{C_S}\).
    \item The missing element is \(1\). Then, \(I = \{a_1, \dots,
      a_K\} = \{2, \dots, K+1\}\). Assume without loss of generality that
      \(a_1 = 2, \dots, a_K = K+1\). Then, as before, we get the equations:
      \begin{align*}
        x_1 + \phantom{x_2}      + \dots + x_{K-1} + x_K &= b_1 - a_1\\
        &\dots \\
        x_1 + x_2 + \dots + x_{K-1} + \phantom{x_K}     &= b_{K-1} - a_{K-1}\\
        x_1 + \phantom{x_2 + \dots + x_{K-1} + x_K} &= b_K - a_K
      \end{align*}
      This set of equations also has
      a solution. First, the last equation directly gives \(x_1\). Now
      subtract the last equation from all the previous \(K-1\)
      equations. This yields exactly \(K-1\) \emph{cyclic} equations
      in \(K-1\) variables \(x_2, \dots, x_K\). These can be solved
      exactly as in the previous case.
    \item The missing element is between \(2\) and \(K\) (both inclusive).
      Without loss of generality, assume that it is \(K\). Then,
      \(I = \{a_1, \dots, a_K\} = \{1, \dots, K-1, K+1\}\). Again,
      assume that \(a_1 = 1, \dots, a_{K-1} = K-1\), and \(a_K = K+1\).
      Then, we get the equations:
      \begin{align*}
        \phantom{x_1} + x_2 + \dots + x_{K-2} + x_{K-1} + x_K &= b_1 - a_1\\
        x_1 + \phantom{x_2} + \dots + x_{K-2} + x_{K-1} + x_K &= b_2 - a_2\\
        &\dots \\
        x_1 + x_2 + \dots + x_{k-2} + \phantom{x_{K-1}} + x_K &= b_{K-1} - a_{K-1}\\
        x_1 + \phantom{x_2 + \dots + x_{K-2} + x_{K-1} + x_K} &= b_K - a_K
      \end{align*}
      Solving these equations is also easy. \(x_1\) is determined by the last equation.
      Adding the first and last equations gives the value of \(x_1 + \dots + x_K\).
      Subtracting the remaining equations from this one, one by one, yields
      \(x_2, \dots, x_{K-1}\). Then, \(x_K\) follows from the first equation.
  \end{itemize}
  
\end{proof}

\begin{lemma}\label{thm:ex-rkkhsp}
  The compiler \(\cmp{\cdot}\) is not \rkkhsp.
\end{lemma}
\begin{proof}
  We construct a concrete \(K+1\) prefix set \(S\) and a \(\trg{C_T}\)
  such \(\trg{C_T\hole{\cmp{P}}}\) can produce all prefixes of \(S\), but no
  \(\src{C_S\hole{P}}\)) can do the same. The proof
  relies on the fact that, in the source, each of
  \(\src{f_1()},\dots,\src{f_K()}\) must be a constant function, so we
  can have \(k+1\) inconsistent equations for these \(k\)
  constants. In the target, the equations are not required to be constant since any
  \(\trg{f_i}\) can return a value based on the private input \trg{x}.

  Let \(c\) be the constant \(K-1\). Consider now the following falsifying
  prefix set:
  \begin{align*}
  S = \{
    &[1, 1 + c],\\
    &[2, 2 + c],\\
    &\dots\\
    &[K, K + c],\\
    &[K + 1, K + 1]
    \}
  \end{align*}

  So, for inputs \(x = 1, \ldots, K\), the output is the input value plus \(c\) (i.e., \(K-1\)),
  but for input \(K + 1\) the output is the input value \(K + 1\) itself.

  The following target context \trg{C_T} generates this prefix set \(S\):
  \begin{align*}
    \trg{f_1()} &= \trg{\ifte{P.x = K + 1}{0}{1}}\\
    \trg{f_2()} &= \trg{1}\\
    &\cdots\\
    \trg{f_K()} &= \trg{1}
  \end{align*}
  where \trg{P.x} is the private value \(\trg{x}\) of \(\trg{P}\).

  The function \(\trg{f_1()}\) returns \(1\) except when the private
  input \(\trg{x}\) is \(K+1\), when it returns \(0\). It is easy to
  see that for inputs \(\trg{x}\) between \(1\) and \(K\) the output
  of \(\trg{P}\) is exactly \(x + (K-1)\), whereas for input \(x =
  K+1\) the output of \trg{P} is \(K+1 + \trg{f_1()} = K+1 + 0 = K+1\).
  Hence, \(\trg{C_T\hole{\cmp{P}}}\) generates the entire prefix set \(S\).

  On the other hand, \(\src{C_S\hole{P}}\) cannot generate all prefixes of \(S\) for any
  \src{C_S}. To see this, suppose that there exists some \src{C_S\hole{P}}
  can actually generate all prefixes of \(S\). Let \(\src{f_i()} = x_i\). We get the
  equations:

  \begin{align*}
    \phantom{x_1 + {}} x_2 + \dots + x_{K-1} + x_K &= (1 + K - 1) - 1 = K -1\\
    x_1 + \phantom{x_2} + \dots + x_{K-1} + x_K &= (2 + K - 1) - 2 = K -1\\
    &\dots \\
    x_1 + x_2 + \dots + x_{K-1} + \phantom{x_K} &= (K + K - 1) - K = K - 1\\
    x_1 + \phantom{x_2 + \dots + x_{K-1} + x_K} &= (K + 1) - (K + 1) = 0
  \end{align*}
  However, these equations are inconsistent. The first \(K\) equations
  (which are cyclic) force that \(x_1 = \dots = x_K = 1\), while the last equation
  requires \(x_1 = 0\). This contradicts our hypothesis on the existence of \src{C_S\hole{P}}.
\end{proof}

\begin{theorem}
  For any $K$, there is a compiler that satisfies \rkhsp but not \rkkhsp.
\end{theorem}

\begin{proof}
  It follows from \autoref{thm:ex-rkhsp} and \autoref{thm:ex-rkkhsp}
  that \(\cmp{\cdot}\) is \rkhsp but not \rkkhsp.
\end{proof}



\subsection{Robust Non-Relational Property Preservation Does Not Imply Robust Relational Property Preservation}
In this section we prove that, as stated in \oldautoref{thm:rhp-doesnt-imply-r2rsp}
from \oldref{sec:rel:sep-rel}, no non-relational preservation criterion
implies any relational preservation criterion. We do this
constructively, by showing a source language, a target language, and a
compiler between them such that:
\begin{itemize}
\item The compiler satisfies the \emph{strongest} non-relational
  preservation criterion (Robust Hyperproperty Preservation, \rhp).
\item The compiler does not satisfy the \emph{weakest} relational
  preservation criterion (Robust 2-Relational Safety Preservation,
  \rtrsp). Because the languages will satisfy the conditions that make
  \rtrsp imply Robust Trace Equivalence Preservation (\rtep), we shall
  simply show that the compiler does not satisfy \rtep, and use the
  result from the next .
\end{itemize}

The source language we shall consider is a standard \emph{while}
language with read and write events
to standard I/O. It has traces comprising exactly two events: one input (read)
followed by one output (write). This language works over integers (\emph{not}
natural numbers) and has exactly two programs, \(\src{P1}\) and \(\src{P2}\),
shown below, that are only different in that the second adds some dead code
to the first:

  \begin{center}
    \begin{BVerbatim}[formatcom=\color{\stlccol}]
P1:
x = read();
y = f();
write (x + y)

P2:
x = read();
y = f();
... some dead code here ...
write (x + y)
    \end{BVerbatim}
  \end{center}
  Here, \(\src{f()}\) is a function provided by the context.

  The target language is the same, but additionally allows the context
  to read the compiled code as a value.

  The compiler under consideration, \(\cmp{\cdot}\), is the
  identity.

  \begin{lemma}\label{thm:rel-cex-rhp}
    The compiler \(\cmp{\cdot}\) satisfies \rhp.
  \end{lemma}
  \begin{proof}
    We need to show that \(\forall \trg{C_T}~\src{P}. \exists\src{C_S}.
    \behav{\src{C_S\hole{P}}} = \behav{\trg{C_T\hole{\cmp{P}}}} \).
    For this,
    pick a \(\trg{C_T}\) and a \(\src{P}\). Note that \(\trg{\cmp{P}}
    = \src{P}\) by definition of the compiler.  To produce \src{C_S},
    modify the \trg{C_T} so that wherever \trg{C_T} reads the code of
    \trg{\cmp{P}}, \src{P} (which is the same as \trg{\cmp{P}})
    is hard-coded in \src{C_S} instead.

     It is trivial to see that \src{C_S\hole{P}} and
     \trg{C_T\hole{\cmp{P}}} have exactly the same behaviors.
  \end{proof}

  \begin{lemma}\label{thm:rel-cex-not-rtep}
    The compiler \(\cmp{\cdot}\) does not satisfy \rtep.
  \end{lemma}
  \begin{proof}
    Since \src{P1} and \src{P2} differ only in the presence of some dead code, which a source
    context cannot examine, it is trivially the case that
    \(\forall \src{C_S}. \behav{\src{C_S}\hole{\src{P1}}} = \behav{\src{C_S}\hole{\src{P2}}}\).

    On the other hand, we can construct a target context \(\trg{C_T}\) whose
    \(\trg{f()}\) checks whether the compiled code is \(\trg{P1}\) or \(\trg{P2}\),
    and returns either \(0\) or \(1\), respectively.
    Then, \(\trg{C_T\hole{P1}}\) produces \([\rdl{0}, \wrl{0}]\) as a trace,
    while \(\trg{C_T\hole{P2}}\) does not have this trace. Hence, the
    compiler is not \rtep.
  \end{proof}

\begin{theorem}\label{thm:rhp-not-rtep}
  There exists a compiler between two languages that satisfy the assumptions of \Cref{thm:rtrsp-rtep}
    that has \rhp, but not \rtep.
\end{theorem}
\begin{proof}
  Both language clearly satisfy determinacy and input totality.
  Furthermore, given an infinite trace \(t\) not produced by some whole program \(\trg{W}\),
  the prefix needed is either the trace produced by the program, or the empty prefix.
  
  The theorem follows immediately from \autoref{thm:rel-cex-rhp} and
  \autoref{thm:rel-cex-not-rtep}.
\end{proof}

\begin{theorem}
There exists a compiler that satisfies \rhp but not \rtrsp.
\end{theorem}
\begin{proof}
Follows directly from \autoref{thm:rhp-not-rtep}.
\end{proof}

\paragraph{The Full Story}
More generally, if we take any source language
in which the context cannot examine the code and compile it to a
target language that is similar, but where the context can
examine the code as an added capability, then the identity compiler satisfies
every non-relational criterion including \rhp, since \emph{for a single
program}, the target context's additional ability to observe the code is
inconsequential. More formally, non-relational preservation
criteria, including \rhp, allow the simulating source context \src{C_S} to
depend on the compiled program \src{P}, so that program code can be hard-coded into
\src{C_S} wherever \trg{C_T} examines the code.
However, it is extremely unlikely that this compiler satisfies any
relational preservation criterion since the target context can branch
on the program being executed and provide different values to each of
the programs.


\subsection{\rtep Does Not Imply \rsp or \rdp}
\label{sec:rtep-useless}

\renewcommand{\tp}{\formatCompilers{TP}}
\renewcommand{\scc}{\formatCompilers{SCC}}
\renewcommand{\ccc}{\formatCompilers{CCC}}

In this section, we give a counterexample compilation chain for showing
a generalization of \oldautoref{thm:rtep-useless} from \oldautoref{sec:fa}.

First, we recall three notions of correctness, from
\oldautoref{sec:rtp}. For \ccc{} we explicitly mention the condition
that $\src{C_S}$ should be linkable with $\src{P}$, which is a
technical hypothesis that we omitted in the main paper text.
\begin{definition}[Backward Simulation (\pf{\tp})]
  $$
  \pf{\tp}:\quad \forall \src{W}\ldotp \trg{\cmp{W} \sem}\ t \Rightarrow
  \src{W \sem}\ t
  $$
\end{definition}
\begin{definition}[Separate Compiler Correctness (\scc)]
  \[
  \begin{multlined}
  \scc:\quad
  \forall\src{P}.~ \forall\src{C_S}.~ \forall t.~
    \trg{\cmp{C_S}\hole{\cmp{P}}} \mathrel{\trg{\sem}} t \Rightarrow
    \src{C_S\hole{P}} \mathrel{\src{\sem}} t
  \end{multlined}
  \]
\end{definition}
\begin{definition}[Compositional Compiler Correctness (\ccc)]
\[\begin{multlined}
 \ccc:~
  \forall\src{P}~\trg{C_T}~\src{C_S}~t.~
  \trg{C_T} {\approx} \src{C_S} \wedge
  \src{C_S}\text{ is linkable with }\src{P} \wedge
    \trg{C_T\hole{\cmp{P}}} \mathrel{\trg{\sem}} t \Rightarrow
    \src{C_S\hole{P}} \mathrel{\src{\sem}} t
\end{multlined}\]
\end{definition}

\begin{theorem}\label{thm:counterexamplertep}
There exists a compiler between two deterministic
languages that satisfies \rtep, \tp, \scc, and \ccc
but that satisfies neither \rsp nor \rdp.
\end{theorem}


\paragraph{Trace Model} We consider languages where exactly one event
is produced containing a natural number that represents the final result of the
computation. Allowed traces are final result singletons and silent
divergence.

\paragraph{Source Language}

A source language program consists of one function obtaining one input from the context (a natural number or a boolean), perfoming basic
computations on it, and returning a natural number as a result.
\begin{align*}
  \mi{Program}~\src{P} \bnfdef&\ \src{f(x : \Nats) \mapsto e} \mid \src{f(x : \Bools) \mapsto e}
  \\
  \mi{Expression}~\src{e} \bnfdef&\ \src{\ifte{x}{e}{e}} \mid \src{\ifte{x<n}{e}{e}} \mid \src{n} \mid \src{f(e)}
  \\
  \mi{Context}~\src{C} \bnfdef&\ \src{f(n)} \mid \src{f(b)}
\end{align*}

In this example, the composition of a program and a context of
incompatible types is statically disallowed. We do not consider them
in the criteria we prove, and implicitely assume that the criteria only
apply when the operations \src{\plug{\cdot}{\cdot}} and \trg{\plug{\cdot}{\cdot}} are defined.
See \Cref{sec:proofs-appendix} for a
example where we take into account the fact that not all components are
linkable by using simple variants of our criteria.

\paragraph{Target Language}
The target language is identical to the source, but it only admits natural numbers as inputs.

\begin{align*}
  \mi{Program}~\trg{P} \bnfdef&\ \trg{f(x : \Nats) \mapsto e}
  \\
  \mi{Expression}~\trg{e} \bnfdef&\ \trg{\ifte{x<n}{e}{e}} \mid \trg{n} \mid \trg{f(e)}
  \\
  \mi{Context}~\trg{C} \bnfdef&\ \trg{f(n)}
\end{align*}

\paragraph{Compiler}
We first define the compilation of expressions:
\begin{align*}
  \trg{\cmp{(\ifte{x}{e_1}{e_2})}} &= \trg{\ifte{x<1}{\cmp{e_1}}{\cmp{e_2}}}\\
  \trg{\cmp{(\ifte{x<n}{e_1}{e_2})}} &= \trg{\ifte{x<n}{\cmp{e_1}}{\cmp{e_2}}}\\
  \trg{\cmp{n}}                 &= \trg{n}\\
  \trg{\cmp{\trues}} &= \trg{0}\\
  \trg{\cmp{\falses}} &= \trg{1}\\
  \trg{\cmp{f(e)}} &= \trg{f(\cmp{e})}
\end{align*}

Then, we define the compilation of partial programs:
\begin{align*}
  \trg{\cmp{(f(x : \Nats) \mapsto e)}} &= \trg{f(x : \Nats) \mapsto \cmp{e}}\\
  \trg{\cmp{(f(x : \Bools) \mapsto e)}} &=
    \trg{f(x : \Nats) \mapsto \ifte{x<2}{\cmp{e}}{\ifte{x<3}{f(2)}{42}}}
\end{align*}

We can trivially extend this compiler to contexts, and whole programs as well.

\begin{lemma}\label{lem:sep-rtep-rsp-tp}
  \(\cmp{\cdot}\) satisfies {\tp}, {\scc}, and {\ccc}.
\end{lemma}
\begin{proof}
  We start by noting that for this compilation chain,
  \tp is implied by \scc. Indeed, let \(\src{W}\) be a source
  whole program. Then, by definition, there exists \src{C} and \src{P}
  such that \(\src{W} = \src{C\hole{P}}\).  Furthermore,
  \(\cmp{\src{W}} = \trg{\cmp{C}\hole{\cmp{P}}}\) by definition of the
  compiler. Hence, to show \tp, one can simply apply \scc to this
  decomposition of \src{W} into two components.

  We then show that for this compilation chain
  \ccc is also implied by \scc. To show this, we need to explicit the
  instantiation of \(\approx\) and of the notion of linkable components
  in the definition of \ccc:
  \begin{itemize}
  \item \(\approx\) is defined by \(\trg{C_T} \approx \src{C_S}\) iff
    \(\trg{C_T} = \cmp{\src{C_S}}\).
  \item \(\src{P}\) and \(\src{C_S}\) are linkable iff they agree on the
    argument type of the program's function.
  \end{itemize}
  Then, by unfolding these definitions in \ccc, and substituting equalities,
  we obtain exactly \scc.
  
  So all that is left to prove is \scc.
  Suppose \(\trg{\cmp{C}\hole{\cmp{P}} \sem}\ t\).
  We will show that \(\src{C\hole{P} \sem}\ t\),
  proceeding by case analysis on \(\src{P}\):
  \begin{itemize}
  \item \(\src{P} = \src{f(x : \Nats) \mapsto e}\).  By induction on
    \(\src{e}\), we can prove that \(\cmp{e} = \src{e}\): indeed,
    because the function argument \(\src{x}\) must be a natural number, and because \(\src{x}\)
    is the only ``variable'' ever in scope, \(\src{e}\) cannot contain a subexpression
    \(\src{\ifte{x}{e_1}{e_2}}\).
    Otherwise, it would mean that \(\src{x}\) is a boolean, this would be a contradiction.
  \item \(\src{P} = \src{f(x : \Bools) \mapsto e}\).  In this case,
    \(\trg{\cmp{P}} = \trg{f(x : \Nats) \mapsto
      \ifte{x<2}{\cmp{e}}{\ifte{x<3}{f(2)}{42}}}\).  Since the target
    context is a compiled context, this steps to the expression
    \(\trg{\cmp{e}}\).  Now, \(\src{f}\) cannot be called recursively
    in \(\src{e}\), because expressions are natural numbers, but
    \(\src{f}\) expects a boolean.
    Hence, the same is true in the compiled version. Now, we can conclude by
    proving the thesis by induction on \(\src{e}\).
    In particular, note that
    \(\trg{\cmp{(\ifte{x}{e_1}{e_2})}} = \trg{\ifte{x<1}{\cmp{e_1}}{\cmp{e_2}}}\),
    that if \(\src{x} = \trues\) then \(\trg{x<1}\) is true, and if \(\src{x} = \falses\)
    then \(\trg{x<1}\) evaluates to false, where \(\trg{x}\) is the compilation of
    \(\src{x}\).    
    \jt{not sure I can do better without expliciting the semantics}
   \end{itemize}
\end{proof}


\begin{lemma}
  \(\cmp{\cdot}\) satisfies \rtep.
\end{lemma}
\begin{proof}
  We prove the contrapositive form of the statement of \rtep.

  Let \(\src{P_1}\) and \(\src{P_2}\) be two programs and suppose
  their compilations are not observationally equivalent; let \(\trg{C}
  = \trg{f(n)}\) be the distinguishing context.
  We consider three cases:
  \begin{itemize}
  \item \(\src{P_1}\) and \(\src{P_2}\) both expect a natural number.
    Take \(\src{C} = \src{f(n)}\). Since the compiler is the identity
    for programs that expect a natural number, we obtain the desired result.
  \item \(\src{P_1}\) and \(\src{P_2}\) both expect a boolean.
    If \(\src{n} = 0\) or \(\src{n} = 1\), then by
    taking \(\src{C} = \src{f(\trues)}\) in the first case, and \(\src{C} = \src{f(\falses)}\)
    in the second case, we obtain the desired result by compiler correctness.
    Otherwise, this case is discharged by contradiction: \(\trg{C\hole{\cmp{P_1}}}\) and
    \(\trg{C\hole{\cmp{P_2}}}\) have the same behavior, by definition of
    the compiler.
  \item \(\src{P_1}\) and \(\src{P_2}\) have different input types.
    This is a case we do not consider, because then the source context
    cannot be linked with both programs.
  \end{itemize}
\end{proof}

\begin{lemma}
  \(\cmp{\cdot}\) does not satisfy \rsp.
\end{lemma}
\begin{proof}
  Consider the program \(\src{P} = \src{f(x : \Bools) \mapsto 1}\).
  This program satisfies the safety property ``never outputs 42,''
  but its compilation does not (it violates it with input \trg{3}, for instance).
\end{proof}

\begin{lemma}
  \(\cmp{\cdot}\) does not satisfy \rdp.
\end{lemma}
\begin{proof}
  Consider the same program \(\src{P} = \src{f(x : \Bools) \mapsto 1}\).
  This program satisfies the safety property ``never silently diverge,''
  but its compilation does not (it violates it with input \trg{2}).
\end{proof}

The proof of \Cref{thm:counterexamplertep} is immediate from the previous lemmas.

We now extend this compilation chain to show that \rtep does not imply
\rtinip either. We introduce a new command at the target level,
\trg{leak}, to model information leakage.  The semantics of this new
instruction is simple: \trg{leak} reduces to a non-deterministically
chosen natural number.
Hence, this command can model looking-up the value of a secret inside
memory, and outputting it publicly.
All outputs are considered public.

Then, we also modify the compilation of partial programs having a boolean argument:
\begin{align*}
    \trg{\cmp{(f(x : \Bools) \mapsto e)}} &=
    \trg{f(x : \Nats) \mapsto \ifte{x<2}{\cmp{e}}{\ifte{x<3}{f(2)}{\ifte{x<4}{42}{leak()}}}}
\end{align*}

It is easy to show that the previous lemmas still hold.
Now is left to show that the compiler does not satisfy \rtinip.
\begin{itemize}
\item Every source whole program trivially satisfies
  termination-insenstive noninterference, because whole source
  programs are completely deterministic.
\item Now, consider a source program \(\src{P}\) with a boolean
  argument, and the target context \(\trg{C_T} = \trg{f(3)}\).  Then,
  \(\trg{\plug{C_T}{\cmp{P}}} \leadsto 0\), and
  \(\trg{\plug{C_T}{\cmp{P}}} \leadsto 1\).  The public inputs are
  identical, but not the public output. Hence the compiler does not
  satisfy \rtinip.
\end{itemize}


\section{Context Composition by Full Reflection or Internal
  Nondeterminism in the Source Language}
\label{sec:composing-contexts}

In this section we prove the theorems from
\oldautoref{sec:context-composition}, where we analyzed how certain
features of the source language can greatly influence the partial
order in \Cref{fig:order}.
In \Cref{sec:java} we assume source programs can completely examine their own
code, a mechanism that is sometimes called \emph{full reflection}.
First of all we need to introduce relational subset-closed hyperproperties, 
the classes of relational hyperproperties that are downward-closed
in each of its arguments.
Then in \Cref{sec:pical}, we assume it is possible to build a source
context $\src{C}$ whose behaviors approximate two given source
contexts $\src{C_1}$ and $\src{C_2}$.
This is the case when an operator for internal nondeterministic choice
is available.


\begin{definition}[\ii{2rSCH}] 
Given $R \in 2^{(2^{\ii{Trace}} \times ~2^{\ii{Trace}})}$
  \begin{align*}
     ~R \in \ii{2rSCH} \iff \forall (b_1, b_2) \in R \ldotp
     \forall s_1 \subseteq b_1, ~ s_2 \subseteq b_2 \ldotp (s_1, s_2) \in R \\  
  \end{align*}
  
\end{definition}

\begin{definition}[\criterion{R2rSCH}] 

  \begin{align*}
    \criterion{R2rSCH} : \quad \forall \src{P_1} \src{P_2} ~ R \in \mathit{2rSCH} ~ &  \forall \src{C_s} \ldotp (\behav{\src{C_s \hole{P_1}}}, \behav{\src{C_s \hole{P_2}}}) \in R \\                                                                                                                    &  \forall \trg{C_t} \ldotp (\behav{\trg{C_T\hole{\cmp{P_1}}}}, \behav{\trg{C_T\hole{\cmp{P_1}}}}) \in R
  \end{align*}
  
\end{definition}

\begin{definition} [\pf{\criterion{R2rSCH}}] 
  \begin{align*}
    \pf{\criterion{R2rSCH}} : \forall \src{P_1} \src{P_2} \trg{C_T} \ldotp ~\exists \src{C_S} \ldotp ~ & \behav{\src{C_s \hole{P_1}}} \subseteq (\behav{\trg{C_T\hole{\cmp{P_1}}}} \wedge \\
                                                                                     & \behav{\src{C_s \hole{P_2}}} \subseteq (\behav{\trg{C_T\hole{\cmp{P_2}}}}
  \end{align*}  
\end{definition}

\begin{lemma} 
 \criterion{R2rSC} $\iff$ \pf{\criterion{R2rSC}}  
\end{lemma}
\begin{proof}
  See file Criteria.v, theorem R2rSCHC\_R2rSCHP. 
\end{proof}

As usual, it is possible to generalize these definitions from binary relations
to relations of finite or arbitrary arities.

\begin{definition}[\criterion{RrSCH}]
	\begin{align*}
		\criterion{RrSCH}:~
  \forall R \in 2^{(\srcAll \to \mathit{SCH})}.~&
    (\forall\src{C_S} \ldotp
  (\lambda \src{P}\ldotp \src{\behav{C_S\hole{P}}}) \in R)
  \Rightarrow \\
  &(\forall \trg{C_T} \ldotp (\lambda \src{P}\ldotp \trg{\behav{C_T\hole{\cmp{P}}}})
  \in R)
	\end{align*}
\end{definition} 


\subsection{Context Composition by Full Reflection}\label{sec:java}

In this section we discuss our criteria assuming source programs can
fully examine their own code, as is enabled by the use of {\em full
  reflection} mechanisms in languages Lisp~\cite{Smith84} and Smalltalk.
More precisely, details we assume that given two distinct source programs
$\src{P_1}, \src{P_2}$ it is possible to compose two source contexts
$\src{C_1}, \src{C_2}$, we write $\src{C} = \src{C_1} \otimes \src{C_2}$ 
such that $\behav{\src{C\hole{P_i}}} = \behav{\src{C_{i}\hole{P_i}}}, ~ i = 1, 2$. 
\Cref{fig:order} reduces to the following diagram: 
\medskip
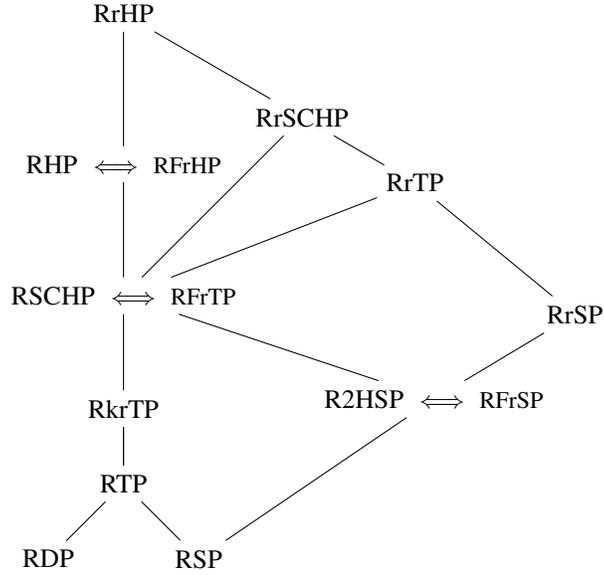
\begin{center}
\begin{tikzpicture}[scale = 0.5]
		\node  (0) at (-9, 10) {RrHP};
		\node  (1) at (-9, 6) {RHP \small{$\iff$} RFrHP};
		\node  (2) at (-9, 2.5) {RSCHP \small{$\iff$} RFrTP};
		\node  (4) at (-9, -2.5) {RTP};
		\node  (5) at (-11, -4.5) {RDP};
		\node  (6) at (-7, -4.5) {RSP};
		\node  (7) at (-0.75, -0.25) {R2HSP \small{$\iff$} RFrSP};
		\node  (8) at (3, 2) {RrSP};
		\node  (9) at (-9, -0.5) {RkrTP};
		\node  (10) at (-1.25, 5.5) {RrTP};
		\node  (11) at (-4.25, 7.25) {RrSCHP};
		\draw (0) to (11);
		\draw (11) to (10);
		\draw (10) to (8);
		\draw (8) to (7);
		\draw (7) to (6);
		\draw (4) to (6);
		\draw (4) to (5);
		\draw (9) to (4);
		\draw (2) to (1);
		\draw (1) to (0);
		\draw (2) to (9);
		\draw (9) to (4);
		\draw (11) to (2);
		\draw (10) to (2);
                \draw (2) to (7);
\end{tikzpicture}
\end{center}
\medskip

%
The file FullReflection.v contains proofs of the following collapses
\begin{itemize}
\item $\rthsp \Rightarrow \rtrsp$                         (theorem R2HSP\_R2rSP)
\ch{Premise didn't match, so I fixed the diagram. Please double check.}
\item $\rhp \Rightarrow \rtrhp$                           (theorem RHP\_R2rHP)
\item $\criterion{RSCH} \Rightarrow \criterion{R2rSCH}$ (theorem RSCHP\_R2rSCHP)
\end{itemize}
\medskip
To sketch a proof of $\rhp \Rightarrow \rtrhp$, consider their \emph{property-free} characterizations.
For $\src{P_1}, \src{P_2}$ distinct and $\trg{C}$ apply twice $\pf{\rhp}$ and get two source contexts $\src{C_1}, \src{C_2}$.
Then $\src{C_1} \otimes \src{C_2}$ satisfies the thesis. We can generalize these facts to finitary relations. 

\begin{theorem}
 $\rthsp \Rightarrow \rfrsp$  
\end{theorem}
\begin{proof}
  Same argument used in reflection.v, theorem R2HSP\_R2rSP. 
\end{proof}

\begin{theorem}
 $\rhp \Rightarrow \criterion{RFrH}$
\end{theorem}
\begin{proof}
  Same argument used in reflection.v, theorem RHP\_R2rHP.
\end{proof}

\begin{theorem}
 $\criterion{RSCH} \Rightarrow \criterion{RFrSCH}$ 
\end{theorem}
\begin{proof}
  Same argument used in reflection.v, theorem RSCHP\_R2rSCHP.
\end{proof}

Some of the variants of the results in this section where previously stated
in \oldautoref{thm:main-reflection}.

\subsection{Context Composition by Internal Nondeterministic Choice} \label{sec:pical}

In this section we discuss our criteria in presence of source contexts
that can nondeterministically behave like one of two already existing
source contexts.
Many criteria, in general stronger,
become equivalent to weaker ones. For instance an $\pf{\rsp}$ compiler preserves much more than the 
robust satisfaction of safety properties, including \(2\)-hypersafety. 
Formally we assume to have an operator $\oplus : \src{\ctx} \times \src{\ctx} \rightarrow \src{\ctx}$ such that 
\begin{equation*}
  \forall \src{C_1} \src{C_2} \src{P}. ~ \behav{\src{C_1} \oplus \src{C_2}} \src{\hole{P}}) \supseteq  \behav{\src{C_1 \hole{P}}} \cup \behav{\src{C_2 \hole{P}}}
\end{equation*} 
\Cref{fig:order} reduces to the following diagram: 
\medskip
\begin{center}
\begin{tikzpicture}[scale = 0.50]
		\node (0) at (-11, 11.25) {RrHP};
		\node (1) at (-11, 8) {RFrHP};
		\node (2) at (-11, 5) {RHP};
		\node (3) at (-11, 2) {RSCHP \small{$\iff$} RFrSCHP};
		\node (4) at (-11, -1.25) {RTP $\iff$ RFrTP};
		\node (5) at (-14, -4) {RDP};
		\node (6) at (-8, -4) {RSP \small{$\iff$} RFrSP};
		\node (7) at (-6.75, 9) {RrSCHP};
		\node (8) at (-4, 7.25) {RrTP};
		\node (9) at (-1.75, 5.75) {RrSP};
		\draw (0) to (1);
		\draw (1) to (2);
		\draw (2) to (3);
		\draw (3) to (4);
		\draw (4) to (5);
		\draw (4) to (6);
		\draw (0) to (7);
		\draw (7) to (8);
		\draw (8) to (9);
		\draw (9) to (6);
		\draw (7) to (2);
		\draw (8) to (3);
\end{tikzpicture}
\end{center}
\medskip

%
\medskip
The file InternalNondet.v contains proofs of the following binary collapses
\begin{itemize} 
 \item $\rschp \Rightarrow \criterion{R2rSCH}$ (theorem RSCHP\_R2rSCHP)
 \item $\rtp \Rightarrow \rtrtp$                 (theorem RTP\_R2rTP)
 \item $\rsp \Rightarrow \rtrsp$                 (theorem RSP\_R2rSP)
\end{itemize}
To sketch a proof for $\rsp \Rightarrow \rtrsp$ consider their \emph{property-free} characterizations, 
and assume $\trg{C\hole{\cmp{P}} \rightsquigarrow}\ m_1, m_2$. Apply twice $\pf{\rsp}$ and 
get two, possibly different, source contexts $\src{C_1}, \src{C_2}$, then $\src{C} = \src{C_1} \oplus \src{C_2}$ satisfies the thesis.
We can generalize these facts to finitary relations. \\

\begin{theorem}
  $\rschp \Rightarrow \criterion{RFrSCH}$  
\end{theorem}
\begin{proof}
  Same argument used in nd\_ctxs.v, theorem RSCHP\_R2rSCHP. 
\end{proof}

\begin{theorem}
  $\rtp \Rightarrow \criterion{RFrT}$  
\end{theorem}
\begin{proof}
  Same argument used in nd\_ctxs.v, theorem RTP\_R2rTP.
\end{proof}

\begin{theorem}
  $\rsp \Rightarrow \rfrsp$  
\end{theorem}
\begin{proof}
  Same argument used in nd\_ctxs.v, theorem RSP\_R2rSP.
\end{proof}

Some of the variants of the results in this section where previously stated
in \oldautoref{thm:main-nondet-choice}.

\section{Proof Techniques for \(\pf{\rrhp}_{\bowtie}\) and \(\pf{\rfrxp}_{\bowtie}\)}
\label{sec:proofs-appendix}

This section presents the formal details of \oldautoref{sec:example}.
As explained in the main paper, we use two different proof techniques, one
that is ``context-based'', and the other ``trace-based'', to prove two different
security criteria for the same compilation chain.
We argue that one of these techniques, the trace-based one, while less powerful,
still gives us an interesting criterion, and should be more generic, as it
relies less on the details of the languages.

\paragraph{A remark on the security criteria used in this section}
In the languages used in this example, not all programs and contexts can be
linked together.
In order for it to be the case, they have to satisfy some interfacing
constraints.
Here, these constraints are the existence of functions called but not defined by
the context, and, in the source language, also agreement on the types of these
functions.
We introduce the operators \(\src{\bowtie}\) and \(\trg{\bowtie}\) to represent
these constraints. For instance, \(\src{P\bowtie C}\) means that \(\src{C}\) is
linkable with \(\src{P}\).
We prove two 
(\autoref{sec:rel-hyper}) and \pf{\rfrxp}
(\autoref{sec:fin-rel-Xsafety}), named \(\pf{\rrhp}_{\bowtie}\)
(\autoref{def:rrhc-linkable}) and \(\pf{\rfrxp}_{\bowtie}\)
(\autoref{def:rfrxc-linkable}), that take into account these
linkability predicates.


\subsection{The Source Language \texorpdfstring{\Lt}{}}\label{sec:inst-src}

A list of elements $e_1,\cdots,e_n$ is indicated as \OB{e}, the empty list is $\emptyset$.

\subsubsection{Syntax}
\begin{align*}
	\mi{Program}~\src{P} \bnfdef&\ \src{\OB{I};\OB{F}}
	\\
	\mi{Contexts}~\src{C} \bnfdef&\ \src{e}
	\\
	\mi{Interfaces}~\src{I} \bnfdef&\ \src{f:\tau\to\tau}
	\\
	\mi{Functions}~\src{F} \bnfdef&\ \src{f(x:\tau):\tau \mapsto \ret e}
	\\
	\mi{Types}~\src{\tau} \bnfdef&\ \Bools \mid \Nats
	\\
	\mi{Operations}~\src{\op} \bnfdef&\ \src{+} \mid \src{-}
	\\
	\mi{Values}~\src{v} \bnfdef&\ \trues \mid \falses \mid \src{n}\in\mb{N}
	\\
	\mi{Expressions}~\src{e} \bnfdef&\ \src{x} \mid \src{v} \mid \src{e \op e} \mid \src{\letin{x:\tau}{e}{e}} \mid \src{\ifte{e}{e}{e}} \mid \src{e \geq e}
	\\
	\mid&\ \src{\call{f}~e} \mid \src{\readexp} \mid \src{\writeexp{e}} \mid \fails 
	\\
	\mi{Runtime\ Expr.} ~\src{e} \bnfdef&\ \cdots \mid \src{\ret e}
	\\
	\mi{Eval.\ Ctxs.}~\src{\evalctx} \bnfdef&\ \src{\hole{\cdot}} \mid \src{e \op \evalctx} \mid \src{\evalctx \op n} \mid \src{\letin{x}{\evalctx}{e}} \mid \src{\ifte{\evalctx}{e}{e}} \mid \src{e \geq \evalctx} \mid \src{\evalctx \geq n}
	\\
	\mid&\ \src{\call{f}~\evalctx} \mid \src{\writeexp{\evalctx}} \mid \src{\ret \evalctx}
	\\
	\mi{Substitutions}~\src{\rho} \bnfdef&\ \subs{v}{x}
	\\
	\mi{Prog.\ States}~\src{\Omega}\bnfdef&\ \src{P\triangleright e } \mid \fails
	\\ 
	\mi{Environments}~\src{\Gamma} \bnfdef&\ \srce \mid \src{\Gamma; (x:\tau)}
	\\
	\mi{Labels}~\bl{\lambda} \bnfdef&\ \bl{\epsilon} \mid \bl{\alpha}
	\\
	\mi{Actions}~\bl{\alpha} \bnfdef&\ \bl{\rdl{n}} \mid \bl{\wrl{n}} \mid \terc \mid \divc \mid \failactc
	\\
        \mi{Interactions}~\bl{\gamma} \bnfdef&\ \bl{\cl{f}{v}} \mid \bl{\rt{v}} 
	\\
	\mi{Behaviors}~\bl{\beta} \bnfdef&\ \bl{\OB{\alpha}}
        \\
        \mi{Traces}~\bl{\sigma} \bnfdef&\ \emptyset \mid \bl{\sigma\alpha} \mid \bl{\sigma\gamma}
\end{align*}

\subsubsection{Static Semantics}\label{sec:typ-src}
The static semantics follows these typing judgements.
\begin{align*}
	&\vdash \src{P}
	&&\text{Program \src{P} is well-typed.}
	\\
	&\src{P}\vdash\src{F} : \src{\tau\to\tau} 
	&&\text{Function \src{F} has type \src{\tau\to\tau} in program \src{P}.}
	\\
	&\src{\Gamma}\vdash\diamond 
	&&\text{Environment \src{\Gamma} is well-formed.}
	\\
	&\src{P;\Gamma}\vdash \src{e} : \src{\tau}
	&&\text{Expression \src{e} has type \src{\tau} in \src{\Gamma} and \src{P}.}
\end{align*}

\mytoprule{\vdash \src{P}}
\begin{center}
	\typerule{T\Lt-component}{
		\src{P}\equiv\src{\OB{I} ; \OB{F}} 
		&
		\src{P}\vdash\src{\OB{F}}:\src{\tau\to\tau}
		&
		\dom{\src{\OB{F}}}\subseteq\src{\OB{I}}
	}{
		\vdash \src{P}
	}{ts-co}
\end{center}
\botrule
\mytoprule{\src{P}\vdash\src{F} : \src{\tau\to\tau}}
\begin{center}
	\typerule{T\Lt-function}{
		\src{F}\equiv \src{f(x:\tau):\tau'\mapsto \ret e}
		&
		\src{P};\src{x:\tau}\vdash \src{e} : \src{\tau'}
	}{
		\src{P}\vdash\src{F}:\src{\tau\to\tau'}
	}{ts-fu}
\end{center}
\botrule
\mytoprule{\src{P;\Gamma}\vdash \src{e} : \src{\tau}}
\begin{center}
	\typerule{T\Lt-true}{
		\src{\Gamma}\vdash\diamond
	}{
		\src{P;\Gamma}\vdash\trues:\Bools
	}{ts-true}
	\typerule{T\Lt-false}{
		\src{\Gamma}\vdash\diamond
	}{
		\src{P;\Gamma}\vdash\falses:\Bools
	}{ts-false}
	\typerule{T\Lt-nat}{
		\src{\Gamma}\vdash\diamond
	}{
		\src{P;\Gamma}\vdash\src{n}:\Nats
	}{ts-nat}
	\typerule{T\Lt-var}{
		\src{x:\tau}\in\src{\Gamma}
	}{
		\src{P;\Gamma}\vdash\src{x}:\src{\tau}
	}{ts-var}
	\typerule{T\Lt-op}{
		\src{P;\Gamma}\vdash \src{e} : \Nats
		\\
		\src{P;\Gamma}\vdash \src{e'} : \Nats
	}{
		\src{P;\Gamma}\vdash \src{e \op e'} : \Nats
	}{ts-op}
	\typerule{T\Lt-geq}{
		\src{P;\Gamma}\vdash \src{e} : \Nats
		\\
		\src{P;\Gamma}\vdash \src{e'} : \Nats
	}{
		\src{P;\Gamma}\vdash \src{e \geq e'} : \Bools
	}{ts-geq}
	\typerule{T\Lt-letin}{
		\src{P;\Gamma}\vdash \src{e} : \src{\tau} 
		\\
		\src{P;\Gamma;x:\tau}\vdash \src{e'} : \src{\tau'} 
	}{
		\src{P;\Gamma}\vdash \src{\letin{x:\tau}{e}{e'}} : \src{\tau'} 
	}{ts-vardef}
	\typerule{T\Lt-if}{
		\src{P;\Gamma}\vdash \src{e} : \Bools
		\\
		\src{P;\Gamma}\vdash \src{e_t} : \src{\tau} 
		&
		\src{P;\Gamma}\vdash \src{e_f} : \src{\tau} 
	}{
		\src{P;\Gamma}\vdash \src{\ifte{e}{e_t}{e_f}} : \src{\tau} 
	}{ts-if}
	\typerule{T\Lt-function-call}{
		(\src{f(x:\tau):\tau' \mapsto \ret e}\in\src{\OB{F}})
		\\
		\src{P}\equiv \src{\OB{I};\OB{F}}
		&
		\src{P;\Gamma}\vdash \src{e} : \src{\tau}
	}{
		\src{P;\Gamma}\vdash \src{\call{f}~e} : \src{\tau'}
	}{ts-fun}
	\typerule{T\Lt-context-function-call}{
		(\src{f(x:\tau):\tau' \mapsto \ret e}\in\src{\OB{I}})
		\\
		\src{P}\equiv \src{\OB{I};\OB{F}}
		&
		\src{P;\Gamma}\vdash \src{e} : \src{\tau}
	}{
		\src{P;\Gamma}\vdash \src{\call{f}~e} : \src{\tau'}
	}{ts-fun-c}
        \typerule{T\Lt-read}{
	}{
		\src{P;\Gamma}\vdash\src{\readexp}:\Nats
	}{ts-read}
	\typerule{T\Lt-write}{
		\src{P;\Gamma}\vdash \src{e} : \Nats
	}{
		\src{P;\Gamma}\vdash\src{\writeexp{e}}:\Nats
	}{ts-write}
	\typerule{T\Lt-fail}{
	}{
		\src{P;\Gamma}\vdash\fails:\src{\tau}
	}{ts-fail}
\end{center}
\botrule

\mytoprule{\src{P \bowtie C} }
\begin{center}
  \typerule{Link-\Lt}{
		\src{P}\equiv \src{\OB{I};\OB{F}}
		&
		\src{C}\equiv\src{e}
		&\
		\vdash\src{P}
		&\
		\src{P};\src{\Gamma}\vdash\src{e}:\src{\tau}
		\\
		\fails \notin \src{P}
		&
		\src{\readexp}, \src{\writeexp{\_}}\notin\src{C}
		&
		\forall \src{\call{f}}\in\src{C}, \src{f}\in\src{\OB{I}}
  }{
    \src{P \bowtie C}
  }{linkable-src}
\end{center}
\botrule

\subsubsection{Dynamic Semantics}\label{sec:sem-src}
\begin{align*}
	&\src{\Omega \xtosb{\lambda} \Omega'} 
	&& \text{Program state \src{\Omega} steps to \src{\Omega'} emitting action \src{\lambda}.}
	\\
	&\src{\Omega \Xtosb{\beta} \Omega'} 
	&& \text{Program state \src{\Omega} steps to \src{\Omega'} with behavior \src{\beta}.}
\end{align*}

\mytoprule{\src{P\triangleright e \xtosb{\lambda} P\triangleright e'} }
\begin{center}
	\typerule{E\Lt-op}{
		n \op n'= n''
	}{
		\src{P\triangleright n\op n' \xtosb{\epsilon} P\triangleright n''} 
	}{ev-s-op}
	\typerule{E\Lt-geq-true}{
		n \geq n'
	}{
		\src{P\triangleright n\geq n' \xtosb{\epsilon} P\triangleright \trues} 
	}{ev-s-geq-t}
	\typerule{E\Lt-geq-false}{
		n < n'
	}{
		\src{P\triangleright n\geq n' \xtosb{\epsilon} P\triangleright \falses} 
	}{ev-s-geq-f}
	\typerule{E\Lt-if-true}{
	}{
		\src{P\triangleright \ifte{\trues}{e}{e'} \xtosb{\epsilon} P\triangleright e} 
	}{ev-s-if-t}
	\typerule{E\Lt-if-false}{
	}{
		\src{P\triangleright \ifte{\falses}{e}{e'} \xtosb{\epsilon} P\triangleright e'} 
	}{ev-s-if-f}
	\typerule{E\Lt-let}{
	}{
		\src{P\triangleright \letin{x}{v}{e} \xtosb{\epsilon} P\triangleright e\subs{v}{x}} 
	}{ev-s-let}
        \typerule{E\Lt-call-internal}{
		\src{f(x : \tau_1) :\tau_2 \mapsto \ret e} \in \src{P}
	}{
		\src{P\triangleright_{\OB{f}} \call{f}~v \xtosb{\epsilon} P\triangleright_{\OB{f},f} \ret e\subs{v}{x}} 
	}{ev-s-call-inte}
	\typerule{E\Lt-call-in}{
		\src{f(x : \tau_1) :\tau_2 \mapsto \ret e} \in \src{P}
	}{
		\src{P\triangleright_{\epsilon} \call{f}~v \xtosb{\cl{f}{v}} P\triangleright_{f} \ret e\subs{v}{x}} 
	}{ev-s-call-in}
        \typerule{E\Lt-ret-internal}{
	}{
		\src{P\triangleright_{\OB{f},f,f'} \ret v \xtosb{\epsilon} P\triangleright_{\OB{f},f} v} 
	}{ev-s-ret-inte}
        \typerule{E\Lt-ret-out}{
	}{
		\src{P\triangleright_{f} \ret v \xtosb{\rt{v}} P\triangleright_{} v} 
	}{ev-s-ret-out}
	\typerule{E\Lt-read}{
	}{
		\src{P\triangleright \readexp \xtosb{\rdl{n}} P\triangleright n} 
	}{ev-s-read}
	\typerule{E\Lt-write}{
	}{
		\src{P\triangleright \writeexp{n} \xtosb{\wrl{n}} P\triangleright n} 
	}{ev-s-write}
	\typerule{E\Lt-ctx}{
		\src{P\triangleright e} \xtosb{\epsilon} \src{P\triangleright e'}
	}{
		\src{P \triangleright \evalctxs{e}} \xtosb{\epsilon} \src{P\triangleright \evalctxs{e'}}
	}{ev-s-cth}
	\typerule{E\Lt-fail}{
	}{
		\src{P \triangleright \fails} \xtosb{\failactc} \fails
	}{ev-s-fail}
\end{center}
\botrule

\mytoprule{\src{P\triangleright e \Xtosb{\beta} P\triangleright e'} }
\begin{center}
	\typerule{E\Lt-refl}{
	}{
		\src{\Omega} \Xtosb{} \src{\Omega}
	}{ev-s-refl}
	\typerule{E\Lt-terminate}{
		\src{\Omega} \nXtos{} \src{\_}
	}{
		\src{\Omega} \Xtosb{\termc} \src{\Omega}
	}{ev-s-term}
	\typerule{E\Lt-diverge}{
		\forall \src{n}.~ \src{\Omega} \src{\xtos{\epsilon}\redapp{n}} \src{\Omega_n'}
	}{
		\src{\Omega} \Xtosb{\divrc} \src{\Omega}
	}{ev-s-divr}
	\typerule{E\Lt-silent}{
		\src{\Omega}\xtosb{\epsilon}\src{\Omega'}
	}{
		\src{\Omega}\Xtosb{}\src{\Omega'}
	}{beh-s-silent}
	\typerule{E\Lt-single}{
		\src{\Omega}\xtosb{\alpha}\src{\Omega'}
	}{
		\src{\Omega}\Xtosb{\alpha}\src{\Omega'}
	}{beh-s-sin}
        \typerule{E\Lt-silent-act}{
		\src{\Omega}\xtosb{\gamma}\src{\Omega'}
	}{
		\src{\Omega}\Xtosb{}\src{\Omega'}
	}{beh-s-silent-a}
	\typerule{E\Lt-cons}{
		\src{\Omega}\Xtosb{\beta}\src{\Omega''}
		&
		\src{\Omega''}\Xtosb{\beta'}\src{\Omega'}
	}{
		\src{\Omega}\Xtosb{\beta\beta'}\src{\Omega'}
	}{beh-s-cons}	
\end{center}
\botrule

\mytoprule{\src{P\triangleright e \Xtolsb{t} P\triangleright e'} }
\begin{center}	
	\typerule{E\Lt-silent}{
		\src{\Omega}\xtosb{\epsilon}\src{\Omega'}
	}{
		\src{\Omega\Xtolsb{} \Omega'}
	}{tr-s-silent}
	\typerule{E\Lt-action}{
		\src{\Omega}\Xtosb{\alpha}\src{\Omega'}
	}{
		\src{\Omega\Xtolsb{\alpha}\Omega'}
	}{tr-s-act}
	\typerule{E\Lt-single}{
		\src{\Omega}\xtosb{\gamma}\src{\Omega'}
	}{
		\src{\Omega\Xtolsb{\gamma}\Omega'}
	}{tr-s-sin}
	\typerule{E\Lt-cons}{
		\src{\Omega\Xtolsb{\sigma}\Omega''}
		\\
		\src{\Omega''\Xtolsb{\sigma'}\Omega'}
	}{
		\src{\Omega\Xtolsb{\sigma\sigma'}\Omega'}
	}{tr-s-cons}	
\end{center}
\botrule

\subsubsection{Auxiliaries and Definitions}

\mytoprule{\text{Helpers}}
\begin{center}
	\typerule{\Lt-Initial State}{
	  \src{P \bowtie C} &
          \src{C} \equiv \src{e}
	}{
		\SInits{\src{\plug{C}{P}}} = \src{P\triangleright e}
	}{s-ini}
\end{center}
\botrule

\begin{definition}[Program Behaviors]\label{def:src-prog-beh}
	\begin{align*}
		\behavs{\src{P}}  
		&= \myset{ \bl{\beta} }{ \exists\src{\Omega'}. \SInits{P}\Xtosb{\beta}\src{\Omega'} }		
	\end{align*}	
\end{definition}

\begin{theorem}[Progress]
If $\src{P;\Gamma}\vdash\src{e}:\src{\tau}$ then either $\src{e}\equiv\src{v}$ or $\exists\src{e'}. \src{P\triangleright e\reds P\triangleright e'}$.
\end{theorem}

\begin{theorem}[Preservation]
If $\src{P;\Gamma}\vdash\src{e}:\src{\tau}$ and $\src{P\triangleright e\reds P\triangleright e'}$ then $\src{P;\Gamma}\vdash\src{e'}:\src{\tau}$.
\end{theorem}


\subsection{The Target Language \texorpdfstring{\Ld}{}}\label{sec:inst-trg-dyn}

\subsubsection{Syntax}
\begin{align*}
	\mi{Program}~\trg{P} \bnfdef&\ \trg{\OB{I};\OB{F}}
	\\
	\mi{Contexts}~\trg{C} \bnfdef&\ \trg{e}
	\\
	\mi{Interfaces}~\trg{I} \bnfdef&\ \trg{f}
	\\
	\mi{Functions}~\trg{F} \bnfdef&\ \trg{f(x) \mapsto \ret e}
	\\
	\mi{Types}~\trg{\tau} \bnfdef&\ \Boolt \mid \Natt
	\\
	\mi{Operations}~\trg{\op} \bnfdef&\ \trg{+} \mid \trg{-}
	\\
	\mi{Values}~\trg{v} \bnfdef&\ \truet \mid \falset \mid \trg{n}\in\mb{N}
	\\
	\mi{Expressions}~\trg{e} \bnfdef&\ \trg{x} \mid \trg{v} \mid \trg{e \op e} \mid \trg{\letin{x}{e}{e}} \mid \trg{\ifte{e}{e}{e}} \mid \trg{e \geq e}
	\\
	\mid&\ \trg{\call{f}~e} \mid \trg{\readexp} \mid \trg{\writeexp{e}} \mid \failt \mid \trg{e \checkty{\tau}}
	\\
	\mi{Runtime\ Expr.} ~\trg{e} \bnfdef&\ \cdots \mid \trg{\ret e}
	\\
	\mi{Eval.\ Ctxs.}~\trg{\evalctx} \bnfdef&\ \trg{\hole{\cdot}} \mid \trg{e \op \evalctx} \mid \trg{\evalctx \op n} \mid \trg{\letin{x}{\evalctx}{e}} \mid \trg{\ifte{\evalctx}{e}{e}} \mid \trg{e \geq \evalctx} \mid \trg{\evalctx \geq n}
	\\
	\mid&\ \trg{\call{f}~\evalctx} \mid \trg{\writeexp{\evalctx}} \mid \trg{\ret \evalctx} \mid \trg{\evalctx \checkty{\tau}}
	\\
	\mi{Substitutions}~\trg{\rho} \bnfdef&\ \subt{v}{x}
	\\
	\mi{Prog.\ States}~\trg{\Omega}\bnfdef&\ \trg{P\triangleright_{\OB{f}} e } \mid \failt
	\\
	\mi{Labels}~\bl{\lambda} \bnfdef&\ \bl{\epsilon} \mid \bl{\alpha} \mid \bl{\gamma}
	\\
	\mi{Actions}~\bl{\alpha} \bnfdef&\ \bl{\rdl{n}} \mid \bl{\wrl{n}} \mid \terc \mid \divrc \mid \failactc
	\\
	\mi{Interactions}~\bl{\gamma} \bnfdef&\ \bl{\cl{f}{v}} \mid \bl{\rt{v}} 
	\\
	\mi{Behaviors}~\bl{\beta} \bnfdef&\ \bl{\OB{\alpha}}
	\\
	\mi{Traces}~\bl{\sigma} \bnfdef&\ \emptyset \mid \bl{\sigma\alpha} \mid \bl{\sigma\gamma}
\end{align*}
Program states carry around the stack of called functions (the \trg{\OB{f}} subscript) in order to correctly characterise calls and returns that go in Traces.
We mostly omit this stack when it just clutters the presentation without itself changing and make it explicit only when it is needed.

We define the linkability operator as follows:

\mytoprule{\trg{P \bowtie C} }
\begin{center}
  \typerule{Link-\Ld}{
		\trg{P}\equiv \trg{\OB{I};\OB{F}}
		&
		\trg{C}\equiv \trg{e}
		\\
		\trg{\readexp}, \trg{\writeexp{\_}}\notin\trg{C}
		&
		\forall \trg{\call{f}}\in\trg{C}, \trg{f}\in\trg{\OB{I}}
  }{
    \trg{P \bowtie C}
  }{linkable-src}
\end{center}
\botrule

\subsubsection{Dynamic Semantics}
\begin{align*}
	&\trg{\Omega \xtotb{\lambda} \Omega'} 
	&& \text{Program state \trg{\Omega} steps to \trg{\Omega'} emitting action \trg{\lambda}.}
	\\
	&\trg{\Omega \Xtotb{\beta} \Omega'} 
	&& \text{Program state \trg{\Omega} steps to \trg{\Omega'} with behavior \trg{\beta}.}
	\\
	&\trg{\Omega \Xtoltb{\sigma} \Omega'} 
	&& \text{Program state \trg{\Omega} steps to \trg{\Omega'} with trace \trg{\sigma}.}
\end{align*}

\mytoprule{\trg{P\triangleright e \xtotb{\lambda} P\triangleright e'} }
\begin{center}
	\typerule{E\Ld-op}{
		n \op n'= n''
	}{
		\trg{P\triangleright n\op n' \xtotb{\epsilon} P\triangleright n''} 
	}{ev-t-op}
	\typerule{E\Ld-geq-true}{
		n \geq n'
	}{
		\trg{P\triangleright n\geq n' \xtotb{\epsilon} P\triangleright \truet} 
	}{ev-t-geq-t}
	\typerule{E\Ld-geq-false}{
		n < n'
	}{
		\trg{P\triangleright n\geq n' \xtotb{\epsilon} P\triangleright \falset} 
	}{ev-t-geq-f}
	\typerule{E\Ld-if-true}{
	}{
		\trg{P\triangleright \ifte{\truet}{e}{e'} \xtotb{\epsilon} P\triangleright e} 
	}{ev-t-if-t}
	\typerule{E\Ld-if-false}{
	}{
		\trg{P\triangleright \ifte{\falset}{e}{e'} \xtotb{\epsilon} P\triangleright e'} 
	}{ev-t-if-f}
	\typerule{E\Ld-let}{
	}{
		\trg{P\triangleright \letin{x}{v}{e} \xtotb{\epsilon} P\triangleright e\subt{v}{x}} 
	}{ev-t-let}
	\typerule{E\Ld-read}{
	}{
		\trg{P\triangleright \readexp \xtotb{\rdl{n}} P\triangleright n} 
	}{ev-t-read}
	\typerule{E\Ld-write}{
	}{
		\trg{P\triangleright \writeexp{n} \xtotb{\wrl{n}} P\triangleright n} 
	}{ev-t-write}
	\typerule{E\Ld-ctx}{
		\trg{P\triangleright e} \xtotb{\epsilon} \trg{P\triangleright e'}
	}{
		\trg{P \triangleright \evalctxt{e}} \xtotb{\epsilon} \trg{P\triangleright \evalctxt{e'}}
	}{ev-t-cth}
	\typerule{E\Ld-check-bool-true}{
		\trg{v}\equiv\truet \vee \trg{v}\equiv\falset
	}{
		\trg{P\triangleright v \checkty{\Bool} \xtotb{\epsilon} P\triangleright \truet} 
	}{ev-t-ch-b-t}
	\typerule{E\Ld-check-bool-false}{
	}{
		\trg{P\triangleright n \checkty{\Bool} \xtotb{\epsilon} P\triangleright \falset} 
	}{ev-t-ch-b-f}
	\typerule{E\Ld-check-nat-true}{
	}{
		\trg{P\triangleright n \checkty{\Nat} \xtotb{\epsilon} P\triangleright \truet} 
	}{ev-t-ch-n-t}
	\typerule{E\Ld-check-nat-false}{
		\trg{v}\equiv\truet \vee \trg{v}\equiv\falset
	}{
		\trg{P\triangleright v \checkty{\Nat} \xtotb{\epsilon} P\triangleright \falset} 
	}{ev-t-ch-n-f}
	\typerule{E\Ld-call-internal}{
		\trg{f(x) \mapsto \ret e} \in \trg{P}
	}{
		\trg{P\triangleright_{\OB{f}} \call{f}~v \xtotb{\epsilon} P\triangleright_{\OB{f},f} \ret e\subt{v}{x}} 
	}{ev-t-call-inte}
	\typerule{E\Ld-call-in}{
		\trg{f(x) \mapsto \ret e} \in \trg{P}
	}{
		\trg{P\triangleright_{\epsilon} \call{f}~v \xtotb{\cl{f}{v}} P\triangleright_{f} \ret e\subt{v}{x}} 
	}{ev-t-call-in}
	\typerule{E\Ld-ret-internal}{
	}{
		\trg{P\triangleright_{\OB{f},f,f'} \ret v \xtotb{\epsilon} P\triangleright_{\OB{f},f} v} 
	}{ev-t-ret-inte}
	\typerule{E\Ld-ret-out}{
	}{
		\trg{P\triangleright_{f} \ret v \xtotb{\rt{v}} P\triangleright_{} v} 
	}{ev-t-ret-out}
	\typerule{E\Ld-op-fail}{
		\trg{v}\equiv\truet \vee \trg{v}\equiv\falset \vee \trg{v'}\equiv\truet \vee \trg{v'}\equiv\falset
	}{
		\trg{P\triangleright v\op v' \xtotb{\failactc} \failt} 
	}{ev-t-op-fail}
	\typerule{E\Ld-geq-fail}{
		\trg{v}\equiv\truet \vee \trg{v}\equiv\falset \vee \trg{v'}\equiv\truet \vee \trg{v'}\equiv\falset
	}{
		\trg{P\triangleright v\geq v' \xtotb{\failactc} \failt} 
	}{ev-t-geq-fail}
	\typerule{E\Ld-if-fail}{
	}{
		\trg{P\triangleright \ifte{n}{e}{e'} \xtotb{\failactc} \failt} 
	}{ev-t-if-fail}
	\typerule{E\Ld-fail}{
	}{
		\trg{P \triangleright \failt} \xtotb{\failactc} \failt
	}{ev-t-fail}
\end{center}
\botrule

\mytoprule{\trg{P\triangleright e \Xtotb{\beta} P\triangleright e'} }
\begin{center}
	\typerule{E\Ld-refl}{
	}{
		\trg{\Omega} \Xtotb{} \trg{\Omega}
	}{ev-t-refl}
	\typerule{E\Ld-terminate}{
		\trg{\Omega} \nXtot{} \trg{\_}
	}{
		\trg{\Omega} \Xtotb{\termc} \trg{\Omega}
	}{ev-t-term}
	\typerule{E\Ld-diverge}{
		\forall \trg{n}.~ \trg{\Omega} \trg{\xtot{\epsilon}\redapp{n}} \trg{\Omega_n'}
	}{
		\trg{\Omega} \Xtotb{\divrc} \trg{\Omega}
	}{ev-t-divr}
	\typerule{E\Ld-silent}{
		\trg{\Omega}\xtotb{\epsilon}\trg{\Omega'}
	}{
		\trg{\Omega}\Xtotb{}\trg{\Omega'}
	}{beh-t-silent}
	\typerule{E\Ld-silent-act}{
		\trg{\Omega}\xtotb{\gamma}\trg{\Omega'}
	}{
		\trg{\Omega}\Xtotb{}\trg{\Omega'}
	}{beh-t-silent-a}
	\typerule{E\Ld-single}{
		\trg{\Omega}\xtotb{\alpha}\trg{\Omega'}
	}{
		\trg{\Omega}\Xtotb{\alpha}\trg{\Omega'}
	}{beh-t-sin}
	\typerule{E\Ld-cons}{
		\trg{\Omega}\Xtotb{\beta}\trg{\Omega''}
		&
		\trg{\Omega''}\Xtotb{\beta'}\trg{\Omega'}
	}{
		\trg{\Omega}\Xtotb{\beta\beta'}\trg{\Omega'}
	}{beh-t-cons}	
\end{center}
\botrule

\mytoprule{\trg{P\triangleright e \Xtoltb{\sigma} P\triangleright e'} }
\begin{center}	
	\typerule{E\Ld-silent}{
		\trg{\Omega}\xtotb{\epsilon}\trg{\Omega'}
	}{
		\trg{\Omega\Xtoltb{} \Omega'}
	}{tr-t-silent}
	\typerule{E\Ld-action}{
		\trg{\Omega}\Xtotb{\alpha}\trg{\Omega'}
	}{
		\trg{\Omega\Xtoltb{\alpha}\Omega'}
	}{tr-t-act}
	\typerule{E\Ld-single}{
		\trg{\Omega}\xtotb{\gamma}\trg{\Omega'}
	}{
		\trg{\Omega\Xtoltb{\gamma}\Omega'}
	}{tr-t-sin}
	\typerule{E\Ld-cons}{
		\trg{\Omega\Xtoltb{\sigma}\Omega''}
		\\
		\trg{\Omega''\Xtoltb{\sigma'}\Omega'}
	}{
		\trg{\Omega\Xtoltb{\sigma\sigma'}\Omega'}
	}{tr-t-cons}	
\end{center}
\botrule

\subsubsection{Auxiliaries and Definitions}

\mytoprule{\text{Helpers}}
\begin{center}
	\typerule{\Ld-Initial State}{
	  \trg{P \bowtie C} &
          \trg{C} \equiv \trg{e}
	}{
		\SInitt{\trg{\plug{C}{P}}} = \trg{P\triangleright e}
	}{t-ini}
\end{center}
\botrule

\begin{definition}[Program Behaviors]\label{def:trg-prog-beh}
	\begin{align*}
		\behavt{\trg{P}}  
		&= \myset{ \bl{\beta} }{ \exists\trg{\Omega'}. \SInitt{P}\Xtotb{\beta}\trg{\Omega'} }		
	\end{align*}	
\end{definition}

\begin{definition}[Program Traces]\label{def:trg-prog-tr}
	\begin{align*}
		\trt{\trg{P}}  
		&= \myset{ \bl{\sigma} }{ \exists\trg{\Omega'}. \trg{\SInitt{P}\Xtoltb{\sigma}\Omega'} }		
	\end{align*}	
\end{definition}


\subsection{\texorpdfstring{\comptd{\cdot}: {}}{} A Compiler from \texorpdfstring{\Lt}{Source} to \texorpdfstring{\Ld}{Target}}\label{sec:inst-comp-td}

\begin{align*}
	\tag{\comptd{\cdot}-Prog}\oldlabel{app:tr:comptd-prog}
	\comptd{\src{I_1,\cdots,I_m;F_1,\cdots,F_n}} 
		=&\ 
		\trg{\comptd{\src{I_1}},\cdots,\comptd{I_m};\comptd{F_1},\cdots,\comptd{F_n}}
	\\
	\tag{\comptd{\cdot}-Intf}\label{tr:comptd-intf}
	\comptd{\src{f:\tau\to\tau'}} 
		=&\ 
		\trg{f}
	\\
	\tag{\comptd{\cdot}-Fun}\oldlabel{app:tr:comptd-fun}
	\comptd{\src{f(x:\tau):\tau'\mapsto \ret e}} 
		=&\ 
		\trg{
			f(x)\mapsto \ret \ifte{x \checkty{\comptd{\tau}}}{\comptd{e}}{\failt}
		}
	\\
	\tag{\comptd{\cdot}-Nat}\label{tr:comptd-nat}
	\comptd{\src{n}} 
		=&\ 
		\trg{n}
	\\
	\tag{\comptd{\cdot}-True}\label{tr:comptd-true}
	\comptd{\trues} 
		=&\ 
		\truet
	\\
	\tag{\comptd{\cdot}-False}\label{tr:comptd-false}
	\comptd{\falses} 
		=&\ 
		\falset
	\\
	\tag{\comptd{\cdot}-Var}\label{tr:comptd-var}
	\comptd{\src{x}} 
		=&\ 
		\trg{x}
	\\
	\tag{\comptd{\cdot}-Op}\label{tr:comptd-op}
	\comptd{\src{e\op e'}} 
		=&\ 
		\trg{\comptd{e}\op \comptd{e'}}
	\\
	\tag{\comptd{\cdot}-Geq}\label{tr:comptd-geq}
	\comptd{\src{e\geq e'}} 
		=&\ 
		\trg{\comptd{e}\geq \comptd{e'}}
	\\
	\tag{\comptd{\cdot}-Let}\label{tr:comptd-let}
	\comptd{\src{\letin{x:\tau}{e}{e'}}} 
		=&\ 
		\trg{\letin{x}{\comptd{e}}{\comptd{e'}}}
	\\
	\tag{\comptd{\cdot}-If}\label{tr:comptd-if}
	\comptd{\src{\ifte{e}{e'}{e''}}} 
		=&\ 
		\trg{\ifte{\comptd{e}}{\comptd{e'}}{\comptd{e''}}}
	\\
	\tag{\comptd{\cdot}-Call}\oldlabel{app:tr:comptd-call}
	\comptd{\src{\call{f}~e}} 
		=&\ 
		\trg{\call{f}~\comptd{e}}
	\\
	\tag{\comptd{\cdot}-Read}\label{tr:comptd-read}
	\comptd{\src{\readexp}} 
		=&\ 
		\trg{\readexp}
	\\
	\tag{\comptd{\cdot}-Write}\label{tr:comptd-write}
	\comptd{\src{\writeexp{e}}} 
		=&\ 
		\trg{\writeexp{\comptd{e}}}
	\\
	\tag{\comptd{\cdot}-Ty-Nat}\label{tr:comptd-ty-nat}
	\comptd{\Nats} 
		=&\ 
		\Natt
	\\
	\tag{\comptd{\cdot}-Ty-Bool}\label{tr:comptd-ty-bool}
	\comptd{\Bools} 
		=&\ 
		\Boolt
\end{align*}

\subsection{Proof That \texorpdfstring{\comptd{\cdot}}{the Compiler} Is
  \texorpdfstring{\(\protect\pf{\rrhp}_{\bowtie}\)}{RRhP}}\label{sec:inst-proof-ct-dyn}

We prove that the compiler satisfies the following variant of \pf{\rrhp}:
\begin{definition}[\(\pf{\rrhp}_{\bowtie}\)]\label{def:rrhc-linkable}
  \begin{align*}
    \pf{\rrhp}_{\bowtie}:\quad
    \forall \src{\OB{I}}.~\forall\trg{C_T}.~ &\exists \src{C_S}.~\forall \src{P : \OB{I}}.~ \trg{\cmp{P} \bowtie C_T} \implies\\
    &\src{P \bowtie C_S} \wedge
     \trg{\behav{C_T\hole{\cmp{P}}}} = \src{\behav{C_S\hole{P}}}
  \end{align*}
\end{definition}

All programs must satisfy the same interface \(\src{\OB{I}}\) in order for
the linkability with a single \src{C_S} to be possible.

We also give the following property-full criteria:
\begin{definition}[\(\rrhp_{\bowtie}\)]\label{def:rrhp-linkable}
  $$\begin{multlined}
    \rrhp_{\bowtie}:\quad
    \forall \src{\OB{I}}.~\forall R \in 2^{(\behavAll^\omega)}.~\forall\src{P_1},..,\src{P_K : \OB{I}},...~\\
    (\forall\src{C_S} \ldotp (\forall i, \src{P_i \bowtie C_S}) \implies
    (\src{\behav{{C_S}\hole{P_1}}},..,\src{\behav{C_S\hole{P_K}}},..) \in R)
    \Rightarrow\\
    (\forall\trg{C_T} \ldotp (\forall i, \trg{\cmp{P_i} \bowtie C_T}) \implies
    (\trg{\behav{{C_T}\hole{\cmp{P_1}}}},..,\trg{\behav{C_S\hole{\cmp{P_K}}}},..) \in R)
  \end{multlined}$$
\end{definition}

The proof of the equivalence of these two criteria is similar to the proof of \autoref{thm:rrhp-eq}.




\subsubsection{\texorpdfstring{\backtrdt{\cdot}: {}}{} Backtranslation of Contexts from \texorpdfstring{\Ld}{Target} to \texorpdfstring{\Lt}{Source}}


Technically, the backtranslation needs one additional parameter to be passed around, the list of functions defined by the compiled component \src{\OB{I}}, 
we elide it for simplicity when it is not necessary.
\begin{align*}
	\tag{\backtrdt{\cdot}-Nat}\label{tr:backtrdt-nat}
	\backtrdt{\trg{n}} 
		=&\ \src{n+2}
	\\
	\tag{\backtrdt{\cdot}-True}\label{tr:backtrdt-true}
	\backtrdt{\truet} 
		=&\ \src{1}
	\\
	\tag{\backtrdt{\cdot}-False}\label{tr:backtrdt-false}
	\backtrdt{\falset} 
		=&\ \src{0}
	\\
	\tag{\backtrdt{\cdot}-Var}\label{tr:backtrdt-var}
	\backtrdt{\trg{x}} 
		=&\ \src{x}
	\\
	\tag{\backtrdt{\cdot}-Op}\label{tr:backtrdt-op}
	\backtrdt{\trg{e \op e'}} 
		=&\ \src{
			\begin{aligned}[t]
				&
				\letins{\src{x1}:\Nats}{\extract{\Nats}(\backtrdt{\trg{e}})
				\\
				&\
				}{\letins{\src{x2}:\Nats}{\extract{\Nats}(\backtrdt{\trg{e'}})
				\\
				&\ \ 
				}{\inject{\Nats}(\src{x1\op x2})}}
			\end{aligned}
		}
	\\
	\tag{\backtrdt{\cdot}-Geq}\label{tr:backtrdt-geq}
	\backtrdt{\trg{e \geq e'}} 
		=&\ \src{
			\begin{aligned}[t]
				&
				\letins{\src{x1}:\Nats}{\extract{\Nats}(\backtrdt{\trg{e}})
				\\
				&\
				}{\letins{\src{x2}:\Nats}{\extract{\Nats}(\backtrdt{\trg{e'}})
				\\
				&\ \ 
				}{\inject{\Bools}(\src{x1\geq x2})}}
			\end{aligned}
		}
	\\
	\tag{\backtrdt{\cdot}-Let}\label{tr:backtrdt-let}
	\backtrdt{\trg{\letin{x}{e}{e'}}} 
		=&\ \src{\letin{x:\Nats}{\backtrdt{\trg{e}}}{\backtrdt{\trg{e'}}}}
	\\
	\tag{\backtrdt{\cdot}-If}\label{tr:backtrdt-if}
	\backtrdt{\trg{\ifte{e}{e'}{e''}}} 
		=&\ \src{\ifte{\extract{\Bools}(\backtrdt{\trg{e}})}{\backtrdt{\trg{e'}}}{\backtrdt{\trg{e''}}}
		}
	\\
	\tag{\backtrdt{\cdot}-Call}\label{tr:backtrdt-call}
	\backtrdt{\trg{\call{f}~e}} 
		=&\ 
		\src{\inject{\tau'}(\call{f}~\extract{\tau}(\backtrdt{\trg{e}}))} 
	\\
		&\
		\text{if }\src{f:\tau\to\tau'}\in\src{\OB{I}}
	\\
	\tag{\backtrdt{\cdot}-Check}\label{tr:backtrdt-check}
	\backtrdt{\trg{e \checkty{\tau}}} 
		=&\ 
			\begin{cases}
				\src{\letin{x:\Nats}{\backtrdt{\trg{e}}}{\ifte{x \geq 2}{0}{1}}}
				&
				\text{if } \trg{\tau}\equiv\trg{\Bool}
				\\
				\src{\letin{x:\Nats}{\backtrdt{\trg{e}}}{\ifte{x \geq 2}{1}{0}}}
				&
				\text{if } \trg{\tau}\equiv\trg{\Nat}
			\end{cases}
\end{align*}

\paragraph{Helper functions}
The back-translation type is \Nats but the encoding is not straight from \Nats but it is \Nats shifted by 2.
\src{inject_{\tau} (e)} takes an expression \src{e} of type \src{\tau} and returns an expression whose type is the back-translation type.
\src{extract_{\tau} (e)} takes an expression \src{e} of back-translation type and returns an expression whose type is \src{\tau}.
\begin{align*}
	\src{\inject{\Nats} (e)}
		=&\
		\src{e+2}
	\\
	\src{\inject{\Bools} (e)}
		=&\
		\src{\ifte{e}{1}{0}}
	\\
	\src{\extract{\Nats} (e)}
		=&\
		\src{\letin{x}{e}{\ifte{x\geq 2}{x-2}{\fails}}}
	\\
	\src{\extract{\Bools} (e)}
		=&\
		\src{\letin{x}{e}{\ifte{x \geq 2}{\fails}{\ifte{x+1\geq 2}{\trues}{\falses}}}}
\end{align*}


\subsubsection{Cross-Language Logical Relation}

%
\paragraph{Language De-sugaring}
\begin{align*}
	\src{v} \bnfdef
		&\
		\ldots \mid \src{\call{f}}
	\\
	\src{e} \bnfdef
		&\
		\ldots \mid \src{\call{f}~ e} 
	\\
	\mi{Types}~ \src{\tau} \bnfdef
		&\
		\src{\sigma} \mid \src{\sigma\to\sigma}
	\\
	\mi{Base~ Types}~ \src{\sigma} \bnfdef
		&\
		\Nats \mid \Bools
\end{align*}

Replace \Cref{tr:ts-fun} with these below.
\begin{center}
	\typerule{T\Lt-call}{
		\src{f(x:\sigma):\sigma' \mapsto \ret e}\in\dom{\src{\OB{F}}}
	}{
		\src{P;\Gamma}\vdash \src{\call{f}} : \src{\sigma\to\sigma'} 
	}{ts-call-v}
	\typerule{T\Lt-app}{
		\src{P;\Gamma}\vdash \src{\call{f}} : \src{\sigma'\to\sigma} 
		\\
		\src{P;\Gamma}\vdash \src{e'} : \src{\sigma'} 
	}{
		\src{P;\Gamma}\vdash \src{\call{f}~e'} : \src{\sigma} 
	}{ts-app}
\end{center}

%
Apply the same changes above to \Ld too.

Context well-formedness ensures that expressions are never turned into \trg{\call{f}} values.
\begin{align*}
	\trg{\Gamma} \bnfdef &\ \trge \mid \trg{\Gamma,x}
\end{align*}
\begin{center}
	\typerule{Ctx-\Ld-true}{}{
		\trg{P;\Gamma}\vdash \truet
	}{ctx-ld-true}
	\typerule{Ctx-\Ld-false}{}{
		\trg{P;\Gamma}\vdash \falset
	}{ctx-ld-false}
	\typerule{Ctx-\Ld-nat}{}{
		\trg{P;\Gamma}\vdash \trg{n}
	}{ctx-ld-nat}
	\typerule{Ctx-\Ld-var}{
		\trg{x}\in\dom{\trg{\Gamma}}
	}{
		\trg{P;\Gamma}\vdash \trg{x}
	}{ctx-ld-var}
	\typerule{Ctx-\Ld-app}{
		\trg{P;\Gamma}\vdash \trg{e'}
		&
		\trg{e'}\nequiv\trg{\call{f}}
		\\
		\trg{f(x)\mapsto\ret e}\in\trg{P}
	}{
		\trg{P;\Gamma}\vdash \trg{\call{f}~e'}
	}{ctx-ld-app}
	\typerule{Ctx-\Ld-op}{
		\trg{P;\Gamma}\vdash \trg{e}
		&
		\trg{P;\Gamma}\vdash \trg{e'}
		\\
		\trg{e},\trg{e'}\nequiv\trg{\call{f}}
	}{
		\trg{P;\Gamma}\vdash \trg{e\op~e'}
	}{ctx-ld-op}
	\typerule{Ctx-\Ld-geq}{
		\trg{P;\Gamma}\vdash \trg{e}
		&
		\trg{P;\Gamma}\vdash \trg{e'}
		\\
		\trg{e},\trg{e'}\nequiv\trg{\call{f}}
	}{
		\trg{P;\Gamma}\vdash \trg{e\geq~e'}
	}{ctx-ld-geq}
	\typerule{Ctx-\Ld-letin}{
		\trg{P;\Gamma}\vdash \trg{e}
		&
		\trg{P;\Gamma,x}\vdash \trg{e'}
		\\
		\trg{e},\trg{e'}\nequiv\trg{\call{f}}
	}{
		\trg{P;\Gamma}\vdash \trg{\letin{x}{e}{e'}}}{ctx-ld-letin}
	\typerule{Ctx-\Ld-if}{
		\trg{P;\Gamma}\vdash \trg{e}
		&
		\trg{P;\Gamma}\vdash \trg{e'}
		&
		\trg{P;\Gamma}\vdash \trg{e''}
		\\
		\trg{e},\trg{e'},\trg{e''}\nequiv\trg{\call{f}}
	}{
		\trg{P;\Gamma}\vdash \trg{\ifte{e}{e'}{e''}}
	}{ctx-ld-if}
	\typerule{Ctx-\Ld-check}{
		\trg{P;\Gamma}\vdash \trg{e}
		\\
		\trg{e}\nequiv\trg{\call{f}}
	}{
		\trg{P;\Gamma}\vdash \trg{e \checkty{\tau}}
	}{ctx-ld-check}
\end{center}

%
Replace \Cref{tr:comptd-call} with these below.
\begin{align*}
	\tag{\comptd{\cdot}-Call-v}\label{tr:comptd-call-v}
	\comptd{\src{\call{f}}} 
		=&\ 
		\trg{\call{f}}
	\\
	\tag{\comptd{\cdot}-App}\label{tr:comptd-app}
	\comptd{\src{e~e'}} 
		=&\ 
		\trg{\comptd{e}~\comptd{e'}}
\end{align*}

\paragraph{Worlds}
\begin{align*}
	\mi{World}~W \bnfdef
		&\ 
		(n,(\src{P},\trg{P}))
	\\
	\stepsfun{(n,\_)} =
		&\
		n
	\\
	\progsfun{(\_,(\src{P},\trg{P}))} =
		&\
		(\src{P},\trg{P})
	\\
	\srcprogfun{(\_,(\src{P},\trg{P}))} =
		&\
		\src{P}
	\\
	\trgprogfun{(\_,(\src{P},\trg{P}))} =
		&\
		\trg{P}
	\\
	\laterfun{(0,\_)}=
		&\ 
		(0,\_)
	\\
	\laterfun{(n+1,\_)} =
		&\
		(n,\_)
	\\
	W\futw W' =
		&\
		\stepsfun{W'}\leq\stepsfun{W}
	\\
	W\strfutw W' =
		&\
		\stepsfun{W'}<\stepsfun{W}
	\\
	\obsfun{W}{\underlogrel}\isdef
		&\
		\myset{(\src{e},\trg{e})}{ 
			\begin{aligned}
				&
				\text{if } \stepsfun{W}=n \text{ and } \progsfun{W}=(\src{P},\trg{P}) 
				\\
				&
				\text{ and } \src{P\triangleright e\Xtosb{\beta}\redapp{n} P\triangleright e'} 
				\\
				&
				\text{ then } \exists \trg{k}.~ \trg{P\triangleright e \Xtotb{\beta}\redapp{k} P\triangleright e'}
			\end{aligned}
		}
	\\
	\obsfun{W}{\overlogrel}\isdef
		&\
		\myset{(\src{e},\trg{e})}{ 
			\begin{aligned}
				&
				\text{if } \stepsfun{W}=n \text{ and } \progsfun{W}=(\src{P},\trg{P}) 
				\\
				&
				\text{ and } \trg{P\triangleright e \Xtotb{\beta}\redapp{n} P\triangleright e'}
				\\
				&
				\text{ then } \exists \src{k}.~ \src{P\triangleright e\Xtosb{\beta}\redapp{k} P\triangleright e'} 
			\end{aligned}
		}
	\\
	\obsfun{W}{\bothlogrel}\isdef
		&\
		\obsfun{W}{\underlogrel}\cap\obsfun{W}{\overlogrel}
	\\
	\later R \isdef
		&\
		\myset{ (W,\src{v},\trg{v}) }{ \text{if } \stepsfun{W}>0 \text{ then } (\laterfun{W},\src{v},\trg{v}) \in R }
	\\
	\monotfun{R} \isdef
		&\ 
		\myset{ (W, \src{v_1}, \trg{v_2}) }{ \forall W'\futw W. (W', \src{v_1}, \trg{v_2}) \in R} 
	\\
	\text{ for } R \text{ a world-values relation}
\end{align*}

%
\paragraph{The Back-translation Type and Pseudo Types}
We index the logical relation by a pseudo type, which captures all the standard types as well as the type of backtranslated stuff.
\begin{align*}
	\psd{\tau} \bnfdef
		&\
		\src{\tau} \mid \emuldv
\end{align*}
Function \toemul{\cdot} takes a \trg{\Gamma} and returns a \src{\Gamma} that has the same domain but where variables all have type \Nats.

%
\paragraph{Value, Context, Expression and Environment relation}
\begin{align*}
	\valrel{\Bools} \isdef
		&\
		\{ (W,\trues,\truet),(W,\falses,\falset) \}
	\\
	\valrel{\Nats} \isdef
		&\
		\{ (W,\src{n},\trg{n}) \}
	\\
	\valrel{\psd{\tau}\to\psd{\tau'}} \isdef
		&\
		\myset{ (W,\src{\call{f}},\trg{\call{f}}) }{
			\begin{aligned}
				&
				\src{f(x:\tau):\tau'\mapsto \ret e} \in \srcprogfun{W} \text{ and } 
				\\
				&
				\trg{f(x)\mapsto \ret e} \in \trgprogfun{W}
				\\
				&
				\forall W', \src{v'}, \trg{v'}. \text{ if } W'\strfutw W \text{ and } (W',\src{v'},\trg{v'})\in\valrel{\psd{\tau}} \text{ then } 
				\\
				&\ 
				(W', \src{\ret e}\subs{v}{x}, \trg{\ret e}\subt{v}{x})\in\termrel{\psd{\tau'}}
			\end{aligned}
		}
	\\
	\valrel{\emuldv} \isdef
		&\
		\{ (W,\src{n+2},\trg{n}),(W,\src{1},\truet),(W,\src{0},\falset) \}
	\\
	\contrel{\psd{\tau}} \isdef
		&\
		\myset{ (W,\src{\evalctx},\trg{\evalctx}) }{
			\begin{aligned}
				&
				\forall W',\src{v},\trg{v}.~ \text{if } W'\futw W \text{ and } (W',\src{v},\trg{v})\in\valrel{\psd{\tau}} \text{ then }
				\\
				&
				(\evalctxs{v},\evalctxt{v})\in\obsfun{W'}{\anylogrel}
			\end{aligned}
		 }
	\\
	\termrel{\psd{\tau}} \isdef
		&\
		\myset{ (W,\src{t},\trg{t}) }{ \forall\src{\evalctx},\trg{\evalctx}.~ \text{if } (W,\src{\evalctx},\trg{\evalctx})\in\contrel{\psd{\tau}} \text{ then } (\evalctxs{t},\evalctxt{t})\in\obsfun{W}{\anylogrel} }
	\\
	\envrel{\srce} \isdef
		&\
		\{ (W,\srce,\trge) \}
	\\
	\envrel{\src{\psd{\Gamma},x:\psd{\tau}}} \isdef
		&\
		\myset{ (W,\src{\gamma\subs{v}{x}},\trg{\gamma\subt{v}{x}}) }{ (W,\src{\gamma},\trg{\gamma})\in\envrel{\psd{\Gamma}} \text{ and } (W,\src{v},\trg{v})\in\valrel{\psd{\tau}} }
\end{align*}

%
\paragraph{Relation for Open and Closed Terms and Programs}
\begin{definition}[Logical relation up to n steps]\label{def:logrel-n-steps}
	\begin{align*}
		\psd{\Gamma};\src{P};\trg{P}\vdash\src{e}\anylogreln{n}\trg{e}:\psd{\tau} \isdef
			&\
			{\psd{\Gamma}};\src{P}\vdash\src{e}:{\psd{\tau}}
		\\
		\text{and }
			&\
			\forall W.~
		\\
		\text{if }
			&\
			\stepsfun{W}\geq n 
			\text{ and }
			\progsfun{W} = (\src{P},\trg{P})
		\\
		\text{then }
			&\
			\forall \src{\gamma},\trg{\gamma}.~
			(W,\src{\gamma},\trg{\gamma})\in\envrel{\psd{\Gamma}}, 
		\\
			&\
			(W,\src{e\gamma},\trg{e\gamma})\in\termrel{\psd{\tau}}
	\end{align*}
\end{definition}

\begin{definition}[Logical relation for expressions]\label{def:logrel-expr}
	\begin{align*}
		\psd{\Gamma};\src{P};\trg{P}\vdash\src{e}\anylogrel\trg{e}:\psd{\tau} \isdef
			&\
			\forall n\in\mb{N}.~ \psd{\Gamma};\src{P};\trg{P}\vdash\src{e}\anylogreln{n}\trg{e}:\psd{\tau} 
	\end{align*}
\end{definition}

\begin{definition}[Logical relation for programs]\label{def:logrel-prog}
	\begin{align*}
		\vdash\src{P}\anylogrel\trg{P}\isdef
			&\
			\src{f(x:\sigma'):\sigma\mapsto \ret e}\in\src{P} \text{ iff } \trg{f(x)\mapsto\ret e}\in\trg{P}
		\\
			&\
			\src{x:\sigma'};\src{P};\trg{P}\vdash\src{e}\anylogrel\trg{e}:\src{\sigma} 
	\end{align*}
\end{definition}

%
\paragraph{Auxiliary Lemmas from Existing Work}
\begin{lemma}[No observation with 0 steps]\label{thm:no-obs-if-no-steps}
\begin{align*}
	\text{if }
		&\
		\stepsfun{W}=0
	\\
	\text{then }
		&\
		\forall \src{e},\trg{e}. (\src{e},\trg{e})\in\obsfun{W}{\anylogrel}
\end{align*}
\end{lemma}
\begin{proof}
	Trivial adaptation of the same proof in~\cite{DevriesePP16,domipoplta}.
\end{proof}

\begin{lemma}[No steps means relation]\label{thm:no-steps-impl-rel}
\begin{align*}
	\text{if } 
		&\
		\stepsfun{W} = n
	\\
		&\
		\src{P\triangleright e \Xtosb{\beta}\redapp{n} \_}
	\\
		&\
		\trg{P\triangleright e \Xtotb{\beta}\redapp{n} \_}
	\\
	\text{then }
		&\
		(\src{e},\trg{e}) \in \obsfun{W}{\anylogrel}
\end{align*}
\end{lemma}
\begin{proof}
	Trivial adaptation of the same proof in~\cite{DevriesePP16,domipoplta}.
\end{proof}

\begin{lemma}[Later preserves monotonicity]\label{thm:lat-pres-mono}
\begin{align*}
	\text{if }
		&\
		\forall R, R\subseteq\monotfun{R}
	\\
	\text{then } 
		&\
		\later R \subseteq \monotfun{\later R}
\end{align*}
\end{lemma}
\begin{proof}
	Trivial adaptation of the same proof in~\cite{DevriesePP16,domipoplta}.
\end{proof}

\begin{lemma}[Monotonicity for environment relation]\label{thm:env-rel-mono}
\begin{align*}
	\text{if }
		&\
		W' \futw W
	\\
		&\
		(W,\src{\gamma},\trg{\gamma})\in\envrel{\Gamma}
	\\
	\text{then }
		&\
		(W',\src{\gamma},\trg{\gamma})\in\envrel{\Gamma}
\end{align*}
\end{lemma}
\begin{proof}
	Trivial adaptation of the same proof in~\cite{DevriesePP16,domipoplta}.
\end{proof}

\begin{lemma}[Monotonicity for continuation relation]\label{thm:cont-rel-mono}
\begin{align*}
	\text{if } 
		&
		W'\futw W
	\\
		&\
		(W,\src{\ctx},\trg{\ctx})\in\contrel{\psd{\tau}}
	\\
	\text{then }
		&\
  		(W',\src{\ctx},\trg{\ctx})\in\contrel{\psd{\tau}}
\end{align*}
\end{lemma}
\begin{proof}
	Trivial adaptation of the same proof in~\cite{DevriesePP16,domipoplta}.
\end{proof}

\begin{lemma}[Monotonicity for value relation]\label{thm:val-rel-mono}
\begin{align*}
	\valrel{\psd{\tau}}\subseteq\monotfun{\valrel{\psd{\tau}}}
\end{align*}
\end{lemma}
\begin{proof}
	Trivial adaptation of the same proof in~\cite{DevriesePP16,domipoplta}.
\end{proof}

\begin{lemma}[Value relation implies term relation]\label{thm:val-rel-impl-term-rel}
\begin{align*}
	\forall\psd{\tau},\valrel{\psd{\tau}}\subseteq\termrel{\psd{\tau}}
\end{align*}
\end{lemma}
\begin{proof}
	Trivial adaptation of the same proof in~\cite{DevriesePP16,domipoplta}.
\end{proof}

\begin{lemma}[Adequacy for $\underlogrel$]\label{thm:log-rel-adeq-less}
\begin{align*}
	\text{if }
		&
		\srce;\src{P};\trg{P} \vdash \src{e} \underlogrel_n \trg{e} : \src{\tau}
	\\
		&
		\src{P\triangleright e \Xtosb{\beta}\redapp{m} P\triangleright e'} \text{ with } n \geq m
	\\
	\text{then }
		&\
		\trg{P\triangleright e \Xtotb{\beta} P\triangleright \_}.
\end{align*}
\end{lemma}
\begin{proof}
	By \Thmref{def:logrel-expr} we have that $(W,\src{e},\trg{e})\in\genrel{E}{\tau}{}{\underlogrel}$  for a $W$ such that $\stepsfun{W}=n$.

	By taking $(W,\src{\hole{\cdot}},\trg{\hole{\cdot}})\in\genrel{K}{\tau}{}{\underlogrel}$ we know that $(\src{e},\trg{e})\in\obsfun{W}{\underlogrel}$.

	By definition of $\obsfun{\cdot}{\underlogrel}$, with the HP of the source reduction, we conclude the thesis.
\end{proof}

\begin{lemma}[Adequacy for $\overlogrel$]\label{thm:log-rel-adeq-gt}
\begin{align*}
	\text{if }
		&
		\srce;\src{P};\trg{P} \vdash \src{e} \overlogrel_n \trg{e} : \src{\tau}
	\\
		&
		\trg{P\triangleright e \Xtotb{\beta}\redapp{m} P\triangleright e'}. \text{ with } n \geq m
	\\
	\text{then }
		&\
		\src{P\triangleright e \Xtosb{\beta} P\triangleright \_}
\end{align*}
\end{lemma}
\begin{proof}
	By \Thmref{def:logrel-expr} we have that $(W,\src{e},\trg{e})\in\genrel{E}{\tau}{}{\overlogrel}$  for a $W$ such that $\stepsfun{W}=n$.

	By taking $(W,\src{\hole{\cdot}},\trg{\hole{\cdot}})\in\genrel{K}{\tau}{}{\overlogrel}$ we know that $(\src{e},\trg{e})\in\obsfun{W}{\overlogrel}$.

	By definition of $\obsfun{\cdot}{\overlogrel}$, with the HP of the target reduction, we conclude the thesis.
\end{proof}

\begin{lemma}[Observation relation is closed under antireduction]\label{thm:obs-rel-clos-antireds}
\begin{align*}
	\text{if } 
		&\
		\src{P\triangleright e \Xtosb{\beta}\redapp{i} P\triangleright e'}
	\\
		&\
		\trg{P\triangleright e \Xtotb{\beta}\redapp{j} P\triangleright e'}
	\\
		&\
		(\src{e'},\trg{e'}) \in \obsfun{W'}{\anylogrel} \text{ for } W' \futw W
	\\
		&\
		\progsfun{W}=\progsfun{W'}=(\src{P},\trg{P})
	\\
		&\
		\stepsfun{W'} \geq \stepsfun{W} - \fun{min}{i,j} 
	\\
		&\
		(\text{ that is: } \stepsfun{W} \leq \stepsfun{W'}+\fun{min}{i,j})
	\\
	\text{then }
		&\
	(\src{e},\trg{e}) \in \obsfun{W}{\anylogrel}
\end{align*}
\end{lemma}
\begin{proof}
	Trivial adaptation of the same proof in~\cite{DevriesePP16,domipoplta}.
\end{proof}

\begin{lemma}[Closedness under antireduction]\label{thm:log-rel-clos-antired}
\begin{align*}
	\text{if }
		&\
		\src{P\triangleright \ctxhs{e} \Xtosb{\beta}\redapp{i} P\triangleright \ctxhs{e'}}
	\\
		&\
		\trg{P\triangleright \ctxht{e} \Xtotb{\beta}\redapp{i} P\triangleright \ctxht{e'}}
 	\\
  		&\
  		(W',\src{e'},\trg{e'}) \in \termrel{\psd{\tau}}
  	\\
  		&\
  		W' \futw W
  	\\
  		&\
  		\stepsfun{W'} \geq \stepsfun{W} - \fun{min}{i,j}
  	\\
  		&\
  		(\text{ that is } \stepsfun{W} \leq \stepsfun{W'}+\fun{min}{i,j})
  	\\
  	\text{then }
  		&\
  		(W,\src{e},\trg{e}) \in \termrel{\psd{\tau}}
\end{align*}
\end{lemma}
\begin{proof}
	Trivial adaptation of the same proof in~\cite{DevriesePP16,domipoplta}.
\end{proof}

\begin{lemma}[Related terms plugged in related contexts are still related]\label{thm:rel-ter-plug-rel-ctx-still-rel}	
\begin{align*}
	\text{if }
		&\
		(W,\src{e},\trg{e}) \in \termrel{\psd{\tau'}}
	\\
		&\
		\text{and if }
		W' \futw W
	\\
		&\
		(W',\src{v},\trg{v}) \in \valrel{\psd{\tau'}}
	\\
		&\ 
		\text{then }
		(W',\ctxhs{v},\ctxht{v})\in\termrel{\psd{\tau}}
	\\
	\text{then } 
		&\
		(W,\ctxhs{e},\ctxht{e}) \in \termrel{\psd{\tau}}
\end{align*}
\end{lemma}
\begin{proof}
	Trivial adaptation of the same proof in~\cite{DevriesePP16,domipoplta}.
\end{proof}

\begin{lemma}[Related functions applied to related arguments are related terms]\label{thm:rel-fun-appl-rel-arg-still-rel}
\begin{align*}
	\text{if }
		&\ 
		(W,\src{v},\trg{v}) \in \valrel{\psd{\tau'}\to\psd{\tau}}
	\\
		&\
		(W,\src{v'},\trg{v'})\in\valrel{\psd{\tau'}}
	\\
	\text{then }
		&\
		(W,\src{v~v'},\trg{v~v'})\in\termrel{\psd{\tau}}
\end{align*}
\end{lemma}
\begin{proof}
	Trivial adaptation of the same proof in~\cite{DevriesePP16,domipoplta}.
\end{proof}

%
\paragraph{Auxiliary Results}
\begin{lemma}[If Extract reduces, it preserves relatedness]\label{thm:extract-red-pres-rel}
	\begin{align*}
		\text{if }
			&\
			(W,\src{v},\trg{v})\in\valrel{\emuldv}
		\\
			&\
			\src{P\triangleright \extract{\sigma} (v) \reds\redstars P\triangleright v'} 
		\\
		\text{then }
			&\
			(W,\src{v'},\trg{v})\in\valrel{\sigma}
	\end{align*}
\end{lemma}
\begin{proof}
	Trivial case analysis:
	\begin{description}
		\item[$\src{\sigma}=\Bools$] means that \src{v}=\src{0} or \src{1}, so by definition of \valrel{\emuldv} \trg{v}=\falset or \truet (respectively).

			Consider the \src{0} and \falset case, the other is analogous.

			By definition the reduction of extract goes as follows.
			\begin{align*}
				&
				\src{P\triangleright \extract{\Bools} 0 }
			\\
			\equiv
				&
				\src{
					P\triangleright
					\letin{x}{0}{\ifte{x \geq 2}{\fails}{\ifte{x+1\geq 2}{\trues}{\falses}}}
				}
			\\
			\reds\reds
				&
				\src{
					P\triangleright
					\ifte{1\geq 2}{\trues}{\falses}
				}
			\\
			\reds
				&
				\src{
					P\triangleright
					\falses
				}
			\end{align*}
			We need to show that $(W,\falses,\falset)\in\valrel{\Bools}$, which follows from its definition.

		\item[$\src{\sigma}=\Nats$] means that \src{v}=\src{n+2} and \trg{v}=\trg{n}

			By definition the reduction of extract goes as follows. (we write n+2 as a value, not as an expression to simplify this)
			\begin{align*}
					&
					\src{P\triangleright \extract{\Nats} n+2 }				
				\\
				\equiv
					&
					\src{P\triangleright \letin{x}{n+2}{\ifte{x\geq 2}{x-2}{\fails}} }
				\\
				\reds
					&
					\src{P\triangleright \ifte{n+2\geq 2}{x-2}{\fails} }
				\\
				\reds
					&
					\src{P\triangleright n }
			\end{align*}
			We need to show that $(W,\src{n},\trg{n})\in\valrel{\Nats}$, which follows from its definition.
	\end{description}
\end{proof}

\begin{lemma}[Inject reduces and preserves relatedness]\label{thm:inject-red-pres-rel}
	\begin{align*}
		\text{if }
			&\
			(W,\src{v},\trg{v})\in\valrel{\sigma}
		\\
			&\
			\src{P\triangleright \inject{\sigma} v \reds\redstars P\triangleright v'}
		\\
		\text{then }
			&\
			(W,\src{v'},\trg{v})\in\valrel{\emuldv}
	\end{align*}
\end{lemma}
\begin{proof}
	Trivial case analysis on \src{\sigma}.
	\begin{description}
		\item[$\src{\sigma}=\Bools$] 
			By definition of \valrel{\Bools} we have \src{v}=\trues and \trg{v}=\truet or \falses/\falset. 
			We consider the first case only, the second is analogous.

			By definition of inject we have:
			\begin{align*}
				&
				\src{P\triangleright \src{\ifte{\trues}{1}{0}}}
			\\
			\reds
				&
				\src{P\triangleright 1}
			\end{align*}

			So we need to prove that $(W,\src{1},\truet)\in\valrel{\emuldv}$ which follows from its definition.

		\item[$\src{\sigma}=\Nats$] 
			By definition of \valrel{\Nats} we have \src{v}=\src{n} and \trg{v}=\trg{n}.

			By definition of inject, we have:
			\begin{align*}
				&
				\src{P\triangleright \src{n+2} }
				\\
			\reds
				&
				\src{P\triangleright \src{n+2} }
			\end{align*}
			(we keep the value as a sum for simplicity)

			So we need to prove that $(W,\src{n+2},\trg{n})\in\valrel{\emuldv}$ which follows from its definition.
	\end{description}
\end{proof}

%
\paragraph{Compatibility Lemmas for \src{\tau} Types}
\begin{lemma}[Compatibility lemma for calls]\label{thm:compat-call-v}
	\begin{align*}
		\text{if }
			&\
			\src{\Gamma,x:\sigma'};\src{P};\trg{P}\vdash\src{e}\anylogreln{n}\trg{e}:\src{\sigma}
		\\
			&
			\src{f(x:\sigma'):\sigma\mapsto\ret e}\in\src{P}
		\\
			&
			\trg{f(x)\mapsto\ret \ifte{x\checkty{\sigma'}}{e}{\failt}}\in\trg{P}
		\\
		\text{then }
			&\
			\src{\Gamma};\src{P};\trg{P}\vdash\src{\call{f}}\anylogreln{n}\trg{\call{f}}:\src{\sigma'\to\sigma}
	\end{align*}
\end{lemma}
\begin{proof}
	We need to prove that 
	\begin{align*}
		\src{\Gamma};\src{P};\trg{P}\vdash\src{\call{f}}\anylogreln{n}\trg{\call{f}}:\src{\sigma'\to\sigma}
	\end{align*}

	Take $W$ such that $\stepsfun{W}\leq n$ and HG $(W,\src{\gamma},\trg{\gamma})\in\envrel{\toemul{\trg{\Gamma}}}$, the thesis is: 
	\begin{itemize}
		\item $(W, \src{\call{f}} , \trg{\call{f}})\in\termrel{\sigma'\to\sigma}$
	\end{itemize}

	By \Thmref{thm:val-rel-impl-term-rel} the thesis is:
	\begin{itemize}
		\item $(W, \src{\call{f}} , \trg{\call{f}})\in\valrel{\sigma'\to\sigma}$
	\end{itemize}

	By definition of the \valrel{\cdot} we take HV $(W',\src{v},\trg{v})\in\valrel{\sigma'}$ such that $W'\strfutw W$ and the thesis is:
	\begin{itemize}
		\item $(W', \src{\ret e\subst{v}{x}\gamma} , \trg{\ret \ifte{x\checkty{\sigma'}}{e}{\failt}\subst{v}{x}\gamma })\in\termrel{\sigma}$
	\end{itemize}

	The reductions proceed as:
	\begin{align*}
		&
		\trg{P\triangleright \trg{\ret \ifte{x\checkty{\sigma'}}{e}{\failt}\subst{v}{x}\gamma }}
		\\
		\equiv
		&
		\trg{P\triangleright \trg{\ret \ifte{v\checkty{\sigma'}}{(e\subst{v}{x}\gamma)}{\failt} }}
		\\
		\redt
		&
		\trg{P\triangleright \trg{\ret \ifte{\truet}{(e\subst{v}{x}\gamma)}{\failt} }}
		\\
		\redt
		&
		\trg{P\triangleright \trg{\ret (e\subst{v}{x}\gamma)}}
	\end{align*}

	By \Cref{thm:log-rel-clos-antired} the thesis becomes:
	\begin{itemize}
		\item $(W', \src{\ret e\subst{v}{x}\gamma} , \trg{\ret e\subst{v}{x}\gamma })\in\termrel{\sigma}$
	\end{itemize}

	This follows from the definition of logical relation if 
	\begin{itemize}
		\item $(W', \src{\subst{v}{x}\gamma} , \trg{\subst{v}{x}\gamma })\in\envrel{\Gamma,x:\sigma'}$
	\end{itemize}
	This follows from HG with \Cref{thm:env-rel-mono} and by HV and \Cref{thm:val-rel-mono} and by the definition of \envrel{\cdot}.
\end{proof}
\begin{lemma}[Compatibility lemma for application]\label{thm:compat-app}
	\begin{align*}
		\text{if }
			&\
			\src{\Gamma};\src{P};\trg{P}\vdash\src{e}\anylogreln{n}\trg{e}:\src{\sigma'\to\sigma}
		\\
			&
			\src{\Gamma};\src{P};\trg{P}\vdash\src{e'}\anylogreln{n}\trg{e'}:\src{\sigma'}
		\\
		\text{then }
			&\
			\src{\Gamma};\src{P};\trg{P}\vdash\src{e~e'}\anylogreln{n}\trg{e~e'}:\src{\sigma}
	\end{align*}
\end{lemma}
\begin{proof}
	This is standard using \Cref{thm:val-rel-impl-term-rel}, \Cref{thm:val-rel-mono}, \Cref{thm:rel-ter-plug-rel-ctx-still-rel} and \Cref{thm:log-rel-clos-antired}.
\end{proof}
\begin{lemma}[Compatibility lemma for op]\label{thm:compat-op}
	\begin{align*}
		\text{if }
			&\
			\src{\Gamma};\src{P};\trg{P}\vdash\src{e}\anylogreln{n}\trg{e}:\src{\Nats}
		\\
			&
			\src{\Gamma};\src{P};\trg{P}\vdash\src{e'}\anylogreln{n}\trg{e'}:\src{\Nats}
		\\
		\text{then }
			&\
			\src{\Gamma};\src{P};\trg{P}\vdash\src{e\op e'}\anylogreln{n}\trg{e\op e'}:\src{\Nats}
	\end{align*}
\end{lemma}
\begin{proof}
	This is standard and analogous to the proof of \Cref{thm:compat-app}.
\end{proof}
\begin{lemma}[Compatibility lemma for geq]\label{thm:compat-geq}
	\begin{align*}
		\text{if }
			&\
			\src{\Gamma};\src{P};\trg{P}\vdash\src{e}\anylogreln{n}\trg{e}:\src{\Nats}
		\\
			&
			\src{\Gamma};\src{P};\trg{P}\vdash\src{e'}\anylogreln{n}\trg{e'}:\src{\Nats}
		\\
		\text{then }
			&\
			\src{\Gamma};\src{P};\trg{P}\vdash\src{e\geq e'}\anylogreln{n}\trg{e\geq e'}:\src{\Bools}
	\end{align*}
\end{lemma}
\begin{proof}
	This is standard and analogous to the proof of \Cref{thm:compat-app}.
\end{proof}
\begin{lemma}[Compatibility lemma for letin]\label{thm:compat-letin}
	\begin{align*}
		\text{if }
			&\
			\src{\Gamma};\src{P};\trg{P}\vdash\src{e}\anylogreln{n}\trg{e}:\src{\sigma}
		\\
			&
			\src{\Gamma,x:\sigma};\src{P};\trg{P}\vdash\src{e'}\anylogreln{n}\trg{e'}:\src{\sigma'}
		\\
		\text{then }
			&\
			\src{\Gamma};\src{P};\trg{P}\vdash\src{\letin{x}{e}{e'}}\anylogreln{n}\trg{\letin{x}{e}{e'}}:\src{\sigma'}
	\end{align*}
\end{lemma}
\begin{proof}
	This is standard and analogous to the proof of \Cref{thm:compat-app}.
\end{proof}
\begin{lemma}[Compatibility lemma for if]\label{thm:compat-if}
	\begin{align*}
		\text{if }
			&\
			\src{\Gamma};\src{P};\trg{P}\vdash\src{e}\anylogreln{n}\trg{e}:\src{\Bools}
		\\
			&
			\src{\Gamma};\src{P};\trg{P}\vdash\src{e'}\anylogreln{n}\trg{e'}:\src{\sigma}
		\\
			&
			\src{\Gamma};\src{P};\trg{P}\vdash\src{e''}\anylogreln{n}\trg{e''}:\src{\sigma}
		\\
		\text{then }
			&\
			\src{\Gamma};\src{P};\trg{P}\vdash\src{\ifte{e}{e'}{e''}}\anylogreln{n}\trg{\ifte{e}{e'}{e''}}:\src{\sigma}
	\end{align*}
\end{lemma}
\begin{proof}
	This is standard and analogous to the proof of \Cref{thm:compat-app}.
\end{proof}
\begin{lemma}[Compatibility lemma for read]\label{thm:compat-read}
	\begin{align*}
		\text{if }
		\\
		\text{then }
			&\
			\src{\Gamma};\src{P};\trg{P}\vdash\src{\readexp}\anylogreln{n}\trg{\readexp}:\src{\Nats}
	\end{align*}
\end{lemma}
\begin{proof}
	By definition of the $\obsfun{W}{\anylogrel}$.
\end{proof}
\begin{lemma}[Compatibility lemma for write]\label{thm:compat-write}
	\begin{align*}
		\text{if }
			&\
			\src{\Gamma};\src{P};\trg{P}\vdash\src{e}\anylogreln{n}\trg{e}:\src{\Nats}
		\\
		\text{then }
			&\
			\src{\Gamma};\src{P};\trg{P}\vdash\src{\writeexp{e}}\anylogreln{n}\trg{\writeexp{e}}:\src{\Nats}
	\end{align*}
\end{lemma}
\begin{proof}
	We need to prove that 
	\begin{align*}
		\src{\Gamma};\src{P};\trg{P}\vdash\src{\writeexp{e}}\anylogreln{n}\trg{\writeexp{e}}:\src{\Nats}
	\end{align*}

	Take $W$ such that $\stepsfun{W}\leq n$ and $(W,\src{\gamma},\trg{\gamma})\in\envrel{\toemul{\trg{\Gamma}}}$, the thesis is: (we omit substitutions as they don't play an active role)
	\begin{itemize}
		\item $(W, \src{\writeexp{e}} , \trg{\writeexp{e}})\in\termrel{\Nats}$
	\end{itemize}

	By \Thmref{thm:rel-ter-plug-rel-ctx-still-rel} with HE, we have that for HW $W'\futw W$, and HV $(W',\src{n},\trg{n})\in\valrel{\Nats}$, the thesis becomes:
	\begin{itemize}
		\item $(W', \src{\writeexp{n}} , \trg{\writeexp{n}})\in\termrel{\Nats}$
	\end{itemize}
	The reductions proceed as:
	\begin{align*}
		\src{P\triangleright \writeexp{n}}
		\Xtosb{\wrl{n}}
		\src{P\triangleright n}
	\end{align*}
	and
	\begin{align*}
		\trg{P\triangleright \writeexp{n}}
		\Xtotb{\wrl{n}}
		\trg{P\triangleright n}
	\end{align*}
	By \Thmref{thm:log-rel-clos-antired} the thesis is:
	\begin{itemize}
		\item $(W', \src{n} , \trg{n})\in\termrel{\Nats}$
	\end{itemize}
	So the theorem holds by \Thmref{thm:val-rel-impl-term-rel} with HV.
\end{proof}

%
\paragraph{Semantic Preservation Results}
\begin{theorem}[\comptd{\cdot} is semantics preserving for expressions]\label{thm:comp-sem-pres-expr}
	\begin{align*}
		\text{if }
			&\
			\src{P;\Gamma}\vdash\src{e}:\src{\tau}
		\\
			&\
			\vdash\src{P}\anylogreln{n}\trg{P}
		\\
		\text{then }
			&\
			\forall n.~
			\src{\Gamma};\src{P};\trg{P}\vdash\src{e}\anylogreln{n}\comptd{\src{e}}:\src{\tau}
	\end{align*}
\end{theorem}
\begin{proof}
	The proof proceeds by induction on the type derivation.
	\begin{description}
		\item[true, false, nat] By definition of \valrel{\cdot}.
		\item[var] By definition of \envrel{\cdot}.
		\item[call] By \Thmref{thm:compat-call-v}.
		\item[app] By IH with \Thmref{thm:compat-app}.
		\item[op] By IH with \Thmref{thm:compat-op}.
		\item[geq] By IH with \Thmref{thm:compat-geq}.
		\item[letin] By IH with \Thmref{thm:compat-letin}.
		\item[if] By IH with \Thmref{thm:compat-if}.
		\item[read] By \Thmref{thm:compat-read}.
		\item[write] By IH with \Thmref{thm:compat-write}.
	\end{description}
\end{proof}

\begin{theorem}[\comptd{\cdot} is semantics preserving for programs]\label{thm:comp-sem-pres-prog}
	\begin{align*}
		\text{if }
			&\
			\vdash\src{P}
		\\
		\text{then }
			&\
			\vdash\src{P}\anylogrel\comptd{\src{P}}
	\end{align*}
\end{theorem}
\begin{proof}
	By induction on the size of \src{P} and then \Cref{tr:comptd-prog} and with \Thmref{thm:comp-sem-pres-expr} on each function body.
\end{proof}

%
\paragraph{Compatibility Lemmas for Pseudo Types}
\begin{lemma}[Compatibility lemma for backtranslation of op]\label{thm:compat-backtr-op}
	\begin{align*}
		\text{if }
			&\ (HE)~
			\toemul{\trg{\Gamma}};\src{P};\trg{P}\vdash\src{e}\anylogreln{n}\trg{e}:\emuldv
		\\
			&\ (HEP)~
			\toemul{\trg{\Gamma}};\src{P};\trg{P}\vdash\src{e'}\anylogreln{n}\trg{e'}:\emuldv
		\\
		\text{then }
			&\
			\toemul{\trg{\Gamma}};\src{P};\trg{P}\vdash
				\begin{aligned}[t]
					&
					\letins{\src{x1}:\Nats}{\extract{\Nats}(\src{e})
					\\
					&\
					}{\letins{\src{x2}:\Nats}{\extract{\Nats}(\src{e'})
					\\
					&\ \ 
					}{\inject{\Nats}(\src{x1\op x2})}}
				\end{aligned}
			\anylogreln{n} \trg{e\op e'} : \emuldv
	\end{align*}
\end{lemma}
\begin{proof}
	We need to prove that 
	\begin{align*}
		\toemul{\trg{\Gamma}};\src{P};\trg{P}\vdash
		\src{
			\begin{aligned}[t]
				&
				\letins{\src{x1}:\Nats}{\extract{\Nats}(\src{e})
				\\
				&\
				}{\letins{\src{x2}:\Nats}{\extract{\Nats}(\src{e'})
				\\
				&\ \ 
				}{\inject{\Nats}(\src{x1\op x2})}}
			\end{aligned}
		}
		\anylogrel\trg{e\op e'}:\emuldv
	\end{align*}

	Take $W$ such that $\stepsfun{W}\leq n$ and $(W,\src{\gamma},\trg{\gamma})\in\envrel{\toemul{\trg{\Gamma}}}$, the thesis is:
	\begin{itemize}
		\item $(W, \src{
			\begin{aligned}[t]
				&
				\letins{\src{x1}:\Nats}{\extract{\Nats}(\src{e})
				\\
				&\
				}{\letins{\src{x2}:\Nats}{\extract{\Nats}(\src{e'})
				\\
				&\ \ 
				}{\inject{\Nats}(\src{x1\op x2})}}
			\end{aligned}
		}
		,
		\trg{e\op e'})\in\termrel{\emuldv}$
	\end{itemize}

	By \Thmref{thm:rel-ter-plug-rel-ctx-still-rel} with HE we need to prove that $\forall W'\futw W$, given IHV $(W',\src{v},\trg{v})\in\valrel{\emuldv}$
	\begin{itemize}
		\item $(W', \src{
			\begin{aligned}[t]
				&
				\letins{\src{x1}:\Nats}{\extract{\Nats}(\src{v})
				\\
				&\
				}{\letins{\src{x2}:\Nats}{\extract{\Nats}(\src{e'})
				\\
				&\ \ 
				}{\inject{\Nats}(\src{x1\op x2})}}
			\end{aligned}
		}
		,
		\trg{v\op e'})\in\termrel{\emuldv}$
	\end{itemize}
	By IHV we perform a case analysis on \trg{v}:
	\begin{itemize}
		\item \truet / \falset and thus \src{v} is \src{1}/\src{0} respectively.

		We show the case for \truet, \src{1} the other is analogous.

		In this case we have:
		\begin{align*}
			\trg{P\triangleright \truet\op e'}
			\Xtotb{\failact}
			\failt
		\end{align*}
		and
		\begin{align*}
				&
				\src{P\triangleright \extract{\Nats}(1)}
			\\
			\equiv
				&
				\src{\letin{x}{1}{\ifte{x\geq 2}{x-2}{\fails}}}
			\\
			\reds
				&
				\src{{\ifte{1\geq 2}{x-2}{\fails}}}
			\\
			\Xtosb{\failact}
				&
				\fails
		\end{align*}

		So this case follows from the definition of $\obsfun{W'}{\anylogrel}$ as both terms perform the same visible action ($\failact$).

		\item \trg{n} and thus \src{v} is \src{n+2}.

		In this case we have:
		\begin{align*}
				&
				\src{P\triangleright \extract{\Nats}(n+2)}
			\\
			\equiv
				&
				\src{\letin{x}{n+2}{\ifte{x\geq 2}{x-2}{\fails}}}
			\\
			\reds
				&
				\src{{\ifte{n+2\geq 2}{x-2}{\fails}}}
			\\
			\reds
				&
				\src{n}
		\end{align*}
		And by \Thmref{thm:extract-red-pres-rel} with IHV we know that IHN $(W',\src{n},\trg{n})\in\valrel{\Nats}$.

		Analogously, \src{e'} and \trg{e'} follow the same treatment.
		So we apply \Thmref{thm:rel-ter-plug-rel-ctx-still-rel} with HEP, perform a case analysis, in one case they fail and in the other they reduce to \src{n'}/\trg{n'} such that IHNP $(W',\src{n'},\trg{n'})\in\valrel{\Nats}$.

		So the reductions are :
		\begin{align*}
				&
				\src{P\triangleright}
				\letins{\src{x1}:\Nats}{\extract{\Nats}(\src{e})
				}{\letins{\src{x2}:\Nats}{\extract{\Nats}(\src{e'})
				\\
				&\ \ 
				}{\inject{\Nats}(\src{x1\op x2})}}
			\\
			\reds\redstars
				&
				\src{P\triangleright}
				\letins{\src{x1}:\Nats}{\extract{\Nats}(\src{n})
				}{\letins{\src{x2}:\Nats}{\extract{\Nats}(\src{e'})
				\\
				&\ \ 
				}{\inject{\Nats}(\src{x1\op x2})}}
			\\
			\reds
				&
				\src{P\triangleright}
				\letins{\src{x2}:\Nats}{\extract{\Nats}(\src{e'})
				\\
				&\ \ 
				}{\inject{\Nats}(\src{n\op x2})}
			\\
			\reds\redstars
				&
				\src{P\triangleright}
				\letins{\src{x2}:\Nats}{\extract{\Nats}(\src{n'})
				\\
				&\ \ 
				}{\inject{\Nats}(\src{n\op x2})}
			\\
			\reds
				&
				\src{P\triangleright
				\inject{\Nats}(\src{n\op n'})}
		\end{align*}
		and
		\begin{align*}
				\trg{P\triangleright e\op e'}
			\redt\redstart
				\trg{P\triangleright n\op e'}
			\redt\redstart
				\trg{P\triangleright n\op n'}
		\end{align*}

		By \Thmref{thm:log-rel-clos-antired} the thesis becomes:
		\begin{itemize}
			\item $(W', \src{\inject{\Nats}(\src{n\op n'}) }, \trg{n\op n'})\in\termrel{\emuldv}$
		\end{itemize}
		If the $\stepsfun{W'}=0$ the thesis follows from \Thmref{thm:no-steps-impl-rel}, otherwise:

		By \Cref{tr:ev-s-op} and \Cref{tr:ev-t-op} we can apply \Thmref{thm:log-rel-clos-antired} (with IHN and IHNP in the term relation by \Thmref{thm:val-rel-impl-term-rel}) and the thesis becomes:
		\begin{itemize}
			\item $(W', \src{\inject{\Nats}(\src{n''}) }, \trg{n''})\in\termrel{\emuldv}$
		\end{itemize}
		The reductions proceed as follows:
		\begin{align*}
			\src{P\triangleright \inject{\Nats}(\src{n''}) \reds P\triangleright n''+2}
		\end{align*}
		By \Thmref{thm:log-rel-clos-antired} and then \Thmref{thm:val-rel-impl-term-rel} the thesis becomes:
		\begin{itemize}
			\item $(W', \src{ n'''+2 } , \trg{n''})\in\valrel{\emuldv}$
		\end{itemize}

		By \Thmref{thm:inject-red-pres-rel} the thesis becomes:
		\begin{itemize}
			\item $(W', \src{ n'' } , \trg{n''})\in\valrel{\Nats}$
		\end{itemize}
		which follows from the definition of \valrel{\Nats}.
	\end{itemize}
\end{proof}

\begin{lemma}[Compatibility lemma for backtranslation of geq]\label{thm:compat-backtr-geq}
	\begin{align*}
		\text{if }
			&\
			\toemul{\trg{\Gamma}};\src{P};\trg{P}\vdash\src{e}\anylogreln{n}\trg{e}:\emuldv
		\\
			&\
			\toemul{\trg{\Gamma}};\src{P};\trg{P}\vdash\src{e'}\anylogreln{n}\trg{e'}:\emuldv
		\\
		\text{then }
			&\
			\toemul{\trg{\Gamma}};\src{P};\trg{P}\vdash
				\begin{aligned}[t]
					&
					\letins{\src{x1}:\Nats}{\extract{\Nats}(\src{e})
					\\
					&\
					}{\letins{\src{x2}:\Nats}{\extract{\Nats}(\src{e'})
					\\
					&\ \ 
					}{\inject{\Bools}(\src{x1\geq x2})}}
				\end{aligned}
			\anylogreln{n} \trg{e\geq e'} : \emuldv
	\end{align*}
\end{lemma}
\begin{proof}
	Analogous to the proof of \Cref{thm:compat-backtr-op}.
\end{proof}

\begin{lemma}[Compatibility lemma for backtranslation of letin]\label{thm:compat-backtr-letin}
	\begin{align*}
		\text{if }
			&\
			\toemul{\trg{\Gamma}};\src{P};\trg{P}\vdash\src{e}\anylogreln{n}\trg{e}:\emuldv
		\\
			&\
			\toemul{\trg{\Gamma}}\src{,x:\Nats};\src{P};\trg{P}\vdash\src{e'}\anylogreln{n}\trg{e'}:\emuldv
		\\
		\text{then }
			&\
			\toemul{\trg{\Gamma}};\src{P};\trg{P}\vdash \src{\letin{x:\Nats}{\src{e}}{\src{e'}}} \anylogreln{n} \trg{\letin{x}{e}{e'}} : \emuldv
	\end{align*}
\end{lemma}
\begin{proof}
	This is a trivial application of \Thmref{thm:rel-ter-plug-rel-ctx-still-rel} and \Thmref{thm:log-rel-clos-antired} and definitions.
\end{proof}

\begin{lemma}[Compatibility lemma for backtranslation of if]\label{thm:compat-backtr-if}
	\begin{align*}
		\text{if }
			&\ (HE)~
			\toemul{\trg{\Gamma}};\src{P};\trg{P}\vdash\src{e}\anylogreln{n}\trg{e}:\emuldv
		\\
			&\ (HEP)~
			\toemul{\trg{\Gamma}};\src{P};\trg{P}\vdash\src{e'}\anylogreln{n}\trg{e'}:\emuldv
		\\
			&\
			\toemul{\trg{\Gamma}};\src{P};\trg{P}\vdash\src{e''}\anylogreln{n}\trg{e''}:\emuldv
		\\
		\text{then }
			&\
			\toemul{\trg{\Gamma}};\src{P};\trg{P}\vdash \src{\ifte{\extract{\Bools}(e)}{e'}{e''}} \anylogreln{n} \trg{\ifte{e}{e'}{e''}} : \emuldv
	\end{align*}
\end{lemma}
\begin{proof}
	We need to prove that 
	\begin{align*}
		\toemul{\trg{\Gamma}};\src{P};\trg{P}\vdash
		\src{\ifte{\extract{\Bools}(e)}{e'}{e''}}\anylogrel\trg{\ifte{e}{e'}{e''}}:\emuldv
	\end{align*}

	Take $W$ such that $\stepsfun{W}\leq n$ and $(W,\src{\gamma},\trg{\gamma})\in\envrel{\toemul{\trg{\Gamma}}}$, the thesis is: (we omit substitutions as they don't play an active role)
	\begin{itemize}
		\item $(W, \src{\ifte{\extract{\Bools}(e)}{e'}{e''}} , \trg{\ifte{e}{e'}{e''}})\in\termrel{\emuldv}$
	\end{itemize}

	By \Thmref{thm:rel-ter-plug-rel-ctx-still-rel} with HE, we have that for HW $W'\futw W$, and HV $(W',\src{v},\trg{v})\in\valrel{\emuldv}$, the thesis becomes:
	\begin{itemize}
		\item $(W',\src{\ifte{\extract{\Bool}(v)}{e'}{e''}},\trg{\ifte{v}{e'}{e''}})\in\termrel{\emuldv}$
	\end{itemize}
	We perform a case analysis based on HV:
	\begin{itemize}
		\item \trg{v}=\truet/\falset and \src{v}=\src{1}/\src{0}

		We consider the case \truet/\src{1} the other is analogous.

		The reductions proceed as follows:
		\begin{align*}
				&
				\src{P\triangleright \extract{\Bool}(1)}
			\\
			\equiv
				&
				\src{P\triangleright \letin{x}{1}{\ifte{x \geq 2}{\fails}{\ifte{x+1\geq 2}{\trues}{\falses}}} }
			\\
			\reds
				&
				\src{P\triangleright \ifte{1 \geq 2}{\fails}{\ifte{1+1\geq 2}{\trues}{\falses}} }
			\\
			\reds
				&
				\src{P\triangleright \ifte{1+1\geq 2}{\trues}{\falses} }
			\\
			\reds\reds
				&
				\src{P\triangleright \trues }
		\end{align*}
		By \Thmref{thm:log-rel-clos-antired} the thesis becomes:
		\begin{itemize}
			\item $(W',\src{\ifte{\trues}{e'}{e''}},\trg{\ifte{\truet}{e'}{e''}})\in\termrel{\emuldv}$
		\end{itemize}
		If the $\stepsfun{W'}=0$ the thesis follows from \Thmref{thm:no-steps-impl-rel}, otherwise:

		We can reduce based on \Cref{tr:ev-s-if-t} and \Cref{tr:ev-t-if-t}.
		By \Thmref{thm:log-rel-clos-antired} the thesis becomes:
		\begin{itemize}
			\item $(W',\src{e'},\trg{e'})\in\termrel{\emuldv}$
		\end{itemize}
		If the $\stepsfun{W'}=0$ the thesis follows from \Thmref{thm:no-steps-impl-rel}, otherwise by HEP.

		\item \trg{v}=\trg{n} and \src{v}=\src{n+2}

		In this case we have that:
		\begin{align*}
				&
				\src{P\triangleright \extract{\Bool}(n+2)}
			\\
			\equiv
				&
				\src{P\triangleright \letin{x}{n+2}{\ifte{x \geq 2}{\fails}{\ifte{x+1\geq 2}{\trues}{\falses}}} }
			\\
			\reds
				&
				\src{P\triangleright \ifte{n+2 \geq 2}{\fails}{\ifte{x+1\geq 2}{\trues}{\falses}} }
			\\
			\Xtosb{\failact}
				&
				\fails
		\end{align*}
		and
		\begin{align*}
			\trg{P\triangleright \ifte{n}{e'}{e''}}
			\Xtotb{\failact}
			\failt
		\end{align*}
		So this case holds by definition of $\obsfun{W'}{\anylogrel}$.
	\end{itemize}
\end{proof}


\begin{lemma}[Compatibility lemma for backtranslation of application]\label{thm:compat-backtr-app}
	\begin{align*}
		\text{if }
			&\
			\toemul{\trg{\Gamma}};\src{P};\trg{P}\vdash\src{e}\anylogreln{n}\trg{e}:\emuldv
		\\
			&\
			\src{f(x:\sigma'):\sigma\mapsto\ret e}\in\src{P} 
		\\
			&\ (HP)~
			\src{P};\trg{P}\vdash \src{\call{f}} \anylogreln{n} \trg{\call{f}} : \src{\sigma'\to\sigma}
		\\
		\text{then }
			&\
			\toemul{\trg{\Gamma}};\src{P};\trg{P}\vdash \src{\inject{\tau'}(\call{f}~\extract{\tau}(\src{e}))}  \anylogreln{n} \trg{\call{f}~e} :\emuldv	
	\end{align*}
\end{lemma}
\begin{proof}
	We need to prove that
	\begin{align*}
		\toemul{\trg{\Gamma}};\src{P};\trg{P}\vdash \src{\inject{\tau'}(\call{f}~\extract{\tau}(\src{e}))}  \anylogreln{n} \trg{\call{f}~e} :\emuldv
	\end{align*}

	Take $W$ such that $\stepsfun{W}\leq n$ and $(W,\src{\gamma},\trg{\gamma})\in\envrel{\toemul{\trg{\Gamma}}}$, the thesis is: (we omit substitutions as they don't play an active role)
	\begin{itemize}
		\item $(W, \src{\inject{\tau'}(\call{f}~\extract{\tau}(\src{e}))} , \trg{\call{f}~e})\in\termrel{\emuldv}$
	\end{itemize}
	By \Thmref{thm:rel-ter-plug-rel-ctx-still-rel} with HE we have that for HW $W'\futw W$, and HV $(W',\src{v},\trg{v})\in\valrel{\emuldv}$, the thesis becomes:
	\begin{itemize}
		\item $(W, \src{\inject{\tau'}(\call{f}~\extract{\tau}(\src{v}))} , \trg{\call{f}~v})\in\termrel{\emuldv}$
	\end{itemize}
	We perform a case analysis based on HV:

	\begin{itemize}
		\item \trg{v}=\truet/\falset and \src{v}=\src{1}/\src{0} (respectively).

		We consider the first case only, the other is analogous.

		We perform a case analysis on \src{\tau}:
		\begin{itemize}
			\item \src{\tau}=\Bools

			The thesis is:
			\begin{itemize}
				\item $(W', \src{\inject{\tau'}(\call{f}~\extract{\Bools}(\src{v}))} , \trg{\call{f}~v})\in\termrel{\emuldv}$
			\end{itemize}

			By definition of \extract{\Boolt} we have
			\begin{align*}
				&
				\src{P\triangleright \inject{\tau'}(\call{f}~\extract{\Bools}(1))}
				\\
				\equiv
				&
				\src{P\triangleright \inject{\tau'}(\call{f}~\letin{x}{1}{\ifte{x \geq 2}{\fails}{\ifte{x+1\geq 2}{\trues}{\falses}}})}
				\\
				\reds
				&
				\src{P\triangleright \inject{\tau'}(\call{f}~{\ifte{1 \geq 2}{\fails}{\ifte{1+1\geq 2}{\trues}{\falses}}})}
				\\
				\reds
				&
				\src{P\triangleright \inject{\tau'}(\call{f}~{{\ifte{1+1\geq 2}{\trues}{\falses}}})}
				\\
				\reds
				&
				\src{P\triangleright \inject{\tau'}(\call{f}~\trues)}
			\end{align*}
			So by \Thmref{thm:log-rel-clos-antired} the thesis becomes:
			\begin{itemize}
				\item $(W', \src{\inject{\tau'}(\call{f}~\trues)} , \trg{\call{f}~\truet})\in\termrel{\emuldv}$
			\end{itemize}
			If the $\stepsfun{W'}=0$ the thesis follows from \Thmref{thm:no-steps-impl-rel}, otherwise:

			By HP and by the Hs on the function bodies, and by the relatedness of \trues and \truet and by the \Thmref{thm:val-rel-mono} we have that HF:
			\begin{align*}
				(W', \src{\ret e}\subs{\trues}{x}, \trg{\ret e}\subt{\truet}{x})\in\termrel{\psd{\tau'}}
			\end{align*}

			By \Thmref{thm:rel-ter-plug-rel-ctx-still-rel} with HF we have that for HW $W''\futw W'$, and HV $(W'',\src{v'},\trg{v'})\in\valrel{\tau'}$, the thesis becomes:
			\begin{itemize}
				\item $(W', \src{\inject{\tau'}(v')} , \trg{v'})\in\termrel{\emuldv}$
			\end{itemize}
			This case follows from \Thmref{thm:val-rel-impl-term-rel} and by \Thmref{thm:inject-red-pres-rel} with HV.

			\item \src{\tau} = \Nats

			By definition of \extract{\Nats} we have:
			\begin{align*}
				&
				\src{P\triangleright \inject{\tau'}(\call{f}~\extract{\Nats}(1))}
				\\
				\equiv
				&
				\src{P\triangleright \inject{\tau'}(\call{f}~ \letin{x}{1}{\ifte{x\geq 2}{x-2}{\fails}})}
				\\
				\reds
				&
				\src{P\triangleright \inject{\tau'}(\call{f}~ {\ifte{1\geq 2}{1-2}{\fails}})}
				\\
				\reds
				&
				\src{P\triangleright \inject{\tau'}(\call{f}~ {\fails}})
				\\
				\reds
				&
				\fails
			\end{align*}
			and by definition of the function bodies and \Cref{tr:comptd-fun}:
			\begin{align*}
				&
				\trg{P\triangleright \trg{\call{f}~\truet}}
				\\
				\redt
				&
				\trg{P\triangleright \ret \ifte{\truet \checkty{\comptd{\Nats}}}{\comptd{e}}{\failt}}
				\\
				\equiv
				&
				\trg{P\triangleright \ret \ifte{\truet \checkty{\Nat}}{\comptd{e}}{\failt}}
				\\
				\redt
				&
				\trg{P\triangleright \ret \ifte{\falset}{\comptd{e}}{\failt}}
				\\
				\redt
				&
				\trg{P\triangleright \ret \failt}
				\\
				\redt
				&
				\failt
			\end{align*}
			So this case holds by definition of $\obsfun{W'}{\anylogrel}$.
		\end{itemize}

		\item \trg{v}=\trg{n} and \src{v}=\src{n+2}

		Case analysis on \src{\tau}

		\begin{itemize}
			\item \src{\tau}=\Bools 

			This is analogous to the case for naturals above.

			\item \src{\tau} = \Nats

			This is analogous to the case for booleans above.
		\end{itemize}
	\end{itemize}
\end{proof}

\begin{lemma}[Compatibility lemma for backtranslation of check]\label{thm:compat-backtr-check}
	\begin{align*}
		\text{if }
			&\ (HE)~
			\toemul{\trg{\Gamma}};\src{P};\trg{P}\vdash\src{e}\anylogreln{n}\trg{e}:\emuldv
		\\
		\text{then }
			1
			&\
			\toemul{\trg{\Gamma}};\src{P};\trg{P}\vdash \src{\letin{x:\Nats}{\src{e}}{\ifte{x \geq 2}{0}{1}}} \anylogreln{n} \trg{e\checkty{\Boolt}} : \emuldv
		\\
			2
			&\
			\toemul{\trg{\Gamma}};\src{P};\trg{P}\vdash \src{\letin{x:\Nats}{\src{e}}{\ifte{x \geq 2}{1}{0}}} \anylogreln{n} \trg{e\checkty{\Nat}} : \emuldv	
	\end{align*}
\end{lemma}
\begin{proof}
	We need to prove that 
	\begin{align*}
		1
			&\
			\toemul{\trg{\Gamma}};\src{P};\trg{P}\vdash \src{\letin{x:\Nats}{\src{e}}{\ifte{x \geq 2}{0}{1}}} \anylogreln{n} \trg{e\checkty{\Boolt}} : \emuldv
		\\
			2
			&\
			\toemul{\trg{\Gamma}};\src{P};\trg{P}\vdash \src{\letin{x:\Nats}{\src{e}}{\ifte{x \geq 2}{1}{0}}} \anylogreln{n} \trg{e\checkty{\Nat}} : \emuldv	
	\end{align*}
	We only show case 1, the other is analogous.

	Take $W$ such that $\stepsfun{W}\leq n$ and $(W,\src{\gamma},\trg{\gamma})\in\envrel{\toemul{\trg{\Gamma}}}$, the thesis is: (we omit substitutions as they don't play an active role)
	\begin{enumerate}
		\item $(W, \src{\letin{x:\Nats}{\src{e}}{\ifte{x \geq 2}{0}{1}}} , \trg{e\checkty{\Boolt}})\in\termrel{\emuldv}$
	\end{enumerate}

	By \Thmref{thm:rel-ter-plug-rel-ctx-still-rel} with HE we have that for HW $W'\futw W$, and HV $(W',\src{v},\trg{v})\in\valrel{\emuldv}$, the thesis becomes:
	\begin{itemize}
		\item $(W',\src{\letin{x:\Nats}{\src{v}}{\ifte{x \geq 2}{0}{1}}},\trg{v\checkty{\Boolt}})\in\termrel{\emuldv}$
	\end{itemize}
	We perform a case analysis based on HV:
	\begin{itemize}
		\item \trg{v}=\truet/\falset and \src{v}=\src{1}/\src{0} (respectively).

		We consider only the first case, the other is analogous.

		We have that
		\begin{align*}
			&
			\src{P\triangleright \src{\letin{x:\Nats}{\src{1}}{\ifte{x \geq 2}{0}{1}}}}
			\\
			\reds
			&
			\src{P\triangleright \src{\ifte{1 \geq 2}{0}{1}}}
			\\
			\reds
			&
			\src{P\triangleright \src{1}}
		\end{align*}
		and
		\begin{align*}
			&
			\trg{P\triangleright \trg{\truet\checkty{\Boolt}}}
			\redt
			\trg{P\triangleright \truet}
		\end{align*}
		This case holds by \Thmref{thm:log-rel-clos-antired} and \Thmref{thm:val-rel-impl-term-rel} and by the definition of \valrel{\emuldv}.

		\item \trg{v}=\trg{n} and \src{v}=\src{n+2}

		In this case we have that:
		\begin{align*}
			&
			\src{P\triangleright \src{\letin{x:\Nats}{\src{n+2}}{\ifte{x \geq 2}{0}{1}}}}
			\\
			\reds
			&
			\src{P\triangleright \src{\ifte{n+2 \geq 2}{0}{1}}}
			\\
			\reds
			&
			\src{P\triangleright \src{0}}
		\end{align*}
		and
		\begin{align*}
			&
			\trg{P\triangleright \trg{n\checkty{\Boolt}}}
			\redt
			\trg{P\triangleright \falset}
		\end{align*}
		This case holds by \Thmref{thm:log-rel-clos-antired} and \Thmref{thm:val-rel-impl-term-rel} and by the definition of \valrel{\emuldv}.
	\end{itemize}

\end{proof}

%
\paragraph{Semantic Preservation of Backtranslation}
\begin{theorem}[\backtrdt{\cdot} is semantics preserving]\label{thm:backtr-sem-pres}
	\begin{align*}
		\text{if }
			&\
			\trg{\Gamma}\vdash\trg{e}
		\\
			&\ (HP)~
			\vdash \src{P}\anylogrel\trg{P}
		\\
		\text{then }
			&\
			\toemul{\trg{\Gamma}};\src{P};\trg{P}\vdash\backtr{\trg{e}}\anylogreln{n}\trg{e}:\emuldv
	\end{align*}
\end{theorem}
\begin{proof}
	The proof proceeds by induction on the derivation of $\trg{\Gamma}\vdash\trg{e}$.
	\begin{description}
		\item[Base cases]
		\begin{description}
			\item[true,false,nat] By definition of the \valrel{\emuldv}
			\item[var] By definition of the \envrel{\cdot}.
			\item[call] This case cannot arise.     
		\end{description}
		\item[Inductive cases]  
		\begin{description}
			\item[app] By IH and HP and \Thmref{thm:compat-backtr-app}.
			\item[op] By IH and \Thmref{thm:compat-backtr-op}.
			\item[geq] Analogous to the case above.
			\item[if] By IH and \Thmref{thm:compat-backtr-if}.
			\item[letin] By IH and \Thmref{thm:compat-backtr-letin}.
			\item[check] By IH and \Thmref{thm:compat-backtr-check}.
		\end{description}
	\end{description}
\end{proof}

%
\paragraph{Theorems that Yield \pf{\rrhp}}
\begin{theorem}[\comptd{\cdot} preserves behaviors]\label{thm:comptd-pres-beh}
	\begin{align*}
		\text{ if }
			&\ (HT)~
			\trg{\comptd{P}\triangleright e \Xtotb{\beta} \comptd{P}\triangleright e'} 
		\\
		\text{ then }
			&
			\src{P\triangleright \backtrdt{\trg{e}} \Xtosb{\beta} P\triangleright e'} 
	\end{align*}
\end{theorem}
\begin{proof}
	By \Thmref{thm:comp-sem-pres-prog} we have HPP:
	\begin{itemize}
		\item $\vdash\src{P}\anylogrel\comptd{\src{P}}$
	\end{itemize}

	Given that $\trge\vdash\trg{e}$, by \Thmref{thm:backtr-sem-pres} with HPP we have HPE:
	\begin{itemize}
		\item $\toemul{\trg{\Gamma}};\src{P};\comptd{\src{P}}\vdash\backtr{\trg{e}}\anylogreln{n}\trg{e}:\emuldv$
	\end{itemize}

	The thesis follows by \Thmref{thm:log-rel-adeq-gt} with HT.
\end{proof}

\begin{theorem}[\comptd{\cdot} reflects behaviors]\label{thm:comptd-refl-beh}
	\begin{align*}
		\text{ if }
			&\ (HS)~
			\src{P\triangleright \backtrdt{\trg{e}} \Xtosb{\beta} P\triangleright e'} 
		\\
		\text{ then }
			&
			\trg{\comptd{P}\triangleright e \Xtotb{\beta} \comptd{P}\triangleright e'} 
	\end{align*}
\end{theorem}
\begin{proof}
	By \Thmref{thm:comp-sem-pres-prog} we have HPP:
	\begin{itemize}
		\item $\vdash\src{P}\anylogrel\comptd{\src{P}}$
	\end{itemize}

	Given that $\trge\vdash\trg{e}$, by \Thmref{thm:backtr-sem-pres} with HPP we have HPE:
	\begin{itemize}
		\item $\toemul{\trg{\Gamma}};\src{P};\comptd{\src{P}}\vdash\backtr{\trg{e}}\anylogreln{n}\trg{e}:\emuldv$
	\end{itemize}

	The thesis follows by \Thmref{thm:log-rel-adeq-less} with HS.
\end{proof}


\subsubsection{Proof That \texorpdfstring{\comptd{\cdot}}{the Compiler} Satisfies \texorpdfstring{\Thmref{def:rrhc-linkable}}{\pf{\rrhp}}}
\begin{align*}
	\forall\trg{e}.\exists\src{e}.&\forall\src{P} \text{ such that }\src{P \bowtie e}, \forall \beta
	\\
	&\
	\trg{\comptd{P}\triangleright e \Xtotb{\beta} \comptd{P}\triangleright e'} 
	\\
	\iff
	&\
	\src{P\triangleright e \Xtosb{\beta} P\triangleright e'} 
\end{align*}
We instantiate \src{e} with \backtrdt{\trg{e}} then two cases arise.

\begin{description}
	\item[$\Rightarrow$ direction]
		By \Thmref{thm:comptd-pres-beh}
	\item[$\Leftarrow$ direction]  
		By \Thmref{thm:comptd-refl-beh}.
\end{description}

\subsection{Proof That \texorpdfstring{\(\comptd{\cdot}\)}{the Compiler} Is \texorpdfstring{\(\pf{\rfrxp}_{\bowtie}\)}{RFrXC}} 
This section focuses on giving a high-level overview of the proof
technique that we use to prove that our compiler satisfies the
following variant of the criterion \pf{\rfrxp}:
\begin{definition}[\(\pf{\rfrxp}_{\bowtie}\)]\label{def:rfrxc-linkable}
  \begin{align*}
    \pf{\rfrxp}_{\bowtie}:\quad &\forall K.~\forall\src{P_1}\ldots\src{P_K : \OB{I}}.~ \forall\trg{C_T}.~ \forall x_1 \ldots x_K.\\
    &(\trg{\cmp{P_1} \bowtie C_T} \wedge \ldots \wedge \trg{\cmp{P_K} \bowtie C_T})
    \implies\\
    &(\trg{C_T\hole{\cmp{P_1}} \rightsquigarrow}\ x_1
    \mathrel{\wedge} \ldots \mathrel{\wedge}
    \trg{C_T\hole{\cmp{P_K}} \rightsquigarrow}\ x_K) \implies\\
    &\exists \src{C_S}\ldotp
     \src{P_1 \bowtie C_S} \wedge \ldots \wedge \src{P_K \bowtie C_S} \wedge\\
    &\phantom{\exists \src{C_S}\ldotp} 
     \src{C_S\hole{P_1}\rightsquigarrow}\ x_1
      \mathrel{\wedge} \ldots \mathrel{\wedge}
      \src{C_S\hole{P_K}\rightsquigarrow}\ x_K
  \end{align*}
\end{definition}

This criterion is equivalent to the following property-full criterion:
\begin{definition}[\(\rfrxp_{\bowtie}\)]\label{def:rfrxp-linkable}
  \begin{align*}
    \rfrxp_{\bowtie}:\quad &
    \forall \src{\OB{I}}.~\forall K,\src{P_1},\cdots,\src{P_k : \OB{I}},\com{R} \in 2^{(\ii{XPref^K})}.
    \\
    &
    (\forall \src{C_S}, \com{x_1}, \cdots, \com{x_K}, 
    (
    \src{P_1 \bowtie C_S} \wedge \src{C_S\hole{P_1}\sem\com{x_1}} 
    \wedge \cdots \wedge
	   {\src{P_K \bowtie C_S} \wedge \src{C_S\hole{P_K}\sem\com{x_K}}}
	   )
	   \\
	   &\quad \Rightarrow
	   (\com{x_1},\cdots,\com{x_K})\in\com{R}) \Rightarrow\\	   
	   &
	   (\forall \trg{C_T}, x_1, \cdots, x_K
	   (
	       {\trg{\cmp{P_1} \bowtie C_T} \wedge \trg{C_T\hole{\compgen{\src{P_1}}}\sem\com{x_1}}}
	       \wedge \cdots \wedge \trg{\cmp{P_K} \bowtie C_T} \wedge {\trg{C_T\hole{\compgen{\src{P_K}}}\sem\com{x_K}}}
	       )
	       \\
	       &\quad
	       \Rightarrow (\com{x_1},\cdots,\com{x_K})\in\com{R})
  \end{align*}
\end{definition}
The proof of the equivalence of these two criteria is very similar to \autoref{thm:rfrxp-eq}.


\subsubsection{Overview of the Proof Technique}


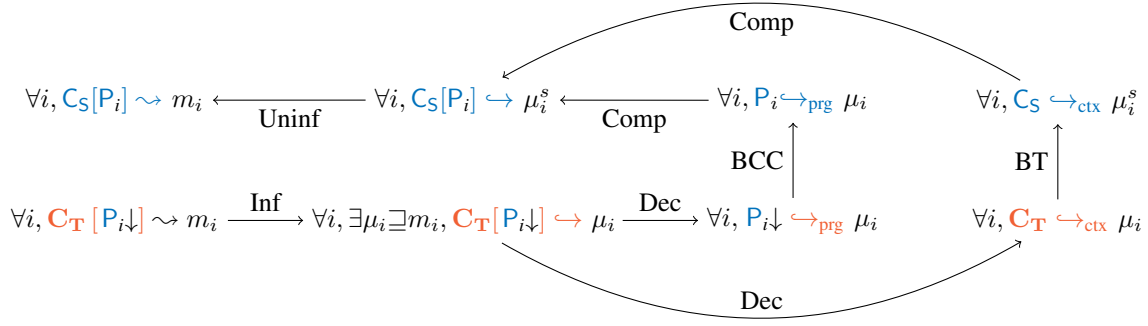
\begin{figure}
\begin{center}
  \begin{tikzpicture}[auto]
    \node(cp) {\( \forall i, \trg{C_T\left[\compgen{\src{P}_{\com{i}}}\right]} \leadsto m_i\) };
    \node[right = of cp,text width=4cm] (cpi) { \(\forall i, \exists \mu_i {\sqsupseteq} m_i,
      \trg{\plug{C_T}{\compgen{\src{P}_{\com{i}}}} \hookrightarrow}\ \mu_i\) };
    \node[right = of cpi] (pi) {\(\forall i, \trg{\compgen{\src{P}_{\com{i}}} \hookrightarrow_{\text{prg}}}\ \mu_i\)};
    \node[right = of pi] (ci) {\(\forall i, \trg{C_T \hookrightarrow_{\text{ctx}}}\ \mu_i\)};
    
    \node[above = of ci] (Ci) {\(\forall i, \src{C_S \hookrightarrow_{\text{ctx}}}\ \mu_i^s\)};
    \node[above = of pi] (Pi) {\(\forall i, \src{P}_i \src{\hookrightarrow_{\text{prg}}}\ \mu_i\)};
    \node[above = of cpi] (CPi) {\(\forall i, \src{\plug{C_S}{P_{\com{i}}} \hookrightarrow}\ \mu_i^s\)};
    \node[above = of cp] (CP) {\(\forall i, \src{\plug{C_S}{P_{\com{i}}} \leadsto}\ m_i\)};
    
    \draw[->] (cp) to node {Inf} (cpi);
    \draw[->] (cpi) to node {Dec} (pi);
    \draw[->,bend right] (cpi) to node {Dec} (ci);
    
    \draw[->] (ci) to node {BT} (Ci);
    \draw[->] (pi) to node {BCC} (Pi);
    
    \draw[->] (Pi) to node {Comp} (CPi);
    \draw[->,bend right] (Ci) to node {Comp} (CPi);
    \draw[->] (CPi) to node {Uninf} (CP);
  \end{tikzpicture}
\end{center}
  \caption{Proposed proof technique}\label{hlo:proof-technique}
\end{figure}

Our proof technique for this is described in Figure~\ref{hlo:proof-technique}. At the heart of this technique is the back-translation
of a finite set of finite trace prefixes into a source context. In particular, this back-translation technique do not inspect
the code of the target context.
The first steps consist in transforming the trace prefixes into prefixes that can be back-translated easily, and separating the target
context from the compiled programs. Then, we build a back-translation that provides us with a source context that can be composed with the initial
source programs to generate the initial traces.

The reason for requiring all programs to share the same interface \(\src{\OB{I}}\) is that it allows us to produce a well-typed context.
Otherwise, two programs could contain the same function, but one returning a natural number and the other a boolean. 
The back-translation would be immediatly impossible.


\subsubsection{Informative Traces}\label{hlo:inf}
The first step of the proof is to augment the existing operational semantics with new events that allow to precisely
track the behavior of the program and of the context. This new semantics is called \emph{informative semantics} and produce
\emph{informative traces}. They are defined at both the source level and the target level.
The relations \(\comc{\hookrightarrow}\) are the equivalent of \(\leadsto\) for these informative semantics, and are defined as:
\begin{align*}
\src{\plug{C}{P} \hookrightarrow}\ \mu &\iff \exists \src{e}, \src{P\triangleright C \Xtolsb{\mu} P\triangleright e}\\
\trg{\plug{C}{P} \hookrightarrow}\ \mu &\iff \exists \trg{e}, \trg{P\triangleright C \Xtoltb{\mu} P\triangleright e}
\end{align*}

We can state the theorem for passing to informative traces as follow
\begin{restatable}[Informative traces]{theorem}{isem}\label{thm:isem}
  Let \(\trg{C_T}\) be a target context and \(\trg{P_T}\) a target program that are linkable.
  Then,
  \[
  \forall m, \trg{\plug{C_T}{P_T} \leadsto}\ m \implies \exists \mu \sqsupseteq m, \trg{\plug{C_T}{P_T} \hookrightarrow}\ \mu
  \]
\end{restatable}
where
\[ \mu \sqsupseteq m \iff |\mu|_{\text{I/O/termination}} = m.\]
\begin{proof}
  Let \trg{C_T} be a target context, \trg{P_T} a target program and \(m\) a finite prefix.
  We are going to show that if there exists \( \trg{e} \) such that
  \trg{P_T \triangleright C_T \Xtotb{m} P_T\triangleright e}, then there exists \(\mu\) such that
  \( |\mu|_{\text{I/O}} = m \) and
  \(\trg{P_T \triangleright C_T \Xtoltb{\mu} P_T\triangleright e}\).

  Let us proceed by induction on the relation
  \(\trg{P_T\triangleright C_T \Xtotb{m} P_T\triangleright e}\).
  \begin{description}
  \item[\Cref{tr:ev-t-refl}] Immediate.
  \item[\Cref{tr:ev-t-term}] This is true by taking \(\mu = \termc\), because the informative semantics can progress if and only if the non-informative semantics can.
  \item[\Cref{tr:ev-t-divr}] This is true by taking \(\mu = \divrc\), because the informative semantics can only diverge when executing the program part (the context
    can not loop or do recursion), and calls from the program part do not generate any event.
  \item[\Cref{tr:beh-t-silent}] Then \(\trg{P_T\triangleright C_T \xtotb{\epsilon} P_T\triangleright e}\) according to the non-informative semantics. Since the semantics only differ on the
    events that are generated, we have two cases.
    Either \(\trg{P_T\triangleright C_T \xtotb{\epsilon} P_T\triangleright e}\) according to the informative semantics, in which case we can take \(\mu = \epsilon\).
    Or \(\trg{P_T\triangleright C_T\xtotb{\alpha} P_T\triangleright e}\) according to the informative semantics, in which case we can take \(\mu = \alpha\). This \(\alpha\) must be
    a call or return event by definition of the informative semantics, hence the result.
  \item[\Cref{tr:beh-t-sin}] Since \(\trg{P_T\triangleright C_T \xtotb{\alpha} P_T\triangleright e}\) according to the non-informative semantics, this is also the case
    according to the informative semantics, hence the result.
  \item[\Cref{tr:beh-t-cons}] Then \(\trg{P_T\triangleright C_T\Xtotb{m_1} P_T\triangleright e'}\) and \(\trg{P_T\triangleright e' \Xtotb{m_2} P_T \triangleright e}\) with
    \( m = m_1m_2\).
    By applying the induction hypothesis, there exists \(\mu_1\) and \(\mu_2\) such that \(\trg{P_T\triangleright C_T\Xtoltb{\mu_1} e'}\), \(\trg{P_T\triangleright e'\Xtoltb{\mu_2} e}\),
    \(|\mu_1|_{\text{I/O/termination}} = m_1\), and \(|\mu_2|_{\text{I/O/termination}} = m_2\).

    Therefore by applying \Cref{tr:tr-t-cons}, \( \trg{P_T\triangleright C_T\Xtoltb{\mu_1\mu_2} e}\).
    It is easy to see
    that \( |\mu_1\mu_2|_{\text{I/O/termination}} = m_1m_2\). We are done.
  \end{description}
\end{proof}

\subsubsection{Decomposition}
This decomposition step relies on the definition of \emph{partial semantics}, one for programs and one for contexts.
These partial semantics describe the possible behaviors of a program in any context and
of a context with respect to any program. Partial semantics can often be defined by abstracting away one
part of the whole program (the context for the partial semantics of programs, and the program for
the partial semantics of contexts), by introducing non-determinism for modeling the abstracted part.

We index our relations by either ``ctx'' or ``prg'' to denote the partial semantics.
The partial semantics for contexts defined as:
\begin{center}
  \typerule{E\Lt-ctx-call}{
  }{
    \src{\call{f}~v \xtosbctx{\cl{f}{v}} \ret e} 
  }{ev-s-ctx-call}
  \typerule{E\Ld-ctx-call}{
  }{
    \trg{\call{f}~v \xtotbctx{\cl{f}{v}} \ret e} 
  }{ev-t-ctx-call}

  \typerule{E\Lt-ctx-ret}{
  }{
    \src{\ret v \xtosbctx{\epsilon} v} 
  }{ev-s-ctx-ret}
  \typerule{E\Ld-ctx-ret}{
  }{
    \trg{\ret v \xtotbctx{\epsilon} v} 
  }{ev-t-ctx-ret}
\end{center}
and the relations \(\Xtolsbctx{\cdot}\) and \(\Xtoltbctx{\cdot}\) are defined in the same manner
as the complete semantics.

The partial semantics for programs are defined in terms of the complete semantics, and are parameterized by the interface
of the program \(\src{\OB{I}}\).
Informally, we define \( \comc{P\hookrightarrow_{\text{prg}}}\ \mu \) to mean that the program \(\comc{P}\) is able
to produce each part of the trace \(\mu\) that comes from the program, i.e. each part that starts
with a call event \(\cl{f}{v}\) and ends before or with the corresponding return event, when it is put into the context that
simply calls this function \(\comc{f}\) with this value \(\comc{v}\). For every ``subtrace'' \(\mu'\)
of \(\mu\) starting with a call event \(\bl{\cl{f}{v}}\) and stopping at the latest at the next (corresponding)
return event, it must be that \( \comc{P \triangleright \call{f}~v \hookrightarrow}\ \mu'\).

\begin{definition}[Partial semantics for programs]
\( \src{P \hookrightarrow_{\text{prg}}}\ \mu\) if and only if:
\begin{itemize}
  \item for any trace \(\mu_{f,v,v'} = \cl{f}{v};\mu';\rt{v'}\) such that \(\mu = \mu_1;\mu_{f,v,v'};\mu_2\), such that
    there is no event \(\ret{\dots}\) in \(\mu'\), and such that \(\src{f : \tau \rightarrow \tau' \in \OB{I}}\)
    with \(\src{v \in \tau}\),
    we have \[\src{P_T \triangleright \call{f}~v \Xtolsb{\mu_{f,v,v'}} P\triangleright v'};\]
  \item for any trace \(\mu_{f,v} = \cl{f}{v};\mu'\) such that \(\mu = \mu_1; \mu_{f,v}\), such that there is no
    event \(\ret{\dots}\) in \(\mu'\), and such that \(\src{f : \tau \rightarrow \tau' \in \OB{I}}\) with
    \(\src{v \in \tau}\),
    there exists \(\src{e}\) such that \[\src{P_T\triangleright\call{f}~v \Xtolsb{\mu_{f,v}} P\triangleright e}.\]
\end{itemize}

\( \trg{P \hookrightarrow_{\text{prg}}}\ \mu\) if and only if:
\begin{itemize}
  \item for any trace \(\mu_{f,v,v'} = \cl{f}{v};\mu';\rt{v'}\) such that \(\mu = \mu_1;\mu_{f,v,v'};\mu_2\), such that
    there is no event \(\ret{\dots}\) in \(\mu'\), and such that \(\trg{f} \in \src{\OB{I}}\)
    we have \[\trg{P_T \triangleright \call{f}~v \Xtoltb{\mu_{f,v,v'}} P\triangleright v'};\]
  \item for any trace \(\mu_{f,v} = \cl{f}{v};\mu'\) such that \(\mu = \mu_1; \mu_{f,v}\), such that there is no
    event \(\ret{\dots}\) in \(\mu'\), and such that \(\trg{f} \in \src{\OB{I}}\)
    there exists \(\trg{e}\) such that \[\trg{P_T\triangleright\call{f}~v \Xtoltb{\mu_{f,v}} P\triangleright e}.\]
\end{itemize}
\end{definition}
We must restrict this definition to the well-typed calls in the source level: indeed, a badly-typed call does
not make sense in the source language.

Our decomposition theorem talks about both programs and contexts:
\begin{restatable}[Decomposition]{theorem}{decthm}\label{thm:dec}
  Let \(\trg{C_T}\) be a target context and \(\trg{P_T}\) a target program that are linkable.
  Then,
  \[
  \forall \mu, \trg{C_T\left[P_T\right] \hookrightarrow}\ \mu \implies \trg{C_T\hookrightarrow_{\text{ctx}}}\ \mu \wedge
  \trg{P_T\hookrightarrow_{\text{prg}}}\ \mu
  \]
\end{restatable}

We are going to prove two different lemmas, one for contexts and one for programs.
\begin{lemma}
  Let \(\trg{C_T}\) be a target context and \(\trg{P_T}\) be a target program, \(\mu\) an informative
  trace and \(\trg{e}\) a target expression. Then,
  \[
  \trg{C_T\left[P_T\right] \Xtoltb{\mu} e} \implies
  \trg{C_T \Xtoltbctx{\mu} e}
  \]
\end{lemma}
\begin{proof}
  By induction on the relation \(\trg{\plug{C_T}{P_T} \Xtoltb{\mu} P_T\triangleright e}\)
  \begin{description}
  \item[\Cref{tr:tr-t-silent}]
    Therefore \(\trg{\plug{C_T}{P_T} \xtotb{\epsilon} P_T\triangleright e}\). By case analysis, it is
    also the case that \(\trg{C_T \xtotbctx{\epsilon} e}\) hence the result.
  \item[\Cref{tr:tr-t-act}] \(\trg{\plug{C_T}{P_T} \xtotb{\alpha} P_T\triangleright e}\). We proceed
    by case analysis on this relation: if \(\alpha\) is an I/O operation, correct termination or
    failure event, then we indeed have \(\trg{C_T \xtotbctx{\alpha} e}\).

    Otherwise, \(\alpha = \divrc\). Therefore, \(\forall n, \exists \trg{e_n},
    \trg{\plug{C_T}{P_T} \xtotb{\epsilon}\redapp{n} P_T \triangleright e_n}\). Now, by induction on
    \(n\), we can prove that \(\forall n, \exists \trg{e_n}, \trg{C_T \xtotbctx{\epsilon}^n e_n}\).
    Hence the result. 
  \item[\Cref{tr:tr-t-sin}] Then \(\trg{\plug{C_T}{P_T} \xtotb{\beta} P_T\triangleright e}\).
    We proceed by case analysis on this relation:
    \begin{itemize}
    \item If \(\beta = \cl{f}{v}\), then \(\trg{C_T} = \evalctxt{\call{f}~v}\) and \( \trg{e} = \evalctxt{\ret e'}\) for some evaluation
      context \(\trg{\evalctx}\) and some expression \(\trg{e'}\).
      Therefore, \(\trg{e \xtotbctx{\cl{f}{v}} \evalctxt{\ret{e'}}}\) by the partial semantics, hence the result.
    \item If \(\beta = \rt{f}{v}\), then \(\trg{C_T} = \evalctxt{\ret{v}}\) for some evaluation context
      \(\trg{\evalctx}\).
      Therefore, \(\trg{e \xtotbctx{\rt{f}{v}} \evalctxt{v}}\) according to the partial semantics, hence the result.
    \end{itemize}
  \item[\Cref{tr:tr-t-cons}] We have that \( \trg{P_T \triangleright C_t \Xtoltbctx{\mu_1} e'} \)
    and \(\trg{P_T\triangleright e' \Xtoltbctx{\mu_2} e}\).
    Then, by applying the induction hypothesis to the two relations, we are done.
  \end{description}
\end{proof}

Then, we prove a similar lemma for programs:
\begin{lemma}
  Let \(\trg{P_T}\) be a target program, \(\trg{C_T}\) a target context and \(\mu\) an informative trace.
  Suppose that \(\trg{C_T\left[P_T\right] \hookrightarrow}\ \mu\).
  Then:
  \begin{itemize}
  \item for any trace \(\mu_{f,v,v'} = \cl{f}{v};\mu';\rt{v'}\) such that \(\mu = \mu_1;\mu_{f,v,v'};\mu_2\) and such that
    there is no event \(\ret{\dots}\) in \(\mu'\), \(\trg{P_T \triangleright \call{f}~v \Xtotb{\mu_{f,v,v'}} v'}\)
  \item for any trace \(\mu_{f,v} = \cl{f}{v};\mu'\) such that \(\mu = \mu_1; \mu_{f,v}\) and such that there is no
    event \(\ret{\dots}\) in \(\mu'\), there exists \(\trg{e}\) such that \(\trg{P_T\triangleright\call{f}~v \Xtotb{\mu_{f,v}} e}\).
  \end{itemize}
\end{lemma}
\begin{proof}
  Consider the first case for instance. From the fact that \(\mu_{f,v,v'}\) appears in \(\mu\), we can deduce the fact
  that there exists an evaluation context \(\trg{\evalctx}\) such that \(\trg{ P\triangleright \evalctxt{\call{f}~v} \Xtoltbctx{\mu_{f,v,v'}} \evalctxt{v'}}\).

  From this, we can reason by induction and use \Cref{tr:ev-t-cth} to obtain the result.
\end{proof}

\subsubsection{Backward Compiler Correctness for Programs}
\begin{restatable}[Backward Compiler Correctness]{theorem}{bcc}\label{thm:bcc}
  Let \(\src{P}\) be a source program. Then,
  \[
  \forall \mu, \trg{\comptd{\src{P}}\ \hookrightarrow_{\text{prg}}}\ \mu \implies \src{P \hookrightarrow_{\text{prg}}}\ \mu.
  \]
\end{restatable}

Before proving the theorem, we state a preliminary lemma:
\begin{lemma}
  Suppose that \(\trg{\comptd{\src{P}} \triangleright\ \call{f}~v \Xtoltb{\cl{f}{v};\mu} \comptd{\src{P}} \triangleright e'}\) where the call is well-typed.
  
  Then, \(\trg{\comptd{\src{P}} \triangleright\ \call{f}~v \Xtoltb{\cl{f}{v}} \comptd{\src{P}}\triangleright \comptd{\src{e}}\subt{x}{v}}\) and:
  \begin{itemize}
  \item \(\trg{\comptd{\src{P}} \triangleright\ \comptd{\src{e}}\subt{x}{v} \Xtoltb{\mu} \comptd{\src{P}}\triangleright e'}\),
  \item or, \(\mu = \epsilon\) and \(\trg{\comptd{\src{P}} \triangleright\ \call{f}~v \Xtoltb{\cl{f}{v}} \comptd{\src{P}}\triangleright e'}\)
  \end{itemize}
  where \(\src{e}\) is the code of the function \(f\) in the source program.
\end{lemma}
\begin{proof}
  By induction on \(\trg{\comptd{\src{P}} \triangleright\ \call{f}~v \Xtoltb{\cl{f}{v};\mu} e'}\).
  \begin{description}
  \item[\Cref{tr:tr-t-sin}] In this case, \(\mu = \epsilon\). The result is obtained by direct application of the
    semantics.
  \item[\Cref{tr:tr-t-cons}] There exists \(\mu_1\) and \(\mu_2\) such that \(\mu_1\mu_2 = \mu\) and
    \[\trg{\comptd{\src{P}} \triangleright\ \call{f}~v \Xtoltb{\cl{f}{v};\mu_1} e_1}\] and
    \[\trg{\comptd{\src{P}} \triangleright\ e_1 \Xtoltb{\mu_2} e}.\]
    By applying the induction hypothesis to the first relation, we obtain the result.
  \item[Other cases:] these cases are impossible
  \end{description}
\end{proof}

We can now prove the backward compiler correctness theorem:
\bcc
\begin{proof}
  Let \(\src{P}\) be a source program and \(\mu\) an informative trace. Suppose that
  \(\trg{\comptd{\src{P}}\ \hookrightarrow_{\text{prg}}}\ \mu\), we will prove that \(\src{P \hookrightarrow_{\text{prg}}}\ \mu\).

  Let \(\mu_{f,v,v'} = \cl{f}{v};\mu';\rt{v'}\) be a trace as defined by the source partial semantics. Let us show that
  \[
  \src{P \triangleright \call{f}~v \Xtolsb{\mu_{f,v,v'}} v'},
  \]
  knowing that
  \[
  \trg{\comptd{\src{P}} \triangleright\ \call{f}~v \Xtoltb{\mu_{f,v,v'}} v'}.
  \]

  By the preliminary lemma, and since \(\mu' \neq \epsilon\), we have that
  \[\trg{\comptd{\src{P}} \triangleright\ \call{f}~v \Xtoltb{\cl{f}{v}} \comptd{\src{e}}}\subt{x}{v}\] where
  \(\src{e}\) is the source of \(\src{f}\) in the source program, because the call is well-typed and
  \(\trg{\comptd{\src{P}} \triangleright\ \comptd{\src{e}}\subt{x}{v} \Xtoltb{\mu'; \rt{v'}} v'}\).

  Now, we can conclude by induction on \(\src{e}\).

\end{proof}

\subsubsection{Back-Translation of a Finite Set of Finite Trace Prefixes}

The theorem we wish to prove in this section is the following theorem:
\begin{restatable}{theorem}{backtranslation}
  Let \(\trg{C_T}\) be a target context and \( \{\mu_i\}\) be a finite set of trace prefixes such that
  \(\forall i, \trg{C_T \hookrightarrow_{\text{ctx}}}\ \mu_i\).
  Then, \[
  \exists \src{C_S}, \forall i, \src{C_S \hookrightarrow_{\text{ctx}}}\ \mu_i^s
  \]
\end{restatable}
where the relation between \(\mu_i\) and \(\mu_i^s\) is explicited later.

We will construct a function \(\src{\uparrow}\) such that if \(F\) is a set of finite prefixes, \(F\src{\uparrow}\) is a
source context such that:
\[
\forall \mu \in F, F\src{\uparrow} \hookrightarrow_{\text{ctx}} \mu^s.
\]
where \(\mu^s\), defined later, is the trace \(\mu\) with the possibility of swapping failure and calls events, as described
previously.

We only consider traces that do not have any I/O. Indeed, I/O is produced only by the programs in these
languages, hence do not affect the backtranslation of a source context.
First, we explicit the tree structure that is found in \(F\) by defining the following inductive construction:
\begin{align*}
  T \bnfdef&\ \epsilon \mid \termc \mid \failactc \mid \divrc \\
           \mid &\ (\cl{f}{v}, (v_1, T_1), (v_2, T_2), \dots, (v_i, T_i))
\end{align*}

From a set of trace \(F\), we define a relation \(F \vDash T\) as follow:
\begin{center}
  \typerule{Tree-Empty}{
      F = \emptyset \vee \forall \mu \in F, \mu = \epsilon
  }{F \vDash \epsilon}{tree-empty}
  \typerule{Tree-Term}{
    \forall \mu \in F, \mu \neq \epsilon \implies \mu = \termc
  }{F \vDash \termc}{tree-term}
  \typerule{Tree-Divr}{
    \forall \mu \in F, \mu \neq \epsilon \implies \mu = \divrc
  }{F \vDash \divrc}{tree-div}
  \typerule{Tree-Fail}{
    \forall \mu \in F, \mu \neq \epsilon \implies \mu = \failactc
  }{F \vDash \failactc}{tree-fail}
  \typerule{Tree-Fail-Type}{
    \forall \mu \in F, \mu \neq \epsilon \implies \mu = \cl{f}{v};\mu' \wedge \src{f : \tau \rightarrow \tau'} \wedge v \notin \tau
  }{F \vDash \failactc}{tree-type}
  \typerule{Tree-Call-Ret}{
        \forall i, \exists \mu \in F, \mu = \bl{\cl{f}{v}};\bl{\rt{v_i}};\mu'\\
        \{ \mu' \mid \bl{\cl{f}{v}};\bl{\rt{v_i}};\mu' \in F \} \vDash T_i\\
        \bigcup_{1\leq j \leq i} \{ \bl{\cl{f}{v}};\bl{\rt{v_j}};\mu' \in F \}
        \cup \{\cl{f}{v}; \divrc\} \cup \{\bl{\cl{f}{v}}\} \cup \{ \epsilon \} \supseteq F
  }{F \vDash (\cl{f}{v}, (v_1, T_1), \dots, (v_i, T_i))}{tree-call}
\end{center}

This relation means that the tree \(T\) represents the set of traces \(F\). The first five rules represent the base
cases from the point of view of the context: \Cref{tr:tree-empty} is the case where every trace is empty or there are no trace in \(F\).
\Cref{tr:tree-term} represent the case where all traces terminate. \Cref{tr:tree-div} is a case that should never happen, because the context
should never diverge. \Cref{tr:tree-fail} is the case where all traces fail in the context. \Cref{tr:tree-type} represent the case
where all traces call a function with an incorrect argument and must fail.

The last rule, \Cref{tr:tree-call}, represent the case where some traces may be cut, and the others shall call a function. The next event
must be either divergence, which is ignored because it is part of the program, or a return event.
Then, the remaining traces are separated into groups receiving the same return value: these traces are then considered on their own to
construct subtrees \(T_i\).
The third condition is required to ensure that no trace is forgotten.

The fact that this object is indeed defined is directly derived from the determinacy of the context. Indeed, let \(F\) be a
set of informative traces produced by the same context. They must either be empty, or start by the same event,
by determinacy, and this event has to be a call event. If this call in not correctly typed, then we are in the fifth case.
Otherwise, we are necessarily in the last case, and the \(T_i\) exist by induction.

The back-translation of \(F\) is defined by induction on the tree \(T\) such that \(F \vDash T\):
\begin{definition}[Backtranslation of the tree \(T\)]
  \[
  T\!\uparrow =
  \begin{cases}
    \fails &\text{if \(T = \epsilon\) or \(T = \failactc\)}\\
    \src{0} &\text{if \(T = \termc\)}\\
    \fails &\text{if \(T = \divrc\)}\\
    \src{\letin{x}{\call{f}~v}{\left(\ 
        \begin{minipage}{0.4\textwidth}\(
          \src{{if}~x = v_1~{then}~T_1\!\uparrow~}\\
          \src{{else}~{if}~x=v_2~{then}~\dots}\\
          \src{{else}~{if}~x=v_i~{then}~T_i\!\uparrow{else}~\fails}
        \)\end{minipage}\!\right)
    }} &\
    \begin{minipage}{0.4\textwidth}
      \begin{flushleft}
      if \( T = \left(\bl{\cl{f}{v}}, (v_1, T_1), \dots, (v_i, T_i)\right)\) and
      \(\src{f : \tau \rightarrow \tau'}\) and \( v \in \tau\)
      \end{flushleft}
    \end{minipage}\\
    \fails &\text{otherwise}
  \end{cases}
  \]
\end{definition}

\begin{lemma}
  The back-translation of a set of traces \(F\) generated by a single context is well-typed and linkable.
\end{lemma}
\begin{proof}
  By induction on the relation \(F \vDash T\).
\end{proof}

We define what it means for a trace to be ``part'' of such a tree:
\begin{definition}[Trace extract from a tree]
  We say that a trace \(\mu\) is extracted from a tree \(T\) if:
  \begin{enumerate}
  \item \(\mu = \epsilon \)
  \item \(\mu = \termc \) and \( T = \termc \)
  \item \(\mu = \failactc \) and \( T = \failactc \)
  \item \(\mu = \bl{\cl{f}{v}} :: \epsilon \), \( \typeop(v) \neq \inputtypeop(f) \) and \( T = \failactc \)
  \item \(\mu = \bl{\cl{f}{v}} :: \failactc \), \( \typeop(v) \neq \inputtypeop(f) \) and \( T = \failactc \)
  \item \(\mu = \bl{\cl{f}{v}} :: \epsilon \) or \(\mu = \cl{f}{v}; \divrc\), \(T = \left(\bl{\cl{f}{v}}, \dots\right)\) and
    \( \typeop(v) = \inputtypeop(f) \)
  \item \(\mu = \bl{\cl{f}{v}} :: \bl{\rt{v'}} :: \mu'\),
    \(T = \left(\bl{\cl{f}{v}}, (v_1, T_1), \dots, (v_i, T_i)\right)\),
    and \( \exists j,\) such that \(v_j = v' \) and \(\mu'\) is extracted from \(T_j\)
  \item \(\mu = \bl{\cl{f}{v}} :: \epsilon \) or \(\mu = \cl{f}{v}; \failactc\), \(T = \failactc\) and
    \( \typeop(v) \neq \inputtypeop(f) \)
  \end{enumerate}
\end{definition}

We are going to prove that any such trace extracted from a tree can be produced by the back-translated context, modulo the behaviors allowed at the target level but not at the source level.

\begin{definition}
  \[
  \mu^s = \begin{cases}
    \mu'\failactc\ &\text{if \( \mu = \mu'\bl{\cl{f}{v}} \) such that \(\inputtypeop(f) \neq \typeop(v)\)}\\
    \mu'\failactc\ &\text{if \( \mu = \mu'\bl{\cl{f}{v}}\failactc \) such that \(\inputtypeop(f) \neq \typeop(v)\)}\\
    \mu\ &\text{otherwise}
  \end{cases}
  \]
\end{definition}

\begin{theorem}[Correction of the backtranslation]\label{thm:backtranslation}
  Let \(T\) be a tree and \(\mu\) a trace extracted from \(T\).
  Then, \(T\bttrace \leadsto \mu^s\).
\end{theorem}
\begin{proof}
  We are going to prove by induction on the relation ``\(\mu\) is extracted from \(T\)'' that there exists \(\src{e}\) such that \(T\bttrace \Xtolsb{\mu^s} \src{e}\).
  \begin{enumerate}
  \item \(\mu = \epsilon \): OK.
  \item \(\mu = \termc \) and \( T = \termc \): \(T\bttrace = 0\). OK.
  \item \(\mu = \failactc \) and \( T = \failactc \): \(T\bttrace = \src{\fails}\). OK.
  \item \(\mu = \bl{\cl{f}{v}} ; \epsilon \), \( \typeop(v) \neq \inputtypeop(f) \) and \( T = \failactc \) We are in the first case for \(\mu^s\): OK.
  \item \(\mu = \bl{\cl{f}{v}} ; \failactc \), \( \typeop(v) \neq \inputtypeop(f) \) and \( T = \failactc \) We are in the second case for \(\mu^s\): OK.
  \item \(\mu = \bl{\cl{f}{v}} ; \epsilon \), \(T = \left(\bl{\cl{f}{v}}, \dots\right)\) and
    \( \typeop(v) = \inputtypeop(f) \):
    \(T\bttrace = \src{\letin{x}{\call{f}~v}{\dots}} \). OK.
    Idem with \(\divrc\) instead of \(\epsilon\).
  \item \(t = \bl{\cl{f}{v}} ; \bl{\rt{v'}} ; \mu'\),
    \(T = \left(\bl{\cl{f}{v}}, (v_1, T_1), \dots, (v_i, T_i)\right)\),
    and \( \exists j,\) such that \(v_j = v' \) and \(\mu'\) is extracted from \(T_j\):
    Then: \[
    T\bttrace = \src{\letin{x}{\call{f}~v}{\ifte{\dots}{\ifte{x = v_j}{T_j\bttrace}{\dots}}{\dots}}}.
    \]

    By application of the partial semantics: \[
    T\bttrace \Xtosbctx{\bl{\cl{f}{v}}; \bl{\rt{v_j}}} \src{\ifte{x=v_j}{T_j\bttrace}{\dots}\subs{v_j}{x}}
    \]
    and therefore by substituting and application of the partial semantics: \[
    T\bttrace \Xtosbctx{\bl{\cl{f}{v}};{\bl{\rt{v_j}}}} T_j\bttrace.
    \]
    By induction hypothesis, we are done.
  \item \(\mu = \bl{\cl{f}{v}} :: \epsilon \) or \(\mu = \cl{f}{v}; \failactc\), \(T = \failactc\) and
    \( \typeop(v) \neq \inputtypeop(f) \). The result is immediate
  \end{enumerate}
\end{proof}

Now, we can prove that any of the initial traces that are used to construct the tree can be found in this tree,
and then the theorem applies to them.
\begin{lemma}
  Let \(F\) be a set of traces and \(T\) such that \(F\vDash T\).
  Then, any trace \(\mu \in F\) is extracted from the tree \(T\).
\end{lemma}
\begin{proof}
  Let us prove by induction on \(T\) that if there exists \(F\) such that \(T = T(F)\), then \( \forall \mu \in F\), \(\mu\) is extracted
  from \(T\). Since the trace \(\epsilon\) is always extracted from any tree, we ignore this case.
  \begin{description}
  \item[\( T = \epsilon\):] OK.
  \item[\( T = \termc\):] Then \(\mu = \termc\). OK.
  \item[\( T = \divrc\):] Then \(\mu = \divrc\). OK.
  \item[\( T = (\bl{\cl{f}{v}}, (v_1, T_1), \dots, (v_i, T_i)) \):] By induction hypothesis.
  \end{description}
\end{proof}

\subsubsection{Composition}

The composition theorem states that if a context and a program can partially produce two related informative
traces, then plugging the program into the context gives a whole program that can produce
one of the traces.
The relation between the two traces captures the fact that the way things fail in the source is not the same as in the target,
as seen in the back-translation section.
The theorem is stated as follows:
\begin{restatable}[Composition]{theorem}{composition}\label{thm:comp}
  Let \(\src{C_S}\) be a source context, \(\src{P_S}\) be a source program, \(\mu_i \sim \mu_i^s\) two related traces, and suppose \(\src{P_S \bowtie C_S}\).
  Then, if \(\src{C_S\hookrightarrow_{\text{ctx}}}\ \mu_i^s\) and \(\src{P_S\hookrightarrow_{\text{prg}}}\ \mu_i\),
  then \(\src{C_S\left[P_S\right] \hookrightarrow}\ \mu_i^s\).
\end{restatable}

We state a preliminary lemma:
\begin{lemma}
  If \(\src{P \hookrightarrow_{\text{prg}}}\ \mu_i\), then \(\src{P \hookrightarrow_{\text{prg}}}\ \mu_i^s\).
\end{lemma}
\begin{proof}
  This is by definition of \(\mu_i^s\).
\end{proof}

\begin{lemma}
  Let \(\src{C_S}\) be a source context, \(\src{P_S}\) be a source program, \(\mu_i \sim \mu_i^s\) two related traces such that \(\mu_i\) was produced by \(\comptd{P_S}\) and some target context,
  and \(\src{e}\) an expression.
  Then, if \(\src{C_S \Xtolsb{\mu_i^s} e}\) and \(\src{P_s \hookrightarrow_{\text{prg}}}\ \mu_i^s\),
  then \(\src{P_S\triangleright C_S \Xtolsb{\mu_i^s} e'}\) where \(\src{C_S \Xtolsb{\mu_i^s} e'}\).
\end{lemma}
\begin{proof}
  We will prove by induction \(n\) that
  \(\forall n, \forall \mu, |\mu| = n, \forall \src{e}, \forall \src{P_S},
  \src{e \hookrightarrow_{\text{ctx}}}\ \mu \wedge \src{P_S \hookrightarrow_{\text{prg}}}\ \mu
  \implies \exists \src{e'}, \src{e \Xtolsb{\mu} e'} \wedge \src{P_S \triangleright e \Xtolsb{\mu} e'}
  \)
  \begin{description}
  \item[Base case] If \(n = 0\), this is trivially true.
  \item[Inductive case] Let \(n\in\mathbb{N}\), \(\mu\) of length \(n\), \(\src{e}\) and \(\src{P_S}\) such that
    \(\src{e \hookrightarrow_{\text{ctx}}}\ \mu\) and \(\src{P_S \hookrightarrow_{\text{prg}}}\ \mu\).

    We consider only one case, but the other cases are similar:
    \[ \mu = \mu_1 \mu_2 \mu_3 \] where \(\mu_2 = \cl{f}{v} \mu_2' \rt{v'}\) is defined as in the
    definition of \(\src{\hookrightarrow_{\text{prg}}}\).
    \begin{itemize}
    \item First, \(\src{e \hookrightarrow_{\text{ctx}}}\ \mu_1\) and
      \(\src{e \hookrightarrow_{\text{prg}}}\ \mu_1\), by definition of these relations. Therefore, by
      induction hypothesis, \(\exists \src{e'}, \src{e \Xtolsb{\mu_1} e'}\) and
      \(\src{P_S\triangleright e \Xtolsb{\mu_1} e'}\).
      In particular, \(\src{e'}\) is of the form \(\evalctxs{\call{f}~v}\) by determinism of the
      execution of the context (since the read/writes are set by the trace),
      such that \(\src{e' \hookrightarrow_{\text{ctx}} \mu_2\mu_3}\).
    \item We have that \(\src{P_S \triangleright e' \Xtolsb{\mu_2} \evalctxs{v'}}\) by definition
      of the partial semantics for programs, and the rules of evaluations inside contexts.
    \item We can again apply the induction hypothesis to \(\mu_3\).
    \end{itemize}
    Hence, we obtain the result: \(\src{P_S \triangleright e \Xtolsb{\mu_1\mu_2\mu_3} e''}\) where
    \(\src{e \Xtolsb{\mu_1\mu_2\mu_3} e''}\).
  \end{description}

\end{proof}

By using these two lemmas, we can prove the composition theorem.

\subsubsection{Back to Non-Informative Traces}
The last step of the proof is to go back to the non-informative trace model. In particular,
we must take into account that the trace \(\mu_i^s\) that is generated by the whole program
is not exactly equal to the original trace \(\mu_i\).

\begin{restatable}[Back to non-informative traces]{theorem}{noninf}\label{thm:noninf}
  Let \(\src{C_S}\) be a source context, \(\src{P_S}\) be a source program, \(m\) a non-informative trace and
  \(\mu\) an informative trace such that \(\mu \sqsupseteq m\).

  Then, \( \src{C_S\left[P_S\right] \hookrightarrow}\ \mu^s \implies \src{C_S\left[P_S\right] \leadsto}\ m\).
\end{restatable}
The proof is immediate by definition of \(\mu^s\).

\subsubsection{Proving the Secure Compilation Criterion}


The proof follows the scheme depicted by \Cref{hlo:proof-technique}.
\begin{proof}
Let \(\src{P}_1\dots \src{P}_k\) be \(k\) programs and \(m_1\dots m_k\) be
\(k\) finite trace prefixes. Let \(\trg{C_T}\) be a target context and suppose the following holds:
\[
  \forall i, \trg{\plug{C_T}{\compgen{\src{P_{\com{i}}}}} \leadsto}\ m_i
\]

We can pass to informative traces by applying Theorem~\ref{thm:isem} to each \(m_i\)
\[
\forall i, \exists \mu_i \sqsupseteq m, \trg{\plug{C_T}{\compgen{\src{P_{\com{i}}}}} \hookrightarrow} \mu_i.
\]

From here, we can apply the decomposition theorem (Theorem~\ref{thm:dec}) to each \(\mu_i\):
\[
\forall i, \trg{C_T \hookrightarrow_{\text{ctx}}}\ \mu_i \wedge
\trg{\compgen{\src{P_{\com{i}}}} \hookrightarrow_{\text{prg}}}\ \mu_i.
\]

By the backward compiler correctness theorem (Theorem~\ref{thm:bcc}) for programs applied to each program, we obtain that:
\[
\forall i, \src{P_{\com{i}} \hookrightarrow_{\text{prg}}}\ \mu_i.
\]

Also, by applying the back-translation theorem, we can produce a source context:
\[
\exists \src{C_S}, \forall i, \src{C_S \hookrightarrow_{\text{ctx}}}\ \mu_i^s.
\]
Moreover, this \src{C_S} is well-typed and linkable with the \(\src{P_i}\).

Now, we are able to apply the composition theorem (Theorem~\ref{thm:comp}) to each program:
\[
\forall i, \src{\plug{C_S}{P_{\com{i}}} \hookrightarrow}\ \mu_i^s
\]

Finally, we can go back to the non-informative traces by the last theorem (Theorem~\ref{thm:noninf}):
\[
\forall i, \src{\plug{C_S}{P_{\com{i}}} \leadsto}\ m_i.
\]
\end{proof}

\paragraph{Remarks on the proof technique}
This proof technique should be fairly generic and could be adapted to other languages. if needed, it is possible to
change the top-level statement by introducing
a more complex relation between source and target, that could for instance model the exchange between failure and calls
that might happen in our instance, or to model non-determinism in a non-deterministic language.
While decomposition and composition are natural properties that we expect
to hold for most languages, and while backward correctness can reasonably be expected from a secure compiler, the
back-translation seems to be the hardest part of the proof and the most subject to change between languages.

\ifanon
\newcommand{\rtcone}{\ensuremath{\pf{\rtp}^{\sim}}\xspace}
\newcommand{\sigmartp}{\ensuremath{\rtp^{\sigma}}\xspace}
\newcommand{\taurtp}{\ensuremath{\rtp^{\tau}}\xspace}

\newcommand{\tildesigmartp}{\ensuremath{\rtp^{\tilde{\sigma}}}\xspace} 
\newcommand{\tildetaurtp}{\ensuremath{\rtp^{\tilde{\tau}}}\xspace} 
\newcommand{\weaktilde}{\ensuremath{\rtp^{\tilde{\sigma}}_{\ii{weak}}}\xspace}

\newcommand{\simst}{\ensuremath{\sim_{\tau}^{\sigma}}\xspace}
\newcommand{\rtctwo}{\ensuremath{\pf{\rtp}^{\simst} }\xspace}

\newcommand{\eundef}{\ensuremath{\mathcal{U}\ii{ndef}}\xspace}

\section{Differing Source and Target Trace Models}
\label{sec:different-traces}


So far we have assumed that the source and target languages share the
sets of traces and properties.
This is a reasonable assumption if observable events are coarse enough
to be common to the source and target language (\EG as it happens
CompCert~\cite{Leroy09}).
There are various settings in which this assumption does not hold though:
\begin{itemize}
\item when values appear in traces and the source and target don't
  share the same notion of values~\cite{PatrignaniG17};
\item when the source language has ``undefined behavior'', which is a
  special event only occurring in source traces, and which leads to
  an arbitrary trace continuation in the target~\cite{Leroy09};
\item when the target language has more events because it is the
  target of multiple source languages;
\item when traces model observations that could become more precise at
  lower levels, \EG because of side-channels like
  timing~\cite{CostanzoSG16};\ch{But with traces formed of such
    leakages even compiler correctness is hopeless~\cite{BartheGL18},
    even more our extensions of it to adversarial contexts!  So maybe
    this is not such a great example after all. The only kind of work
    that could fit here is the one by Zhong Shao~\cite{CostanzoSG16},
    where the source and target observations are different, but still
    equivalent. So started watering this down by adding a ``could'' above.}
\item in general, when there is a big difference between the
  abstractions of the source and target languages.
\end{itemize}


%
\ch{I find this bijection story a confusing digression,
  but would still like to understand why would it be relevant here}%
Using standard set-theoretical arguments we know that, if the two sets of events are both finite or in a bijection,
then so are the sets of traces and properties. In general such a bijection between traces is not constructive,
so that it is interesting to study variants of our criteria handling the discrepancies between trace events.  
In this section we focus on \rtp and \pf{\rtp}, leaving a generalization to the remaining criteria as future work. 


In this section, typographic conventions for source and target elements extend
to traces and properties, previously common to both levels.
From a security point of view, \pf{\rtp} captures the fact that if
some $\trg{C_T}$ can mount an attack $\trg{t_T}$ in the target
language, then the same attack is possible in the source.
A natural generalization of \pf{\rtp} requires that the target attack
$\trg{t_T}$ can be simulated in the source by producing a source trace
$\src{t_S}$ that is {\em related} to the target trace $\trg{t_T}$.


\begin{definition}[\rtcone] \label{defn:rtcone}
Given a relation between source and target traces, $\sim \,\subseteq \src{\ii{Trace}_S} \times \trg{\ii{Trace}_T}$, we define   
\begin{align*}
  \rtcone: \quad\ \forall\src{P}.~ \forall\trg{C_T}.~ \forall \trg{t_T}.~
        \mathrel{\trg{C_T\hole{\cmp{P}} \sem}} \trg{t_T} \Rightarrow
        \mathrel{\exists\src{C_S}\ldotp \exists\src{t_S} \sim \trg{t_T} \ldotp\src{C_S\hole{P}\sem}} \src{t_S}
\end{align*}
\end{definition}


\smallskip

As $\pf{\rtp}$, the statement of \rtcone is ``property-free''. 
In order to bring properties into the picture,
two important questions arise:
(1) what $\trg{\pi_T}$ will robustly satisfied in the
target if we have established (\EG by verification) that
some $\src{\pi_S}$ is robustly satisfied in the source, and dually
(2) what is the $\src{\pi_S}$ that must be robustly satisfied in
the source in order to guarantee the robust satisfaction of
a fixed $\trg{\pi_T}$ we want to obtain in the target.
Reflecting this duality, we therefore provide two different
generalizations of $\rtp$ to different source and target traces:

\begin{definition}[\taurtp and \sigmartp{}] \label{defn:RTP}
\  Given an arbitrary pair of mappings between source and target properties,  $\tau :  2^\src{\ii{Trace}_S} \rightarrow 2^\trg{\ii{Trace}_T}$, 
  $~\sigma : 2^\trg{\ii{Trace}_T} \rightarrow 2^\src{\ii{Trace}_S} $ we define:
%
%

  \begin{align*}
    \taurtp: \quad
    \forall \src{\pi_S} \in 2^\src{\ii{Trace}_S}.~\forall\src{P}.~
    &
      (\forall\src{C_S}~\src{t_S} \ldotp
      \mathrel{\src{C_S\hole{P} \sem}} \src{t_S} \Rightarrow \src{t_S}\in\src{\pi_S})
      \Rightarrow
    \\
    &
      (\forall \trg{C_T}~\trg{t_T} \ldotp
      \mathrel{\trg{C_T\hole{\cmp{P}} \sem}} \trg{t_T} \Rightarrow \trg{t_T}\in\tau(\src{\pi_S}))
    \\ 
    \\
    \sigmartp: \quad
    \forall \trg{\pi_T} \in 2^\trg{\ii{Trace}_T}.~\forall\src{P}.~
    &
      (\forall\src{C_S}~ \src{t_S} \ldotp
      \mathrel{\src{C_S\hole{P} \sem}} \src{t_S} \Rightarrow \src{t_S}\in \sigma(\trg{\pi_T}))
      \Rightarrow
    \\
    &
      (\forall \trg{C_T}~\trg{t_T} \ldotp
      \mathrel{\trg{C_T\hole{\cmp{P}} \sem}} \trg{t_T} \Rightarrow \trg{t_T}\in\trg{\pi_T})
\end{align*}
\end{definition}
%
\bigskip
%
The mapping $\tau$ fixes an interpretation of source properties in the
target language, while $\sigma$ fixes an interpretation of target
properties in the source language.
For example, $\tau$ tells us how to 
interpret the source property \emph{``The program does not encounter an undefined behavior"} in a target 
language with no symbol for undefined behavior in it events alphabet.

\ch{We are still badly missing intuition why a Gallois connection is a
  reasonable thing to require/expect here. Simply making the proof
  below work does not count as intuition! Also we don't yet have an
  explanation why this Gallois connection goes in the opposite
  direction compared to what the intuition would say.}

We are going to show, in \Cref{thm:galois}, that certain conditions on $\tau$ and $\sigma$ provide the equivalence 
of \taurtp and \sigmartp. First of all let us recall the definition and a useful characterization of \emph{Galois connections} \cite{CousotCousot77-1}.

\begin{definition}[Galois connection] Let $(X, \leq)$ and $(Y, \sqsubseteq)$ be two posets.
  A pair of maps, $\alpha: X \to Y$, $\gamma: Y \to X$ is a Galois connection \ii{iff}
  it satisfies the following \emph{adjunction law}:
                   \begin{equation*}
                     \forall x \in X. ~ \forall y \in Y. ~ \alpha(x) \sqsubseteq y \iff x \leq \gamma(y)
                   \end{equation*}
$\alpha$ is referred to as the \emph{lower adjoint} and $\gamma$ as the \emph{upper adjoint}. 
We will often write  $\alpha: (X, \leq) \leftrightarrows (Y, \sqsubseteq) :\gamma$ to denote a Galois connection.
\end{definition}

\begin{lemma}[Characteristic property of Galois connections] \label{lem:chG}
  $\alpha: (X, \leq) \leftrightarrows (Y, \sqsubseteq) :\gamma$ is a Galois connection if and only if $\alpha, \gamma$ are monotone and
  they satisfy the following properties:
  \begin{align*}
   i)& \quad \forall x \in X.  ~ x \leq \gamma(\alpha (x)) \\ 
   ii)& \quad \forall y \in Y. ~ \alpha(\gamma (y)) \sqsubseteq y
  \end{align*}
\end{lemma}
\bigskip
\ch{zero intuition here}%
It turns out that if $\tau \leftrightarrows \sigma$ is a Galois connection, the two
criteria in \Cref{defn:RTP} are equivalent.
\ca{TODO: not sure if I fully integrated RB comments}

\begin{theorem}[\taurtp and \sigmartp coincide] \label{thm:galois}

 If  $\tau :  (2^\src{\ii{Trace}_S}, \subseteq) \rightleftarrows (2^\trg{\ii{Trace}_T}, \subseteq) : \sigma $ 
 is a Galois connection, with $\tau$ its lower adjoint and $\sigma$ its upper adjoint, then
 $ \taurtp \iff \sigmartp $.

\end{theorem}

\begin{proof}
  As a preliminary remark observe that, both in the source and target,
  if a program robustly satisfies a property $\pi$, then it robustly satisfies every extension 
  $\pi' \supseteq \pi$. Using this we prove the claim of the theorem:
\ch{Please make it more explicit below where the remark above
  and the precise conditions of the Galois connection are used:}

  ($\Rightarrow$)\quad Assume \taurtp and that $\src{P}$ robustly satisfies $\sigma(\trg{\pi_T})$. Apply
  \taurtp to $\src{P}$ and $\sigma(\trg{\pi_T})$ and deduce that $\cmp{P}$ robustly satisfies 
  $\tau(\sigma (\trg{\pi_T})) \subseteq \trg{\pi_T}$.

  ($\Leftarrow$)\quad Assume \sigmartp and that $\src{P}$ robustly 
  satisfies $\src{\pi_S} \subseteq \sigma( \tau (\src{\pi_S} ) )$. Apply \sigmartp to $\src{P}$ 
  and $\sigma( \tau (\src{\pi_S} ) )$ deducing $\cmp{P}$ that robustly satisfies $\tau (\src{\pi_S} )$.
\end{proof}

\ch{Does the other direction of the theorem above hold?  Can we obtain
  that $\tau$ and $\sigma $ form a Galois connection from a proof that
  $ \taurtp \iff \sigmartp $?}
 
So far we discussed how to generalize the criterion for the
preservation of robust satisfaction of all trace properties, $\rtp$.
We showed that it is possible to start either from the property-free
characterization $\pf{\rtp}$ and a relation between source and target
traces (\Cref{defn:rtcone}), or from $\rtp$ and an interpretation
of properties of one language into the other (\Cref{defn:RTP}).
Below we investigate the relation between \rtcone and
\taurtp/\sigmartp.  In \Cref{sec:tilde} we define two natural mappings
$\tilde{\tau}$, $\tilde{\sigma}$ starting from a given relation $\sim$
and discuss under which conditions the three criteria presented are
all equivalent.
These extra conditions may restrict the class of target properties for
which a one can use source-level reasoning, even when a secure
compilation chain is available.
In \Cref{sec:connection} we start from a Galois connection
$\tau \rightleftarrows \sigma$ and define a relation that ensures the
equivalence.

\subsection{Deriving Property Mappings from a Relation} \label{sec:tilde}

In this section we show how a relation $\sim$ between source and target traces induces a pair of mappings from properties of one language
to properties of the other language. We specialize \Cref{defn:RTP} with such interpretation and show that one of the criteria is
equivalent to $\rtcone$ with no extra assumption (\Cref{thm:rtprtc}), while in general the other one is stronger (\Cref{lem:left}).   

\begin{definition}[Induced mappings] \label{defn:tilde} 

  Let  $\sim \,\subseteq \src{\ii{Trace}_S} \times \trg{\ii{Trace}_T}$

  \begin{align*}
    \tilde{\sigma} & = \lambda ~\trg{\pi_T}. ~ \myset{ \src{t_S} }{ \exists \trg{t_T} \in \trg{\pi_T}. ~ \src{t_S} \sim \trg{t_T} } \\ 
    \tilde{\tau}   & = \lambda ~\src{\pi_S}. ~ \myset{ \trg{t_T} }{ \exists \src{t_S} \in \src{\pi_S}. ~ \src{t_S} \sim \trg{t_T} }             
  \end{align*}

We denote with {\tildetaurtp}, \tildesigmartp the criteria in
\Cref{defn:RTP} for $\tilde{\tau}$ and $\tilde{\sigma}$ respectively.

\end{definition}

\begin{theorem}[\tildetaurtp and \rtcone always coincide] \label{thm:rtprtc} 
  Let  $\sim \,\subseteq \src{\ii{Trace}_S} \times \trg{\ii{Trace}_T}$, then $\tildetaurtp \iff \rtcone$.
\end{theorem}


\begin{proof}
($\Rightarrow$)\quad Assume \tildetaurtp in contrapositive form and assume $\mathrel{\trg{C_T\hole{\cmp{P}} \sem}} \trg{t_T}$.
We can apply \tildetaurtp to $\src{P}$ and $\src{\pi_S} \equiv \myset{\src{t_S}}{\src{t_S} \not\sim \trg{t_T}}$.
If it is non-empty, $\tilde{\tau}(\src{\pi_S})$ contains $\trg{t_T}$,
in both cases $\cmp{P}$ does not robustly satisfy it. 
We therefore deduce $\exists \src{C_S}\exists \src{t_S}. ~\mathrel{\src{C_S\hole{P} \sem}} \src{t_S} $ and $\src{t_S} \notin \src{\pi_S}$, from which $\src{t_S} \sim \trg{t_T}$.

($\Leftarrow$)\quad Assume \rtcone and that for $\src{P}, \src{\pi_S}$ there exists $\trg{C_T}$ 
such that $\mathrel{\trg{C_T\hole{\cmp{P}} \sem}} \trg{t_T}$ with $\trg{t_T} \notin \tilde{\tau}(\src{\pi_S})$.
Apply \rtcone and deduce $\mathrel{\exists\src{C_S}\ldotp \exists\src{t_S} \sim \trg{t_T} \ldotp\src{C_S\hole{P}\sem}} \src{t_S}$. If $\src{t_S} \in \src{\pi_S}$,
then $\trg{t_T} \in \tilde{\tau}(\src{\pi_S})$, which is a contradiction and therefore $\src{t_S} \notin \src{\pi_S}$.
\end{proof}

Note that \Cref{thm:rtprtc} does not require any extra assumptions. 
On the other hand, in general \tildesigmartp and \rtcone are not equivalent (as shown by \Cref{ex:length} below),
although the equivalence can be shown under some extra conditions,
which are subsumed by the adjunction law of the Galois connection.\ch{Is one
  of the conditions in \autoref{lem:chG} enough?}
 
\begin{lemma}[\tildesigmartp stronger than \rtcone]\label{lem:left} 
 Let  $\sim \,\subseteq \src{\ii{Trace}_S} \times \trg{\ii{Trace}_T}$ be such that $ \forall \src{t_S}. ~\exists \trg{t_T}. ~ \src{t_S} \sim \trg{t_T}$,
 then \sigmartp $\Rightarrow$ \rtcone.
\end{lemma}

\begin{proof}
  Assume  $\mathrel{\trg{C_T\hole{\cmp{P}} \sem}} \trg{t_T}$, and consider 
  the target property $\trg{\pi_T} = \myset{\trg{t_T'}}{\trg{t_T'} \neq \trg{t_T}}$. By \sigmartp deduce 
  $\mathrel{\src{C_S\hole{P} \sem}} \src{t_S}$ for some $\src{C_S}$ and some $\src{t_S} \not\in \tilde{\sigma}(\trg{\pi_T})$. 
  This means that $\src{t_S}$ is not related to target traces different from $\trg{t_T}$, and from our hypothesis we deduce $\src{t_S} \sim \trg{t_T}$,
  which concludes the proof. 
\end{proof}

We introduce the following, weaker variant of \tildesigmartp, which will be shown to follow from \rtcone.

\begin{definition}[Weak $\tildesigmartp$] \label{def:weaker}

  Let  $\sim \,\subseteq \src{\ii{Trace}_S} \times \trg{\ii{Trace}_T}$,
  
  \begin{align*}
    \weaktilde \equiv \quad
    \forall \trg{\pi_T} \in \mathcal{G}.~\forall\src{P}.~
    &
      (\forall\src{C_S}~ \src{t_S} \ldotp
      \mathrel{\src{C_S\hole{P} \sem}} \src{t_S} \Rightarrow \src{t_S}\in \sigma(\trg{\pi_T}))
      \Rightarrow
    \\
    &
      (\forall \trg{C_T}~\trg{t_T} \ldotp
      \mathrel{\trg{C_T\hole{\cmp{P}} \sem}} \trg{t_T} \Rightarrow \trg{t_T}\in\trg{\pi_T}) 
  \end{align*}

\bigskip

where  $ \mathcal{G} = \{ \trg{\pi_T} | ~ \forall \trg{t_1} \trg{t_2} \src{t_S}. ~ 
                                \src{t_S} \sim \trg{t_1} \wedge \src{t_S} \sim \trg{t_2} \Rightarrow 
                                (\trg{t_1} \in \trg{\pi_T} \iff \trg{t_2} \in \trg{\pi_T}) \} $

\end{definition}
%
\bigskip
It is possible to show, by straightforward manipulations, that the class $\mathcal{G}$ coincides with the class of all 
properties  $\trg{\pi_T}$ such that $\tilde{\tau} (\tilde{\sigma} (\trg{\pi_T})) \subseteq \trg{\pi_T}$. 

\begin{lemma}[\rtcone stronger than \weaktilde] \label{lem:right}
  Let  $\sim \,\subseteq \src{\ii{Trace}_S} \times \trg{\ii{Trace}_T}$, \rtcone $\Rightarrow$ \weaktilde.
\end{lemma}

\begin{proof}
 We prove the contrapositive of {\weaktilde}. Assume $\mathrel{\trg{C_T\hole{\cmp{P}} \sem}} \trg{t_T}$ with $\trg{t_T} \not\in \trg{\pi_T}$, 
 for an arbitrary $\trg{\pi_T} \in \mathcal{G}$. Applying \rtcone we deduce $\mathrel{\src{C_S\hole{P} \sem}} \src{t_S}$
 for some $\src{C_S}$ and some $\src{t_S} \sim \trg{t_T}$.  We just need to show that $\src{t_S} \not\in \tilde{\sigma}(\trg{\pi_T})$.  
 If $\src{t_S} \in \tilde{\sigma}(\trg{\pi_T})$ 
 then  $\trg{t_T} \in \tilde{\tau}(\{ \src{t_S} \}) \subseteq \tilde{\tau}(\tilde{\sigma}(\trg{\pi_T})) \subseteq \trg{\pi_T}$, 
 a contradiction.
\end{proof}

When the two maps form a Galois connection it is possible to show that every source trace is related to some target trace,
and that $\mathcal{G} = 2^\trg{\ii{Trace}_T}$, so that \tildesigmartp and its weaker variant coincide. 

\begin{corollary}[\tildetaurtp, \rtcone and \tildesigmartp all coincide]
   If  $\tilde{\tau} \rightleftarrows \tilde{\sigma}$ is a Galois connection then 
   $\tildetaurtp \iff \rtcone \iff \tildesigmartp$.
\end{corollary} 

\begin{proof} 
  The thesis\ch{??? please be explicit what!}
  follows from \Cref{thm:galois} and \Cref{thm:rtprtc}, we also propose a direct proof
  that shows the adjunction law ensures that $\forall \src{t_S}. ~ \exists \trg{t_T}. ~ \src{t_S} \sim \trg{t_T}$, 
  and that $\mathcal{G} = 2^{\trg{\ii{Trace}_T}}$.

  Let $\src{t_S}$ be a source trace, trivially $\tilde{\tau}(\{ \src{t_S} \}) \subseteq \tilde{\tau}(\{ \src{t_S} \})$
  so that by the adjunction law we deduce $ \{ \src{t_S} \} \subseteq \tilde{\sigma}(\tilde{\tau}(\{ \src{t_S} \}))$, 
  that means there exists some $\trg{t_T}$ such that $\src{t_S} \sim \trg{t_T}$.

  Finally the characteristic property of Galois connections in \Cref{lem:chG}, 
  $\forall \trg{\pi_T}. ~ \tilde{\tau} (\tilde{\sigma} (\trg{\pi_T})) \subseteq \trg{\pi_T}$.
  Thesis\ch{????? please be explicit what!}
  hence follows  from \Cref{lem:left} and \Cref{lem:right}. 
\end{proof}


\smallskip

\subsection{Deriving a Relation from Maps} \label{sec:connection}


In this section, given a Galois connection between source and target
properties we define a relation over source and target traces ensuring the
equivalence between \rtcone and \sigmartp and \taurtp.

\begin{definition}[Induced relation] \label{defn:relation}
  Let $\tau :  2^\src{\ii{Trace}_S} \rightarrow 2^\trg{\ii{Trace}_T}$ and  $~\sigma : 2^\trg{\ii{Trace}_T} \rightarrow 2^\src{\ii{Trace}_S}$, define 
   $\simst \,\subseteq \src{\ii{Trace}_S} \times \trg{\ii{Trace}_T}$ as following

  \begin{equation*}
    \src{t_S} \simst \trg{t_T} \iff \trg{t_T} \in \tau(\{ \src{t_S} \})
  \end{equation*}
  
\end{definition}

\medskip

\begin{theorem}[\taurtp, \rtctwo and \sigmartp all coincide]
  If  $\tau \rightleftarrows \sigma$ is a Galois connection then  
          $$ \taurtp \iff  \rtctwo \iff \sigmartp $$
\end{theorem}


\begin{proof}
  It suffices to show the equivalence between \rtctwo and \sigmartp, then by \Cref{thm:galois}
  we get the other equivalence.

  ($\Rightarrow$)\quad
  Assume \rtctwo and assume that $\mathrel{\trg{C_T\hole{\cmp{P}} \sem}} \trg{t_T}$ and
  $\trg{t_T} \notin \trg{\pi_T}$ for some $\src{P}$ and some $\trg{\pi_T}$. 
  By \rtctwo deduce $\mathrel{\src{C_S\hole{P} \sem}} \src{t_S}$ for some $\src{C_S}$ and 
  some $\src{t_S}$ such that $\trg{t_T} \in \tau(\{ \src{t_S} \} )$.
  It now suffices to show that $\src{t_S} \notin \sigma(\{ \trg{\pi_T}\})$. 
  If $\src{t_S} \in \sigma(\{ \trg{\pi_T}\})$, by the adjunction law we get 
  $\tau(\{ \src{t_S} \}) \subseteq \trg{\pi_T}$ and hence $\trg{t_T} \in \trg{\pi_T}$, a contradiction.
  
  ($\Leftarrow$)\quad
  Assume \sigmartp and $\mathrel{\trg{C_T\hole{\cmp{P}} \sem}} \trg{t_T}$
  for some $\src{P}$ and some $\trg{t_T}$.
  By \sigmartp applied to  $\trg{\pi_T} \equiv \{ \trg{t'} | ~\trg{t'} \neq \trg{t_T} \}$
  deduce the existence of some $\src{C_S}$ such that $\mathrel{\src{C_S\hole{P} \sem} \src{t_S}}$
  and $\src{t_S} \notin \sigma(\trg{\pi_T})$. We still need to show that $\trg{t_T} \in \tau(\{ \src{t_S} \})$. 
  By $\src{t_S} \notin \sigma(\trg{\pi_T})$ and the adjunction law deduce
  $\tau(\{ \src{t_S} \}) \not\subseteq \trg{\pi_T}$ that implies $\tau(\{ \src{t_S} \}) = \{ \trg{t_T} \}$.
  \MP{ it'd be nice to have REFS for stuff like the adjunct law, so you can refer to them and the reader can more easily follow the proof.
  You can use command : thmref for that.}
\end{proof}

\subsection{Unsafe Languages and Side-Channels} \label{sec:examples}

In this section we show how the criteria above deal with two
interesting case studies.
In \Cref{ex:undef} we consider an unsafe source
language~\cite{AbateABEFHLPST18}, while in \Cref{ex:length} a timer is
available to measure time spent for computations of target programs.
To model this, in the first case we assume the set of source events
contains a special symbol denoting that the program encountered an
\emph{undefined behavior}~\cite{Leroy09}.
In the second case we equip target traces with a natural number,
intuitively denoting the number of time units, e.g. ms, needed by the
program to produce the trace.
In both cases source and target traces differ, but are related by
some relation $\sim$.
We explain the guarantees an \rtcone compilation chain provides, and
discuss whether and how \tildetaurtp and \tildesigmartp match the same guarantees.


\begin{example}[Relating traces when the source has undefined behaviour] \label{ex:undef}
Let $\Sigma$ be some set of events and an $\eundef$ a symbol not contained by $\Sigma$. Assume observable events in the source
are $\src{\Sigma} = \Sigma \cup \{ \eundef \}$, while in the target $\trg{\Sigma} = \Sigma$.
If a source program encounters an undefined behavior
then the source semantics immediately terminates its computation, so
that source traces can exhibit the symbol $\eundef$ only at the end
of a (finite) trace.
Notice that properties, traces and finite prefixes in the target are all valid properties, traces and 
finite prefixes in the source.

We say that a target trace \emph{refines} a source trace, and write $\src{t_S} \sim \trg{t_T}$
\emph{iff} $\src{t_S} = \trg{t_T} \vee \exists m \leq \trg{t_T}. ~ \src{t_S} = {m} :: \eundef$.
\rtcone{} holds for a compilation chain if, whenever a trace $\trg{t_T}$ is produced in some target context
by some $\cmp{P}$, then there is a source context in which $\src{P}$ can either faithfully 
reproduce $\trg{t_T}$ or encounter an undefined behavior while trying.
\tildesigmartp holds if, in order to claim
$\cmp{P}$ robustly satisfies a target property, it suffices to show
that computations of $\src{P}$ satisfy the same property or stop
prematurely exhibiting an undefined behavior.
%
Finally, \tildetaurtp holds if whenever a source property is robustly satisfied,
a certain \emph{refinement}
of the same property is robustly satisfied in the target.
%
The equivalence of these three criteria can be rigorously checked by simply unfolding our definitions 
of $\tilde{\tau}$ and $\tilde{\sigma}$. Indeed,
\begin{align*}
  \tilde{\tau}(\src{\pi_S}) &=  \myset{\trg{t_T}}{\trg{t_T} \in \src{\pi_S}} \cup \myset{\trg{t_T}}{\exists m \leq \trg{t_T}. ~ m :: \eundef \in \src{\pi_S}} \\ 
  \tilde{\sigma}(\trg{\pi_T}) &= \trg{\pi_T} \cup \myset{ m :: \eundef}{ \exists \trg{t_T} \in \trg{\pi_T}. ~ m \leq \trg{t_T}} 
\end{align*}
Notice that every source trace is related to some target trace and that for an arbitrary $\trg{\pi_T}$, $\tilde{\tau}(\tilde{\sigma} (\trg{\pi_T})) = \trg{\pi_T}$,
so that \tildesigmartp coincides with its weaker variant and
$\tildetaurtp \iff \rtcone \iff \tildesigmartp$. 
\end{example}

\begin{example}[With side-channels target traces are more informative than source ones] \label{ex:length}
In this example source and target share the same set of observable events, but we assume a timer is available in  the target. 
To model this we equip source traces with a natural number that intuitively gives us an estimation of the time needed to 
produce the trace.

In more detail, let $\src{\ii{Trace}_S}$ denote the set of traces in the source, $\trg{\ii{Trace}_T} = \src{\ii{Trace}_S} \times (\mathbb{N} \cup \{ \omega \})$. 
Write  $\trg{W_T \leadsto}_n ~\src{t_S}$ to denote that $\trg{W_T}$ produces $\src{t_S}$ in at most $n$ units of time, e.g. ms.
If $n = \omega$ then we mean that $\trg{W_T}$ produces $\src{t_S}$ in an arbitrary, maybe infinite, amount of time. 
Define the relation $\sim$ as follows: 

\ch{I find the second conjunct below silly.
  And it's still not even explained. How about removing it?
  Or do the results below depend on it? If so it's fishy.}
\begin{center}
  $\src{t_S} \sim \trg{t_T}$ \emph{iff} $\exists n. ~\trg{t_T} = (\src{t_S}, n) \wedge ~\exists \trg{W_T}. ~ \trg{W_T \leadsto}_n ~\src{t_S}$.
\end{center}
Notice that a trace produced by some target program is
related to the source trace that 
simply ``forgets'' the time spent in the computation. 
\rtcone{} requires that, if an attack can be mounted at the target level, the same attack can be
simulated, no matter in how much time, in the source.
\ch{I would start by explaining that $\tilde{\sigma}$ and $\tilde{\tau}$ do not
  form a Galois connection.}%
It is easy to show that for the set of properties $\mathcal{G}$ we get
from \autoref{def:weaker}\ch{Re-Explain the intuition of this set}
$\mathcal{G} \neq 2^{\trg{Trace}}$ and 
that \weaktilde{} is indeed weaker than $\tildesigmartp.$
Intuitively a $\trg{\pi_T}$ may contain spurious pairs $(\src{t_S}, n)$,
where $n$ is not large enough to represent a sufficient amount of execution time, and $\tilde{\sigma}$ removes all of them. 
For instance let $\src{t_S}$ be an infinite trace, $\tilde{\sigma}(\{ (\src{t_S}, 0)\}) = \emptyset$ because
there is no way to produce an infinite trace in $0$ units of time.
\weaktilde{} tells us that $\cmp{P}$ robustly 
satisfies $\{ (\src{t_S}, 0)\}$ if $\src{P}$ robustly satisfies the empty property, which is not possible. 
\end{example}

\ch{So what's the conclusion of this example? If I care about side-channels,
  should I not use any of these definitions? Os should I only give up on
  $\weaktilde$?}
\fi

\twocolumn
\fi 

\ifieee
\bibliographystyle{abbrvnaturl}
\footnotesize
\else 

\ifcamera
\bibliographystyle{ACM-Reference-Format}
\citestyle{acmauthoryear}   
\else
\bibliographystyle{abbrvnaturl}
\fi

\fi 

\bibliography{local,mp,safe}

\end{document}